\documentclass[11pt]{article}
\usepackage[a4paper,left=20mm,top=20mm,right=20mm,bottom=30mm]{geometry}

\usepackage{amsmath,amsopn,amsthm,amssymb}

\usepackage{mathtools}
\usepackage{mathrsfs}  
\usepackage{enumerate}
\usepackage[shortlabels]{enumitem}
\usepackage{rotating}
\usepackage{pifont}

\usepackage{algorithm}
\usepackage{algpseudocode}% http://ctan.org/pkg/algorithmicx
\algtext*{EndWhile}% Remove "end while" text
\algtext*{EndIf}% Remove "end if" text

\usepackage{verbatim}

\usepackage{authblk}
\usepackage[numbers,sort&compress]{natbib}

\newtheorem{theorem}{Theorem}[section]
\newtheorem{lemma}[theorem]{Lemma}
\newtheorem{proposition}[theorem]{Proposition}
\newtheorem{corollary}[theorem]{Corollary}

\newtheorem{fact}[theorem]{Observation}

\newtheorem*{repeatedlemma}{Lemma}

\newenvironment{claim-proof}%
{\begin{description}[leftmargin = 0.2cm, labelsep = 0.2cm]
		\item \emph{Proof of Claim:}}{\hfill$\diamond$\end{description}}

\newtheorem{remark}{Remark}%
\newtheorem{problem}[theorem]{Problem}%[section]

\newtheorem{definition}[theorem]{Definition}%

\makeatletter
\def\moverlay{\mathpalette\mov@rlay}
\def\mov@rlay#1#2{\leavevmode\vtop{%
		\baselineskip\z@skip \lineskiplimit-\maxdimen
		\ialign{\hfil$\m@th#1##$\hfil\cr#2\crcr}}}
\newcommand{\charfusion}[3][\mathord]{
	#1{\ifx#1\mathop\vphantom{#2}\fi
		\mathpalette\mov@rlay{#2\cr#3}
	}
	\ifx#1\mathop\expandafter\displaylimits\fi}
\DeclareRobustCommand\bigop[1]{%
	\mathop{\vphantom{\sum}\mathpalette\bigop@{#1}}\slimits@
}
\newcommand{\bigop@}[2]{%
	\vcenter{%
		\sbox\z@{$#1\sum$}%
		\hbox{\resizebox{\ifx#1\displaystyle.9\fi\dimexpr\ht\z@+\dp\z@}{!}{$\m@th#2$
		}}%
	}%
}
\makeatother
\newcommand{\cupdot}{\charfusion[\mathbin]{\cup}{\cdot}}

\DeclareMathOperator{\parent}{par}
\DeclareMathOperator{\lca}{lca}
\DeclareMathOperator{\LCA}{LCA}
\DeclareMathOperator{\indeg}{indeg}
\DeclareMathOperator{\outdeg}{outdeg}
\DeclareMathOperator{\cl}{cl}
\DeclareMathOperator{\CC}{\mathtt{C}}
\DeclareMathOperator{\Top}{Top}
\DeclareMathOperator{\child}{child}

\newcommand{\IC}{\mathcal{I}(\mathscr{C})}
\newcommand{\Hasse}[1][]{\mathfrak{H}\ifthenelse{\equal{#1}{}}{}{(#1)}}

\DeclareMathOperator{\expand}{\textsc{expd}}
\DeclareMathOperator{\contract}{\textsc{cntr}}
\DeclareMathOperator{\phylo}{\textsc{phylo}}

\usepackage[T3,T1]{fontenc}
\DeclareSymbolFont{tipa}{T3}{cmr}{m}{n}
\DeclareMathAccent{\invbreve}{\mathalpha}{tipa}{16}

\usepackage{xcolor}

\providecommand{\keywords}[1]{\textbf{\textit{Keywords: }} #1}

  \title{Clustering Systems of Phylogenetic Networks}

\author[1,*]{Marc Hellmuth} 
\author[2]{David Schaller}
\author[2-5]{Peter F.\ Stadler}

\affil[1]{Department of Mathematics, Faculty of Science,
	Stockholm University, SE - 106 91 Stockholm,   Sweden \newline 
	\texttt{mhellmuth@mailbox.org}}

\affil[2]{Bioinformatics Group, Department of Computer Science \&
	Interdisciplinary Center for Bioinformatics, Universit{\"a}t Leipzig,
	H{\"a}rtelstra{\ss}e~16--18, D-04107 Leipzig, Germany.}

\affil[3]{Max Planck Institute for Mathematics in the Sciences,
	Inselstra{\ss}e 22, D-04103 Leipzig, Germany}

\affil[4]{Institute for Theoretical Chemistry, University of Vienna,
	W{\"a}hringerstrasse 17, A-1090 Wien, Austria}

\affil[5]{Facultad de Ciencias, Universidad National de Colombia, Sede
	Bogot{\'a}, Colombia}

\affil[6]{Santa Fe Insitute, 1399 Hyde Park Rd., Santa Fe NM 87501,
	USA}

\affil[*]{corresponding author}

\date{\ }

\begin{document}
\sloppy

\maketitle

\abstract{ 
	Rooted acyclic graphs appear naturally when the phylogenetic
	relationship of a set $X$ of taxa involves not only speciations but also
	recombination, horizontal transfer, or hybridization, that cannot be
	captured by trees. A variety of classes of such networks have been
	discussed in the literature, including phylogenetic, level-1, tree-child,
	tree-based, galled tree, regular, or normal networks as models of
	different types of evolutionary processes. Clusters arise in models of
	phylogeny as the sets $\CC(v)$ of descendant taxa of a vertex $v$.  The
	clustering system $\mathscr{C}_N$ comprising the clusters of a network
	$N$ conveys key information on $N$ itself. In the special case of rooted
	phylogenetic trees, $T$ is uniquely determined by its clustering system
	$\mathscr{C}_T$. Although this is no longer true for networks in general,
	it is of interest to relate properties of $N$ and $\mathscr{C}_N$. Here,
	we systematically investigate the relationships of several well-studied
	classes of networks and their clustering systems. The main results are
	correspondences of classes of networks and clustering system of the
	following form: If $N$ is a network of type $\mathbb{X}$, then
	$\mathcal{C}_N$ satisfies $\mathbb{Y}$, and conversely if $\mathscr{C}$
	is a clustering system satisfying $\mathbb{Y}$ then there is network $N$
	of type $\mathbb{X}$ such that $\mathscr{C}\subseteq\mathscr{C}_N$.This,
	in turn, allows us to investigate the mutual dependencies between the
	distinct types of networks in much detail.
}
\smallskip

\noindent 
\keywords{compatibility, level-k, hybrid, evolution, cluster, network,
	phylogenetics, least common ancestor
}

\sloppy

\section{Introduction}

Networks used to model phylogenetic relationships typically are directed
acyclic graphs (DAGs) with a single root, i.e., a unique vertex from which
all other vertices can be reached from. As usual in phylogenetics, the
subset $X$ of vertices without descendants (the leaves of the network)
represents the extant taxa, while the remaining vertices model their
ancestors. In this contribution, we are interested in the relationships
between the structure of networks $N$ with leaf set $X$ and their
associated clustering systems $\mathscr{C}_N$, which contains, for each
vertex $v$ of $N$, the subset $\CC(v)\subseteq X$ of leaves that can be
reached from $v$ \cite{Nakhleh:05,Huson:08}.  In the literature on
phylogenetic networks, the sets $C\in\mathscr{C}_N$ are often called the
``hardwired clusters'' of $N$.

As a special case, there is a well-known 1-to-1 correspondence between
rooted phylogenetic trees $T$ and hierarchies $\mathscr{C}$
\cite{sem-ste-03a}, i.e., clustering systems that do not contain pairs of
overlapping clusters. Therefore, $T$ is uniquely determined by
$\mathscr{C}_T=\mathscr{C}$.  This simple correspondence, however, is no
longer true for (phylogenetic) networks.  Nevertheless, it is not difficult
to find some network $N$ for a given clustering system $\mathscr{C}$ such
that $\mathscr{C}_N= \mathscr{C}$: If suffices to determine the Hasse
diagram $N = \Hasse[\mathscr{C}]$ of the inclusion partial order of the
clustering system $\mathscr{C}$ to obtain such a network.  For a
phylogenetic trees $T$, the Hasse diagram $\Hasse[\mathscr{C}_T]$ and $T$
are isomorphic. For general networks, however, the situation is much more
complicated \cite{Nakhleh:05,Huson:08,Zhang:19}.

A broad array of different types of networks have been studied in the
literature in order to model different modes of non-tree-like evolution
such as horizontal gene transfer, recombination, or hybridization, see
\cite{Kong:22} for a current review.  Naturally, the question arises how
much information about the structure of $N$ is contained in the clustering
system $\mathscr{C}_N$.  We will in particular be concerned with the
following, inter-related questions\smallskip

\begin{enumerate}
\item Which types of networks $N$ satisfy $N\simeq \Hasse[\mathscr{C}_N]$?
\item What are necessary properties of the clustering systems
  $\mathscr{C}_N$ obtained for networks $N$ of a given class?
\item Which types of networks $N$, if any, can be characterized in
  terms of properties of their clustering systems $\mathscr{C}_N$?
\item When is a network $N$ uniquely determined by $\mathscr{C}_N$ or
  by the corresponding multiset of clusters $\mathscr{M}_N$?
\item When is a clustering system $\mathscr{C}$ compatible with a specified
  type of network $N$ in the sense that there is network $N$ of given type
  such that  $\mathscr{C} \subseteq \mathscr{C}_N$? \smallskip
\end{enumerate}

While addressing these questions, we will also consider the implications
between the defining properties of the various network classes. To help the
reader navigate this contribution, we summarize the properties of interest
in Table~\ref{tab:sum-def-N} and point to their formal definitions.
Complementarily, properties of clustering systems are compiled in
Table~\ref{tab:sum-def-C}. Many of the results established here
are summarized in Table~\ref{tab:summary-table}.

It is important to note that the literature on phylogenetic networks does
not always utilize the same nomenclature. In particular, properties such as
\emph{binary}, \emph{separated}, \emph{conventional}, or
\emph{phylogenetic} are -- more often than not -- taken for granted in a
given publication and explicitly or tacitly included in the definition of
``phylogenetic network''.  Here, we start from a very general setting of
rooted DAGs, called ``networks'' throughout. All additional properties are
made explicit throughout. We furthermore strive to prove all statements as
general as possible. The reader will therefore on occasion find results
that are well known in the field, although earlier proofs pertain to a more
restrictive setting.

This paper is organized as follows. In Section \ref{sec:prelim}, we provide
the basic terminology and definitions used throughout this paper. In
Section \ref{sec:NetworksClusteringSystems}, we start with a closer look at
phylogenetic networks (Sect.~\ref{ssec:BasicConcept-N}) and related
concepts which includes graph modifications such as arc-contractions or
expansions (Sect.~\ref{sec:def}) as well as the structural properties of
non-trivial biconnected components (called blocks) in networks
(Sect.\ \ref{sec:blocks}).  We then continue in
Section~\ref{sec:cluster-hasse} to characterize the structure of the Hasse
diagram of clustering systems. In particular, we provide new
characterizations of regular networks, that is, networks that are
isomorphic to the Hasse diagram of some clustering system.

In Section~\ref{sec:semi-reg}, we consider semi-regular networks, i.e.,
networks that do not contain so-called shortcuts and satisfy the
path-cluster-comparability (PCC) property as introduced in
Section~\ref{ssec:PCC}.  (PCC) simply ensures that one of the clusters
$\CC(u)$ and $\CC(v)$ is subsets of the other one whenever the vertices $u$
and $v$ are connected by a directed path in $N$. Although this property
does not seem to have been studied so far, it turns out to play a
fundamental role in the relationships between networks and their clustering
systems.  Regular networks, as it turns out, are precisely the semi-regular
networks that do not contain vertices with outdegree $1$.  In addition, we
show how to obtain regular networks $N'$ from networks $N$ that only
satisfy (PCC) such that $\mathscr{C}_N = \mathscr{C}_{N'}$.  In
Section~\ref{ssec:separated}, we consider separated networks (networks for
which each vertex with indegree greater than $1$ has outdegree $1$) and
cluster networks in the sense of \cite{Huson:08} (whose definition is
somewhat more involved). As we shall see, cluster networks are precisely
the networks that are semi-regular, separated, and phylogenetic.  We then
show in Section \ref{sec:Multi-SemiReg}, that semi-regular networks are
uniquely determined by their multiset of clusters and that, in turn,
cluster networks, as a subclass of regular networks, are uniquely
determined by their clustering systems.
Section~\ref{sec:tree-child-normal-tree-based} makes a short excursion to
so-called tree-child, normal, and tree-based networks and 
their mutual relationships.

In Section~\ref{sec:lcaN-main}, we then have a closer look at the concept
of least common ancestors ($\lca$) in networks. In contrast to rooted
trees, the $\lca$ of a pair of leaves (or more generally, a subset of
leaves) is in general not uniquely defined. We introduce in
Section~\ref{ssec:LCA-N-basic} several classes of networks in which the
$\lca$ is unique for at least certain subsets of leaves.  This leads to the
cluster-lca property (CL) which is satisfied by a network $N$ whenever
$\lca(\CC(v))$ is uniquely determined for all $v\in V(N)$.  We shall see
that every network that satisfies (PCC), and this in particular includes
all normal networks, also satisfy (CL).  In Section~\ref{ssec:lcaN}, we
consider lca-networks, i.e., networks in which $\lca(A)$ is uniquely
determined for all leaf sets $A\subseteq X$. Among other results, we show
that a clustering system $\mathscr{C}$ is closed (under intersection) if
only if it is the clustering systems $\mathscr{C}_N$ of an lca-network.  We
then consider in Section~\ref{ssec:lcaN-strong} the subclass strong
lca-networks, in which, for all $A\subseteq X$, it holds that
$\lca(A)=\lca(\{x,y\})$ for a suitably chosen pair of leaves $x$ and
$y$. These are closely related to weak hierarchies.

A very prominent role in phylogenetics is played by level-1 networks.
Section \ref{sec:LEVEL-1} is devoted to establishing structural results for
level-1 networks and their underlying clustering systems.
After establishing basic results, we derive in Section~\ref{ssec:Level1-L} a
simple condition, called property (L), for clustering systems that is defined
in terms of the intersection of its elements.  As a main result of this
contribution, we obtain in Section~\ref{ssec:level1-char} a simple
characterization of the clustering systems of (phylogenetic, separated)
level-1 network as the ones that are closed and satisfy (L). We then
show in Section \ref{sec:compatibility} that property (L) is sufficient to
ensure that clustering systems are ``compatible'' with a (phylogenetic,
separated) level-1 network.
Moreover, we provide a polynomial-time algorithm to check if
$\mathscr{C}$ is compatible with some level-1 network and, in the
affirmative case, to construct such a network. We finally consider in
Section \ref{sec:specialL1} several subclasses of
level-1 networks as e.g.\ galled trees or binary, conventional, separated
or quasi-binary level-1 networks and characterize
their clustering systems. Finally, we show that quasi-binary level-1
networks are encoded by their multisets of clusters.
In Section \ref{sec:sum}, we provide a summary of the main results (see also
Table \ref{tab:summary-table}).

\section{Preliminaries}
\label{sec:prelim}

The power set of a given set $X$ is denoted by $2^X$. Two sets $A$ and $B$
overlap if $A\cap B\notin \{\emptyset, A,B\}$.

We consider graphs $G=(V,E)$ with vertex set $V(G)\coloneqq V$ and arc set
$E(G)\coloneqq E$. A graph $G$ is \emph{undirected} if $E$ is a subset of
the set of two-element subsets of $V$ and $G$ is \emph{directed} if
$E\subseteq (V\times V)\setminus \{(v,v)\mid v\in V \}$.  Thus, arcs $e\in
E$ in an undirected graph $G$ are of the form $e=\{x,y\}$ and in directed
graphs of the form $e=(x,y)$ with $x,y\in V$ being distinct.  The
\emph{degree} of a vertex $v\in V$ in an undirected or directed graph $G$,
denoted by $\deg_G(v)$, is the number of arcs that are incident with
$v$. If $G$ is directed, we furthermore distinguish the indegree
$\indeg_{G}(v) = \vert \left\{u \mid (u,v)\in E\right\}\vert$ and the
outdegree $\outdeg_{G}(v) = \vert \{u \mid (v,u)\in E\}\vert$.  We write
$H \subseteq G$ if $H$ is a subgraph of $G$ and $G[W]$ for the subgraph in
$G$ that is induced by some subset $W \subseteq V$.  Moreover $G-v$ denotes
the induced subgraph $G[V\setminus \{v\}]$.

A path $P$ in an undirected (resp.\ directed) graph $G$ is a subgraph of
$k\ge 1$ vertices and arcs $\{v_i,v_{i+1}\}\in E$
(resp.\ $(v_i,v_{i+1})\in E$) for all $1\le i\le k-1$.  Moreover,
(directed) paths connecting two vertices $x$ and $y$ are also called
\emph{$xy$-paths}. We will often write \emph{undirected path} for a
subgraph $P$ of a directed graph $G$ that has vertices
$\{v_1,v_2,\dots,v_k\}$, $k\ge 1$, and the \emph{forward arc}
$(v_i,v_{i+1})$ or the corresponding \emph{backward arc} $(v_{i+1},v_i)$
for all $1\le i\le k-1$.  The vertices $v_1$ and $v_k$ in a directed or
undirected path $P$ are the \emph{endpoints} of $P$ and all other vertices
(in $P$) are its \emph{inner vertices}.  A path $P$ with vertices
$\{v_1,v_2,\dots,v_k\}$ in a directed graph G is \emph{induced (in $G$)} if
$(v_i,v_j) \in E(G)$ precisely if $j = i+1$, for all $i\in
\{1,\dots,k-1\}$.

A directed cycle $K$ in a directed graph $G$ is a subgraph with vertices
$\{v_1,v_2,\dots,v_k\}$, $k\ge 2$, and arcs $(v_i,v_{i+1})\in E$ for all
$1\le i\le k-1$ and additionally $(v_k, v_1)\in E$.  In analogy to
undirected paths, an \emph{undirected cycle} $K$ in a directed graph $G$ is
a subgraph with $k\ge 3$ vertices that can be ordered in the form
$\{v_1,v_2,\dots,v_k\}$ such that the forward arc $(v_i,v_{i+1})$ or the
corresponding backward arc $(v_{i+1},v_i)$ for $1\le i\le k-1$, as well as
the forward arc $(v_k,v_1)$ or the backward arc $(v_1,v_k)$ are exactly
the arcs of $K$.

An undirected graph $G=(V,E)$ is \emph{bipartite} if there is a partition
of $V$ into subsets $W$ and $W'$ such that every arc in $G$ connects one
vertex in $W$ to one vertex in $W'$.  If in addition, $x\in W$ and $x'\in
W'$ implies $\{x,x'\}\in E$, then $G$ is \emph{complete bipartite}.

\paragraph{Graph connectivity.}

An undirected graph is \emph{connected} if, for every two vertices $u,v\in
V$, there is a path connecting $u$ and $v$. A directed graph is
\emph{connected} if its underlying undirected graph is connected.  A
connected component of $G$ is a maximal induced subgraph that is connected.
A vertex $v$ is a cut vertex in a graph $G$ if $G[V(G)\setminus\{v\}]$
consists of more connected components than $G$. Similarly, a directed or
undirected arc $(u,v)$ is a cut arc in $G$ if the graph $G'$ with vertex
set $V(G')=V(G)$ and arc set $E(G')=E(G)\setminus\{(u,v)\}$ consists of
more connected components than $G$.

An undirected or directed graph is \emph{biconnected} if it contains no
vertex whose removal disconnects the graph. A \emph{block} of an undirected
or a directed graph is a maximal biconnected subgraph. A block $B$ is
called \emph{non-trivial} if it contains an (underlying undirected)
cycle. Equivalently, a block is non-trivial if it is not a single vertex or
a single arc. An arc that is at the same time a (trivial) block is a cut
arc.  For later reference, we state here the following observations that 
are immediate consequences of the fact that two distinct blocks in a graph 
share at most one vertex \cite[Prop.~4.1.19]{West:01}:
\begin{fact}
  \label{obs:identical-block}
  If two biconnected subgraphs share two vertices, then their union is
  contained in a common block.
\end{fact}
\begin{fact}\label{obs:biConn-arc-disjoint}
  If $B$ and $B'$ are distinct blocks of a directed graph, then $B$ and
  $B'$ are arc-disjoint.
\end{fact}
The latter is justified by the fact that if blocks $B$ and $B'$ share a
common arc, then $B\cup B'$ is biconnected and thus, $B=B'$ since blocks
are always \emph{maximal} biconnected subgraphs.  
We will frequently make use of
\begin{theorem}\textnormal{\cite[Thm.~4.2.4]{West:01}}
  For a graph $G$ with at least three vertices, the following statements are 
  equivalent:
  \begin{enumerate}[noitemsep, nolistsep]
    \item $G$ is biconnected.
    \item For all $x,y\in V(G)$, there are at least two internally 
    vertex-disjoint (undirected) paths connecting $x$ and $y$.
    \item For all $x,y\in V(G)$, there is an (undirected) cycle containing $x$ 
    and $y$.
  \end{enumerate}
\end{theorem}

\begin{corollary}
  \label{cor:block-cycle}
  Any two vertices of a non-trivial block $B$ lie on a common (undirected)
  cycle in $B$.
\end{corollary}
The following well-known result will also be useful throughout:
\begin{proposition}\textnormal{\cite[Prop.~3.1.1]{Diestel:17}}
  \label{prop:H-path}
  Let $H$ be a biconnected subgraph of $G$ and $P$ be a path in $G$ that only
  shares its endpoints with $H$. Then the subgraph obtained by adding $P$ to
  $H$ is again biconnected.
\end{proposition}

\paragraph{Directed acyclic graphs.}

A directed graph $G=(V,E)$ is \emph{acyclic} if it does not contain a
directed cycle. In particular, every undirected cycle in a directed acyclic
graph (DAG) contains at least one forward and one backward arc.  In a DAG
$G$, a vertex $u\in V$ is called an \emph{ancestor} of $v\in V$ and $v$ a
\emph{descendant} of $u$, in symbols $v \preceq_G u$, if there is a
directed path (possibly reduced to a single vertex) in $G$ from $u$ to
$v$. We write $v \prec_G u$ if $v \preceq_N u$ and $u\neq v$. If $u
\preceq_N v$ or $v \preceq_G u$, then $u$ and $v$ are
\emph{$\preceq_G$-comparable} and otherwise,
\emph{$\preceq_G$-incomparable}. Moreover, if $(u,v)\in E$, we say that $u$
is a \emph{parent} of $v$, $u\in\parent_{G}(v)$, and $v$ is a \emph{child}
of $u$, $v\in \child_{G}(u)$.  Following \cite{Huber:19}, we call a vertex
$v$ that is $\preceq_{G}$-minimal in a block $B$ a \emph{terminal
vertex} (of $B$). Note that every terminal vertex $v$ of a non-trivial
block $B$ must always have indegree at least $2$ since, by
  Cor.~\ref{cor:block-cycle}, $v$ lies on some undirected cycle in $B$ and,
  by $\preceq_{G}$-minimality of $v$ in $B$, its two incident vertices on
  this cycle must both be in-neighbors. Below we will consider DAGs
  in which terminal vertices are a type of so-called hybrid vertices.

An arc $(u,w)$ in a DAG $G$ is a \emph{shortcut} if there is a vertex
$v\in\child(u)\setminus\{w\}$ such that $w\prec_G v$ (or, equivalently, if
there is a vertex $v'\in V(G)$ such that $w\prec_G v'\prec_G u$).  A DAG
without shortcuts is \emph{shortcut-free}.
\begin{fact}
  \label{obs:shortcut}
  Let $G$ be a DAG. The following statements are equivalent:
  \begin{enumerate}
    \item $N$ is shortcut-free.
    \item For all $u\in V(N)$, $v,w\in\child_G(u)$ are
      $\preceq_G$-comparable if and only if $v=w$.
    \item For all $u\in V(N)$, $v,w\in\parent_G(u)$ are
      $\preceq_G$-comparable if and only if $v=w$.
  \end{enumerate}
\end{fact}

\section{Networks and Clustering Systems}
\label{sec:NetworksClusteringSystems}

\subsection{Basic Concepts}
\label{ssec:BasicConcept-N}
We define (phylogenetic) networks here
as a slightly more general class of
DAGs than what is customarily considered in most of the literature on the
topic.
\begin{definition}
  \label{def:N}
  A \emph{(rooted) network} is a directed acyclic graph $N=(V,E)$ such that
  \begin{itemize}
  \item[(N1)] There is a unique vertex $\rho_N$, called the \emph{root}
    of $N$, with $\indeg(\rho_N)=0$.
  \end{itemize}
  A network is \emph{phylogenetic} if
  \begin{itemize}
  \item[(N2)] There is no vertex $v\in V$ with
    $\outdeg(v)=1$ and $\indeg(v)\le 1$.
  \end{itemize}
  A vertex with $v\in V$ is a \emph{leaf} if $\outdeg(v)=0$, a
  \emph{hybrid vertex} if $\indeg(v)>1$, and \emph{tree vertex} if
  $\indeg(v)\le 1$. The set of leaves is denoted by $X$.
\end{definition}
We note that a leaf $x\in X$ is always either a hybrid vertex or a tree
vertex. A leaf $x\in X$ is a \emph{strict descendant} of $v\in V$ if every
directed path from the root $\rho_N$ to $x$ contains $v$. In contrast to
the even more general definition \cite[Def.~3]{Huson:11}, we use the term
``phylogenetic'' here to mean that vertices with indegree $1$ and outdegree
$1$ do not appear. Moreover, the root is either the single leaf or has
$\outdeg(\rho_N)\ge 2$. Rooted phylogenetic networks thus generalize rooted
phylogenetic trees. Since the root is an ancestor of all vertices, $N$ is
connected.  We emphasize that all networks considered here are rooted and
thus, we always use the term ``network'' instead of ``rooted network''.

For a vertex $v$ of $N$, the \emph{subnetwork $N(v)$ of $N$ rooted at $v$},
is the network obtained from the subgraph $N[W]$ induced by the vertices in
$W\coloneqq \{w\in V(N)\mid w\preceq_N v\}$ and by suppression of $w$ if it
has indegree $0$ and outdegree $1$ in $N[W]$ or hybrid-vertices of $N$ that
have in- and outdegree $1$ in $N[W]$.

\begin{table}[t]\footnotesize
  \caption{Summary of networks considered in this paper.}
  \label{tab:sum-def-N}
  \setlength{\tabcolsep}{8pt} % Default value: 6pt
  \renewcommand{\arraystretch}{1.5} % Default value: 1
  \begin{tabular}{p{2.cm}|p{7cm}l}
    \multicolumn{2}{l}{\em Network $N$ is/satisfies}  & Ref. \\  \hline
    \emph{tree (level-$0$)} & $N$ does not contain hybrid vertices & \\
    \emph{shortcut-free} & $N$ does not contain shortcuts & \\
    \emph{phylogenetic} & cf.\ Def.~\ref{def:N} (N2) &
    Def.~\ref{def:N} \\
    \emph{separated} & all hybrid vertices of $N$ have outdegree
    $1$.& Def.~\ref{def:sep}\\
    \emph{binary} &   every tree vertex 	 $v$ is either a leaf or
    has $\outdeg(v)=2$, and every hybrid vertex $v$
    satisfies $\indeg(v)=2$ and $\outdeg(v)=1$. &
    Def.~\ref{def:binary} \\
    \emph{level-$k$} & each block $B$ of $N$ contains at
    most $k$ hybrid vertices (distinct from the ``root''
    of $B$) & Def.~\ref{def:level-k} \\
    \emph{regular} & there is a prescribed isomorphism between
    $N$ and the Hasse diagram $\Hasse[\mathscr{C}_N]$ of the clusters in $N$.
    &Def.~\ref{def:regular-N}\\
    \emph{path-cluster-comparability (PCC)} &for all $u,v\in V(N)$,
    $u$ and $v$ are $\preceq_N$-comparable if and only if
    $\CC(u)\subseteq \CC(v)$ or $\CC(v)\subseteq \CC(u)$.
    & Def.~\ref{def:PCC}\\
    \emph{semi-regular} & $N$ is shortcut-free and satisfies
    (PCC).& Def.~\ref{def:semi-regular}\\
    \emph{cluster network}  & $N$ satisfies (PCC), and, three additional
    properties based on the clusters and respective vertices in $N$ &
    Def.~\ref{def:cluster-network}\\
    \emph{tree-child} & for every
    $v\in V^0(N)$, there is a ``tree-child'', i.e., $u\in\child(v)$ with
    $\indeg(u)=1$. & Def.~\ref{def:tree-child}\\
    \emph{normal} & $N$ is tree-child and shortcut-free. &
    Def.~\ref{def:normal}\\
    \emph{tree-based} &
    there is a base tree $T$ of $N$ that can be obtained from $N$ in a
    prescribed manner & Def.~\ref{def:tree-based}\\
    \emph{cluster-lca (CL)} & $\lca(\CC(v))$ is defined for all $v\in
    V(N)$ & Def.~\ref{def:CL}\\
    \emph{lca-network} & $\lca(A)$ is well-defined,
    i.e., if $\vert\LCA(A)\vert=1$ for all non-empty subsets $A\subseteq X$.
    & Def.~\ref{def:lcaN}\\
    \emph{strong lca-network} & $N$ is an lca-network and, for every
    non-empty subset $A\subseteq X$, there are $x,y\in A$ such that
    $\lca(\{x,y\})=\lca(A)$. & Def.~\ref{def:stronglca}\\
    \emph{galled tree} & every non-trivial block in $N$ is an
    (undirected) cycle. & Def.~\ref{def:galledtree}\\
    \emph{conventional} & (i) all leaves have indegree at
    most $1$ and (ii) every hybrid vertex is contained in a unique
    non-trivial block.
    & Def.~\ref{def:conventional}\\
    \emph{quasi-binary} &
    $\indeg_N(w)=2$ and $\outdeg_{N}(w)=1$ for every hybrid vertex $w\in V(N)$
    and,  additionally, $\outdeg_N(\max B) = 2$
    for every non-trivial block $B$ in $N$. 	&Def.~\ref{def:quasi-bin} \\
    \hline %		\multicolumn{3}{l}{}
  \end{tabular}
\end{table}

\begin{table}[t]\footnotesize
  \caption{Properties of clustering systems considered in this paper.}
  \label{tab:sum-def-C}
  \setlength{\tabcolsep}{8pt} % Default value: 6pt
  \renewcommand{\arraystretch}{1.4} % Default value: 1
  \begin{tabular}{p{2.cm}|p{7cm}l}
    \multicolumn{2}{l}{\em Clustering system $\mathscr{C}$ is/satisfies}  &
    Ref. \\  \hline
    \emph{hierarchy} & for all $C,C'\in\mathscr{C}$, it holds
    $C\cap C'\in\{\emptyset,C,C'\}$. & Def.~\ref{def:Csys}~\\
    \emph{closed} & $\mathcal{C}$ is closed under intersection, i.e.,
    $\bigcap_{C\in \mathscr{C}'} C \in \mathscr{C}\cupdot\{\emptyset\}$ holds
    for all $\mathscr{C}'\subseteq \mathscr{C}$. & Def.~\ref{def:closed}  \\
    \emph{pre-binary} & for every pair $x,y\in X$, there is a unique
    inclusion-minimal cluster $C$ such that $\{x,y\}\subseteq C$ &
    Def.~\ref{def:prebinary} \\
    \emph{binary} & pre-binary and, for every
    $C\in\mathscr{C}$, there is a pair of vertices $x,y\in X$ such that $C$
    is the unique inclusion-minimal cluster containing $x$ and $y$. &
    Def.~\ref{def:wH-and-binary} \\
    \emph{weak hierarchy} & for all $C_1,C_2,C_3\in\mathscr{C}$, it holds
    $C_1\cap C_2\cap C_3 \in\{C_1\cap C_2,C_1\cap C_3, C_2\cap
    C_3,\emptyset\}$. & Def.~\ref{def:wH-and-binary} \\
    \emph{(L)} & $C_1\cap C_2=C_1\cap C_3$ for all $C_1,C_2,C_3\in\mathscr{C}$
    where $C_1$ overlaps both $C_2$ and $C_3$. & Def.~\ref{def:L} \\
    \emph{(N3O)} & $\mathscr{C}$ contains no three distinct pairwise
    overlapping clusters. & Def.~\ref{def:N3O} \\
    \emph{paired hierarchy} & every $C\in\mathscr{C}$ overlaps with at most one
    other
    cluster in $\mathscr{C}$. & Def.~\ref{def:pairedH}\\
    \emph{(2-Inc)} & for all
    clusters $C\in\mathscr{C}$, there are at most two inclusion-maximal
    clusters $A,B\in\mathscr{C}$ with $A,B\subsetneq C$ and at most two
    inclusion-minimal clusters $A,B\in\mathscr{C}$ with $C\subsetneq A,B$. &
    Def.~\ref{def:2-Inc} \\
    \hline
%    \multicolumn{3}{l}{}
  \end{tabular}
\end{table}

A network $N$ with leaf set $X$ is often called a \emph{network on $X$}.
Two networks $N_1 =(V_1,E_1)$ and $N_2=(V_2,E_2)$ on $X$ are \emph{graph
isomorphic}, in symbols $N_1\sim N_2$, if there is a \emph{graph
isomorphism}, i.e., a bijection $\varphi\colon V_1\to V_2$ such that
$(u,v)\in E_1$ if and only if $(\varphi(u),\varphi(v))\in E_2$ for all
$u,v\in V_1$. Moreover, if additionally $N_1$ and $N_2$ are networks on the
same leaf set $X$, we say that $N_1$ and $N_2$ are \emph{isomorphic} in
symbols $N_1\simeq N_2$ if $N_1\sim N_2$ (via the graph isomorphism
$\varphi$) and $\varphi(x)=x$ for all $x\in X$.  We say that a network $N$
on $X$ is \emph{unique} w.r.t.\ some property (or some set of properties),
if $N\simeq N'$ for every network $N'$ that also satisfies the desired
property (or properties).

Many studies into phylogenetic networks require that reticulation events
and speciation events are separated, i.e., $\outdeg(v)=1$ for all hybrid
vertices.
\begin{definition}\label{def:sep}
  A network $N$ is \emph{separated} if all hybrid vertices have outdegree
  $1$.
\end{definition}
In particular, all leaves have indegree $1$ in a separated network (or
indegree zero if the network consists of a single vertex).

The properties \emph{phylogenetic} and \emph{separated} are part of
the definition of networks in many publications in the field, see
e.g.\ \cite{Jetten:18,Pons:19, Zhang:19}.  However, we opted for the more
general definition of networks for several reasons. On the one hand, we aim
to explore which restrictions are actually needed to establish the
relationship of different properties or classes of networks. On the other
hand, separated networks do not include regular networks \cite{Baroni:05},
which are, as we shall see, a class of networks that is intimately linked
with clustering systems.

An even more restrictive class of networks that is often considered are
binary networks \cite{Gambette:12,Bordewich:2016if, Kong:22}:
\begin{definition}\label{def:binary}
  A network is \emph{binary} if every tree vertex $v$ is either a leaf or
  has $\outdeg(v)=2$, and every hybrid vertex $v$ satisfies $\indeg(v)=2$
  and $\outdeg(v)=1$.
\end{definition}
By construction, binary networks are always phylogenetic and separated.

Throughout this paper, several other properties and distinct classes of
networks are considered.  For convenience, all these types are listed in
Table~\ref{tab:sum-def-N}.  More formal definitions or more precise
explanations are given in the remainder of the paper.  A further essential
ingredient to our paper are clusters and clustering systems as defined
next.

\begin{definition}
  \label{def:cluster}
  Let $N$ be a network with vertex set $V$, leaf set $X$ and partial order
  $\preceq_N$. Then, for each $v\in V$, the associated \emph{cluster} is
  $\CC(v)\coloneqq \CC_N(v)\coloneqq \{x\in X\mid x\preceq_N v\}$.
  Furthermore, we write $\mathscr{C}\coloneqq\mathscr{C}_N\coloneqq\{
  \CC(v)\mid v\in V\}$.
\end{definition}
Note, that $\CC(v)=\CC(w)$ may be possible for distinct $v,w\in V$.  However,
$\mathscr{C}$ is considered as a \emph{set} and thus, each cluster appears
only once in $\mathscr{C}_N$.  The clusters in $\mathscr{C}_N$ are usually
called the \emph{hardwired clusters} of $N$, see e.g.\ \cite{Huson:11}.

\begin{definition} \cite{Barthelemy:08,sem-ste-03a}
  A \emph{clustering system} on $X$ is a set $\mathscr{C}\subseteq 2^X$
  such that (i) $\emptyset\notin\mathscr{C}$, (ii) $X\in\mathscr{C}$, and
  (iii) $\{x\}\in\mathscr{C}$ for all $x\in X$.  A clustering system is a
  \emph{hierarchy} if it does not contain pairwise overlapping sets.
  \label{def:Csys}
\end{definition}

\begin{figure}[t]
  \begin{center}
    \includegraphics[width=0.65\textwidth]{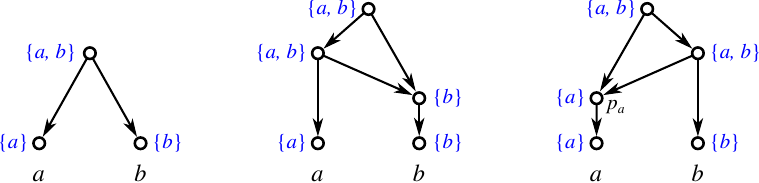}
  \end{center}
  \caption{Three non-isomorphic (binary) level-1 networks with the
    same clustering system $\mathscr{C} =
    \{\{a\},\{b\},\{a,b\}\}$. While they are indistinguishable in
    terms of their clustering systems, they are encoded by their multiset
    of clusters, see Thm.~\ref{thm:level-1-alt}, i.e., they are
    distinguished by the multiplicities of the clusters $\{a\}$, $\{b\}$,
    and $\{a,b\}$.}
  \label{fig:non-isomorphic}
\end{figure}

We will mainly focus on clustering systems $\mathscr{C}_N$ of networks $N$
(cf.\ Lemma~\ref{lem:C-N}). As shown in Fig.~\ref{fig:non-isomorphic}, the
information conveyed by $\mathscr{C}_N$ is often insufficient to determine
$N$, i.e., there are non-isomorphic networks $N$ and $N'$ for which
$\mathscr{C}_N = \mathscr{C}_{N'}$.  A natural generalization is to
consider the \emph{multiset of clusters} $\mathscr{M}_N$, in which each
cluster $C\in \mathscr{C}_N$ appears once for every vertex $v\in V(N)$ with
$\CC(v)=C$. We say that $\mathscr{M}_N$ \emph{encodes} $N$
\textbf\emph{within a given class $\mathbb{P}$ of networks} if
$N'\in\mathbb{P}$ and $\mathscr{M}_{N'}=\mathscr{M}_N$ implies $N'\simeq
N$.

As for networks, we will also consider plenty of different types of clustering
systems equipped with certain properties and, for convenience, list them in
Table~\ref{tab:sum-def-C}.

\subsection{Arc-Expansion and Arc-Contraction}
\label{sec:def}

As mentioned above, often only separated networks are considered,
stipulating that (1) leaves, i.e., vertices $v$ with $\outdeg(v)=0$ have
$\indeg(v)=1$; (2) hybrid vertices $v$ have $\indeg(v)\ge 2$ and
$\outdeg(v)=1$.  Such networks are obtained from the ones in
Def.~\ref{def:N} by means of a simple refinement operation that replaces
every ``offending'' vertex by a pair of vertices connected by single
arc. More precisely, we define the following operation on a network $N$,
which is also part of \cite[Alg.~6.4.2]{Huson:11}:\medskip
\begin{description}
\item[$\expand(v)$] Create a new vertex $v'$, replace arcs $(u,v)$ by
  $(u,v')$ for all $u\in\parent_N (v)$, and add the arc $(v',v)$.
\end{description}\medskip
It will also useful to consider the reversed operation for arcs $(v',v)$
that are not shortcuts:\medskip
\begin{description}
\item[$\contract(v',v)$] Replace arcs $(u,v')$ by $(u,v)$ for all
  $u\in\parent_N (v')\setminus \parent_N (v)$; replace arcs $(v',w)$ by
  $(v,w)$ for all $w\in\child_N (v')\setminus \child_N (v)$; and finally
  delete $(v',v)$ and $v'$.
\end{description}\medskip
The notation $\contract$ is chosen in compliance with the literature where
arc contraction is a commonly used operation. Fur our purpose, however, it
will be useful to have this more formal definition in order to precisely
keep track of the vertex sets upon execution of multiple operations.
Observe that, e.g.\ since $(u,v'),(u,v)\in E$ is possible, applying first
$\contract(v',v)$ and then $\expand(v)$ does not necessarily yield a
network that is isomorphic to the original network.  Furthermore, we remark
that the condition that $(v',v)$ is not a shortcut cannot be dropped since
otherwise directed cycles are introduced
(cf.\ Fig.~\ref{fig:cntr-issues}(A)).

\begin{figure}[t]
  \begin{center}
    \includegraphics[width=0.85\textwidth]{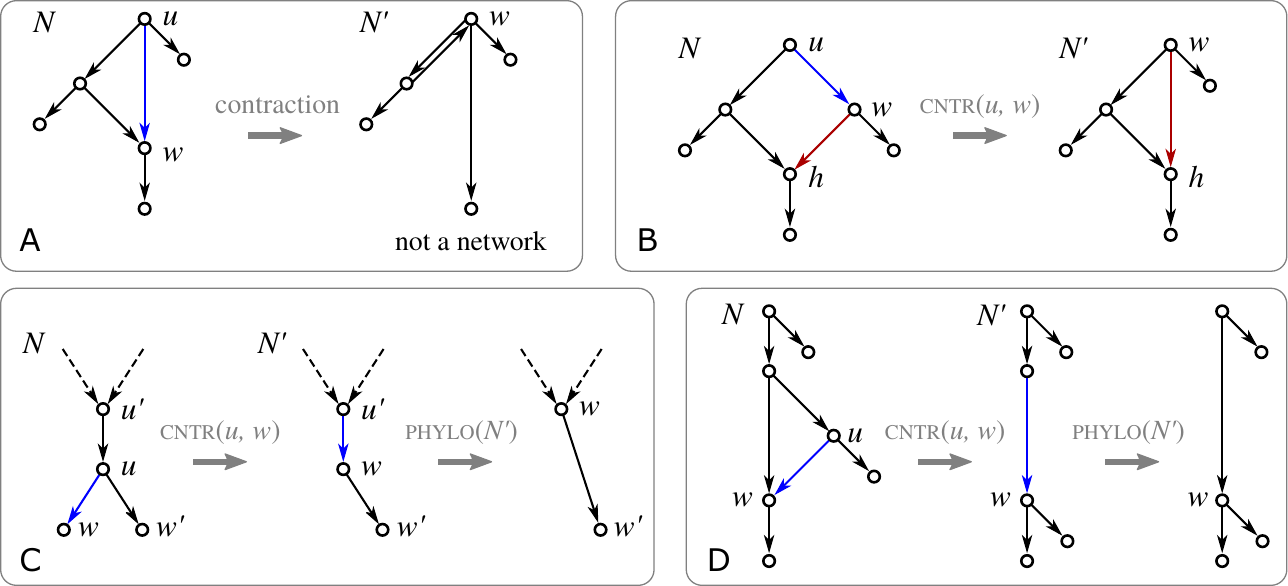}
  \end{center}
  \caption{ Complications arising in the contraction of arcs. The arcs to
    be contracted are highlighted in blue.  (A) Contraction of a shortcut
    $(u,w)$ introduced directed cycles.  (B) Application of
    $\contract(u,w)$ to a shortcut-free network $N$ can result in a network
    $N'$ that contains a shortcut (r.h.s., shortcut indicated by the red
    arc). (C),(D) Contraction of an arc $(u,w)$ in a phylogenetic network
    $N$ can yield a network $N'$ that is no longer
    phylogenetic. Application of $\phylo(N')$ can resolve this issue.}
  \label{fig:cntr-issues}
\end{figure}

We are now in the position to define least-resolved networks:
\begin{definition}
  \label{def:least-resolved}
  A network $N$ is \emph{least-resolved} (w.r.t.\ its clustering system
  $\mathscr{C}\coloneqq\mathscr{C}_N$) if there is no network $N'$ with
  $\mathscr{C}_{N}=\mathscr{C}_{N'}$ that can be obtained from $N$ by a
  non-empty series of shortcut removal and application of $\contract(v',v)$
  for some arc $(v',v)$ that is not a shortcut.
\end{definition}

In many applications, phylogenetic networks are considered.  However,
$\contract(w',w)$ applied on a phylogenetic network may result in a
non-phylogenetic network. By way of example, see
Fig.~\ref{fig:cntr-issues}(C), if $u$ is a tree vertex with parent $u'$ and
two children $w$ and $w'$ which are leaves, then $\contract(u,w)$ will
``locally'' result in a path with arcs $(u',w)$ and $(w,w')$, i.e.,
$\indeg(w) = \outdeg(w)=1$.  Similarly, $\contract(u,w)$ in a block that
contains a shortcut can result in a network $N'$ that is not phylogenetic,
see Fig.~\ref{fig:cntr-issues}(D).  To circumvent this issue, we must
``suppress'' $w$ to obtain a phylogenetic network.  To this end, we define
the following operation to make a network $N$ phylogenetic:\medskip
\begin{description}
  \item[$\phylo(N)$]
  Repeatedly apply $\contract(u,w)$ for an arc $(u,w)$ such that
  $\outdeg(u)=1$ and $\indeg(u)\le 1$ until no such operation is possible.
\end{description}\medskip
\noindent
Now contraction of an arcs $(v',v)$ that is not a shortcut and
``suppression'' of superfluous vertices can be combined in:\medskip
\begin{description}
  \item[$\contract^{\star}(v',v)$]
  Apply $\contract(v',v)$ to obtain $N'$ and then $\phylo(N')$.
\end{description}

\begin{definition}
  Let $N$ and $N'$ be networks such that $V(N')\subseteq V(N)$.  Then, $N$
  and $N'$ are \emph{$(N,N',W)$-ancestor-preserving} if for all
  $v,v'\in W\subseteq V(N')$, it holds that $v\preceq_{N} v'$ if and only
  if $v\preceq_{N'} v'$. If $N$ and $N'$ are
  $(N,N',V(N'))$-ancestor-preserving, then we say for simplicity that $N$
  and $N'$ are $(N,N')$-ancestor-preserving.
\end{definition}

\begin{lemma}
  \label{lem:shortcut-deletion}
  Let $N$ be a network on $X$ and $(u,w)\in E(N)$ be a shortcut.  Then
  removal of $(u,w)$ in $N$ results in a network $N'$ with leaf set $X$ and
  $V(N)=V(N')$. In particular,  $\CC_N(v)=\CC_{N'}(v)$ for all
  $v\in V(N')=V(N)$ and thus, $\mathscr{C}_N=\mathscr{C}_{N'}$.
  Moreover, $N$ and $N'$ are $(N,N')$-ancestor-preserving.
\end{lemma}
\begin{proof}
  Let $N$ be a network on $X$ and $(u,w)$ be a shortcut in $N$.
  Since $(u,w)$ is a shortcut, there is a $w'\in\child(u)\setminus\{w\}$
  such that $w\prec_N w'$.  Hence, there is a $w'w$-path $P$ in $N$. Since
  $N$ is acyclic and $w'\prec_N u$, $u$ is not a vertex in $P$ since
  otherwise $u\preceq_{N} w'$.  Therefore, $w$ has indegree larger than $1$
  in $N$.  In particular, there is a $uw$-path $P'$ in $N'$ formed by the
  arc $(u,w')$ and $w'w$-path $P$.  Since removal of $(u,w)$ only
  decreases the indegree of $w$ and $\indeg_N(w)\ge 2$, $\rho_N=\rho_{N'}$
  is still the only vertex with indegree $0$ in $N'$.  Moreover, removal of
  arcs clearly preserves acyclicity, and thus $N'$ is a rooted network.

  Now let $v,v'\in V(N)=V(N')$. If $v \not\preceq_{N} v'$, then there is no
  $v'v$-path in $N$.  Clearly, removal of arcs changes nothing about this
  and thus $v \not\preceq_{N'} v'$.  Suppose now that $v \preceq_{N} v'$
  and thus let $P_{v'v}$ be a $v'v$-path in $N$.  If $P_{v'v}$ does not
  contain the arc $(u,w)$, then $P_{v'v}$ is still a $v'v$-path in $N'$.
  Otherwise, the path obtained from $P_{v'v}$ by replacing $(u,w)$ by the
  $uw$-path $P'$ is a $v'v$-path in $N'$.  Hence, $v \preceq_{N} v'$ holds
  in both cases.  In summary, we have $v\preceq_{N} v'$ if and only if
  $v\preceq_{N'} v'$.

  By the latter arguments, $N'$ is a network with leaf set $X$
  and we have $x\preceq_{N} v$ if and only if $x\preceq_{N'} v$ for all
  $x\in X$ and all $v\in V(N)=V(N')$. Therefore, $x\in \CC_N(v)$ if and only
  if $x\in \CC_{N'}(v)$ for all $x\in X$ and all $v\in V(N)=V(N')$, and thus,
  $\CC_N(v)=\CC_{N'}(v)$. Together with $V(N)=V(N')$, this implies
  $\mathscr{C}_N=\mathscr{C}_{N'}$.
\end{proof}
Note that deletion of a shortcut from a phylogenetic network does not
necessarily result in a phylogenetic network.

\begin{lemma}
  \label{lem:siblings}
  If a network $N$ is shortcut-free and has no vertex of outdegree $1$,
  then for every vertex $w\in V(N)\setminus\{\rho_N\}$, there is a vertex
  $v\in \child_N(\parent_N(w))$ such that $v$ and $w$ are
  $\preceq_N$-incomparable.  In this case, $N$ is phylogenetic.
\end{lemma}
\begin{proof}
  Since $N$ is shortcut-free, siblings $v',v''\in\child_N(u)$, $v'\ne v''$
  are $\preceq_N$-incomparable. Thus there is $v\in\child_N(\parent_N(w))$
  that is $\preceq_N$-incomparable with $w$ if and only if
  $\parent_N(w)\ne\emptyset$ and $\outdeg(\parent_N(w))>1$. Both conditions are
  satisfied by assumption.
\end{proof}

\begin{lemma}
  \label{lem:contraction}
  Let $N$ be a network on $X$ and $(u,w) \in E(N)$ be an arc that is not a
  shortcut.  Then, $\contract(u,w)$ applied on $N$ results in a network
  $N'$ with leaf set $X$ or $X\setminus\{w\}$ and
  $V(N')=V(N)\setminus\{u\}$.  Moreover, for all $v,v'\in V(N')$,
  \begin{enumerate}
  \item $v\preceq_N v'$ implies $v \preceq_{N'} v'$, and
  \item $v\preceq_{N'} v'$ implies
    (i) $v \preceq_{N} v'$ or
    (ii) $w\preceq_N v'$ and $v\preceq_N w'$ for some
    $w'\in\child_N(u)\setminus \{w\}$ that is $\preceq_{N}$-incomparable
    with $w$.
  \end{enumerate}
  In particular, $v\prec_{N'} v'$ always implies $v\prec_{N} v'$ or $v$ and
  $v'$ are $\preceq_{N}$-incomparable.
\end{lemma}
\begin{proof}
  The proof is rather lengthy and technical and is, therefore, placed to
  Section~\ref{sec:appx-blocks} in the Appendix.
\end{proof}

\begin{lemma}
  \label{lem:outdeg-1-contraction}
  Let $N$ be a network on $X$ and $(u,w)\in E(N)$ such that
  $\outdeg_{N}(u)=1$.  Then $\contract(u,w)$ results in a network $N'$ with
  leaf set $X$ and $V(N')=V(N)\setminus\{u\}$ that is
  $(N,N')$-ancestor-preserving.
  Moreover, $\CC_N(v)=\CC_{N'}(v)$ for all $v\in V(N')=V(N)\setminus\{u\}$ and,
  in particular, $\mathscr{C}_N=\mathscr{C}_{N'}$.
\end{lemma}
\begin{proof}
  Let $N$ be a network on $X$ and $(u,w)\in E(N)$ such that $\outdeg_{N}(u)=1$.
  Since an arc $(u,w)\in E(N)$ with $\outdeg_{N}(u)=1$ cannot be a
  shortcut and satisfies $\child_N(u)\setminus\{w\}=\emptyset$, and thus
  condition~(ii) in Lemma~\ref{lem:contraction} cannot occur, $N'$ is a
  network with leaf set $X$ or $X\setminus\{w\}$ and
  $(N,N')$-ancestor-preserving.  Moreover, since $w$ is the only
  out-neighbor of $u$, we do not add any out-neighbors for $w$. Hence, $N'$
  has leaf set $X$.

  By the latter argument, $N'$ is a network with leaf
  set $X$ and we have $x\preceq_{N} v$ if and only if $x\preceq_{N'} v$ for
  all $x\in X$ and all $v\in V(N')=V(N)\setminus\{u\}$. Therefore,
  $x\in \CC_N(v)$ if and only if $x\in \CC_{N'}(v)$ holds for all $x\in X$
  and all $v\in V(N')=V(N)\setminus \{u\}$.  Hence, we have
  $\CC_N(v)=\CC_{N'}(v)$ for all $v\in V(N')=V(N)\setminus\{u\}$.
  Moreover, since $w$ is the unique out-neighbor of $u$,
  one can easily verify that $\CC_N(u)=\CC_N(w)$ (cf.\
  Obs.~\ref{obs:outdeg-1-cluster} for further arguments) and thus
  $\CC_N(u)=\CC_{N'}(w)\in \mathscr{C}_{N'}$.  Taken together, we obtain
  $\mathscr{C}_N=\mathscr{C}_{N'}$.
\end{proof}

As an immediate  consequence of Lemma~\ref{lem:shortcut-deletion} 
and~\ref{lem:outdeg-1-contraction}, we obtain
\begin{corollary}\label{cor:lrN-sf-out1}
  Every least-resolved network $N$ is shortcut-free and does not
  contain vertices $v$ with $\outdeg_N(v)=1$.
\end{corollary}

\begin{lemma}
  \label{lem:expand}
  Let $N$ be a network and $N'$ be obtained from $N$ by applying
  $\expand(w)$ for some $w\in V(N)$. Then $N'$ is a network such that $N$
  and $N'$ are $(N',N)$-ancestor-preserving.
  Moreover, $\CC_N(v)=\CC_{N'}(v)$ for all $v\in V(N)\subseteq
  V(N')$ and, in particular, it holds $\mathscr{C}_N=\mathscr{C}_{N'}$.
  Moreover, if $N$ is phylogenetic, then $N'$ is phylogenetic if and only if 
  $w$ is a hybrid vertex and $\outdeg_{N}(w)\ne 1$.
\end{lemma}
\begin{proof}
  Let $N$ be a network on $X$.  We show first that $N'$ is a network. By
  construction, $w$ is the only vertex in $N$ whose in-neighborhood changes
  and it has the new vertex $u$ as its unique in-neighbor in $N'$. If $w\ne
  \rho_N$, then $w$ has at least one in-neighbor in $N$, which becomes an
  in-neighbor of $u$. Hence, $\rho_N$ is still the only vertex with
  indegree $0$ in $N'$. If $w=\rho_{N}$, then it has no in-neighbors in $N$
  and thus $u$ has no in-neighbors in $N'$. Together with the fact that $w$
  no longer has indegree $0$, $u$ is the only vertex with indegree $0$ in
  $N'$ in this case.  Now assume that $N'$ contains a directed cycle $K$
  comprising the vertices $v_1, v_2, \dots, v_k$, $k\ge 2$, in this order,
  i.e., $(v_i,v_{i+1})$, $1\leq i\leq k-1$ and $(v_k,v_1)$ are arcs in
  $N'$. If all arcs in $K$ are in $N$, then $K$ is a directed cycle in $N$;
  a contradiction.  If $K$ contains an arc that is not in $N$, then $K$
  must contain the new vertex $u$ since all new arcs are incident with
  $u$. Suppose w.l.o.g.\ that $u=v_1$. Since $u$ has a unique out-neighbor
  $w$ and exactly the vertices in $\parent_N(w)$ as in-neighbors, we must
  have $v_2=w$ and $v_k\in\parent_N(w)$, respectively. In particular, this
  implies $v_2\ne v_k$ and $(v_k,v_2)\in E(N)$.  Since $u$ appears in $K$
  at most once, $(v_k, v_1)$ and $(v_1, v_2)$ are the only arcs in $K$ that
  are incident with $u$, and thus all other arcs of $K$ are also arcs in
  $N$. In particular, there is a $v_2 v_k$-path in $N$.  Together with the
  fact that $(v_k,v_2)\in E(N)$, this implies that $N$ contains a directed
  cycle; a contradiction.  Therefore, $N'$ must be acyclic. Since moreover
  $N'$ has a unique root, it is a network.

  The operation $\expand(w)$ on a network $N$ creates a network $N'$ with
  an additional vertex $u$ such that $w$ is the unique out-neighbor of $u$
  and $\parent_{N'}(u) = \parent_N (w)$.  Therefore, $N$ is recovered from
  $N'$ by applying $\contract(u,w)$.  This observation together with
  Lemma.~\ref{lem:outdeg-1-contraction} implies that $N$ and $N'$ are
  $(N',N)$-ancestor-preserving.

  Suppose now that $N$ is phylogenetic.  Assume first that $w$ is a hybrid
  vertex and $\outdeg_{N}(w)\ne 1$. Then, by construction, the newly
  created vertex $w$ satisfies $\indeg_{N'}(w')=\indeg_{N}(w)\ge 2$ and
  moreover, we have $\outdeg_{N'}(w)=\outdeg_{N}(w)\ne 2$. The only other
  vertices whose neighborhoods are affected are the vertices
  $u\in\parent_N(w)$. More precisely, their outneighbor $w$ is replaced by
  an outneighbor $w'$ and thus $\indeg_{N'}(u)=\indeg_{N}(u)$ for any
  $u\in\parent_N(w)$. Together with the fact that $N$ is phylogenetic, the
  latter arguments imply that there is no vertex $v\in V(N')$ with
  $\outdeg_{N'}(v)=1$ and $\indeg_{N'}(v)\le 1$.  Hence, $N'$ is
  phylogenetic.  Now assume that $w$ is a not hybrid vertex or
  $\outdeg_{N}(w)= 1$.  If $w$ is not a hybrid vertex, then
  $\indeg_{N'}(w')=\indeg_{N}(w)\le 1$.  Moreover, $\outdeg_{N'}(w')=1$
  holds by construction, and thus $N'$ is not phylogenetic. If
  $\outdeg_{N}(w)= 1$, then $\outdeg_{N'}(w)= 1$ since the outneighborhood
  of $w$ does not change. In addition, $w'$ is the unique in-neighbor of
  $w$ in $N'$ by construction. Hence, $N'$ is not phylogenetic.  In
  summary, it holds that $N'$ is phylogenetic if and only if $w$ is a
  hybrid vertex and $\outdeg_{N}(w)\ne 1$.

  By the latter arguments, $N'$ is a network with leaf set $X$.  The newly 
  created vertex $w'$ has a unique child $w$.
  The statement ``$\CC_N(v)=\CC_{N'}(v)$ for all $v\in V(N)\subseteq
  V(N')$ and, in particular, $\mathscr{C}_N=\mathscr{C}_{N'}$''
  therefore follows immediately from 
  Lemma~\ref{lem:outdeg-1-contraction}
  and the fact that $N$ is recovered from $N'$ by applying $\contract(w',w)$.
\end{proof}

The following result shows that the expansion operation does not introduce
  shortcuts and is an immediate consequence of Lemma~\ref{lem:expand-shortcuts}
  in Section~\ref{sec:appx-blocks}.
\begin{corollary}
  \label{cor:expand-shortcuts}
  Let $N$ be a network and $N'$ be the network obtained from $N$ by applying
  $\expand(w)$ for some $w\in V(N)$. Then $N$ is shortcut-free if and only
  if $N'$ is shortcut-free.
\end{corollary}
We remark that an analogue of Cor.~\ref{cor:expand-shortcuts} does not hold
for the contraction operation $\contract(u,w)$.
Fig.~\ref{fig:cntr-issues}(B) shows an example where contraction
introduces a shortcut.

\subsection{Blocks}
\label{sec:blocks}

The blocks of $N$ will play a key role in the following. We first establish
several technical results that will allow us efficiently reason about the
block structure of a network.

\begin{lemma}
  \label{lem:upper-path}
  Let $N$ be a network and $u,v\in V(N)$ be two $\preceq_N$-incomparable
  vertices.  Then $u$ and $v$ are connected by an undirected path $P$ that
  contains at least 3 vertices and of which all inner vertices $w$ satisfy
  $u\prec_N w$ or $v\prec_N w$.  In addition, we have $w\not\preceq_N u$
  and $w\not\preceq_N v$ for every such inner vertex $w$.
\end{lemma}
\begin{proof}
  There are directed paths $P_u$ and $P_v$ from $\rho_N$ to both $u$ and
  $v$, respectively.  Let $w^*$ be the $\preceq_N$-minimal vertex of $P_u$
  that is also a vertex of $P_v$, which exists since at least $\rho_N$ is
  contained in both paths.  It must hold that $w^*\notin \{u,v\}$ since
  otherwise $u$ and $v$ would be $\preceq_N$-comparable.  In particular,
  $u \prec_N w^*$ and $v \prec_N w^*$.  Let $P'_u$ and $P'_v$ be the
  subpaths of $P_u$ and $P_v$ from $w^*$ to $u$ and $v$, respectively. By
  construction, $P'_u$ and $P'_v$ only have vertex $w^*$ in common, which
  moreover is an outer vertex of both paths.  Now consider the path $P$
  that is the union of the underlying undirected version of $P'_u$ and
  $P'_v$.  By construction, $P$ contains at least the three vertices $u$,
  $v$, and $w^*$ and all of its inner vertices $w$ satisfy $u\prec_N w$ or
  $v\prec_N w$.  Assume, for contradiction, that $w\preceq_N u$ for some of
  these inner vertices. Since $u\prec_N w$ is not possible, we must have
  $v\prec_N w$.  But then $v\prec_N w$ and $w\preceq_N u$ imply that $v$
  and $u$ are $\preceq_N$-comparable; a contradiction.  Hence,
  $w\not\preceq_N u$ must hold. One shows analogously that
  $w\not\preceq_N v$.
\end{proof}
Paths of the form described in Lemma \ref{lem:upper-path} connecting two
leaves $u$ and $v$ are called ``up-down-paths'' in \cite{Bordewich:16}.

\begin{lemma}
  \label{lem:block-prec-sandwich}
  Let $B$ be a block in a network $N$ and $u,v\in V(B)$ such that
  $v\preceq_{N} u$. Then every $uv$-path in $N$ is completely contained in
  $B$.
\end{lemma}
\begin{proof}
  Let $P$ be a $uv$-path in $N$, which exists since $v\preceq_{N} u$. The
  statement holds trivially if $B$ is an isolated vertex, $v=u$, or $B$ is
  the arc $(u,v)$.  Thus suppose $B$ is a non-trivial block. Suppose, for
  contradiction, there is a vertex $w\in V(P)\setminus \{V(B)\}$.  Let
  $w_a$ and $w_d$ be the $\preceq_{N}$-minimal ancestor and the
  $\preceq_{N}$-maximal descendant, resp., of $w$ in $P$ (both of which
  exist since $u,v\in V(P)$).  Consider the subpath $P'$ of $P$ from $w_a$
  to $w_d$.  By Prop.~\ref{prop:H-path}, the subgraph of $N$ obtained by
  adding $P'$ to $B$ is again biconnected. Together with $w\in
  V(P')\setminus \{V(B)\}$, this contradict that $B$ is a block. Hence such
  a vertex cannot exist.  Therefore and since blocks are always induced
  subgraphs, the statement follows.
\end{proof}

\begin{lemma}
  \label{lem:max-B-unique}
  Every block $B$ in a network $N$ has a unique $\preceq_N$-maximal vertex
  $\max B$.  In particular, for every $v\in V(B)$, there is a directed path
  from $\max B$ to $v$ and every such path is completely contained in $B$.
\end{lemma}
\begin{proof}
  The statement is trivial for a block that consists only of a single
  vertex or arc. Otherwise, suppose there are two distinct
  $\preceq_N$-maximal vertices $v_1$ and $v_2$ in $B$.  By assumption,
  $v_1$ and $v_2$ must be $\preceq_N$-incomparable.  By
  Lemma~\ref{lem:upper-path}, $v_1$ and $v_2$ are connected by an
  undirected path $P$ that contains at least 3 vertices and of which all
  inner vertices $w$ satisfy $u\prec_N w$ or $v\prec_N w$. By
  $\preceq_N$-maximality of $v_1$ and $v_2$, none of these inner vertices
  can be contained in $B$. By Prop.~\ref{prop:H-path}, adding $P$ to $B$
  preserves biconnectivity, and thus $B$ is not a maximal biconnected
  subgraph; a contradiction.  In particular, for every $v\in V(B)$, we have
  $v\preceq_{N}\max B$, i.e., there is a path from $\max B$ to $v$ and by
  Lemma~\ref{lem:block-prec-sandwich}, each every such path is completely
  contained in $B$.
\end{proof}

\begin{corollary}
  \label{cor:maxB-outdegree}
  If $B$ is a non-trivial block in network $N$, then $\max B$ has at least
  two out-neighbors in $B$.
\end{corollary}
\begin{proof}
  Since $B$ is non-trivial, $\max B$ lies on an undirected cycle in $B$ and
  thus is incident with two distinct vertices in $B$. By
  $\preceq_{N}$-maximality of $\max B$ in $B$, these must be out-neighbors of
  $\max B$.
\end{proof}

\begin{lemma}
  \label{lem:block-identity}
  Let $N$ be a network and suppose that $v\in V(N)$ is contained in the
  blocks $B$ and $B'$ of $N$.  If $v\notin \{\max B, \max B'\}$, then
  $B=B'$.
\end{lemma}
\begin{proof}
  Assume that vertex $v$ is contained in the blocks $B$ and $B'$ of $N$ but
  $v\notin \{\max B, \max B'\}$.  By Lemma~\ref{lem:max-B-unique}, there
  exists a directed path $P$ in $B$ from $\max B$ to $v$. Similarly, there
  is a directed path $P'$ in $B'$ from $\max B'$ to $v$.  Since
  $v\notin \{\max B, \max B'\}$, both $P$ and $P'$ contain at least one
  arc.

  Assume first that $P$ and $P'$ share an arc $e$ and thus, that $B$ and
  $B'$ share the arc $e$. In this case, contraposition of Obs.\
  \ref{obs:biConn-arc-disjoint} implies that $B=B'$.  Hence, in the
  following we assume that $P$ and $P'$ are arc-disjoint.

  Consider first the case $\max B' \preceq_{N} \max B$.  Let $u$ be the
  unique $\preceq_N$-minimal vertex in $P$ such that
  $\max B'\preceq_{N} u$. Together with $v\prec_N \max B'$, this implies
  that $u\ne v$.  Let $P_{u,v}$ be the subpath of $P$ from $u$ to $v$ and
  note that $P_{u,v}$ contains at least one arc.  Since
  $\max B'\preceq_{N} u$, we can find a directed path $P_{u, \max B'}$
  (possible only containing a single vertex $u=\max B'$) from $u$ to
  $\max B'$.  The paths $P_{u,v}$ and $P_{u, \max B'}$ only have vertex $u$
  in common since $u$ is the unique $\preceq_N$-minimal vertex in $P$ with
  $\max B'\preceq_{N} u$.  Since $N$ is acyclic, $P_{u, \max B'}$ and $P'$
  are arc-disjoint.  In summary, $P'$, $P_{u,v}$, and $P_{u, \max B'}$ are
  pairwise arc-disjoint.  Hence, $\max B'$ and $v$ are connected by two
  arc-disjoint undirected paths that correspond to $P'$ and the union of
  $P_{u,v}$ and $P_{u, \max B'}$.  Therefore, $\max B'$ and $v$ are
  contained in a common block $B''$.  In particular, $B$ and $B''$ share
  all arcs in $P_{u,v}$, and thus at least one arc. Similarly, $B'$ and
  $B''$ share all arcs in $P'$, and thus at least one arc.  By Obs.\
  \ref{obs:biConn-arc-disjoint}, it follows that $B=B''=B'$.  Similarly,
  $\max B \preceq_{N} \max B'$ implies $B=B'$.

  Suppose now that $\max B$ and $\max B'$ are $\preceq_N$-incomparable.
  Recall that $P$ and $P'$ are arc-disjoint and each contain at least one
  arc.  Let $\breve{P}$ be the undirected path corresponding to the union
  of $P$ and $P'$ and observe that all of its inner vertices $w$ that
  satisfy $w\preceq_N \max B$ or $w\preceq_N \max B'$.  Since $\max B$ and
  $\max B'$ are $\preceq_N$-incomparable, Lemma~\ref{lem:upper-path}
  implies that they are connected by an undirected path $\invbreve{P}$ that
  contains at least 3 vertices and of which all inner vertices $w'$ satisfy
  $w'\not\preceq_N \max B$ and $w'\not\preceq_N \max B'$.  As a
  consequence, $\breve{P}$ and $\invbreve{P}$ only have their endpoints
  $\max B$ and $\max B'$ in common.  Hence, $\max B$ and $\max B'$ are
  contained in a common block $B''$.  In particular, $B$ and $B''$ share
  all arcs in $P$, and thus at least one arc. Similarly, $B'$ and $B''$
  share all arcs in $P'$, and thus at least one arc.  By Obs.\
  \ref{obs:biConn-arc-disjoint}, it follows that $B=B''=B'$.
\end{proof}

By definition, $N$ is tree if and only if it contains no undirected cycle,
i.e., if all blocks are trivial. Thus $N$ is a tree if and only if there
are no hybrid vertices.

\begin{definition}
  \label{def:intB}
  Let $N$ be a network and $B$ a non-trivial block in $N$ with
  terminal vertices $\{m_1,m_2,\dots, m_h\}$, $h\ge 1$. Then
  \begin{equation}
    B^0 \coloneqq B\setminus\{ \max B, m_1, m_2,\dots m_h\}
  \end{equation}
  is the \emph{interior} of $B$.
\end{definition}
As an immediate consequence of Lemma~\ref{lem:block-identity}, we have
\begin{fact}
  \label{obs:block-identity}
  Let $B_1$ and $B_2$ be two distinct blocks in $N$. Then $B_1^0\cap
  B_2^0=\emptyset$.
\end{fact}

\begin{lemma}
  \label{lem:hybrid-properly-contained}
  Let $N$ be a network and $w\in V(N)$ be a hybrid vertex. Then $w$ and all
  of its in-neighbors are contained in a non-trivial block $B$.
\end{lemma}
\begin{proof}
  Let $w$ be a hybrid vertex, i.e., $\indeg_N(w)\ge 2$, and let $v$
  and $v'$ be two distinct in-neighbors of $w$. If $v'\prec_N v$, then there is
  a directed path $P$ from $v$ to $v'$ that contains at least one arc.
  Moreover, $w$ is not a vertex of $P$ since otherwise $v'\preceq_{N} w$ would
  contradict $w\prec_N v'$. Therefore, $P$ together with $w$ and arcs $vw$ and
  $v'w$ form an undirected cycle.  An analogous argument applies if
  $v\prec_N v'$.  If $v$ and $v'$ are $\preceq_{N}$-incomparable, then
  Lemma~\ref{lem:upper-path} implies that they are connected by an
  undirected path $P$ that contains at least 3 vertices and of which all
  inner vertices $w'$ satisfy $w'\not\preceq_N v$ and $w'\not\preceq_N v'$.
  Together with $w\prec_N v,v'$, this implies that $w$ is not contained in
  $P$.  Therefore, $P$ together with $w$ and arcs $vw$ and $v'w$ form an
  undirected cycle.  In summary, in all cases, $w$ is contained in a
  non-trivial block $B_{v'}$ that, in particular, also contains $v$, $v'$,
  and the arc $vw$.  Since $v'$ was chosen arbitrarily among the
  in-neighbors of $w$ that are distinct from $v$ and the blocks $B_{v'}$
  for all of these vertices share the arc $vw$,
  Obs.~\ref{obs:identical-block} implies that $w$ and all of its
  in-neighbors are contained in a non-trivial block $B$.
\end{proof}
Following \cite{Huson:11}, we say that a hybrid vertex $w$ is
\emph{properly contained} in a block $B$ if all of its in-neighbors are
contained in $B$. As an immediate consequence of
Lemma~\ref{lem:hybrid-properly-contained}, every hybrid vertex is properly
contained in exactly one block.
\begin{lemma}
  \label{lem:properly-contained}
  Let $N$ be a network, $w$ a hybrid vertex in $N$, and $B$ be a block of $N$.
  Then the following statements are equivalent:
  \begin{enumerate}[noitemsep,nolistsep]
  \item $w$ is properly contained in $B$, i.e., $w$ and all of its parents
    are contained in $B$.
  \item $w$ and one of its parents $u$ are contained in $B$.
  \item $w\in V(B)\setminus\{ \max B \}$.
  \end{enumerate}
\end{lemma}
\begin{proof}
  \textit{(3) $\implies$ (2).}
  Since $\max B$ is the unique $\preceq_{N}$-maximal vertex in $B$, we have
  $w\prec_N \max B$. By Lemma~\ref{lem:max-B-unique}, there is a directed
  path from $\max
  B$ to $w$ that is completely contained in $B$. Clearly, $P$ contains a
  parent of $w$, which is there also contained in $B$.
  \textit{(2) $\implies$ (1).}
  By Lemma~\ref{lem:hybrid-properly-contained}, $w$ and all of its parents
  are contained in a non-trivial block $B'$.
  Hence, $B$ and $B'$ share the distinct vertices $w$ and $u$. By
  Obs.~\ref{obs:identical-block}, $B=B'$.
  \textit{(1) $\implies$ (3).}
  If $w$ and all of its (at least two) parents are contained in $B$, then
  clearly $w\in V(B)\setminus\{ \max B \}$.
\end{proof}

As a consequence, if a hybrid vertex $w$ is contained in a block $B$ but
not properly contained, then it must hold $w=\max B$. This motivates the
following definition of level-$k$ networks:
\begin{definition}
  \label{def:level-k}
  A network $N$ is \emph{level-$k$} if each block $B$ of $N$ contains at
  most $k$ hybrid vertices distinct from $\max B$.
\end{definition}

Equivalently, by Lemma~\ref{lem:hybrid-properly-contained}, $N$ is
level-$k$ if each block $B$ of $N$ properly contains at most $k$ hybrid
vertices.  In \cite{CJSS04}, level-$k$ networks are simply defined by
having no more than $k$ hybrid vertices within any given block. We note
that this is equivalent to our definition in a setting where hybrid
vertices are restricted to having outdegree $1$. Def.~\ref{def:level-k}
also accommodates the contraction of  out-arcs of hybrid vertices $v$
with $\outdeg(v)=1$, see Fig.~\ref{fig:network-8}.
\begin{figure}
  \begin{center}
    \includegraphics[width=0.45\textwidth]{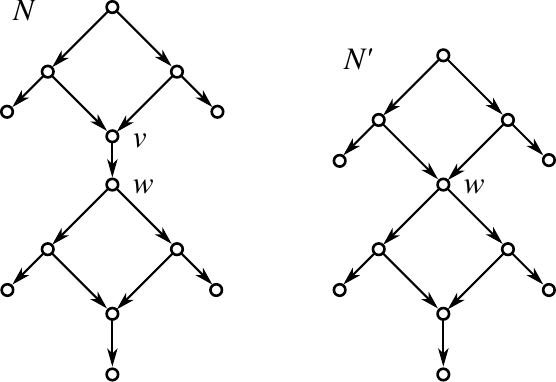}
  \end{center}
  \caption{The network $N$ contains a hybrid vertex $v$ with
    $\outdeg(v)=1$. Network $N'$ is obtained from $N$ by contraction of the
    arc $(v,w)$, i.e., the operation $\contract(v,w)$ which preserves
    vertex $w$.  Vertex $w$ is now a hybrid vertex that is contained in two
    blocks of $N'$.  However, only the upper block properly contains it.}
  \label{fig:network-8}
\end{figure}

The following two lemmas show that neither arc contraction nor expansion
increases the level of a network. Since their proofs are rather lengthy and
technical, they are given in Section~\ref{sec:appx-blocks} in the Appendix.

\begin{lemma}
  \label{lem:contract-level-k}
  Let $N$ be a network, $(w',w) \in E(N)$ be an arc that is not a
  shortcut, and $N'$ be the network obtained from $N$ by applying
  $\contract(w',w)$.  If $N$ is level-$k$, then $N'$ is also level-$k$.
\end{lemma}
\begin{proof}
  See Section~\ref{sec:appx-blocks} in the Appendix.
\end{proof}

\begin{figure}[t]
  \begin{center}
    \includegraphics[width=0.5\textwidth]{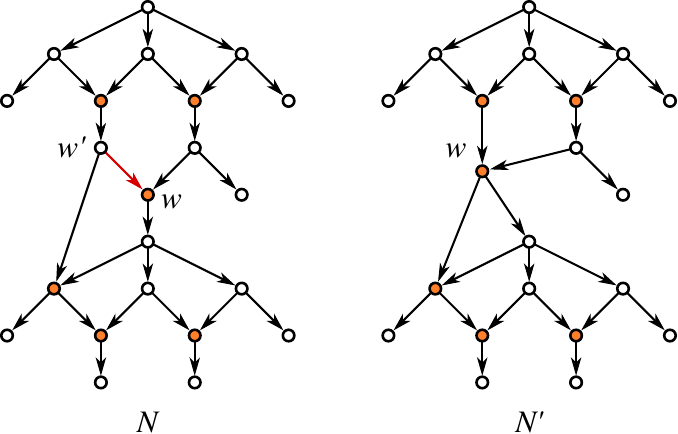}
  \end{center}
  \caption{The level-$3$ network $N'$ is obtained from the level-$6$
    network by application of $\contract(w',w)$. The hybrid vertices are
    highlighted in orange.}
  \label{fig:cntr-decreases-level}
\end{figure}

The converse of Lemma~\ref{lem:contract-level-k}, however, is not true as
shown by the example in Fig.~\ref{fig:cntr-decreases-level}. This example
also shows that even mitigated versions ``\emph{if $N'$ is level-$k$, then
$N$ is level-$(k+1)$}'' do not hold. As an immediate consequence of the
definition of $\phylo(N)$, Lemma~\ref{lem:outdeg-1-contraction} and
\ref{lem:contract-level-k}, we obtain
\begin{corollary}
  Let $N$ be a level-$k$ network. Then the network $N'$ obtained by
  operation $\phylo(N)$ is a phylogenetic level-$k$ network such that
  $\mathscr{C}_{N}=\mathscr{C}_{N'}$.
\end{corollary}

\begin{lemma}
  \label{lem:expand-level-k}
  Let $N$ be a network and $N'$ be the network obtained from $N$ by applying
  $\expand(w)$ for some $w\in V(N)$. Then $N$ is level-$k$ if and only if
  $N'$ is level-$k$.
\end{lemma}
\begin{proof}
  See Section~\ref{sec:appx-blocks} in the Appendix.
\end{proof}

The definition of $\phylo(N)$ and $\contract^{\star}(v',v)$ and
Lemma~\ref{lem:contract-level-k} yield
\begin{corollary}
  Let $N$ be a level-$k$ network.  If $N'$ is the the network obtained from
  $N$ by applying $\phylo(N)$ or $\contract^{\star}(v',v)$ for some arc
  $(v',v) \in E(N)$ that is not a shortcut, then $N'$ is phylogenetic and
  level-$k$.
\end{corollary}

\subsection{Clusters, Hasse Diagrams, and Regular Networks}
\label{sec:cluster-hasse}

In this section we consider general properties of the set of clusters
$\mathscr{C}_N$ of a phylogenetic network as specified in
Def.~\ref{def:cluster} above.

\begin{lemma}
  \label{lem:C-N}
  For all networks $N$ on $X$, the set $\mathscr{C}_N$ is a clustering
  system.
\end{lemma}
\begin{proof}
  Every non-leaf vertex $v\in V\setminus X$ has at least one out-neighbor
  and $N$ is acyclic and finite. Thus every directed path in $N$ can be
  extended to a directed path that eventually ends in a leaf, implying
  $\CC(v)\ne\emptyset$. Since $\CC(v)\ne\emptyset$ for all $v\in V$ and since
  $N$ contains at least a root $\rho_N$ as a vertex, we have
  $\emptyset\notin\mathscr{C}_N$ and thus Condition (i) holds. Since
  $v\preceq_N \rho_N$ for all $v\in V$, we have $\CC(\rho_N)=X$ and (ii) is
  satisfied.  To see that Condition (iii) holds, observe that for all
  $x\in X$, we have $\outdeg(x)=0$ and thus $\CC(x)=\{x\}$.
\end{proof}
This simple observation connects phylogenetic networks to a host of
literature on clustering systems, which have been studied with motivations
often unrelated to evolution or phylogenetics
\cite{Jardine:71,Barthelemy:08,Janowitz:10}.

A particular difficulty in the characterization of certain types of
networks by means of their clustering systems is that even rather simple
clustering systems such as hierarchies can be explained by very
complex networks.
\begin{lemma}\label{lem:complexNsimpleC}
  Let $n$ be a positive integer. Then, for all $k\in \{0,2,\dots,n\}$, there
  is a phylogenetic, shortcut-free level-$k$ network $N$ on $n$ leaves that
  is not level-$(k-1)$ such that $\mathscr{C}_N$ is a hierarchy.
\end{lemma}
\begin{proof}
  If $n=1$, then $k=0$ and the single vertex graphs serves as an example
  (since a network contains at least one vertex and thus a level-$(-1)$
  cannot exist by definition).
  Let $n\geq 2$.  For $k=0$, simply take a tree whose root is adjacent to
  the $n$ leaves only. Again, this tree is level-$0$ but not level-$(-1)$.
  We refer to this tree as a star tree.  For $k\geq
  2$, take a star tree $T$ and randomly collect $k$ of its leaves
  $l_1,\dots l_k$.  Now add new leaves $x_1,\dots,x_k$ and edges such that
  the induced subgraph $N[\{l_1,\dots l_k,x_1,\dots,x_k\}$ is graph
  isomorphic to a complete bipartite graphs where  one part of the
  bipartition contains all $l_1,\dots,l_k$ and  the other part
  all $x_1,\dots,x_k$  (see Fig.~\ref{fig:proof-construction}
  for a generic example).
  It is easy to verify that $N$ is shortcut-free,
  phylogenetic, level-$k$ but not level-$(k-1)$.  In all cases, $\mathscr{C}_N$
  just consist of the clusters $\{x_1,\dots,x_k\}$, $X$, and the singletons
  $\{x\}$,  $x\in X$ and is, therefore, a hierarchy.
\end{proof}

\begin{figure}
  \begin{center}
    \includegraphics[width=0.65\textwidth]{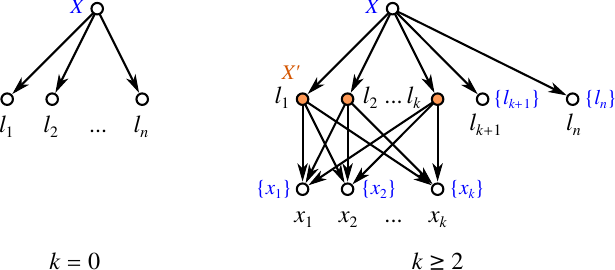}
  \end{center}
  \caption{A generic framework that shows that, for every $n\geq 2$ and
    $k\in \{0,2,\dots,n\}$, there is a phylogenetic, shortcut-free level-$k$
    network $N$ on $n$ leaves that is not level-$(k-1)$ and where
    $\mathscr{C}_N$ is a hierarchy. The network shown left is level-0 and its
    clustering system is trivially a hierarchy. The clustering system
    $\mathscr{C}_N$ of the level-$k$ network $N$ shown right consists of the
    clusters $X'\coloneqq \{x_1,\dots,x_k\}$ (for which the corresponding
    vertices are highlighted in orange),
    $X = X'\cup \{l_{k+1}, \dots, l_n\}$, and the
    singletons and, therefore, $\mathscr{C}_N$ is a hierarchy.}
  \label{fig:proof-construction}
\end{figure}
As we shall in Lemma \ref{lem:complexNsimpleC-level1} there is no
phylogenetic shortcut-free level-1 network $N$ (that is not a tree) for
which $\mathscr{C}_N$ is a hierarchy.

For a clustering system $\mathscr{C}$ on $X$ and a subset $A\subseteq 2^X$,
we define the \emph{closure operator} as the map $\cl\colon 2^X\to 2^X$
defined by
\begin{equation}\label{eq:closure}
  \cl(A)\coloneqq \bigcap_{\substack{C\in\mathscr{C}\\ A \subseteq C}} C.
\end{equation}
It is well defined, isotonic
[$A\subseteq B \implies \cl(A)\subseteq \cl(B)$], enlarging
[$A\subseteq\cl(A)$], and idempotent $\cl(\cl(A))=\cl(A)$. For $\vert
X\vert>1$, we have $\cl(\emptyset)=\emptyset$.

\begin{definition}
  \label{def:closed}
  A clustering system $\mathscr{C}$ is
  \emph{closed} if, for all \emph{non-empty} $A\in 2^X$, the following
  condition holds: $\cl(A)=A \iff A\in\mathscr{C}$.
\end{definition}

The following result is well know in the clustering literature.
\begin{lemma}\label{lem:simple-closed}
  A clustering system $\mathscr{C}$ is closed if and only if
  $A,B \in \mathscr{C}$ and $A\cap B\ne\emptyset$ implies
  $A\cap B\in \mathscr{C}$.
\end{lemma}
\begin{proof}
  For completeness, a proof is provided in Section~\ref{sec:appx-closed} in
  the Appendix.
\end{proof}

We continue with three simple observations concerning the clusters of
networks.
\begin{lemma}
  \label{lem:inclusion}
  Let $N$ be a network. Then $v\preceq_N w$ implies $\CC(v)\subseteq \CC(w)$.
\end{lemma}
\begin{proof}
  By construction, $x\in \CC(v)$ if and only if $x\in X$ and $v$ lies on a
  directed path from the root $\rho_N$ to $x$. Furthermore, $v\preceq_N w$
  implies that $w$ lies on a directed path from $\rho_N$ to $v$. By (N1)
  and since $N$ is a DAG, there is directed path from $\rho_N$ to $x$ that
  contains $w$, and thus $x\preceq_N w$, i.e., $x\in \CC(w)$.
\end{proof}

We note in passing that the converse of Lemma \ref{lem:inclusion} is not
always satisfied (even in level-$1$ networks): If $v$ is a hybrid vertex
with unique child $w$, we have $\CC(v) = \CC(w)$ and thus
$\CC(v)\subseteq\CC(w)$, but $v\not\preceq_N w$, (cf.\ the network $N$ in
Fig.~\ref{fig:network-8}).  A result similar to Lemma~\ref{lem:upper-path}
ensures the existence of a path $P$ connecting $\preceq_N$-incomparable
vertices $u,v\in V(N)$ that contains only vertices that are below $u$ or
$v$.  However, it requires that $u$ and $v$ have at least one descendant
leaf in common, i.e., that $\CC(u)\cap\CC(v)\ne\emptyset$:
\begin{lemma}
  \label{lem:lower-path}
  Let $N$ be a network and $u,v\in V(N)$ be two $\preceq_N$-incomparable
  vertices such that $\CC(u)\cap \CC(v)\ne \emptyset$.  Then, for every
  $x\in \CC(u)\cap \CC(v)$, $u$ and $v$ are connected by an undirected path
  $P=(w_1\coloneqq u, \dots, w_h,\dots, w_k\coloneqq v)$, $1<h<k$, such
  that
  \begin{description}
    \item[(i)] $(w_i, w_{i+1})\in E(N)$ for all $1\le i< h$,
    $(w_{i+1}, w_i)\in E(N)$ for all $h\le i< k$, and $w_h$ is a hybrid
    vertex satisfying $w_h\prec_N u$ and $w_h\prec_N v$.
    \item[(ii)] $x\in \CC(w_h)$,
  \end{description}
  In particular, $k\ge 3$,
  all inner vertices $w_i$ of $P$ satisfy $w_i\prec_N u$ or $w_i\prec_N v$,
  and $P$ is a subgraph of a non-trivial block $B$.
\end{lemma}
\begin{proof}
  There are directed paths $P_u$ and $P_v$ from $u$ and $v$, respectively,
  to the leaf $x$.  Let $w^*$ be the $\preceq_N$-maximal vertex of $P_u$
  that is also a vertex of $P_v$, which exists since at least $x$ is
  contained in both paths.  It must hold that $w^*\notin \{u,v\}$ since
  otherwise $u$ and $v$ would be $\preceq_N$-comparable.  In particular,
  $w^*\prec_N u$ and $w^*\prec_N v$.  Let $P'_u$ and $P'_v$ be the subpaths
  of $P_u$ and $P_v$ from $u$ and $v$, respectively, to $w^*$. By
  construction, $w^*$ must be a hybrid vertex, $x\in \CC(w^*)$, and $P'_u$
  and $P'_v$ only have vertex $w^*$ in common, which moreover is an outer
  vertex of both paths.  Now consider the path
  $P=(w_1\coloneqq u, \dots, w_h\coloneqq w^*,\dots, w_k\coloneqq v)$, that
  is the union of the underlying undirected version of $P'_u$ and $P'_v$.
  It is now easy to verify that $P$ satisfies all of the desired
  properties.  Two see that $P$ is a subgraph of a non-trivial block $B$,
  observe that, by Lemma~\ref{lem:upper-path}, the two
  $\preceq_N$-incomparable vertices $u$ and $v$ are connected by an
  undirected path $\invbreve{P}$ that contains at least 3 vertices and of
  which all inner vertices $w'$ satisfy $w'\not\preceq_N u$ and
  $w'\not\preceq_N v$.  Hence, $P$ and $\invbreve{P}$ cannot have any inner
  vertices in common.  Therefore, $u$ and $v$ are connected by two distinct
  paths that both have at least 3 vertices and that only have the endpoints
  $u$ and $v$ in common. Hence, $u$ and $v$ lie on a cycle $K$ and thus in
  a common block $B$ of $N$. In particular, $P$ is a subgraph of $K$ and
  thus of $B$.
\end{proof}

\begin{lemma}\label{lem:overlap-B0}
  Let $N$ be a network and $u, v\in V(N)$. If $\CC_N(u)$ and $\CC_N(v)$
  overlap, then $u$ and $v$ are $\preceq_{N}$-incomparable and
  $u,v\in B^0$ for a non-trivial block $B$ of $N$.
\end{lemma}
\begin{proof}
If two clusters $\CC_N(u)$ and $\CC_N(v)$ overlap, then
Lemma~\ref{lem:inclusion} implies that $u$ and $v$ are
$\preceq_{N}$-incomparable. Lemma~\ref{lem:lower-path} implies
that $u$ and $v$ are contained in
a common non-trivial block $B$ of $N$.
Since $\CC(u)\subseteq \CC(\max B)$ for all $u\in B$, $\CC(\max B)$
does not overlap any cluster $\CC(v)$ with $v\in B$, and thus $u,v\ne
\max B$. Now suppose $u$ is a terminal vertex of $B$ and $v\in B$ such
that $\CC(u)$ and $\CC(v)$ overlap.  By Lemma~\ref{lem:lower-path} there
is $w\in B$ in $x\in\CC(u)\cap\CC(v)$ with $w\prec_N u$, contradicting
that $u$ is terminal. Therefore $u,v\in B^0$.
\end{proof}

Clusters of outdegree-1 vertices $w$ are redundant in the sense that every
directed path from $w$ to one of its descendant leaves necessarily passes
through the unique child $v$ of $w$. Thus we have $\CC(w)\subseteq \CC(v)$.
Moreover, $v\prec_N w$ and Lemma~\ref{lem:inclusion} imply $\CC(v)\subseteq
\CC(w)$, and thus, $\CC(v)=\CC(w)$. Hence, we have
\begin{fact}
  \label{obs:outdeg-1-cluster}
  Let $N$ be a network. If $v$ is the unique child of $w$ in $N$, then
  $\CC(v)=\CC(w)$.
\end{fact}

\begin{lemma}
  \label{lem:union-hybrid-clusters}
  Let $N$ be a network, $B$ a block in $N$ and $u,v\in V(B)$. Moreover, let
  $H$ be the set of hybrid vertices $h$ that are properly contained in $B$
  and satisfy $h\preceq_{N} u,v$. Then it holds
  $\CC(u)\cap \CC(v)\in \{\CC(u),\CC(v), \bigcup_{h\in H} \CC(h)\}$.
\end{lemma}
\begin{proof}
  It suffices to show that $\CC(u)\cap \CC(v)\notin \{\CC(u),\CC(v)\}$ implies
  $\CC(u)\cap \CC(v)= \bigcup_{h\in H} \CC(h) \eqqcolon C$.  Hence, suppose
  $\CC(u)\cap \CC(v)\notin \{\CC(u),\CC(v)\}$. Then Lemma~\ref{lem:inclusion}
  implies that $u$ and $v$ are $\preceq_{N}$-incomparable.  If $x\in C$,
  then $x\in \CC(h)$ for some $h\in H$.  Since $h\preceq_{N} u,v$,
  Lemma~\ref{lem:inclusion} implies $\CC(h)\subseteq \CC(u)$ and
  $\CC(h)\subseteq \CC(v)$ and thus, $x\in \CC(h)\subseteq \CC(u)\cap \CC(v)$.  Now
  suppose $x\in \CC(u)\cap \CC(v)$.  By Lemma~\ref{lem:lower-path}, the
  $\preceq_{N}$-incomparable vertices $u$ and $v$ are connected by an
  undirected path $P$ which contains a hybrid vertex $h \prec_N u,v$ with
  $x\in \CC(h)$ and is a subgraph of a non-trivial block $B'$ of $N$.  Since
  $B$ and $B'$ share the two distinct vertices $u$ and $v$,
  Obs.~\ref{obs:identical-block} implies $B=B'$.  In particular,
  $u,v,h\in V(B)$ and $h \prec_N u,v$ imply $h\ne \max B$, and thus, $h$
  must be properly contained in $B$ by
  Lemma~\ref{lem:hybrid-properly-contained}.  Hence, we have $h\in H$ and
  thus $x\in \CC(h) \subseteq C$.  In summary, we have $x\in \CC(u)\cap \CC(v)$
  if and only if $x\in C$, and thus $\CC(u)\cap \CC(v)=C$.
\end{proof}
Note that $H=\emptyset$ and thus, $C=\bigcup_{h\in H} \CC(h)=\emptyset$
in Lemma~\ref{lem:union-hybrid-clusters} is possible.

The \emph{Hasse diagram} $\Hasse\coloneqq \Hasse[\mathscr{C}]$ of
$\mathscr{C}$ w.r.t.\ to set inclusion is a DAG whose vertices are the
clusters in $\mathscr{C}$. There is a directed arc $(C,C')\in\Hasse$ if
$C'\subsetneq C$ and there is no $C''\in\mathscr{C}$ with $C'\subsetneq
C''\subsetneq C$. Since $X\in\mathscr{C}$, the Hasse diagram is connected
and has $X$ as its unique root. The singletons $\{x\}$, $x\in X$, are
exactly the inclusion-minimal vertices in $\mathscr{C}$ and thus, they have
outdegree $0$ but not necessarily indegree $1$ in $\Hasse$.  Another simple
property of $\Hasse$ is the following:
\begin{lemma}
  \label{lem:Hasse-outdeg}
  Let $\mathscr{C}$ be a clustering system on $X$. Then every non-singleton set
  $C\in\mathscr{C}$ satisfies $\outdeg_{\Hasse}(C)\ge 2$ in the Hasse diagram
  $\Hasse$ of $\mathscr{C}$.
\end{lemma}
\begin{proof}
  Let $C\in \mathscr{C}$ be a non-singleton set, i.e., $\vert C\vert\ge 2$.
  Therefore and since $\{x\}\in\mathscr{C}$ for all $x\in X$, there is a
  directed path in $\Hasse$ from $C$ to some singleton set
  $\{x'\}\in\mathscr{C}$. In particular, this path contains at least the two
  distinct clusters $C$ and $\{x'\}$, and thus, $C$ has a child $C'$ in
  $\Hasse$ with $\{x'\}\subseteq C'\subsetneq C$.
  Now pick an element $x''\in C\setminus C'\ne \emptyset$. Since
  $\{x''\}\in\mathscr{C}$ and $\{x''\}\subseteq C$, we can argue similarly
  as before to conclude that $C$ has a child $C''$ in $\Hasse$ with
  $\{x''\}\subseteq C''\subsetneq C$.
  Since $x''\notin C'$, we have $C'\ne C''$ and thus, $C$ satisfies
  $\outdeg_{\Hasse}(C)\ge 2$.
\end{proof}
\begin{lemma}
  \label{lem:Hasse-is-Network}
  Let $\mathscr{C}$ be a clustering system on $X$ with corresponding Hasse
  diagram $\Hasse$.  Then $\Hasse$ is a phylogenetic network with leaf set
  $X_{\Hasse}\coloneqq \{ \{x\} \mid x \in X \}$.
\end{lemma}
\begin{proof}
  Clearly, $\Hasse$ is a DAG.  Since $X\in\mathscr{C}$ and $C\subseteq X$
  for all $C\in\mathscr{C}$, $X$ is the unique cluster in $\mathscr{C}$
  with indegree 0, i.e., $X$ is the root in $\Hasse$ and $\Hasse$ satisfies
  (N1).  By definition of clustering systems, we have
  $X_{\Hasse}\subseteq \mathscr{C}$. Now consider a cluster
  $\{x\}\in X_{\Hasse}$. Since $\emptyset\notin \mathscr{C}$, $\{x\}$ has
  outdegree zero in $\Hasse$.  Lemma~\ref{lem:Hasse-outdeg} implies
  $\outdeg_{\Hasse}(C)\ge 2$ for all $C\in\mathscr{C}$ with $\vert C\vert>1$,
  i.e.,
  for all $C\in\mathscr{C}\setminus X_{\Hasse}$.  Taken together, the
  latter arguments imply that the elements in $X_{\Hasse}$ are exactly the
  leaves of $\Hasse$ and that (N2) is satisfied.
\end{proof}
For a given a clustering system $\mathscr{C}$ and a cluster
$C\in\mathscr{C}$, we will moreover make use of the subsets
\begin{equation}\label{eq:D}
  \mathcal{D}(C)\coloneqq\{D\in\mathscr{C}\mid D\subsetneq C\}
  \qquad\text{and} \qquad
  \overline{\mathcal{D}}(C)\coloneqq\{D\in\mathscr{C}\mid D \not\subseteq C\}.
\end{equation}
Note that, by definition, we have
$\mathcal{D}(C)\cupdot \overline{\mathcal{D}}(C)\cupdot\{C\}=\mathscr{C}$
for all $C\in\mathscr{C}$, $\mathcal{D}(C)=\emptyset$ if and only if $C$ is a
singleton, and $\overline{\mathcal{D}}(C)=\emptyset$ if and only if $C=X$.

\begin{lemma}
  \label{lem:cutvertex}
  Let $\Hasse$ be the Hasse diagram of a clustering system $\mathscr{C}$
  and $C\in\mathscr{C}$ such that $C$ does not overlap any other set. Then,
  there is no undirected cycle in $\Hasse$ that intersects both
  $\mathcal{D}(C)$ and $\overline{\mathcal{D}}(C)$. In particular, if
  $C\neq X$ and $\vert C\vert>1$, then $C$ is a cut vertex in $\Hasse$.
\end{lemma}
\begin{proof}
  Suppose that $C$ does not overlap with any other cluster.  If $C=X$ then,
  $\overline{\mathcal{D}}(C) = \emptyset$ and if $\vert C\vert =1$ then
  $\mathcal{D}(C) = \emptyset$ and thus, for any cycle $K$ in $\Hasse$ we
  have $K\cap \overline{\mathcal{D}}(C) = \emptyset$ or $K\cap
  {\mathcal{D}}(C) = \emptyset$. Hence, $K$ cannot intersect both. Now,
  assume that $C\neq X$ and $\vert C\vert>1$.  Since $C$ is neither a singleton
  nor
  $X$, both $\mathcal{D}(C)$ and $\overline{\mathcal{D}}(C)$ are
  non-empty. Furthermore, $\mathcal{D}(C)\cup\overline{\mathcal{D}}(C)=
  \mathscr{C}\setminus\{C\}$.  Let $C_1\in\mathcal{D}(C)$ and
  $C_2\in\overline{\mathcal{D}}(C)$. By assumption, we have $C_1\subsetneq
  C$ and either (i) $C_2\cap C=\emptyset$ or (ii) $C\subsetneq C_2$.  In
  case (i), we have $C_1\cap C_2=\emptyset$ and in case (ii), it holds
  $C_1\subsetneq C\subsetneq C_2$. Therefore and since
  $C_1\in\mathcal{D}(C)$ and $C_2\in\overline{\mathcal{D}}(C)$ were chosen
  arbitrarily, $\Hasse$ contains no arc connecting a cluster in
  $\mathcal{D}(C)$ and a cluster in $\overline{\mathcal{D}}(C)$. Together
  with
  $\mathcal{D}(C)\cup\overline{\mathcal{D}}(C)=\mathscr{C}\setminus\{C\}$,
  this implies that the subgraph of $\Hasse$ obtained by removing $C$ is
  disconnected and thus, $C$ is a cut vertex.  In particular, every
  undirected path connecting a cluster in $\mathcal{D}(C)$ and a cluster in
  $\overline{\mathcal{D}}(C)$ has to pass through $C$ and thus the second
  statement of the lemma follows as an immediate consequence.
\end{proof}

\begin{figure}
  \begin{center}
    \includegraphics[width=0.75\textwidth]{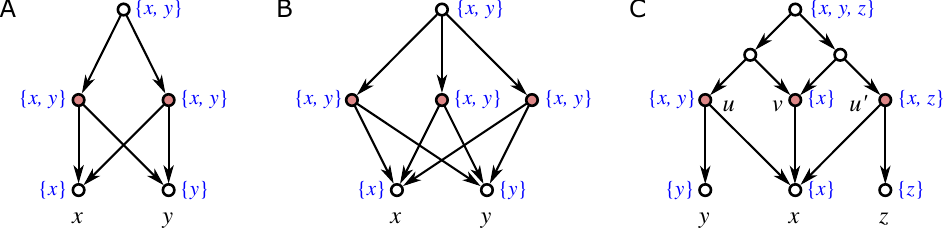}
  \end{center}
  \caption{Both the rooted $K_{2,3}$ (A) and $K_{3,3}$ (B) are a
    phylogenetic networks that have only two leaves, denoted $x$ and $y$
    here. The clustering system therefore consists only of $X=\{x,y\}$ and
    the two singletons $\{x\}$ and $\{y\}$. The same clustering system can
    be represented by a rooted tree with a single root that is adjacent to
    the two leaves $x$ and $y$. In particular, both networks do not satisfy
    (PCC) since the highlighted vertices are $\preceq_{N}$-incomparable but
    share the cluster $\{x,y\}$.  (C) A network showing that
    $\CC(v)\subsetneq \CC(u)$ is also possible for two
    $\preceq$-incomparable vertices $u$ and $v$.}
  \label{fig:no-PCC}
\end{figure}

Every phylogenetic tree $T$ is isomorphic to the Hasse diagram of its
clustering systems $\mathscr{C}$ by virtue of the map
$\varphi \colon V(T)\to \mathscr{C},\, v\mapsto \CC(v)$, see e.g.\
\cite{sem-ste-03a}.  Fig.~\ref{fig:no-PCC} shows that this is not the
case for phylogenetic networks in general. The rooted networks that share
this property with phylogenetic trees have been introduced and studied in
\cite{Baroni:05,Baroni:06,Willson:10}.
\begin{definition} \cite{Baroni:05}
  \label{def:regular-N}
  A network $N=(V,E)$ is \emph{regular} if the map
  $\varphi\colon V\to V(\Hasse[\mathscr{C}_N])\colon v\mapsto \CC(v)$ is a
  graph isomorphism between $N$ and $\Hasse[\mathscr{C}_N]$.
\end{definition}
The graph isomorphism in Def.~\ref{def:regular-N} is quite constrained. In
particular, it is not obvious that an arbitrary graph isomorphism $\varphi$
between the two networks $N$ and $\Hasse[\mathscr{C}_N]$ implies that $N$
is regular, as $\varphi(v)\neq \CC(v)$ may be possible.  As we shall see in
Cor.~\ref{cor:hasse-iso}, however, $N\sim\Hasse[\mathscr{C}_N]$ if and only
if $N$ is regular.  Even more, as noted without proof in \cite{Baroni:05},
a rooted network $N$ with leaf set $X$ is regular if and only if it is
graph isomorphic to the Hasse diagram $\Hasse[\mathscr{C}]$ for some
clustering system $\mathscr{C}\subseteq 2^X$. This result will be an
immediate consequence of the results established below and will be
summarized and proven in Prop.~\ref{prop:N-regular-someC}.

\begin{proposition}
  \label{prop:regular-unique}
  For every clustering system $\mathscr{C}$, there is a unique regular
  network $N$ with $\mathscr{C}_N=\mathscr{C}$.
\end{proposition}
\begin{proof}
  Let $\mathscr{C}$ be a clustering system on $X$. By
  Lemma~\ref{lem:Hasse-is-Network}, $\Hasse[\mathscr{C}]$ is a network with
  leaf set $X_{\Hasse}\coloneqq \{ \{x\} \mid x \in X \}$.
  Replacing all leaves $\{x\}$ in $\Hasse[\mathscr{C}]$ with the single
  element $x$ that they contain clearly yields a network $N$ such that
  $\mathscr{C}_{N}=\mathscr{C}$ and $\varphi\colon V(N)\to
  V(\Hasse[\mathscr{C}])\colon v\mapsto \CC_{N}(v)$ is an isomorphism
  between $N$ and $\Hasse[\mathscr{C}]$. By definition, $N$ is regular.

  Now let $N'$ be a regular network with $\mathscr{C}_{N'}=\mathscr{C}$,
  i.e, there is an isomorphism
  $\varphi'\colon V(N')\to V(\Hasse[\mathscr{C}])\colon v\mapsto \CC_{N'}(v)$
  between $N'$ and $\Hasse[\mathscr{C}]$.
  In particular, we have $\varphi(x)=\varphi'(x)=\{x\}$ for all $x\in X$ and
  thus $\varphi'^{-1}(\varphi(x))=x$.
  Hence, $\varphi'^{-1} \circ \varphi$ is an isomorphism between $N$ and $N$
  that is the identity on $X$. Hence, $N$ is the unique regular network with
  $\mathscr{C}_N=\mathscr{C}$.
\end{proof}

\begin{remark}
  By a slight abuse of notation, we also write $\Hasse[\mathscr{C}]$ for the
  unique regular network of a clustering system $\mathscr{C}$ since it is
  obtained from the Hasse diagram by relabeling all leaves $\{x\}$ with $x$.
\end{remark}

The following characterization is a slight rephrasing of Prop.~4.1 in
\cite{Baroni:05}:
\begin{proposition}
  \label{prop:Baroni}
  A network $N$ is regular if and only if
  \begin{itemize}
  \item[(i)] $\CC(u)\subseteq \CC(v) \iff u\preceq_N v$ for all $u,v\in V$,
    and
  \item[(ii)] $N$ is shortcut-free.
  \end{itemize}
\end{proposition}
\begin{proof}
  Prop.~4.1 in \cite{Baroni:05} states that $N$ is regular if and only if
  the following three conditions hold: (a) $u\ne v$ implies $\CC(u)\ne
  \CC(v)$, i.e., $\CC(u)=\CC(v)\implies u=v$; (b) if $\CC(u)\subsetneq
  \CC(v)$, then there is a directed path from $v$ to $u$, i.e., $u\prec_N
  v$; and (c) if there are two distinct directed paths connecting $u$ and
  $v$, then neither path consists of a single arc, i.e., $(u,v)$ is not a
  shortcut. Clearly, conditions~(ii) and~(c) are equivalent. It therefore
  suffices to show that condition~(i) holds if and only if conditions~(a)
  and~(b) are satisfied.  Together, (a), (b) and Lemma~\ref{lem:inclusion}
  obviously imply (i). Now suppose (i) is satisfied. Then $\CC(u)=\CC(v)$
  implies $\CC(u)\subseteq \CC(v)$ and $\CC(v)\subseteq \CC(u)$ and thus we
  have both $u\preceq_N v$ and $v\preceq_N u$, and hence $u=v$, i.e., (a)
  holds.  Assuming $\CC(u)\subsetneq \CC(v)$, i.e, $\CC(u)\subseteq \CC(v)$
  and $\CC(u)\ne \CC(v)$ implies $u\preceq_N v$ by (i) and $u\ne v$ by (a),
  and thus $u\prec_N v$, i.e., (b) holds as well.
\end{proof}

Prop.~\ref{prop:regular-unique}	and~\ref{prop:Baroni} imply
\begin{corollary} \label{cor:cluster-n-with-strongprec}
  For every clustering system $\mathscr{C}$ there is a network $N$ with
  $\mathscr{C}_N = \mathscr{C}$ such that $\CC(u)\subseteq \CC(v) \iff
  u\preceq_N v$ for all $u,v\in V$.
\end{corollary}

\begin{corollary}
  \label{cor:regular-least-resolved}
  Every regular network is least-resolved.
\end{corollary}
\begin{proof}
  Suppose, for contradiction, that the regular network $N$ is not
  least-resolved, i.e., there is a network $N'$ with
  $\mathscr{C}_N=\mathscr{C}_{N'}$ that can be obtained from $N$ by a non-empty
  sequence of shortcut removals and application of $\contract(v',v)$.
  By Prop.~\ref{prop:Baroni}, $N$ is shortcut-free and, therefore,
  the operation that is applied first must be a contraction.
  Therefore and since no new vertices are introduced, it must hold
  $\vert V(N')\vert <\vert V(N)\vert$. Since $N$ is regular, it holds $\vert
  V(N)\vert =\vert \mathscr{C}_N\vert$.
  Hence, we have $\vert V(N')\vert<\vert \mathscr{C}_N\vert=\vert
  \mathscr{C}_{N'} \vert$; a contradiction.
\end{proof}

The converse of Cor.~\ref{cor:regular-least-resolved}, however, is not
satisfied, see Fig.~\ref{fig:counter-reg-lr}. The network $N$ is
shortcut-free and satisfies $\vert V(N) \vert = \vert \mathscr{C}_N\vert$.
Hence, $\mathscr{C}_N$ is
least-resolved. The unique regular network $N'$ for $\mathscr{C}_N$ on the
r.h.s.\ of Fig.~\ref{fig:counter-reg-lr} is not isomorphic to $N$. Hence, we
obtain
\begin{fact}
  Not every least-resolved network is regular.
\end{fact}
Nevertheless, for level-1 networks the terms least-resolved and regular
coincide as shown in Cor.~\ref{cor:regular-IFF-least-resolved}.

\begin{figure}[t]
  \centering
  \includegraphics[width=0.75\textwidth]{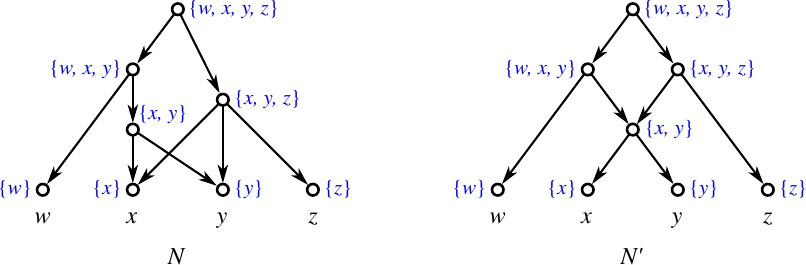}
  \caption{The network $N$  is least-resolved but
    not isomorphic to $N'\sim\Hasse[\mathscr{C}_N]$. Consequently,
    $N$ is not regular.}
  \label{fig:counter-reg-lr}
\end{figure}

A clustering system $\mathscr{C}$, by definition, is a hierarchy if and
only if, for all $C,C'\in\mathscr{C}$ holds $C\cap
C'\in\{\emptyset,C,C'\}$.
\begin{corollary}\label{cor:hasse-tree}
  A clustering system $\mathscr{C}$ on $X$ is a hierarchy if and only if
  $\Hasse[\mathscr{C}]$ is a phylogenetic tree.  Moreover, $N$ is a
  phylogenetic tree if and only if $\Hasse[\mathscr{C}_N]\sim N$ and
  $\mathscr{C}_N$ is a hierarchy.
\end{corollary}
\begin{proof}
  The 1-to-1 correspondence of hierarchies and phylogenetic trees is well
  known, see e.g.\ \cite[Thm.~3.5.2]{sem-ste-03a}.
\end{proof}

\section{Semi-Regular Networks}
\label{sec:semi-reg}

\subsection{Path-Cluster Comparability}
\label{ssec:PCC}

Regularity as characterized in Prop.~\ref{prop:Baroni} is a bit too
restrictive for our purposes. We therefore consider a slightly weaker
condition, which we will call \emph{semi-regularity}. More precisely, we
relax condition (i) in Prop.~\ref{prop:Baroni}:
\begin{definition}\label{def:PCC}
  A network $N$ has the \emph{path-cluster-comparability (PCC)} property if
  it satisfies, for all $u,v\in V(N)$,
  \begin{description}
  \item[(PCC)] $u$ and $v$ are $\preceq_N$-comparable if and only if
    $\CC(u)\subseteq \CC(v)$ or $\CC(v)\subseteq \CC(u)$.
  \end{description}
\end{definition}
\begin{definition}
  \label{def:semi-regular}
  A network is \emph{semi-regular} if it is shortcut-free and satisfies
  (PCC).
\end{definition}
We introduce the term ``semi-regular'' because, as we shall see in
Thm.~\ref{thm:semiregular}, it is a moderate generalization of regularity
that preserves many of useful properties of regular networks. We
emphasize that (PCC) is a quite restrictive property. For instance, the
rooted $K_{3,3}$ in Fig.~\ref{fig:no-PCC}(B) violates (PCC).  The example
in Fig.~\ref{fig:no-PCC}(C) shows that even $\CC(v)\subsetneq\CC(u)$ is
possible for two $\preceq_{N}$-incomparable vertices $u$ and $v$.

\begin{fact}
  \label{obs:subsetneq-implies-below}
  Let $N$ be a network satisfying (PCC) and $u, v\in V(N)$. Then
  $\CC(u)\subsetneq \CC(v)$ implies $u\prec_N v$.
\end{fact}
\begin{proof}
  Suppose $\CC(u)\subsetneq \CC(v)$. Then (PCC) implies that $u$ and $v$
  are $\preceq_{N}$-comparable.  If $v\preceq_N u$, then
  Lemma~\ref{lem:inclusion} implies $\CC(v)\subseteq \CC(u)$; a
  contradiction. Hence, only $u\prec_N v$ is possible.
\end{proof}
Property (PCC) still allows that distinct vertices are associated with the
same clusters. It requires, however, that such vertices lie along a common
directed path.
\begin{lemma}
  \label{lem:equal-outn}
  Let $N$ be a network. Then $\CC(u)=\CC(v)$ and $u\prec_N v$ imply that there
  is $w\in\child_N(v)$ such that $u\preceq_N w \prec_N v$ and $\CC(w)=\CC(v)$.
\end{lemma}
\begin{proof}
  Since $u\prec_N v$, there is a $vw$-path which passes through some child
  $w\in\child_N(v)$. Hence, we have $u\preceq_N w \prec_N
  v$. Lemma~\ref{lem:inclusion} implies $\CC(u)\subseteq \CC(w)\subseteq
  \CC(v)$. Together with $\CC(u)=\CC(v)$, this yields $\CC(v)=\CC(w)=\CC(u)$.
\end{proof}

\begin{lemma}
  \label{lem:outdeg1}
  Let $N$ be a semi-regular network and let $v\in V(N)$. Then, there is a
  vertex $u\in V(N)$ with $\CC(u)=\CC(v)$ and $u\prec_N v$ if and only if
  $\outdeg(v)=1$. If, moreover, $N$ is phylogenetic, then $u$ is the unique
  child of $v$ in this case.
\end{lemma}
\begin{proof}
  Suppose first that $\outdeg(v)=1$ and thus let $u$ be the unique child of
  $v$. Thus, it holds $u\prec_N v$ and, by Obs.~\ref{obs:outdeg-1-cluster},
  we have $\CC(u)=\CC(v)$.  Conversely, suppose that $\CC(u)=\CC(v)$ and
  $u\prec_N v$.  Lemma~\ref{lem:equal-outn} implies that there is
  $w\in\child(v)$ with $u\preceq_N w \prec_N v$ and $\CC(w)=\CC(v)$. Suppose
  there is another child $w'\in\child(v)$ with $w'\ne w$. By
  Lemma~\ref{lem:inclusion}, we have $\CC(w')\subseteq \CC(v)=\CC(w)$. Hence,
  (PCC) implies that $w$ and $w'$ are $\preceq_{N}$-comparable.  But then
  Obs.~\ref{obs:shortcut} and $N$ being shortcut-free imply $w=w'$; a
  contradiction.  Hence, $w$ is the unique child of $v$.

  Now suppose, in addition, that $N$ is phylogenetic and assume, for
  contradiction, that $u\ne w$ and thus $u\prec_N w$.  We can apply similar
  arguments as before to conclude that $w$ has a unique child. Therefore
  and since $N$ is phylogenetic, there must be a vertex
  $v'\in\parent_N(w)\setminus \{v\}$ since otherwise
  $\indeg(w)=\outdeg(w)=1$. By Lemma~\ref{lem:inclusion}, we have
  $\CC(v)=\CC(w)\subseteq \CC(v')$. This together with (PCC) implies that
  $v$ and $v'$ are $\preceq_{N}$-comparable. But then
  Obs.~\ref{obs:shortcut} and $N$ being shortcut-free imply $v=v'$; a
  contradiction.  Therefore, $u=w$ is the unique child of $v$.
\end{proof}

As a consequence of Lemmas~\ref{lem:inclusion} and~\ref{lem:outdeg1}, we
have
\begin{corollary}
  \label{cor:gt2}
  Let $N$ be a semi-regular network and let $v\in V^0$.  Then
  $\CC(u)\subsetneq \CC(v)$ for all $u\in\child_N(v)$ if and only if
  $\outdeg(v)\ge 2$.
\end{corollary}

\begin{lemma}
  \label{lem:C=path}
  Let $N$ be a semi-regular, $u\in V(N)$, and let $Q(u)\coloneqq \{u'\in
  V(N) \mid \CC(u')=\CC(u)\}$.  Then the vertices of $Q(u)$ are pairwise
  $\preceq_N$-comparable and lie consecutively along an induced directed
  path in $N$. Moreover, if $\vert Q(u)\vert>1$, either $Q(u)$ is contained
  in a non-trivial block, or adjacent pairs of vertices in $Q(u)$ form a
  trivial block.
\end{lemma}
\begin{proof}
  By (PCC), $\CC(v)=\CC(w)$ and $v\ne w$ implies $v\prec_N w$ or $w\prec_N
  v$ for $v,w\in Q(u)$, i.e., the vertices in $Q(u)$ are pairwise
  $\preceq_N$-comparable and thus linearly ordered w.r.t.\ $\preceq_N$.
  Using Lemma~\ref{lem:equal-outn}, one easily verifies that $Q(u)$ forms a
  directed path in $N$ and there are unique $\preceq_{N}$-minimal and
  $\preceq_{N}$-maximal vertices $\min Q(u)$ and $\max Q(u)$.  In
  particular, since $N$ is acyclic and shortcut-free, this path must be an
  induced subgraph.

  Now suppose $u$ is a hybrid vertex and suppose there is $v\in Q(u)$ with
  $(v,u)\in E(N)$, i.e., $u$ has a parent in $Q(u)$. Then there is
  $v'\in\parent(u)\setminus\{v\}$.  By Lemma~\ref{lem:inclusion},
  $\CC(u)\subseteq \CC(v')$, and thus by (PCC), $v$ and $v'$ are
  $\preceq_N$-comparable. However, since $N$ is shortcut-free, $v$ and $v'$
  are $\preceq_N$-incomparable by Obs.~\ref{obs:shortcut}; a contradiction.
  Thus only $\max Q(u)$ can be a hybrid vertex in $Q(u)$. By similar
  argument, only $\min Q(u)$ can have outdegree greater that one.  Hence,
  all inner vertices in the directed path $P$ formed by $Q(u)$ have degree
  $2$. Therefore, if one arc in $P$ is contained in an undirected cycle,
  then all arcs in $P$ are contained in this cycle, in which case $Q(u)$
  is contained in a non-trivial block. Otherwise $Q(u)$ consists of a
  sequence of consecutive cut-arcs.
\end{proof}
The examples in Fig.~\ref{fig:no-PCC}(A) and (B) show that (PCC) is
necessary in Lemma~\ref{lem:C=path}. We note, moreover, that the
cardinality $\vert Q(u)\vert$ equals the multiplicity of the cluster
$\CC(u)$ in $\mathscr{M}_N$.

\begin{theorem}
  \label{thm:semiregular}
  A network is regular if and only if it is semi-regular and there is no
  vertex with outdegree $1$.
\end{theorem}
\begin{proof}
  If $N$ is regular, it is in particular also semi-regular, and thus
  satisfies (PCC). Furthermore, then $\CC(u)=\CC(v)$ implies $u=v$ and thus
  there are no two vertices $u,v$ satisfying $\CC(u)=\CC(v)$ and $u\prec_N
  v$. Lemma~\ref{lem:outdeg1} thus implies that there is no vertex with
  outdegree $1$.  Conversely, assume that $N$ is semi-regular (and thus
  shortcut-free) and suppose there is not vertex with outdegree $1$. Then
  Lemma~\ref{lem:outdeg1} implies that there is no pair of vertices $u,v$
  with $\CC(u)=\CC(v)$ and $u\prec_N v$ or $v\prec_N u$, i.e.,
  $\CC(u)=\CC(v)$ implies $u=v$. Therefore and by
  Lemma~\ref{lem:inclusion}, we have $\CC(u)\subseteq \CC(v) \iff
  u\preceq_N v$. By Prop.~\ref{prop:Baroni}, therefore, $N$ is regular.
\end{proof}

By Thm.~\ref{thm:semiregular}, every regular network is semi-regular and thus,
satisfies (PCC). This together with Cor.~\ref{cor:cluster-n-with-strongprec}
and Def.~\ref{def:PCC} implies
\begin{corollary} \label{cor:cluster-n-with-PCC}
  For every clustering system $\mathscr{C}$ there is a network $N$ with
  $\mathscr{C}_N = \mathscr{C}$ that that satisfies (PCC).
\end{corollary}

\begin{proposition}\label{prop:N-regular-someC}
  Let $N$ be a network on $X$.  Then, $N\sim \Hasse[\mathscr{C}]$ for some
  clustering system $\mathscr{C}\subseteq 2^X$ if and only if $N$ is
  regular.
\end{proposition}
\begin{proof}
  Assume first that $N\sim \Hasse[\mathscr{C}]$.  In this case, $N$ is
  shortcut-free, satisfies (PCC), and has no vertex with outdegree $1$
  (since $\Hasse[\mathscr{C}]$ has these properties). By
  Thm.~\ref{thm:semiregular}, $N$ is regular.  Assume now that $N$ is
  regular. By Def.~\ref{def:regular-N}, $N\sim
  \Hasse[\mathscr{C_N}]$. Thus, we can put $\mathscr{C} = \mathscr{C}_N$
  and obtain $N\sim \Hasse[\mathscr{C}]$.
\end{proof}

It should be noted that $N\sim \Hasse[\mathscr{C}]$ does not necessarily
imply that $\mathscr{C}_N = \mathscr{C}$. By way of example, consider a
binary phylogenetic rooted tree $T$ on $X$ with $\mathscr{C}_T = \{\{1\},
\{2\}, \{3\}, \{1,2\}, \{1,2,3\}\}$ and the clustering system $\mathscr{C}
= \{\{1\}, \{2\}, \{3\}, \{2,3\}, \{1,2,3\}\}$. It can easily be verified
that $T\sim \Hasse[\mathscr{C}]$ although $\mathscr{C}_T\neq \mathscr{C}$.
Nevertheless, Prop.\ \ref{prop:N-regular-someC} together with
Def.~\ref{def:regular-N} immediately implies

\begin{corollary}\label{cor:hasse-iso}
  $N\sim \Hasse[\mathscr{C}_N]$ if and only if $N$ is regular.
\end{corollary}

A crucial link between a network and its clustering systems is the ability
to identify the non-trivial blocks. The following result shows that, at
least in semi-regular networks, key information is provided by the overlaps
of clusters.
\begin{lemma}
\label{lem:semiregular-overlap}
  Let $B$ be a non-trivial block in a semi-regular network $N$.  Then for
  every $u\in B^0$ there is a $v\in B^0$ such that $\CC(u)$ and $\CC(v)$
  overlap.
\end{lemma}
\begin{proof}
  Suppose $u\in B^0$ and consider the two disjoint sets $A\coloneqq\{w\in
  V(B) \mid u\prec_{N} w\}$ and $D\coloneqq\{w\in V(B) \mid w\prec_{N}
  u\}$, i.e., the ancestors and descendants, resp., of $u$ in $B$.  Note
  that both sets are non-empty since $u\in B^0$.  There is no arc
  connecting a vertex in $A$ with a vertex in $D$. To see this,
  consider $a\in A$ and $d\in D$. Since $d\prec_{N} u \prec_{N} a$ and $N$
  is acyclic, there is a directed path $P$ from $a$ to $d$ passing through
  $u$.  Thus an arc $(a,d)$ would be a shortcut, contradicting that $N$ is
  semi-regular, and an arc $(d,a)$ would imply $a\prec_{N} d$
  contradicting $d\prec_{N} u \prec_{N} a$.  Since $a,d\in V(B)$ and $B$ is
  a non-trivial block, $a$ and $d$ lie on an undirected cycle $K$ in $B$.
  In particular, they are connected by two undirected paths that do not
  share any inner vertices. Thus there is an undirected path $P=(d=v_1,
  v_2,\dots,a=v_k)$ that does not contain $u$.  Let $v_i$ be the unique
  vertex in $P$ such that $v_i\in D$ and there is no vertex $v_j\in D$ with
  $i<j\le k$. Such a vertex exists since $v_1\in D$. Moreover,
  $v_k=a\in A$ implies $i< k$.  Thus consider the vertex $v\coloneqq
  v_{i+1}$. We have $v\notin D$ by construction and $v\notin A$ since
  $v_i\in D$ is not adjacent to any vertex in $A$. Since $P$ does not contain
  $u$ and $v\in V(B)\setminus (A\cupdot D)$, we see that $u$ and $v$ are
  $\preceq_{N}$-incomparable.  Since $v_i\in D$, we
  have $v_i\prec_N u$. Hence, $(v_i,v_{i+1})$ cannot be an arc in $N$
  since otherwise $v_{i+1}\prec_{N} v_i\prec_N u$ would imply $v_{i+1}\in
  D$; a contradiction. Therefore, it have $(v, v_i)=(v_{i+1}, v_i)\in E(N)$
  and thus $v_i\prec_{N} u,v$.  By Lemma~\ref{lem:inclusion}, this implies
  $\emptyset\ne \CC_{N}(v_i)\subseteq \CC_{N}(u)\cap\CC_{N}(v)$.  Together
  with (PCC) and the fact that $u$ and $v$ are $\preceq_{N}$-incomparable,
  this yields that $\CC(u)$ and $\CC(v)$ overlap.  In particular, $v\ne
  \max B$ since $u$ and $v$ are $\preceq_{N}$-incomparable and $v$ is not a
  terminal vertex since $v_{i}\prec_{N} v$.  Therefore, we have $v\in B^0$.
\end{proof}
We note that semi-regularity cannot be omitted in
Lemma~\ref{lem:semiregular-overlap}, since the statement is not true for
the rooted $K_{3,3}$ of Fig.~\ref{fig:no-PCC}.  Lemma~\ref{lem:overlap-B0}
and Lemma~\ref{lem:semiregular-overlap} together show that in semi-regular
networks all vertices in the interior of non-trivial blocks are identified
by the fact their clusters overlap. It remains an open question, however,
whether the information of cluster overlaps is sufficient to identify the
non-trivial blocks.

We continue by showing that, whenever (PCC) is satisfied, least-resolved
networks are precisely the regular network, To this end, we consider first
the implications given by the operations $\expand$ and $\contract$, and by
the removal of short-cuts, respectively.
\begin{lemma}
  \label{lem:expand-PCC}
  Let $N$ be a network, $w\in V(N)$, and $N'$ the network obtained from $N$
  by applying $\expand(w)$. Then $N$ satisfies (PCC) if and only if $N'$
  satisfies (PCC).
\end{lemma}
\begin{proof}
  By Lemma~\ref{lem:expand}, $N$ and $N'$ are $(N',N)$-ancestor-preserving,
  i.e., $v\preceq_{N} v'$ if and only if $v\preceq_{N} v'$ holds for all
  $v,v'\in V(N)$.  By Lemma~\ref{lem:expand},
  $\CC_N(v)=\CC_{N'}(v)$ for all $v\in V(N)\subseteq V(N')$.  Let $w'$ be
  the unique vertex in $V(N')\setminus V(N)$. By construction, $w$ is the
  unique child of $w'$ in $N'$ and thus, by
  Obs.~\ref{obs:outdeg-1-cluster}, $\CC_{N}(w)=\CC_{N'}(w)=\CC_{N'}(w')$.

  Suppose first $N'$ satisfies (PCC), i.e., for all $u, v\in V(N')$, it
  holds that $u$ and $v$ are $\preceq_{N'}$-comparable if and only if
  $\CC_{N'}(u)\subseteq \CC_{N'}(v)$ or $\CC_{N'}(v)\subseteq \CC_{N'}(u)$.
  To see that $N$ also satisfies (PCC), consider $u, v\in V(N)\subseteq
  V(N')$.  Suppose $u$ and $v$ are $\preceq_{N}$-comparable. Hence, $u$ and
  $v$ are also $\preceq_{N'}$-comparable, and thus
  $\CC_{N}(u)=\CC_{N'}(u)\subseteq \CC_{N'}(v)=\CC_{N}(v)$ or
  $\CC_{N}(v)=\CC(v)_{N'}\subseteq C_{N'}(u)=\CC_{N}(u)$.  Conversely, if
  $\CC_{N}(u)\subseteq \CC_{N}(v)$ or $\CC_{N}(v)\subseteq \CC_{N}(u)$,
  then also $\CC_{N'}(u)\subseteq \CC_{N'}(v)$ or $\CC_{N'}(v)\subseteq
  \CC_{N'}(u)$. Hence, $u$ and $v$ are $\preceq_{N'}$-comparable and thus
  also $\preceq_{N}$-comparable.

  Now suppose $N$ satisfies (PCC).  By similar argument as above, it holds,
  for all $u, v\in V(N)=V(N')\setminus \{w'\}$ that $u$ and $v$ are
  $\preceq_N$-comparable if and only if $\CC_N(u)\subseteq \CC_N(v)$ or
  $\CC_N(v)\subseteq \CC_N(u)$.  To show that $N'$ satisfies (PCC), it
  therefore only remains to consider $w'$ and some vertex $v\in V(N)$.  It
  holds that $v$ and $w'$ are $\preceq_{N'}$-comparable if and only if $v$
  and $w$ are $\preceq_{N'}$-comparable.  To see this, suppose first $v$
  and $w'$ are $\preceq_{N'}$-comparable. If $v\prec_{N'} w'$, then
  $v\preceq_{N'} w$ since $w$ is the unique child of $w'$. If $w'\prec_{N'}
  v$, then $w\prec_{N'} w'$ implies $w\prec_{N'} v$.  Now suppose $v$ and
  $w$ are $\preceq_{N'}$-comparable.  If $v\preceq_{N'} w$, then
  $w\prec_{N'} w'$ implies $v\prec_{N'} w'$.  If $w\prec_{N'} v$, then
  $w'\preceq_{N'} v$ (and thus $w'\prec_{N'} v$) since $w'$ is the unique
  parent of $w$ in $N'$.  Taken together and since $v,w\in V(N)$, the
  arguments so far imply that $v$ and $w'$ are $\preceq_{N'}$-comparable if
  and only if $v$ and $w$ are $\preceq_{N'}$-comparable if and only if
  $\CC_{N'}(v)\subseteq \CC_{N'}(w)=\CC_{N'}(w')$ or
  $\CC_{N'}(w')=\CC_{N'}(w)\subseteq \CC_{N'}(v)$ In summary, therefore,
  $N'$ satisfies (PCC).
\end{proof}

Cor.~\ref{cor:expand-shortcuts} and Lemma~\ref{lem:expand-PCC} imply that
semi-regularity is preserved by $\expand$ applied on arbitrary vertices and
$\contract(u,w)$ applied to arcs $(u,w)$ where $\outdeg(u)=1$.
\begin{corollary}
  \label{cor:expand-semiregular}
  Let $N$ be a network, $w\in V(N)$, and $N'$ the network obtained from $N$
  by applying $\expand(w)$.
  Then $N$ is semi-regular if and only if $N'$ is semi-regular.
\end{corollary}

\begin{proposition}
  \label{prop:edit-N-to-Hasse}
  Let $N$ be a network satisfying (PCC).  The unique regular network
  $\Hasse[\mathscr{C}_N]$ is obtained from $N$ by repeatedly executing one of
  the following operations until neither of them is possible:
  \begin{description}
  \item[(1)] remove a shortcut $(u,w)$, and
  \item[(2)] apply $\contract(u,w)$ for an arc $(u,w)$ with
    $\outdeg(u)=1$.
  \end{description}
\end{proposition}
\begin{proof}
  Let $N'$ be the network obtained by applying one of the following
  operations until neither of them is possible.  By construction, we have
  $V(N')\subseteq V(N)$.  We can repeatedly apply
  Lemma~\ref{lem:shortcut-deletion} and \ref{lem:outdeg-1-contraction}
  to conclude that $N'$ is a network with leaf set $X$ and that, for all
  $v,v'\in V(N')$, it holds $v\preceq_{N} v'$ if and only if $v\preceq_{N'}
  v'$.  Similarly, Lemma~\ref{lem:shortcut-deletion} and
  \ref{lem:outdeg-1-contraction} imply that $\CC_N(v)=\CC_{N'}(v)$
  holds for all $v\in V(N')$ and, in particular,
  $\mathscr{C}_N=\mathscr{C}_{N'}$.

  By assumption, $N$ satisfies (PCC), i.e., for all $u, v\in V(N)$, it
  holds that $u$ and $v$ are $\preceq_N$-comparable if and only if
  $\CC_N(u)\subseteq \CC_N(v)$ or $\CC_N(v)\subseteq \CC_N(u)$.  To see
  that $N'$ also satisfies (PCC), consider $u, v\in V(N')$.  Suppose $u$
  and $v$ are $\preceq_{N'}$-comparable. Hence, $u$ and $v$ are also
  $\preceq_{N}$-comparable, and thus $\CC_{N'}(u)=\CC_N(u)\subseteq
  \CC_N(v)=\CC_{N'}(v)$ or $\CC_{N'}(v)=\CC_N(v)\subseteq
  \CC_N(u)=\CC_{N'}(u)$.  Conversely, if $\CC_{N'}(u)\subseteq \CC_{N'}(v)$
  or $\CC_{N'}(v)\subseteq \CC_{N'}(u)$, then also $\CC_{N}(u)\subseteq
  \CC_{N}(v)$ or $\CC_{N}(v)\subseteq \CC_{N}(u)$. Hence, $u$ and $v$ are
  $\preceq_N$-comparable and thus also $\preceq_{N'}$-comparable.
  Therefore, $N'$ satisfies (PCC).  Since moreover $N'$ is shortcut-free by
  construction, $N'$ is semi-regular.  This together with
  Thm.~\ref{thm:semiregular} and the fact that, by construction, there is
  no vertex $v$ with $\outdeg_{N'}(v)=1$ implies that $N'$ is regular.
  This, together with $\mathscr{C}_N=\mathscr{C}_{N'}$, implies that $N'$
  is the unique regular network $\Hasse[\mathscr{C}]$.
\end{proof}

\begin{proposition}
  \label{prop:lr-and-PCC->regular}
  Let $N$ be a least-resolved network satisfying (PCC). Then $N$ is the
  unique regular network for $\mathscr{C}_N$.
\end{proposition}
\begin{proof}
  By Cor.~\ref{cor:lrN-sf-out1}, $N$ is shortcut-free (and thus, 
  semi-regular) and contains no vertex with outdegree $1$.
  By Thm.~\ref{thm:semiregular}, therefore, $N$ is regular.
\end{proof}
As a consequence of 
Cor.~\ref{cor:regular-least-resolved}, Thm.~\ref{thm:semiregular}, and
Prop.~\ref{prop:lr-and-PCC->regular}, we obtain
\begin{theorem}
  \label{thm:PCC-reg-iff-lr}
  A network $N$ is regular if and only if $N$ is least-resolved and satisfies 
  (PCC).
\end{theorem}

\subsection{Separated Networks and Cluster Networks}
\label{ssec:separated}

Recall that a network $N$ is \emph{separated} if all hybrid vertices have
outdegree $1$.
We have already seen above that the Hasse diagrams of the clustering
systems cannot produce separated networks with hybrid vertices
because $\outdeg(v)=1$ implies that $\CC(v)=\CC(u)$ whenever $u$ is
the only child of $v$. By Thm.~\ref{thm:semiregular}, furthermore, a
regular network does not have any vertex with outdegree~$1$.  Therefore, a
regular network cannot be separated whenever it contains at least one
hybrid vertex and \textit{vice versa}.

The \emph{Cluster-popping} algorithm \cite{Huson:08} constructs a separated
network for a given clustering system $\mathscr{C}$ by first constructing the
Hasse diagram, and thus the unique regular network $\Hasse[\mathscr{C}]$, and
then applying $\expand(w)$ to all hybrid vertices $w\in
V(\Hasse[\mathscr{C}])$. In particular, the resulting network is a so-called
\emph{cluster network} \cite{Huson:08,Huson:11}:
\begin{definition}
  \label{def:cluster-network}
  A network $N$ is a \emph{cluster network} if (i) it satisfies (PCC), and,
  for all $u,v\in V(N)$, it holds
  \begin{enumerate}[noitemsep]
    \item[(ii)] $\CC_{N}(u)= \CC_{N}(v)$ if and only if $u= v$ or $v$ is a
      hybrid vertex and parent of $u$ or \textit{vice versa},
    \item[(iii)] if $u$ is a child of $v$, then there exists no node $w$
      with $\CC_{N}(u) \subsetneq \CC_{N}(w) \subsetneq \CC_{N}(v)$, and
    \item[(iv)] every hybrid vertex $v$ has exactly one child, which is a
      tree node.
  \end{enumerate}
\end{definition}
We note that the definition of cluster networks usually is expressed using
the following condition instead of (PCC):
\begin{itemize}[noitemsep]
\item[(i')] $\CC_{N}(u)\subseteq \CC_{N}(v)$ if and only if
  $u\preceq_{N} v$ for all
  $u,v\in V(N)$ \cite{Huson:08,Zhang:19}.
\end{itemize}
However, this contradicts the existence of hybrid vertices $v$, which are
required to have exactly one child $u$ by (iv). To see this, observe that
$u\prec_{N} v$ and, by Obs.~\ref{obs:outdeg-1-cluster}, we have
$\CC_{N}(u)= \CC_{N}(v)$. The latter means that $\CC_{N}(v)\subseteq_{N}
\CC_{N}(u)$ is satisfied and thus (i') implies $v\preceq_{N} u$; a
contradiction.

We can rephrase conditions~(i)-(iv), and thus the definition of cluster
networks, as follows:
\begin{theorem}
  \label{thm:cluster-network-charac}
  A network $N$ is a cluster network if and only if it is semi-regular,
  separated, and phylogenetic.
\end{theorem}
\begin{proof}
  Suppose first that $N$ is a cluster network, i.e., it satisfies
  conditions~(i)-(iv) in Def.~\ref{def:cluster-network}. By condition~(i)
  and~(ii), resp., $N$ satisfies (PCC) and is separated.  Hence, it remains
  to show that $N$ is shortcut-free and phylogenetic.  Suppose, for
  contradiction, that $(v,u)$ is a shortcut in $N$.  Then there is $w\in
  \child_{N}(v)\setminus\{u\}$ and $w'\in\parent_{N}(u)\setminus \{v\}$
  such that $u\prec_{N}w'\preceq_{N} w$.  By Lemma~\ref{lem:inclusion},
  $\CC_{N}(u)\subseteq \CC_{N}(w') \subseteq \CC_{N}(w) \subseteq
  \CC_{N}(v)$.  If $\CC_{N}(w)=\CC_{N}(v)$, then conditions~(ii) and~(iv)
  imply that $v$ is a hybrid vertex with exactly one child; a
  contradiction. Therefore, $\CC_{N}(w)\subsetneq\CC_{N}(v)$ must hold.  If
  $\CC_{N}(u)=\CC_{N}(w')$, then conditions~(ii) and~(iv) imply that $w'$
  is a hybrid vertex and its unique child $u$ is a tree node, contradicting
  that $w',v\in \parent_{N}(u)$. Hence, we have $\CC_{N}(u)\subsetneq
  \CC_{N}(w') \subseteq \CC_{N}(w) \subsetneq \CC_{N}(v)$, which
  contradicts~(iii).  Therefore, $N$ must be shortcut-free and thus
  semi-regular.  Suppose, for contradiction, that $N$ is not
  phylogenetic. Hence, there is a vertex $v$ with exactly one child $u$ and
  $\indeg_{N}(v)\le 1$. By Obs.~\ref{obs:outdeg-1-cluster}, we have
  $\CC_{N}(u)= \CC_{N}(v)$ and thus $v$ must be a hybrid vertex by~(ii); a
  contradiction.

  Conversely, suppose $N$ is semi-regular, separated, and phylogenetic.
  Hence, $N$ satisfies~(PCC), i.e., condition~(i).  Condition~(iii) must be
  satisfied since an arc $(v,u)$ with $\CC_{N}(u) \subsetneq \CC_{N}(w)
  \subsetneq \CC_{N}(v)$ for some $w\in V(N)$ would be a shortcut by
  Obs.~\ref{obs:subsetneq-implies-below}.  We continue with
  showing~(ii). Suppose $\CC_{N}(u)= \CC_{N}(v)$ and $u\ne v$.  By (PCC),
  it holds $u\prec_{N} v$ or $v\prec_{N} u$.  Suppose w.l.o.g.\ that
  $u\prec_{N} v$. Then Lemma~\ref{lem:outdeg1} implies that
  $\outdeg_{N}(v)=1$ and $u$ is the unique child of $v$. Since $N$ is
  phylogenetic, $v$ must be a hybrid vertex.  Conversely, a hybrid vertex
  $v$ in a separated network has exactly one child $u$ implying
  $\CC_{N}(u)= \CC_{N}(v)$ by Obs.~\ref{obs:outdeg-1-cluster}.  Since $N$
  is separated, it remains to show that the unique child $u$ of a hybrid
  vertex $v$ is a tree node.  Suppose for contradiction that $u$ is a
  hybrid node. Then $u$ again has a unique child $w$. Hence, we have
  $w\prec_{N} u\prec_{N} v$ and, by Obs.~\ref{obs:outdeg-1-cluster}, it
  holds $\CC_{N}(w)=\CC_{N}(u)= \CC_{N}(v)$.  By~(ii), this implies that
  $(v,w)$ is an arc in $N$ and, in particular, a shortcut; a
  contradiction.  Therefore, condition~(iv) is also satisfied.
\end{proof}

We shall see in Thm.~\ref{thm:unique-cluster-network} in the following
section that cluster networks are uniquely determined by their cluster
sets. To obtain this result it will be convenient to make use of the fact
the semi-regular networks are encoded by their multisets of clusters.

\subsection{Cluster Multisets of Semi-Regular Networks}
\label{sec:Multi-SemiReg}

\begin{lemma}
  \label{lem:semireg-1or2-multiplicity}
  Let $N$ be a semi-regular phylogenetic network.  Then, the multiplicity
  of each cluster $C\in\mathscr{C}$ in the cluster multiset $\mathscr{M}_N$
  is either one or two.  In the latter case, the two distinct vertices
  $u,v\in V(N)$ with $\CC_{N}(u)=\CC_{N}(v)=C$ are adjacent.
\end{lemma}
\begin{proof}
  Recall that a semi-regular network $N$ satisfies (PCC) and is
  shortcut-free. Let $C\in\mathscr{C}$. Thus, there is at least one vertex
  $v\in V(N)$ with $\CC_{N}(v)=C$.  Now suppose there is $u\in
  V(N)\setminus \{v\}$ with $\CC_{N}(u)=C$.  By (PCC), it holds $u\prec_{N}
  v$ or $v\prec_{N} u$.  Suppose that $u\prec_{N} v$. Then
  Lemma~\ref{lem:outdeg1} implies that $\outdeg_{N}(v)=1$ and $u$ is the
  unique child of $v$.  Suppose, for contradiction, there is a third vertex
  $w\in V(N)\setminus \{u,v\}$.  By (PCC), it holds $w\prec_{N} v$ or
  $v\prec_{N} w$.  If $w\prec_{N} v$, then Lemma~\ref{lem:outdeg1} implies
  that $w$ is the unique child of $v$; a contradiction. If $v\prec_{N} w$,
  then we have also $u\prec_{N} v\prec_{N} w$. By Lemma~\ref{lem:outdeg1},
  therefore, $v$ and $u$ are both the unique child of $w$, which is not
  possible since $u\ne v$.  One argues similarly if $v\prec_{N} u$. In
  particular, $u$ and $v$ are adjacent in both cases.
\end{proof}

As we shall see in Theorem~\ref{thm:semireg-multiset}, the property of being
phylogenetic, however, is not necessary for semi-regular networks to be
identified by their cluster multisets.  In order to prove this, the
following map will be of useful.
\begin{definition}
  Let $N$ and $\tilde{N}$ be two networks satisfying (PCC) and
  $\mathscr{M}_N=\mathscr{M}_{\tilde{N}}$.
  Then, the map $\varphi_{PCC}\colon V(N)\to V(\tilde{N})$ is given by the
  following steps for all
  $C\in \mathscr{C}_N=\mathscr{C}_{\tilde{N}}$:
  \begin{enumerate}[noitemsep,nolistsep]
  \item[(i)] sort the $k\ge1$ vertices in $N$ with cluster $C$ such that
    $v_1 \prec_N \dots \prec_N v_k$,
  \item[(ii)] sort the $k$ vertices in $\tilde{N}$ with cluster $C$ such
    that $\tilde{v}_1 \prec_{\tilde{N}} \dots \prec_{\tilde{N}}
    \tilde{v}_k$, and
  \item[(iii)] set $\varphi(v_i)=\tilde{v}_i$ for all $1\le i \le k$.
  \end{enumerate}
\end{definition}
In other words, we map the $\preceq_{N}$-larger vertices with cluster $C$
in $N$ to $\preceq_{\tilde{N}}$-larger vertices with cluster $C$ in
$\tilde{N}$, which is possible since, by (PCC), these vertices are totally
ordered w.r.t.\ $\preceq_{N}$ and $\preceq_{\tilde{N}}$, respectively.

\begin{lemma}
  \label{lem:varphi}
  Let $N$ and $\tilde{N}$ be two networks satisfying (PCC) and
  $\mathscr{M}_N=\mathscr{M}_{\tilde{N}}$.
  Then $\varphi_{PCC}$ is a bijection between $V(N)$ and $V(\tilde{N})$ that
  is the identity on the common leaf set $X$.
  Writing $\tilde{v}\coloneqq \varphi_{PCC}(v)$, it moreover holds
  \begin{enumerate}[noitemsep,nolistsep]
  \item $\CC_{N}(v)=\CC_{\tilde{N}}(\tilde{v})$ for all $v\in V(N)$,
  \item $v$ is a leaf if and only if $\tilde{v}$ is a leaf for all
    $v\in V(N)$, and
  \item $u\prec_{N} v$ if and only if $\tilde{u}\prec_{\tilde{N}}
    \tilde{v}$ for all $u,v\in V(N)$.
  \end{enumerate}
\end{lemma}
\begin{proof}
  Since $\mathscr{M}_N=\mathscr{M}_{\tilde{N}}$, the multiplicity of every
  cluster $C\in\mathscr{C}\coloneqq \mathscr{C}_N=\mathscr{C}_{\tilde{N}}$
  is equal in $\mathscr{M}_N$ and $\mathscr{M}_{\tilde{N}}$, i.e., there
  are $k\ge 1$ vertices with cluster $C$ in $N$ and $k$ vertices with
  cluster $C$ in $\tilde{N}$.  One easily verifies that, by construction,
  $\varphi_{PCC}$ is a bijection between $V(N)$ and $V(\tilde{N})$ that is
  the identity on the common leaf set $X$ and satisfies
  $\CC_{N}(v)=\CC_{\tilde{N}}(\tilde{v})$ for all $v\in V(N)$.  To see that
  $u\prec_{N} v$ if and only if $\tilde{u}\prec_{\tilde{N}} \tilde{v}$,
  suppose $u\prec_{N} v$. By Lemma~\ref{lem:inclusion}, this implies
  $\CC_{N}(u)\subseteq \CC_N(v)$. If $\CC_{N}(u)\subsetneq \CC_N(v)$ (and
  thus $\CC_{\tilde{N}}(\tilde{u})\subsetneq \CC_{\tilde{N}}(\tilde{v})$),
  then Obs.~\ref{obs:subsetneq-implies-below} implies
  $\tilde{u}\prec_{\tilde{N}} \tilde{v}$. If, on the other hand,
  $\CC_{N}(u)= \CC_N(v)$, then $\tilde{u}\prec_{\tilde{N}} \tilde{v}$ holds
  by construction of $\varphi$.  Analogously, $\tilde{u}\prec_{\tilde{N}}
  \tilde{v}$ implies $u\prec_{N} v$.  In particular, this implies that, for
  every $v\in V(N)$, $v$ is a leaf if and only if $\tilde{v}$ is a leaf.
  \end{proof}

\begin{theorem}
  \label{thm:semireg-multiset}
  Let $N$ be a semi-regular network. Then $N$ is the unique semi-regular
  network whose cluster multiset is $\mathscr{M}_N$.
\end{theorem}
\begin{proof}
  Suppose $N$ and $\tilde{N}$ are semi-regular networks with
  $\mathscr{M}_N=\mathscr{M}_{\tilde{N}}$.  By assumption, both $N$ and
  $\tilde{N}$ are shortcut-free and satisfy (PCC).  We continue with
  showing that $\varphi_{PCC}\colon V(N)\to V(\tilde{N})$ is a graph
  isomorphism. By Lemma~\ref{lem:varphi}, $\varphi_{PCC}$ is a bijection
  that is the identity on the common leaf set $X$.  In the following, we
  write $\tilde{v}\coloneqq \varphi_{PCC}(v)$ for all $v\in V(N)$.  Suppose
  that $(v,u)\in E(N)$. Thus, we have $u\prec_{N} v$ which implies
  $\tilde{u}\prec_{\tilde{N}} \tilde{v}$ by Lemma~\ref{lem:varphi}(3).
  Assume, for contradiction, that $(\tilde{v},\tilde{u})\notin
  E(\tilde{N})$.  Then there must be $\tilde{z}\in V(\tilde{N})$ such that
  $\tilde{u}\prec_{\tilde{N}} \tilde{z} \prec_{\tilde{N}} \tilde{v}$. By
  Lemma~\ref{lem:varphi}(3), we have $u\prec_{N} z\prec_{N} v$. Hence, the
  arc $(v,u)$ must be a shortcut in $N$; a contradiction. Therefore,
  $(\tilde{v},\tilde{u})\in E(\tilde{N})$.  By analogous arguments,
  $(\tilde{v},\tilde{u})\in E(\tilde{N})$ implies $(v,u)\in E(N)$.  Hence,
  $\varphi_{PCC}$ is a graph isomorphism that is the identity on $X$ and
  thus $N\simeq \tilde{N}$. Therefore, $N$ is the unique semi-regular
  network whose cluster multiset is $\mathscr{M}_N$.
\end{proof}

We emphasize that none of the two conditions (PCC) and shortcut-free that
define semi-regular networks can be omitted in
Theorem~\ref{thm:semireg-multiset} as shown by the examples in
Fig.~\ref{fig:multiset-not-semireg}.

\begin{figure}[t]
  \centering
  \includegraphics[width=\textwidth]{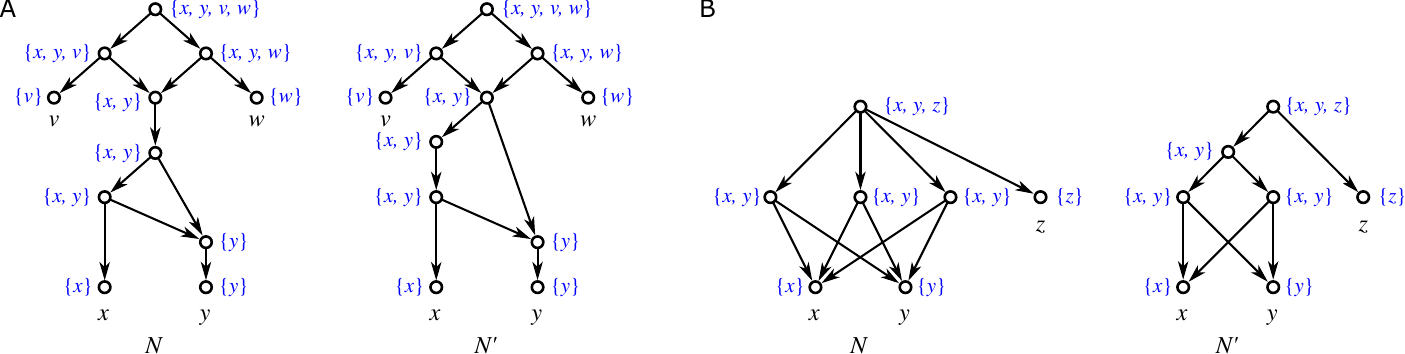}
  \caption{Two pairs of non-isomorphic (phylogenetic) networks $N$ and $N'$
    for which $\mathscr{M}_{N} = \mathscr{M}_{N'}$.  (A) $N$ and $N'$
    satisfy (PCC) but are not shortcut-free.  (B) $N$ and $N'$ are
    shortcut-free but do not satisfy (PCC).}
  \label{fig:multiset-not-semireg}
\end{figure}

It is an easy task to verify that semi-regular networks $N$ (as any other
network for which the property of being phylogenetic has been left out) are
not determined by their clustering systems $\mathscr{C}_N$.  The network
$N$ in Fig.\ \ref{fig:not-TC}(A), for example, is not phylogenetic (but
semi-regular). Suppression of any vertex with in- and outdegree $1$ yields
a network $N'$ with $N\not\simeq N'$ but $\mathscr{C}_N =
\mathscr{C}_{N'}$.  We will see in the following that there is a 1-to-1
correspondence between clustering systems and cluster networks.  To show
this, we will need the following technical result that relates the
occurrence of vertices with equal clusters in such networks to the
structure of the Hasse diagram.

\begin{lemma}
  \label{lem:separated-1or2}
  Let $N$ be a cluster network with clustering system $\mathscr{C}$.  Then,
  for every cluster $C\in\mathscr{C}$, there is either exactly one vertex
  $v\in V(N)$ with $\CC_{N}(v)=C$ or there are exactly two vertices $u,v\in
  V(N)$ with $\CC_{N}(u)=\CC_{N}(v)=C$. The latter case occurs if and only
  if $C$ has indegree at least $2$ in $\Hasse[\mathscr{C}]$.  Moreover, in
  this case, $u$ and $v$ are adjacent in $N$.
\end{lemma}
\begin{proof}
  By Theorem~\ref{thm:cluster-network-charac}, $N$ is phylogenetic,
  separated, and semi-regular, i.e., it satisfies (PCC) and is
  shortcut-free.  By Lemma~\ref{lem:semireg-1or2-multiplicity}, it only
  remains to show that there are two distinct vertices $u,v\in V(N)$ with
  $\CC_{N}(u)=\CC_{N}(v)=C$ if and only if $C$ has indegree at least $2$ in
  $\Hasse[\mathscr{C}]$.

  Suppose $C=\CC_{N}(v)=\CC_{N}(u)$ for two distinct vertices $v,u\in V(N)$
  and assume w.l.o.g.\ that $u\prec_{N} v$. By Lemma~\ref{lem:outdeg1}, $u$
  must be the unique child of $v$. Since $N$ is phylogenetic, this implies
  that $\indeg_{N}(v)\ge 2$.  Thus let $v_1$ and $v_2$ be two distinct
  parents of $v$. Since $N$ is shortcut-free and by
  Obs.~\ref{obs:shortcut}, $v_1$ and $v_2$ are
  $\prec_{N}$-incomparable. Using (PCC), we conclude that $C_1\coloneqq
  \CC_{N}(v_1)$ and $C_2\coloneqq \CC_{N}(v_2)$ are distinct and none of
  them is contained in the other.  Moreover, we have $C\subseteq C_1\cap
  C_2$ by Lemma~\ref{lem:inclusion}, and thus, and $C\subsetneq C_1,C_2$.
  Suppose there is $C'\in \mathscr{C}$ such that $C\subsetneq C'\subsetneq
  C_1$. Let $v'\in V(N)$ be a vertex with $\CC_{N}(v)=C'$.  By
  Obs.~\ref{obs:subsetneq-implies-below}, we have $u\prec_{N} v'\prec_{N}
  v$. Therefore, $(v,u)\in V(N)$ must be a shortcut; a
  contradiction. Hence, there is no $C'\in \mathscr{C}$ with $C\subsetneq
  C'\subsetneq C_1$ and $C_1$ must be a parent of $C$ in
  $\Hasse[\mathscr{C}]$. By similar arguments, $C_2$ is a parent of $C$ in
  $\Hasse[\mathscr{C}]$, which together with $C_1\ne C_2$ implies that $C$
  has indegree at least $2$ in $\Hasse[\mathscr{C}]$.

  Conversely, suppose $C$ has at least two distinct parents $C_1$ and $C_2$
  in $\Hasse[\mathscr{C}]$. Hence, it holds $C\subsetneq C_1$ and
  $C\subsetneq C_2$ and none of $C_1$ and $C_2$ is contained in the other.
  Let $v, v_1, v_2\in V(N)$ be vertices with $C= \CC_{N}(v)$, $C_1=
  \CC_{N}(v_1)$, and $C_2= \CC_{N}(v_2)$.  Clearly, $v$, $v_1$, and $v_2$
  are pairwise distinct.  By Obs.~\ref{obs:subsetneq-implies-below}, we
  have $v \prec_{N} v_1$ and $v \prec_{N} v_2$, i.e., there are a $v_1
  v$-path $P_1$ and a $v_2 v$-path $P_2$ in $N$.  Let $v'$ be the
  $\preceq_{N}$-maximal vertex in $P_1$ that is also a vertex in $P_2$.  We
  distinguish the two cases $v'= v$ and $v'\ne v$.  If $v'= v$, then $v$
  has a parent in each of $P_1$ and $P_2$ which are distinct by
  construction. Thus $v$ is a hybrid vertex. Since moreover $N$ is
  separated, $v$ has a unique child $u$. By
  Obs.~\ref{obs:outdeg-1-cluster}, this implies $\CC_{N}(u)=\CC_{N}(v)=C$.
  Now suppose $v'\ne v$. Lemma~\ref{lem:inclusion} and $v\prec_{N}
  v'\preceq_{N} v_1, v_2$ imply $C\subseteq C'\coloneqq \CC_{N}(v')$,
  $C'\subseteq C_1$ and $C'\subseteq C_2$. Since none of $C_1$ and $C_2$ is
  contained in the other, we must have $C\subseteq C'\subsetneq C_1$ and
  $C\subseteq C'\subsetneq C_2$. Since $C_1$ and $C_2$ are parents of $C$
  in $\Hasse[\mathscr{C}]$, the latter is only possible if $C=C'$.  Hence,
  in both cases, there are two distinct vertices in $N$ with cluster $C$.
  \end{proof}

\begin{theorem}
  \label{thm:unique-cluster-network}
  For every clustering system $\mathscr{C}$, there is a unique cluster
  network $N$ with $\mathscr{C}=\mathscr{C}_N$.  It is obtained from the
  unique regular network $\Hasse[\mathscr{C}]$ of $\mathscr{C}$ by applying
  $\expand(v)$ to all hybrid vertices.  In particular, $N$ is the unique
  semi-regular separated phylogenetic network with clustering system
  $\mathscr{C}$.
\end{theorem}
\begin{proof}
  By Thm.~\ref{thm:semiregular}, the unique regular network
  $\Hasse[\mathscr{C}]$ is shortcut-free, satisfies (PCC), and has no
  vertex with outdegree~$1$.  In particular, $\Hasse[\mathscr{C}]$ is
  phylogenetic.  Now let $N$ be the network obtained from $N$ by repeatedly
  applying $\expand(w)$ to some hybrid vertex whose outdegree is not $1$
  until no such vertex exists.  Clearly this is achieved by applying
  $\expand(w)$ to all hybrid vertices $w\in V(N)$ since they all satisfy
  $\outdeg_{\Hasse[\mathscr{C}]}(w)\ne 1$ and, moreover, no expansion step
  introduces new such vertices but reduces their number by $1$.  We can
  repeatedly (i.e., in each expansion step) apply Lemma~\ref{lem:expand} to
  conclude that the resulting digraph $N$ is a phylogenetic network,
  Lemma~\ref{lem:expand} to conclude that $N$ satisfies
  $\mathscr{C}_{N}=\mathscr{C}_{\Hasse[\mathscr{C}]}=\mathscr{C}$, and
  Cor.~\ref{cor:expand-semiregular} to conclude that $N$ is semi-regular.
  In particular, by construction, all hybrid vertices in $N$ have outdegree
  $1$, i.e., $N$ is separated.  By Theorem~\ref{thm:cluster-network-charac},
  $N$ is a cluster network.

  It remains to show that $N$ is the unique cluster network with clustering
  system $\mathscr{C}$.  To this end, let $\tilde{N}$ be a cluster network
  with $\mathscr{C}=\mathscr{C}_{\tilde{N}}$. By
  Theorem~\ref{thm:cluster-network-charac}, both $N$ and $\tilde{N}$ are
  semi-regular, shortcut-free, and phylogenetic. Moreover, by
  Lemma~\ref{lem:separated-1or2}, for every cluster $C\in\mathscr{C}$, $C$
  has multiplicity $2$ in $\mathscr{M}_N$ if and only if $C$ has indegree
  at least $2$ in $\Hasse[\mathscr{C}]$ if and only if $C$ has multiplicity
  $2$ in $\mathscr{M}_{\tilde{N}}$; and multiplicity $1$ in both
  $\mathscr{M}_N$ and $\mathscr{M}_{\tilde{N}}$ otherwise. Hence, we have
  $\mathscr{M}_N=\mathscr{M}_{\tilde{N}}$. By
  Theorem~\ref{thm:semireg-multiset}, we conclude that $N\simeq \tilde{N}$,
  and thus, $N$ is the unique cluster network.  In particular, by
  Theorem~\ref{thm:cluster-network-charac}, $N$ is the unique semi-regular
  separated phylogenetic network with clustering system $\mathscr{C}$.
\end{proof}

The uniqueness of cluster networks for a given clustering system
$\mathscr{C}$ has been proved in the framework of \emph{reticulate
networks} in \cite[Thm.~3.9]{ALRR:14}, using alternative arguments. A
network $N$ is \emph{reticulate} in the sense of \cite[Def.~2.3]{ALRR:14}
if (a) every hybrid vertex has exactly one child which, moreover, must
be a tree vertex, and (b) if a vertex $v$ has $\outdeg_{N}(v)=1$ and
$\indeg_{N}(v)\le 1$ then $v$ has a unique child and parent, both of which
are hybrid vertices. We observe that the additional condition that $N$
is reticulate does not affect cluster networks, and thus
Thm.~\ref{thm:unique-cluster-network} and \cite[Thm.~3.9]{ALRR:14} are
equivalent.
\begin{proposition}
  If $N$ is a cluster network, then $N$ is reticulate.
\end{proposition}
\begin{proof}
  By Theorem~\ref{thm:cluster-network-charac}, $N$ is semi-regular,
  phylogenetic, and separated. Since $N$ is phylogenetic, it does not
  contain a vertex satisfying condition~(b), and hence (b) is satisfied
  trivially.  Since $N$ is separated, every hybrid vertex $v$ has exactly
  one child $u$. By Obs~\ref{obs:outdeg-1-cluster}, we have
  $\CC(v)=\CC(u)$. Suppose that $u$ is also a hybrid vertex, i.e., there is
  $v'\in\parent_{N}(u)\setminus\{v\}$. Since $N$ is shortcut-free, $v$ and
  $v'$ are $\preceq_{N}$-incomparable. However, by
  Lemma~\ref{lem:inclusion}, we have $\CC(v)=\CC(u) \subseteq \CC(v')$ and
  thus (PCC) implies that $v$ and $v'$ must be $\preceq_{N}$-comparable; a
  contradiction. Hence, $N$ also satisfies condition~(a).  is also
  satisfied for any cluster network.
\end{proof}

\section{Tree-Child, Normal, and Tree-Based Networks}
\label{sec:tree-child-normal-tree-based}

\begin{definition}
  \textnormal{\cite{Cardona:09}}
  \label{def:tree-child}
  A  network $N$ has the \emph{tree-child} property if, for every
  $v\in V^0$, there is a ``tree-child'', i.e., $u\in\child(v)$ with
  $\indeg(u)=1$.
\end{definition}

Tree-child networks are not necessarily phylogenetic, see
Fig.~\ref{fig:not-TC}(A) for an example. As a shown in
\cite[Lemma~2]{Cardona:09}, $N$ is a tree-child network if and only if
every vertex $v\in V$ has a strict descendant, i.e., a leaf $x\in X$ such
that every directed path from the root $\rho_N$ to $x$ contains $v$.

\begin{lemma}
  \label{lem:TC-incomp-overlap}
  Suppose $N$ is tree-child and $u$ and $v$ are $\preceq_N$-incomparable.
  Then there is a vertex $x\in \CC(v)$ such that $x\notin \CC(u)$ and $y\in
  \CC(u)$ such
  that $y\notin \CC(v)$.
\end{lemma}
\begin{proof}
  Since $N$ is tree-child, there is a strict descendant $x$ of $v$.  It
  satisfies $x\in \CC(v)$ and every path from the root $\rho_N$ to $x$ runs
  through $v$. Now suppose $x\in \CC(u)$. Then there is a directed path
  from $\rho_N$ to $x$ that contains $u$. Since any such path also contains
  $v$, the vertices $u$ and $v$ must be $\preceq_N$-comparable; a
  contradiction.  Thus $x\notin \CC(u)$. The same argument shows that there
  is $y\in \CC(v)$ with $y\notin \CC(u)$.
\end{proof}

\begin{corollary}
  \label{cor:TC->PCC}
  Every tree-child network satisfies (PCC).
\end{corollary}
\begin{proof}
  Let $N$ be a tree-child network and $u,v\in V(N)$.  By
  Lemma~\ref{lem:inclusion}, $u\preceq_N v$ implies
  $\CC(u)\subseteq \CC(v)$. On the other hand, if $u$ and $v$ are
  $\preceq_N$-incomparable, then Lemma~\ref{lem:TC-incomp-overlap} implies
  that either $\CC(u)\cap \CC(v)=\emptyset$ or $\CC(u)$ and $\CC(v)$
  overlap and thus neither $\CC(u)\subseteq \CC(v)$ nor
  $\CC(v)\subseteq \CC(u)$ is satisfied.
\end{proof}
The converse is not true. Fig.~\ref{fig:not-TC}(B) shows an example of a
network that satisfies (PCC) but does not have the tree-child property.

\begin{figure}
  \begin{center}
    \includegraphics[width=0.6\textwidth]{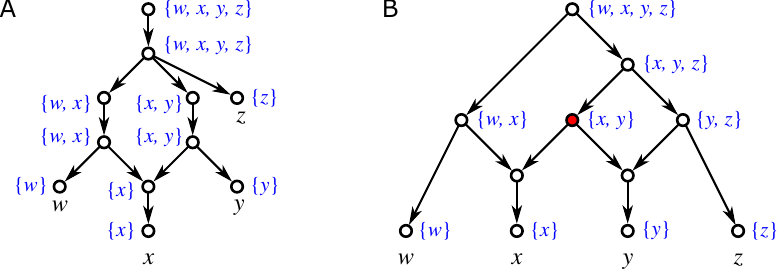}
  \end{center}
  \caption{(A) Example of a network that is semi-regular and
  	tree-child but
    not phylogenetic.  (B) Example of a cluster network, thus satisfying
    (PCC), that is not tree-child. The central node marked in red does not
    have a tree-child. One easily checks that it satisfies (PCC) since
    vertices associated with overlapping pairs of clusters are
    incomparable.}
  \label{fig:not-TC}
\end{figure}

\begin{definition}\label{def:normal}
  A  network is \emph{normal} if it is tree-child and shortcut-free.
\end{definition}

Willson \cite{Willson:10b} studies normal networks in a somewhat different
setting, in which $X$ comprises not only the leaves in our sense but also
the root and all vertices with outdegree $1$. Under this assumption,
\cite[Thm.~3.4]{Willson:10b} states that ``$N$ is regular whenever it is
normal''. The absence of vertices with outdegree $1$ can be included as an
extra condition. The analog of Willson's result in our setting follows
immediately from Cor.~\ref{cor:TC->PCC}, Thm.~\ref{thm:semiregular},
and the absence of shortcuts:
\begin{corollary}
  \label{cor:normal->semireg}
  Let $N$ be a  network. If $N$ is normal, then $N$ is semi-regular.
  If, in addition, there are no vertices with outdegree $1$, then $N$ is
  regular.
\end{corollary}
The converse is not true, there are (semi-)regular networks that are not
normal, see Fig.~1 of \cite{Willson:10b} and also Fig.~\ref{fig:not-TC}(B)
for a semi-regular example. Hence, we have
\begin{remark}
  Not every semi-regular network, and in particular not every cluster
  network, is normal.
\end{remark}

\begin{proposition}
  \label{prop:unique-normal}
  Let $\mathscr{C}$ be a clustering system. If there is a phylogenetic,
  separated, normal network $N$ with $\mathscr{C}=\mathscr{C}_N$, then $N$
  is unique w.r.t.\ these properties. In particular, $N$ is the unique
  cluster network with $\mathscr{C}=\mathscr{C}_N$.
\end{proposition}
\begin{proof}
  Suppose $N$ is phylogenetic, separated, and normal and satisfies
  $\mathscr{C}=\mathscr{C}_N$. By Cor~\ref{cor:normal->semireg}, $N$ is
  semi-regular and thus, by Theorem~\ref{thm:cluster-network-charac}, a
  cluster network.  By Thm.~\ref{thm:unique-cluster-network}, $N$ is
  unique.
\end{proof}

From Prop.~\ref{prop:unique-normal} and the definition of binary networks,
we immediately obtain
\begin{corollary}\label{cor:normal-unique-new}
  Let $N$ be a binary normal network. Then $N$ is the unique binary normal
  network whose cluster set is $\mathscr{C}_N$. In particular, $N$ is a
  cluster network.
\end{corollary}
The following result appears to be well known, see e.g.\
\cite{Murakami:19}. An argument for binary tree-child and level-1
networks can be found \cite{huber2019orienting}. A direct proof for our
more general setting is included here for completeness.
\begin{proposition}
  \label{prop:phy-level1->TC}
  Every phylogenetic level-1 network is tree-child.
\end{proposition}
\begin{proof}
  Let $N$ be a phylogenetic level-$1$ network.  If $u$ is a hybrid vertex,
  then, by Lemma~\ref{lem:hybrid-properly-contained}, there is a
  non-trivial block $B_v$ that contains $v$ and all its parents.  Suppose,
  for contradiction, there is a non-leaf vertex $v$ whose children are all
  hybrid vertices.  Suppose first that $\outdeg(v)=1$. Since $N$ is
  phylogenetic, this implies that $v$ is hybrid vertex.  Let $u$ be the
  unique child of $v$, which is a hybrid vertex by assumption. The hybrid
  vertices $u$ and $v$ are contained in a common non-trivial block $B_u$.
  In particular, $u\prec_{N} v$ implies $u\ne \max B_u$. Additionally,
  $v\ne \max B_u$ since $\outdeg(v)=1$ but $\outdeg(\max B_u)>1$ by
  Lemma~\ref{cor:maxB-outdegree}; contradicting that $N$ is level-1.  Now
  suppose that $\outdeg(v)\ge 2$. Since all $u_i\in\child(v)$ are hybrid
  vertices, $v$ and $u_i$ are contained in blocks $B_i\coloneqq B_{u_i}$.
  Note that $u_i\ne \max B_i$. If there are two distinct
  $u_i,u_j\in\child(v)$ such that $v\notin \{\max B_i, \max B_j\}$, then
  Lemma~\ref{lem:block-identity} implies $B_i=B_j$ and thus $B_i$ contains
  the hybrid vertices $u_i,u_j\ne \max B_i$; this contradicts that $N$ is
  level-1. Otherwise, there is at least one $u_i\in\child(v)$ such that
  $v=\max B_i$. Since $u_i$ is a hybrid vertex, there is $w_i\in
  \parent_{N}(u_i)\setminus \{v\}$, which is also contained in
  $B_i$. Hence, we have $u_i\prec_{N} w_i\prec_{N} \max B_i = v$. Therefore,
  and because $N$ is acyclic, there is $u_j\in\child_{N}(v)$ such that
  $w_i\preceq_{N} u_j\prec_{N} v$. By assumption, $u_j$ is a hybrid vertex
  and moreover $u_j\notin \{u_i, v=\max B_i \}$. By
  Lemma~\ref{lem:block-prec-sandwich}, $w_i\preceq_{N} u_j\prec_{N} v$ and
  $w_i,v\in V(B_i)$ imply $u_j\in V(B_i)$.  Hence, $B_i$ contains two
  distinct hybrid vertices $u_i, u_j\ne\max B_i$, contradicting that $N$ is
  level-1.
\end{proof}

Note that ``phylogenetic'' cannot be omitted in
Prop.~\ref{prop:phy-level1->TC}: Consider a tree-vertex $v$ with a hybrid
child $u\in\child(v)$. Subdivision of the arc $(v,u)$ creates a new tree
vertex $u'\in\child(v)$ with the hybrid $u$ as its only child.  The
modified network is still level-$1$ but no longer tree-child.

Next we consider the overlapping clusters in tree-child networks in
some more detail:
\begin{lemma}
  \label{lem:tree-child-intersections}
  Let $N$ be a tree-child network and suppose $u$ and $v$ are
  $\preceq_N$-incomparable. Then either $\CC(u)\cap \CC(v)=\CC(h)$ for some
  hybrid vertex $h\in V(N)$ or $\CC(u)\cap \CC(v)\notin\mathscr{C}_N$.
\end{lemma}
\begin{proof}
  By Cor.~\ref{cor:TC->PCC}, $N$ satisfies (PCC) and thus either
  $\CC(u)\cap \CC(v)=\emptyset$, in which case the assertion is obviously
  true, or $\CC(u)$ and $\CC(v)$ overlap. In the latter case,
  Lemma~\ref{lem:lower-path} implies that $u,v$ are contained in a common
  non-trivial block $B$.  Set $A\coloneqq \CC(u)\cap \CC(v)$ and assume
  $A\ne \CC(h)$ for any hybrid vertex $h\in V(N)$. By
  Lemma~\ref{lem:uvmin}, this implies that $N$ is level-$k$ for some $k\ge
  2$.  Lemma~\ref{lem:union-hybrid-clusters} implies that $A=\bigcup_{h\in
    H} \CC(h)$ for some set $H$ of $\vert H\vert\ge 2$ hybrid vertices.
  Furthermore, we have $\CC(h)\subsetneq A$ for all $h\in H$.  Now suppose,
  for contradiction, that there is a non-hybrid vertex $w\in V(N)$ such
  that $\CC(w)=\bigcup_{h\in H} \CC(h)$. Then, for all $h\in H$, we have
  $\CC(h)\subsetneq \CC(w)$, which together with
  Cor.~\ref{obs:subsetneq-implies-below} implies $h\prec_N w$. Moreover,
  all elements in $\CC(w)$ are descendants of one of the hybrid vertices in
  $H$.  Since $N$ is tree-child, there is a leaf $x$ that is reachable from
  $w$ along a path consisting entirely of tree vertices and thus cannot be
  a descendant of any hybrid vertex $h\prec_N w$. That is, there is $x\in
  \CC(w)$ and $x\notin \CC(h)$ for all $h\in H$; a contradiction. Therefore
  $A\notin\mathscr{C}_N$.
\end{proof}

In particular, the case $\CC(u)\cap \CC(v)\notin\mathscr{C}_N$ in
Lemma~\ref{lem:tree-child-intersections} can indeed occur even if
$\CC(u)\cap \CC(v)\ne\emptyset$ as the example in Fig.~\ref{fig:not-closed}
shows.  Hence, the clustering system $\mathscr{C}_N$ of a tree-child
network $N$ is not necessarily closed.
\begin{figure}[t]
  \centering
  \includegraphics[width=0.35\textwidth]{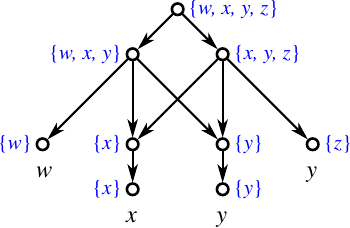}
  \caption{A tree-child network whose clustering system $\mathscr{C}$ is
    not closed since $\{w,x,y\}\cap \{x,y,z\}=\{x,y\}\notin\mathscr{C}$.}
  \label{fig:not-closed}
\end{figure}

\begin{corollary}
  Let $N$ be a tree-child network with clustering system $\mathscr{C}$ and let
  $C_1,C_2\in\mathscr{C}$ be a pair of overlapping clusters. Then $C_1\cap
  C_2\in\mathscr{C}$ if and only if there is a hybrid vertex $h\in V(N)$ such
  that $C_1\cap C_2=\CC(h)$.
\end{corollary}

Another class of networks that has received considerable attention in the
last decade are tree-based networks
\cite{Francis:2015kj,Zhang:16,Jetten:18,Pons:19}. They capture the idea
that networks can be obtained from (the subdivision of) a tree by inserting
additional arcs:
\begin{definition}\label{def:tree-based}
  A network $N$ is called \emph{tree-based} with base tree $T$ if $N$ can be
  obtained from $T$ by (a) subdividing the arcs of $T$ by introducing
  vertices with in- and outdegree $1$ (called \emph{attachment points}),
  and (b) adding arcs (called \emph{linking arcs}) between pairs of
  vertices, so that $N$ remains acyclic.
\end{definition}
This definition further generalizes the original one for non-binary
tree-based networks in \cite{Jetten:18} in the sense that we do not require
the two properties phylogenetic and separated and that we allow to have
additional arcs between non-attachment points and attachment points may
have in- and outdegree greater than one. This generalization ensures that
all trees, i.e., in particular non-phylogenetic trees remain tree-based.

Equivalently, a network $N$ on $X$ is tree-based if and only if there is a
rooted (not necessarily phylogenetic) spanning tree $T$ with leaf set $X$,
i.e., there are no \emph{dummy leaves} in $T$ that correspond to inner
vertices in $N$. Clearly, the unique incoming arc $(u,v)$ of a tree vertex
$v$ must be contained in every rooted spanning tree $T$ of a network $N$
and thus $u$ cannot be a dummy leaf in $T$.  As an immediate consequence,
the well-known fact that tree-child network are always tree-based
\cite{Pons:19}, remains also true in our generalized setting: and does in
particular, not require the properties phylogenetic and separated:
\begin{fact}
  \label{obs:treechild->treebased}
  Every tree-child network is tree-based.
\end{fact}

\begin{figure}[t]
  \begin{center}
    \includegraphics[width=0.65\textwidth]{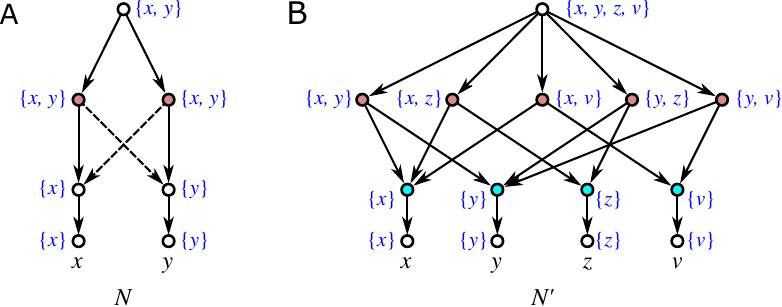}
  \end{center}
  \caption{(A) A network that is tree-based (with a possible base tree
    indicated by the solid-line arcs) but that does not satisfy
    (PCC). (B) A cluster network (thus satisfying (PCC)) that is not
    tree-based, see details in the text.}
  \label{fig:tree-based}
\end{figure}

Fig.~\ref{fig:tree-based}(A) shows a network $N$ that is tree-based since
removal of the dashed arcs results in a rooted spanning tree whose leaves
are exactly the leaves of $N$. However, $N$ does not satisfy (PCC) since
the two vertices highlighted in orange correspond to the same cluster but
are $\preceq_{N}$-incomparable.  Conversely, the example in
Fig.~\ref{fig:tree-based}(B) shows that (PCC), or even semi-regularity,
does not imply that a network is tree-based.  In particular, the network
$N'$ is the cluster network for $\mathscr{C}= \{ \{x\}, \{y\}, \{z\},
\{v\}, \{x,y\}, \{x,z\}, \{x,v\}, \{y,z\}, \{y,v\}, \{x,y,z,v\} \}$.  To
see that $N'$ is not tree-based, consider the set $U$ of all inner vertices
whose children are all hybrid vertices, called \emph{omnians} in
\cite{Jetten:18}. These vertices are highlighted in orange.  In any rooted
spanning tree, each of the four hybrid vertices (highlighted in cyan) has
exactly one parent. Since, in addition, none of the five omnians has a
child that is not one of the four hybrid vertices, one easily verifies that
at least one omnian must be a dummy leaf in every spanning tree. Therefore,
$N'$ is not tree-based. The latter observation is in line with Cor.~3.6 in
\cite{Jetten:18}, which holds for a more restricted definition of
tree-based and states that a (phylogenetic, separated) network is
tree-based if and only if, for all $S\subseteq U$, the number of different
children of the vertices in $S$ is greater than or equal to $\vert S\vert$.
Clearly, the latter is not satisfied for $S=U$ in the example.  We
summarize the latter findings in
\begin{fact}
  Not every tree-based network satisfies (PCC). Moreover, there are cluster
  networks and thus, phylogenetic separated networks that do not satisfy
  (PCC).
\end{fact}

\section{Least Common Ancestors and Lca-Networks}
\label{sec:lcaN-main}

\subsection{Basics}
\label{ssec:LCA-N-basic}

\begin{definition} \cite{Bender:01}\
  \label{def:lca}
  A \emph{least common ancestor} (LCA) of a subset $Y\subseteq V$ in a DAG
  $N$ is a $\preceq_N$-minimal vertex of $V$ that is an ancestor of all
  vertices in $Y$.
\end{definition}
In general DAGs $N$, an LCA does not necessarily exist for a given vertex
set.  Moreover, the LCA is not unique in general. We write $\LCA(Y)$ for the
(possibly) empty set of $\preceq_{N}$-minimal ancestors of the elements in $Y$.
In a (phylogenetic) network $N$, the root $\rho_N$ is an ancestor of all
vertices in $V$, and thus a least common ancestor exists for all $Y\subseteq
V$. The LCA sets retain key information on the partial order $\preceq_N$:
\begin{fact}
  \label{obs:LCA-PO}
  Let $N$ be a network on $X$ and $Z\subseteq Y\subseteq V$. Then for
  every $y\in\LCA(Y)$ there is $z\in\LCA(Z)$ such that $z\preceq_N y$.
\end{fact}
\begin{proof}
  Consider $y\in\LCA(Y)$. Then $y$ is also an ancestor of all vertices in
  $Z$ and thus there is a $\preceq_N$ minimal descendant $z\preceq_N y$
  that is an ancestor of all vertices in $Z$, i.e., $z\in\LCA(Z)$.
\end{proof}

If $\LCA(Y) = \{u\}$ consists of a single element $u$ only, we write
$\lca(Y)=u$. In other words, $\lca(Y) = u$ always implies that the
$\preceq_{N}$-minimal ancestor of the elements in $Y$ exists and is
uniquely determined. We leave $\lca(Y)$ undefined for all $Y$ with
$\vert\LCA(Y)\vert\ne1$.

In \cite{HS18}, least common ancestors $u$ are defined in terms of the fact
that no child of $u$ is an ancestor of all vertices in $Y$. These
definitions are equivalent:
\begin{lemma}
  \label{lem:lca-equivalence}
  Let $N$ be a network and $\emptyset\ne Y \subseteq V$.  Then $u\in V$ is
  a least common ancestor of $Y$ if and only if $u$ is an ancestor of all
  vertices in $Y$ but there is no $v\in\child_N(u)$ that is an ancestor of
  all vertices in $Y$.
\end{lemma}
\begin{proof}
  By definition, $u\in V$ is a least common ancestor of $Y$ if it is
  ancestor of all vertices in $Y$ and $\preceq_{N}$-minimal w.r.t.\ this
  property. Thus the \emph{only if}-part of the statement follows
  immediately.  Conversely, suppose $u$ is an ancestor of all vertices in
  $Y$ but there is no $v\in\child_N(u)$ with this property.  Writing $D_x$
  for the set of descendants of a vertex $x\in V$, suppose conversely that
  $u$ is an ancestor of all vertices in $Y$ but $Y\not\subseteq D_v$ for
  each $v\in\child_N(u)$.  For every vertex $w\in V(N)$ with $w\prec_N u$,
  there is a directed path passing through some child $v\in\child_N(u)$,
  i.e., $w\preceq_N v$. Hence, we have $D_w\subseteq D_v$. Together with
  $Y\not\subseteq D_v$, this implies $Y\not\subseteq D_w$.  Hence, $u$ is a
  least common ancestor of $Y$.
\end{proof}

We will in particular be concerned here with LCAs of leaves, i.e., the sets
$\LCA(A)$ for non-empty subsets $A\subseteq X$. We can then express LCAs in
terms of clusters:
\begin{fact}
$v\in\LCA(A)$ if and only if $A\subseteq \CC(v)$ and there is no
  vertex $u\prec_N v$ such that $A\subseteq \CC(u)$.
\end{fact}
Suppose $\lca(A)\eqqcolon q$ is defined for some non-empty $A\subseteq X$.
Then, by assumption, every vertex $v$ with $A\subseteq\CC(v)$ satisfies
$q\preceq_N v$ and thus $\CC(q)\subseteq\CC(v)$ by
Lemma~\ref{lem:inclusion}. Since every vertex $v'$ for which
$\CC(q)\subseteq \CC(v')$ in particular also satisfies $A\subseteq\CC(v')$,
we conclude that $\lca(\CC(q))=q$. Thus we have
\begin{fact}
  \label{obs:deflcaY}
  Let $N$ be a network, $\emptyset\ne A\subseteq X$, and suppose $\lca(A)$
  is defined. Then the following is satisfied:
  \begin{enumerate}[noitemsep,nolistsep]
  \item[(i)] $\lca(A)\preceq_{N} v$ for all $v$ with $A\subseteq\CC(v)$.
  \item[(ii)] $\CC(\lca(A))$ is the unique inclusion-minimal cluster in
    $\mathscr{C}_N$ containing $A$.
  \item[(iii)] $\lca(\CC(\lca(A)))=\lca(A)$.
  \end{enumerate}
\end{fact}

In much of the literature on least common ancestors in DAGs, only pairwise
LCAs are considered. Networks with unique pairwise LCAs are of interest
because of a close connection with so-called binary clustering systems
\cite{Barthelemy:08} and monotone transit functions
\cite{Changat:19a,Changat:21w}.
\begin{definition} \cite{Barthelemy:08} \label{def:prebinary}
  A clustering system $\mathscr{C}$ on $X$ is \emph{pre-binary} if, for every
  pair $x,y\in X$, there is a unique inclusion-minimal cluster $C$ such that
  $\{x,y\}\subseteq C$.
\end{definition}
From Obs.~\ref{obs:deflcaY}(ii), we immediately obtain
\begin{fact}
  \label{obs:lca->prebinary}
  If $N$ is a network on $X$ such that $\lca(\{x,y\})$ is defined
  for all $x,y\in X$, then $\mathscr{C}_N$ is pre-binary.
\end{fact}

\begin{figure}[t]
  \begin{center}
    \includegraphics[width=0.6\textwidth]{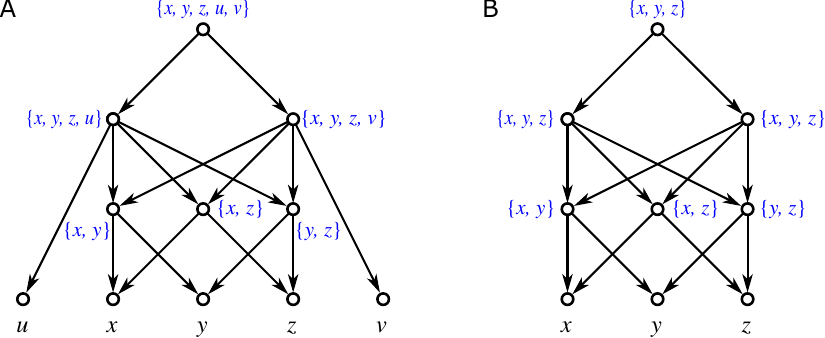}
  \end{center}
  \caption{(A) A regular network that is not an lca-network even though all
    $\lca(\{x,y\})$ is defined for all pairs. For $A=\{x,y,z\}$ both
    children of the root are contained in $\LCA(A)$.  Here, $Y$ is not a
    cluster. (B) The existence of pairwise LCAs is also insufficient to
    ensure that property (CL) is satisfied, i.e., that $\lca(\CC(v))$ is
    defined for all $v\in V(N)$ since $A$ is a cluster in this example.}
  \label{fig:lca-xy-no-Y}
\end{figure}

The first example in Fig.~\ref{fig:lca-xy-no-Y} shows, however, that
unique pairwise LCAs are not sufficient to ensure that $\lca(A)$ is also
defined for larger sets. The second example shows that unique pairwise
LCAs also do not ensure that all clusters $\CC(v)$ have a unique LCA.

Willson \cite[Thm.~3.3]{Willson:10b} showed that $\lca(\CC(v))$ (there
called ``mrca'') is well defined for all vertices of a normal network.
However, \cite{Willson:10b} uses a different definition of $X$ as the
``base set'' comprising the root, leaves, and all vertices with outdegree
$1$. We therefore adapt Thm.~3.3 of \cite{Willson:10b} to our setting
and include a proof for completeness.
\begin{definition}\label{def:CL}
  A network $N$ has the \emph{cluster-lca} property (CL) if
  \begin{itemize}
  \item[(CL)] For every $v\in V(N)$, $\lca(\CC(v))$ is defined.
  \end{itemize}
\end{definition}
\begin{lemma}
  \label{lem:lca(C(v))<=v}
  Suppose $N$ has property (CL). Then for all $v\in V(N)$ holds
  $\lca(\CC(v))\preceq_N v$ and $\CC(\lca(\CC(v)))=\CC(v)$.
\end{lemma}
\begin{proof}
  If $v\notin\LCA(v)$ then there is a descendant $v'\preceq v$ such that
  $\CC(v)=\CC(v')$, and every $\preceq_N$-minimal descendant $v'$ with this
  property satisfies $v'\in\LCA(\CC(v))$. By property (CL), $\LCA(\CC(v))$
  contains only a single vertex $\lca(\CC(v))$, which therefore must
  coincide either with $v$ or one of its descendants. The second statement
  now follows directly from the definition.
\end{proof}

In Lemma~\ref{lem:C=path}, we saw that $Q(u)\coloneqq\{u'\in V(N) \mid
\CC(u')=\CC(u)\}$ forms an induced path in semi-regular
networks. Property (CL) imposes a weaker structure.
\begin{lemma}
  Let $N$ be a network satisfying (CL). Then (i) $Q(u)$ has a unique
  $\preceq_N$-minimal element, namely $\lca(\CC(u))=\min Q(u)$, and
  (ii) if $u,v\in Q(u)$ and $w$ is contained in a directed path from
  $u$ to $v$, then $w\in Q(u)$.
\end{lemma}
\begin{proof}
  The second statement in Lemma~\ref{lem:lca(C(v))<=v} implies that
  $q\coloneqq \lca(\CC(u))\in Q(u)$. By definition $q$ is the unique
  $\preceq_N$-minimal vertex that has all leaves in $\CC(q)$ as its
  descendants and thus $q\preceq u'$ for all $u'\in Q(u)$, establishing
  statement (i). Statement (ii) is a direct consequence of
  Lemma~\ref{lem:inclusion}.
\end{proof}

\begin{figure}[t]
  \begin{center}
    \includegraphics[width=0.13\textwidth]{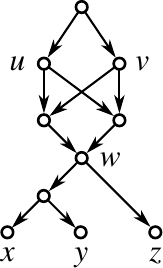}
  \end{center}
  \caption{The network $N$ satisfies (CL) but not (PCC). In addition to the
    singletons, only $\{x,y\}$ and $A=\{x,y,z\}$ appear as clusters. Both
    have unique last common ancestors, hence (CL) is satisfied. In
    particular, $\lca_{N}(A)=w$. However, we have $A=\CC_{N}(u)=\CC_{N}(v)$
    for the two $\preceq_{N}$-incomparable vertices $u$ and $v$. Hence, $N$
    does not satisfy (PCC).}
  \label{fig:CL-no-PCC}
\end{figure}

\begin{lemma}
  \label{lem:PCC->CL}
  If a network $N$ satisfies (PCC), then it satisfies (CL).
\end{lemma}
\begin{proof}
  Suppose that $N$ satisfies (PCC).  Given a cluster $\CC(u)$, we consider
  the non-empty set $W\coloneqq\{w\in V \mid \CC(u)=\CC(w)\}\subseteq V$.
  (PCC) implies that the elements of $W$ are pairwise
  $\preceq_N$-comparable, and thus there is a unique $\preceq_N$-minimal
  element $w\in W$. Furthermore, $\CC(u)=\CC(w)\subsetneq \CC(v)$ implies
  $w\prec_N v$ by Obs.~\ref{obs:subsetneq-implies-below}. Therefore,
  $\lca(\CC(u))=w$. Since $u\in W$ by construction, we have $w\preceq_N u$,
  and thus $\lca(\CC(u))\preceq_N u$.
\end{proof}
Fig.~\ref{fig:CL-no-PCC} shows that the converse is not true, i.e., (CL)
does not imply (PCC).  Cor.~\ref{cor:cluster-n-with-PCC} and
Lemma~\ref{lem:PCC->CL} imply
\begin{proposition} \label{prop:cluster-n-with-CL}
	For every clustering system $\mathscr{C}$ there is a network $N$ with
	$\mathscr{C}_N = \mathscr{C}$ that satisfies (CL).
\end{proposition}

Since a normal network is tree-child, it satisfies (PCC) by
Cor.~\ref{cor:TC->PCC}. This together with Lemma~\ref{lem:PCC->CL} implies
\begin{corollary}
  \label{cor:normal->CL}
  Every normal network $N$ satisfies (CL).
\end{corollary}
Property (CL) however does not imply that $\lca$ is well-defined for all
subsets $A\subseteq X$.

\subsection{Lca-Networks}
\label{ssec:lcaN}

\begin{definition}\label{def:lcaN}
  A network $N$ is an \emph{lca-network} if $\lca(A)$ is well-defined,
  i.e., if $\vert\LCA(A)\vert=1$ for all non-empty subsets $A\subseteq X$.
\end{definition}
From Obs.~\ref{obs:LCA-PO} and uniqueness of the least common ancestors,
we immediately obtain
\begin{fact}
  \label{obs:lca-PO}
  Let $N$ be an lca-network and $A\subseteq B \subseteq X$. Then
  $\lca(A)\preceq_N\lca(B)$.
\end{fact}

\begin{corollary}
  \label{cor:lca-below}
  Let $N$ be an lca-network on $X$ and $\emptyset\ne A\subseteq \CC_{N}(v)$
  for some $v\in V(N)$. Then, it holds $\lca_N(A)\preceq_{N} v$.
\end{corollary}
\begin{proof}
  Since $Y\subseteq \CC_{N}(v)$, there clearly exists some vertex
  $u\preceq_{N} v$ such that $A\subseteq \CC_N(u)$ but $A\not\subseteq
  \CC(w)$ for all vertices $w\in V(N)$ with $w\prec_N u$.  By
  Lemma~\ref{lem:lca-unique}, $u$ must be the unique vertex $\lca_{N}(A)$.
\end{proof}

\begin{lemma}
  \label{lem:PCC+cl=>lca}
  If $N$ is a network with a closed clustering system $\mathscr{C}_N$
  and $N$ satisfies (PCC), then it is an lca-network.
\end{lemma}
\begin{proof}
  Assume, for contradiction, that there is some set $A$ with two distinct
  vertices $u,v\in\LCA(A)$. Then $u,v$ are $\preceq_N$-incomparable and
  $A\subseteq \CC(u)$ and $A\subseteq \CC(v)$. If $\CC(u)\subseteq \CC(v)$,
  then (PCC) implies $u\preceq_N v$ or $v\preceq_N u$; a
  contradiction. Similarly, $\CC(v)\subseteq \CC(u)$ is not possible. Thus
  $\CC(u)$ and $\CC(v)$ overlap. Since $\mathscr{C}_N$ is closed, there is
  a cluster $\CC(w)=\CC(u)\cap \CC(v)$ for some $w\in V$.  Since
  $A\subseteq \CC(w)\subsetneq \CC(u)$ and $A\subseteq \CC(w)\subsetneq
  \CC(v)$, Obs.~\ref{obs:subsetneq-implies-below} implies $w\prec_N u$ and
  $w\prec_N v$, contradicting $u,v\in \LCA(A)$.
\end{proof}

\begin{theorem}
  \label{thm:PCC::lca<->closed}
  Let $N$ be a network satisfying (PCC). Then $N$ is an lca-network if and
  only if its clustering system $\mathscr{C}_N$ is closed.
\end{theorem}
\begin{proof}
  Let $N$ be an lca-network satisfying (PCC).  We show that $\mathscr{C}_N$
  is closed, i.e., for all non-empty $A\in 2^X$, it holds $\cl(A)=A \iff
  A\in \mathscr{C}_N$.  From the definitions of $\LCA$, clusters, and the
  closure operator we obtain
  \begin{equation*}
    A \subseteq
    \bigcap_{\substack{v \in V\\ A\subseteq \CC(v)}} \CC(v) =
    \bigcap_{\substack{C\in\mathscr{C}_N\\ A\subseteq C}} C = \cl(A) \,.
  \end{equation*}
  If $A\in \mathscr{C}_N$, then clearly $\cl(A)=A$.  Now suppose
  $A\notin\mathscr{C}_N$ and assume, for contradiction, that
  $\cl(A)=A$. Thus, we have $\cl(A)\notin\mathscr{C}_N$.  Then there are at
  least two distinct inclusion-minimal clusters $C'$ and $C''$ such that
  $A\subsetneq C',C''$.  By Lemma~\ref{lem:C=path}, there are unique
  $\preceq_{N}$-minimal vertices $u'$ and $u''$ with $C'=\CC(u')$ and
  $C''=\CC(u'')$.  Therefore and by Lemma~\ref{lem:inclusion}, we obtain,
  for all $v\prec_{N} u'$, that $\CC(v)\subsetneq\CC(u')=C'$ and thus
  $A\not\subseteq \CC(v)$ by inclusion-minimality of $C'$. Hence, $u'$ is a
  last common ancestor of $A$.  By analogous arguments, $u''$ is a last
  common ancestor of $A$.  Since $C'\ne C''$, $u'$ and $u''$ are
  distinct. Together with $\{u',u''\}\subseteq\LCA(A)$, this contradicts
  that $N$ is an lca-network.  Therefore, $A\notin\mathscr{C}_N$ implies
  $A\subsetneq \cl(A)$. Isotony of the closure function together with
  the contraposition of the latter statement shows that $A=\cl(A)\implies
  A\in\mathscr{C}_N$. In summary, therefore, $\mathscr{C}_N$ is closed.
  Conversely, by Lemma~\ref{lem:PCC+cl=>lca}, a network satisfying (PCC)
  with a closed clustering system is an lca-network.
\end{proof}
The example in Fig.~\ref{fig:CL-no-PCC} shows that (PCC) cannot be omitted
in Theorem~\ref{thm:PCC::lca<->closed}.  A trivial consequence of
Theorem~\ref{thm:PCC::lca<->closed} is
\begin{corollary}
  A semi-regular networks $N$ is an lca-network if and only if its clustering
  system $\mathscr{C}_N$ is closed.
\end{corollary}
Since every tree-child network satisfies (PCC) we also obtain
\begin{corollary}
  A tree-child network $N$ is an lca-network if and only if its clustering
  system $\mathscr{C}_N$ is closed.
\end{corollary}

Moreover, we have
\begin{proposition}
  \label{prop:lca-iff-closed}
  A clustering system $\mathscr{C}$ is closed if and only if it
  is the clustering system of some lca-network.  In this case, the unique
  regular network and the unique cluster network of $\mathscr{C}$ are
  lca-networks.
\end{proposition}
\begin{proof}
  Let $\mathscr{C}$ be the clustering system.  By
  Prop.~\ref{prop:regular-unique} and Thm.~\ref{thm:unique-cluster-network},
  there is a unique regular network $N$ and a unique cluster network $N'$
  with $\mathscr{C}=\mathscr{C}_N=\mathscr{C}_{N'}$.  By
  Thm.~\ref{thm:semiregular} and \ref{thm:cluster-network-charac},
  resp., both networks are semi-regular and thus satisfy (PCC).  Now apply
  Thm.~\ref{thm:PCC::lca<->closed}.
  \end{proof}

As in the case of trees, there is a simple connection of the LCA with the
closure operator:
\begin{lemma}
  \label{lem:lca-cl-identities}
  Let $N$ be an lca-network with clustering system $\mathscr{C}$.
  Then the following identity holds:
  \begin{equation}
    \label{eq:cl-lca}
    \CC(\lca(Y))=\cl(Y) \quad\text{ for all } \emptyset\ne Y\subseteq X.
  \end{equation}
  Furthermore, we have
  \begin{equation}
    \label{eq:lca-C-equals-C}
    \CC(\lca(C))=\cl(C)=C \quad\text{ for all } C\in\mathscr{C}.
  \end{equation}
\end{lemma}
\begin{proof}
    Let $\emptyset\ne Y\subseteq X$.  By Prop.~\ref{prop:lca-iff-closed},
    $\mathscr{C}$ is closed and thus $\cl(Y)\in \mathscr{C}$.  In
    particular, $\cl(Y)$ is the unique inclusion-minimal cluster in
    $\mathscr{C}$ containing $Y$.  Let $u$ be a vertex in $N$ such that
    $\CC(u)=\cl(Y)$.  Since $Y\subseteq \cl(Y)=\CC(u)$,
    Cor.~\ref{cor:lca-below} yields $\lca(Y)\preceq_{N} u$.  By
    Lemma~\ref{lem:inclusion}, we have $Y\subseteq\CC(\lca(Y))\subseteq
    \CC(u)=\cl(Y)$.  Since $\cl(Y)$ is the unique inclusion-minimal cluster
    in $\mathscr{C}$ containing $Y$, this implies $\CC(\lca(Y))=\cl(Y)$.
    For every cluster $C\in\mathscr{C}$, we have $C=\cl(C)$ by
    Equ.~\eqref{eq:closure} and thus Equ.~\eqref{eq:lca-C-equals-C} follows
    immediately.
\end{proof}

Next we show that two set sets have the same LCA whenever their LCAs are
associated with the same cluster.
\begin{lemma}
  \label{lem:lca-comparable}
  Let $N$ be an lca-network on $X$ and let $Y,Y'\subseteq X$.  Then (i)
  $\lca(\CC(\lca(Y)))=\lca(Y)$ and (ii) $\CC(\lca(Y))=\CC(\lca(Y'))$
  implies $\lca(Y)=\lca(Y')$
\end{lemma}
\begin{proof}
  Since $\cl$ is enlarging, i.e., $Y\subseteq\cl(Y)$, we have
  $Y\subseteq\CC(\lca(Y))$ by Lemma~\ref{lem:lca-cl-identities}.  Thus
  Obs.~\ref{obs:lca-PO} implies $\lca(Y)\preceq_N\lca(\CC(\lca(Y)))$. On
  the other hand, all leaves in $\CC(\lca(Y))$ are descendants of
  $\lca(Y)$, and thus $\lca(\CC(\lca(Y)))\preceq_N \lca(Y)$. Thus statement
  (i) holds.  Now suppose $\CC(\lca(Y))=\CC(\lca(Y'))$. Uniqueness of the
  LCA implies $\lca(\CC(\lca(Y)))=\lca(\CC(\lca(Y')))$ and thus statement
  (i) implies $\lca(Y)=\lca(Y')$.
\end{proof}

\subsection{Strong Lca-Networks and Weak Hierarchies}
\label{ssec:lcaN-strong}

In this section, we consider an interesting subclass of lca-networks.
\begin{definition}
  \label{def:stronglca}
  A network $N$ on $X$ is a \emph{strong lca-network} if it is an
  lca-network and, for every non-empty subset $A\subseteq X$, there are
  $x,y\in A$ such that $\lca(\{x,y\})=\lca(A)$.
\end{definition}
Fig.~\ref{fig:lca-no-xy} shows an lca-network that is not a strong
lca-network.
\begin{figure}[t]
  \begin{center}
    \includegraphics[width=0.2\textwidth]{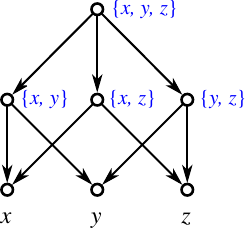}
  \end{center}
  \caption{An lca-network with a subset of leaves $A\coloneqq
    \{x,y,z\}\subseteq X$ in which there are no $x',y'\in A$ such that
    $\lca(\{x',y'\})=\lca(A)$.}
  \label{fig:lca-no-xy}
\end{figure}
We shall see below that strong lca-networks are intimately connected with
well-studied types of clustering systems.
\begin{definition}
  \label{def:wH-and-binary}
  A clustering system $\mathscr{C}$ on $X$ is
  \begin{itemize}
  \item[] a \emph{weak hierarchy} if $C_1\cap C_2\cap C_3 \in\{C_1\cap
    C_2,C_1\cap C_3, C_2\cap C_3,\emptyset\}$ for all
    $C_1,C_2,C_3\in\mathscr{C}$; and
  \item[] \emph{binary} if it is pre-binary and, for every
    $C\in\mathscr{C}$, there is a pair of vertices $x,y\in X$ such that $C$
    is the unique inclusion-minimal cluster containing $x$ and $y$.
  \end{itemize}
\end{definition}

Weak hierarchies have been studied in detail in \cite{Bandelt:89}. Binary
clustering system are considered systematically in \cite{Barthelemy:08}.
We first consider the lca-networks with binary clustering systems:

\begin{lemma}
  \label{lem:binC}
  Let $N$ be an lca-network. Then the following conditions are equivalent:
  \begin{itemize}
    \item[(i)] For all $v\in N$, there is $x,y\in\CC(v)$ such that
      $\lca(\{x,y\})=\lca(\CC(v))$.
    \item[(ii)] $\mathscr{C}_N$ is binary.
  \end{itemize}
\end{lemma}
\begin{proof}
  Since $N$ is an lca-network, it is in particular pre-binary (cf.\
  Obs.~\ref{obs:lca->prebinary}) and satisfies (CL). Property (i) and
  Lemma~\ref{lem:lca(C(v))<=v} imply
  $\CC(v)=\CC(\lca(\CC(v)))=\CC(\lca(\{x,y\}))$ for two vertices
  $x,y\in\CC(v)$. By Obs.~\ref{obs:deflcaY}(ii), $\CC(v)$ is the unique
  inclusion-minimal cluster containing $x$ and $y$, i.e., $\mathscr{C}_N$
  is binary.  Conversely, suppose $N$ is an lca-network with a binary
  clustering system. Then for every $v\in N$, there is $x,y\in X$ such that
  $\CC(v)$ is the unique inclusion-minimal cluster that contains $x$ and
  $y$.  By Obs.~\ref{obs:deflcaY}(ii), this implies
  $\CC(v)=\CC(\lca(\{x,y\}))$.  Hence, we have
  $\lca(\CC(v))=\lca(\CC(\lca(\{x,y\})))$, and thus, by
  Obs.~\ref{obs:deflcaY}(iii),
  $\lca(\CC(v))=\lca(\CC(\lca(\{x,y\})))=\lca(\{x,y\})$, i.e., property (i)
  holds.
\end{proof}

In particular, therefore, strong lca-networks give rise to binary
clustering systems:
\begin{corollary}
  The clustering system of a strong lca-network is binary.
\end{corollary}
The converse is not true in general, since condition (i) in
Lemma~\ref{lem:binC} requires only that the LCAs of clusters but not
necessarily the LCAs of all sets are determined by the LCA of a leaf
pair. The latter, stronger condition, is related to weak hierarchies.  To
investigate this connection, we recall
\begin{proposition}
  \textnormal{\cite[Lemma~1]{Bandelt:89}}
  \label{prop:BD89}
  A clustering system $\mathscr{C}$ on $X$ is a weak hierarchy if and only if
  for every nonempty subset $A\subseteq X$ there exist $x,y\in A$ such that
  $\cl(A)=\cl(\{x,y\})$.
\end{proposition}

\begin{proposition}
  \label{prop:wh-stlca}
  Let $N$ be an lca-network on $N$. Then $N$ is a strong lca-network if and
  only of $\mathscr{C}_N$ is a weak hierarchy.
\end{proposition}
\begin{proof}
  Def.~\ref{def:stronglca} and Equ.~\eqref{eq:cl-lca} imply that for every
  $\emptyset\ne A\subseteq X$ there is $x,y\in A$ such that
  $\cl(A)=\CC(\lca(A))=\CC(\lca(\{x,y\}))=\cl(\{x,y\})$, and thus
  $\mathscr{C}_N$ is a weak hierarchy by Prop.~\ref{prop:BD89}. Conversely,
  if $N$ is an lca-network such that $\mathscr{C}_N$ is a weak hierarchy,
  then for all $\emptyset\ne A\subseteq X$ there is $x,y\in A$ such that
  $\CC(\lca(A))=\CC(\lca(\{x,y\}))$. By Obs.~\ref{obs:deflcaY} we have
  $\lca(A)=\lca(\CC(\lca(A)))=\lca(\CC(\lca(\{x,y\})))=\lca(\{x,y\})$, and
  thus $N$ is a strong lca-network.
\end{proof}

From Prop.~\ref{prop:lca-iff-closed} and Prop.~\ref{prop:wh-stlca}, we obtain
\begin{corollary}
  \label{cor:lca-iff-closed-weak}
  Let $N$ be a network satisfying (PCC). Then $N$ is a strong lca-network
  if and only if $\mathscr{C}_N$ is a closed weak hierarchy.
\end{corollary}
Furthermore, we can use the same arguments in the proof of
Prop.~\ref{prop:lca-iff-closed} together with
Cor.~\ref{cor:lca-iff-closed-weak} to derive the final result of this
section:
\begin{proposition}
  \label{prop:cwh}
  A clustering system $\mathscr{C}$ is a closed weak hierarchy if and only
  if it is the clustering system of a strong lca-network.  In this case,
  the unique regular network and the unique cluster network of
  $\mathscr{C}$ are strong lca-networks.
\end{proposition}

\section{Level-1 Networks}
\label{sec:LEVEL-1}

\subsection{Basic Properties}
\label{sec:LEVEL-1-basics}

We start by showing that all phylogenetic level-1 networks have the
path-cluster-comparability property (PCC).
\begin{lemma}
  \label{lem:orderiff}
  Every phylogenetic level-1 network satisfies (PCC).
\end{lemma}
\begin{proof}
  If $N$ is a phylogenetic level-1 network, then it is tree-child by
  Prop.~\ref{prop:phy-level1->TC}, and in turn every phylogenetic tree-child
  network satisfies (PCC) by Cor.~\ref{cor:TC->PCC}.
\end{proof}

We note that ``phylogenetic'' cannot be dropped in
Lemma~\ref{lem:orderiff}.  To see this, consider the level-$1$ network $N$
in Fig.~\ref{fig:level2-multiset}(B) where the two parents of the hybrid
vertex both correspond to cluster $\{a\}$ but they are
$\preceq_{N}$-incomparable; a violation of (PCC).

\begin{corollary}
  \label{cor:regular-IFF-least-resolved}
  Let $N$ be a level-1 network. Then, $N$ is least-resolved if and only if
  $N$ is regular.
\end{corollary}
\begin{proof}
  By Cor.~\ref{cor:lrN-sf-out1} and Thm.~\ref{thm:semiregular}, resp.,
  least-resolved and regular networks do not contain vertices with
  outdegree $1$, and thus they are phylogenetic. The statement now follows
  immediately from Lemma~\ref{lem:orderiff} and
  Thm.~\ref{thm:PCC-reg-iff-lr}.
\end{proof}
We emphasize, however, that there can exist least-resolved networks $N$ for
a given clustering system $\mathscr{C}$ that are not regular, as the
example in Fig.~\ref{fig:counter-reg-lr} shows. In this example, the
regular network $N'$ is level-$1$.  Next we show that Lemma
\ref{lem:complexNsimpleC} does not hold for level-1 networks:
\begin{lemma}\label{lem:complexNsimpleC-level1}
  Let $n$ be a positive integer. Then, there is no phylogenetic,
  shortcut-free level-$1$ network $N$ on $n$ leaves that is not a tree and
  where $\mathscr{C}_N$ is a hierarchy.
\end{lemma}
\begin{proof}
  Let $N$ be a phylogenetic, shortcut-free level-$1$ that is not a tree. By
  Lemma \ref{lem:orderiff}, $N$ satisfies (PCC). Since, in addition, $N$ is
  shortcut-free, $N$ is semi-regular.  Since $N$ is not a tree, it must
  contain a non-trivial block $B$. By Lemma \ref{lem:semiregular-overlap},
  there are at least two vertices $u$ and $v$ such that $\CC(u)$ and
  $\CC(v)$ overlap. Hence, $\mathscr{C}_N$ is not a hierarchy.
\end{proof}

As an immediate consequence of Thm.~\ref{thm:semiregular} and
Lemma~\ref{lem:orderiff}, we also obtain the following
\begin{proposition}\label{prop:semi-reg}
  A phylogenetic level-1 network is semi-regular if and only if it is
  shortcut-free. Furthermore, a level-1 network is regular if and only if
  it is shortcut-free and has no vertex with outdegree $1$.
\end{proposition}

\begin{lemma}\label{lem:parents-of-hybrid-cluster}
  Let $N$ be a phylogenetic level-1 network and $v$ be a hybrid-vertex of
  $N$.  Then $\CC(v)\subsetneq \CC(u)$ for every $u\in V(N)$ with $v\prec_N
  u$.
\end{lemma}
\begin{proof}
  Let $N$ be a phylogenetic level-1 network and $v$ be a hybrid vertex of
  $N$.  By Lemma~\ref{lem:hybrid-properly-contained}, $v$ and all of its
  (at least two) parents are contained in a common non-trivial block
  $B$. Hence, consider first one of the parents $w_1$ of $v$ such that
  $w_1\ne \max B$.  By Lemma~\ref{lem:inclusion}, $\CC(v)\subseteq
  \CC(w_1)$.  Assume, for contradiction, that $\CC(v) = \CC(w_1)$.  Since
  $v$ and $w_1$ are contained in the same non-trivial block $B$ and $N$ is
  level-1, $w_1$ cannot be a hybrid vertex and thus, since $N$ is
  phylogenetic, we have $\outdeg_N(w_1)\geq 2$. Let $w'\neq v$ be another
  child of $w_1$.  Again, by~Lemma \ref{lem:inclusion}, $\CC(w')\subseteq
  \CC(w_1)$ and thus, $\CC(w')\subseteq \CC(v)$.  Lemma~\ref{lem:orderiff}
  implies that $N$ satisfies (PCC), and thus, $v$ and $w'$ are
  $\preceq_N$-comparable.  Hence, we distinguish the two cases (a)
  $w'\prec_N v$ and (b) $v\prec_{N} w'$.

  In Case~(a), the arc $(w_1,w')$ must be a shortcut since $w'\prec_N v$
  and $v\in\child_{N}(w_1)\setminus\{w'\}$. In particular, $w'$ must be a
  hybrid vertex and there is a directed path from $w_1$ to $w'$ passing
  through $v$, which together with the arc $(w_1,w')$ forms an undirected
  cycle. Hence, $w_1$, $v$, and $w'$ are contained in a common block that
  shares the arc $(w_1,v)$ with $B$ and thus equals $B$. But then $B$
  contains two hybrid vertices $v$ and $w'$ that are distinct from $\max
  B$; a contradiction.

  Now consider Case~(b), i.e., $v\prec_{N} w' \prec_{N} w_1$.  In this
  case, $v$ has a parent $w_2$ such that $v\prec_N w_2\prec_{N} w_1
  \prec_N\max B$. In particular, $w_2$ lies in $B$ and
  Lemma~\ref{lem:inclusion} implies $\CC(v)\subseteq \CC(w_2) \subseteq
  \CC(w_1)$ and thus $\CC(v)= \CC(w_2)$.  Now we can repeat the latter
  arguments for parent $w_2$ and eventually encounter a contradiction as in
  Case~(a) or, if we never obtain such a contradiction, we end in an
  infinite chain of vertices $w_1\succ_{N} w_2 \succ_{N} \dots$; a
  contradiction to $V(N)$ being finite.  The latter together with the fact
  that $w_1\ne \max B$ was chosen arbitrarily implies that
  $\CC(v)\subsetneq \CC(w)$ for every parent $w$ of $v$ that is distinct
  from $\max B$. Now suppose that $w=\max B$ is a parent of $v$. Since $v$
  is a hybrid vertex, it has another parent $w'$, which is also contained
  in $B$ and satisfies $v\prec_{N} w'\prec_{N}
  w$. Lemma~\ref{lem:inclusion} and the arguments above thus yield
  $\CC(v)\subsetneq \CC(w')\subseteq \CC(w)$.

  Finally note that $v\prec_N u$ if and only if $v\prec_N w\preceq_N u$
  where $w$ is a parent of $v$.  This together with Lemma
  \ref{lem:inclusion} implies that $\CC(v)\subsetneq \CC(w)\subseteq
  \CC(u)$ for all $u\in V(N)$ with $v\prec_N w\preceq_N u$ and where $w$ is
  a parent of $v$.
\end{proof}

\begin{lemma}
  \label{lem:eta=min}
  Let $N$ be a level-1 network. Then every block $B$ has a unique
  $\preceq_N$-minimal vertex $\min B$ and a unique $\preceq_N$-maximal
  vertex $\max B$.  In case $B$ is not a single vertex or arc,
  $\min B$ is the unique properly contained hybrid vertex in $B$ and
  $\max B$ is the unique root of $B$.
\end{lemma}
\begin{proof}
  The statement is trivial for a block that consists only of a single
  vertex or arc.  Uniqueness of the $\preceq_N$-maximal vertex in $B$
  follows from Lemma~\ref{lem:max-B-unique}.  Otherwise, every $v\in V(B)$
  lies on an undirected cycle.  Since $B$ is acyclic, a $\preceq_N$-minimal
  vertex $u$ in $B$ does not have an out-neighbor along the cycle, and
  therefore, $u$ has at least two in-neighbors that are contained in $B$.
  Thus $u$ is a hybrid vertex and, by
  Lemma~\ref{lem:hybrid-properly-contained}, $u$ is properly contained in
  $B$.  By definition of level-1, there is at most one such vertex in $B$.
\end{proof}

As an immediate consequence, we have
\begin{corollary}
  \label{cor:min-max-B}
  Let $N$ be a level-1 network and $B$ a block of $N$.  For every $v\in B$,
  it holds $\min B \preceq_N v \preceq_N \max B$ and
  $\CC(\min B)\subseteq \CC(v)\subseteq \CC(\max B)$.
\end{corollary}
The fact that every block in a level-1 contains at most one hybrid vertex, 
implies that $B^0 = B\setminus\{\min B,\max B\}$ for every block $B$ (cf.\ 
Def.~\ref{def:intB}).
Recall that a block is non-trivial if it is not a single vertex or a single
arc. Hence, a block $B$ is non-trivial precisely if $B^0\ne\emptyset$. In
the absence of shortcuts and in case $B$ is non-trivial, the subnetwork
induced by $B^0$ is a forest consisting of at least two non-empty trees.

\begin{lemma}
  \label{lem:uvmin}
  Let $N=(V,E)$ be a level-1 network and suppose $u,v\in V$ are
  $\preceq_N$-incomparable. Then $u$ and $v$ are located in a common block
  $B$ of $N$ if and only if $\CC(u)\cap \CC(v)\ne\emptyset$. In particular,
  they share exactly the descendants of the $\preceq_N$-minimal element
  $\min B$ of $B$, i.e., in this case we have $\CC(u)\cap \CC(v)=\CC(\min
  B)$.
\end{lemma}
\begin{proof}
  If $u$ and $v$ are both located in block $B$, then
  Cor.~\ref{cor:min-max-B} implies $\emptyset\ne \CC(\min B)\subseteq
  \CC(u)\cap \CC(v)$.  Conversely, if $\CC(u)\cap \CC(v)\ne\emptyset$, then
  Lemma~\ref{lem:lower-path} implies that, for every $x\in\CC(u)\cap
  \CC(v)$, $u$ and $v$ are contained in a common block $B$, and $B$
  contains, in addition, a hybrid vertex $w$ such that $w\prec_N u,v$ and
  $x\in \CC(w)$. Since $N$ is level-$1$ and $w\ne \max B$, we have $w=\min
  B$. Moreover, for all $x\in\CC(u)\cap \CC(v)$, the corresponding blocks
  $B$ share $u$ and $v$ and are therefore identical by
  Obs.~\ref{obs:identical-block}. Hence, we obtain $\CC(u)\cap \CC(v)
  \subseteq \CC(\min B)$ and thus $\CC(u)\cap \CC(v) = \CC(\min B)$.
\end{proof}

\subsection{Clusters and Least Common Ancestors}
\label{ssec:level1-clusters}

An important property of level-1 networks that is not true in general
phylogenetic networks is the following.
\begin{lemma}
  \label{lem:lca-unique}
  Every level-$1$ network is an lca-network.
\end{lemma}
\begin{proof}
  Let $N$ be a level-1 network on $X$ and $\emptyset\ne Y\subseteq X$.
  Suppose for contradiction that there are two distinct such vertices $u$
  and $u'$ for which $Y\subseteq \CC(u), \CC(u')$ and such that
  $Y\not\subseteq \CC(v)$ whenever $v\prec_N u,u'$.  Clearly, $u$ and $u'$
  must be $\preceq_N$-incomparable. From Lemma~\ref{lem:uvmin} and
  $\emptyset\ne Y\subseteq \CC(u)\cap \CC(u')$, we obtain that $u$ and $u'$
  are located in the same block $B$ and $\CC(u)\cap \CC(u')=\CC(\min B)$.
  In particular, therefore, $Y\subseteq \CC(\min B)$.  Since $\min
  B\preceq_T u,u'$ by Cor.~\ref{cor:min-max-B} and $u$ and $u'$ are
  $\preceq_N$-incomparable, we have $u\ne\min B$ and $u'\ne\min B$, and
  thus $\min B\prec_N u,u'$; a contradiction to $Y\not\subseteq \CC(v)$ for
  all $v\prec_N u,u'$.  Therefore, the least common ancestor is unique and
  $\lca_{N}(u)$ is well-defined for all $\emptyset\ne Y\subseteq X$.
\end{proof}

By Lemma~\ref{lem:lca-unique}, every leaf set in a level-1 network $N$ has
a unique LCA. As a further consequence of Lemmas~\ref{lem:lca-equivalence}
and~\ref{lem:lca-unique}, the following result, which was stated without
proof in \cite{HS18} for binary level-1 networks, also holds in our more
general setting:
\begin{corollary}
  \label{cor:lca-unique}
  Let $N$ be a level-1 network on $X$ and $\emptyset\ne Y\subseteq X$. Then
  there is a unique vertex $u$ such that $Y\subseteq \CC(u)$ but
  $Y\not\subseteq \CC(v)$ for all $v\in\child(u)$. In this case, $u =
  \lca_N(Y)$.
\end{corollary}

Prop.~1 of \cite{HS18} also states the following result (without proof) for
binary level-1 networks:
\begin{lemma}
  \label{lem:lca-xy}
  Every level-1 network is a strong lca-network.
\end{lemma}
\begin{proof}
  Let $N$ be a level-1 network. By Lemma~\ref{lem:lca-unique}, $N$ is an
  lca-network. Thus, it remains to show that, for every $\emptyset\ne
  Y\subseteq X$, there are leaves $x,y\in X$ such that
  $\lca(Y)=\lca(\{x,y\})$.  The statement holds trivially if $Y=\{x\}$,
  since then $\lca(\{x,x\})=\lca(Y)$.  Hence, suppose now that $\vert
  Y\vert \ge 2$ and thus that $v\coloneqq \lca(Y)$ is not a leaf. By
  Cor.~\ref{cor:lca-unique}, every child $v'\in\child(v)$ satisfies
  $Y\not\subseteq \CC(v')$ and $Y\subseteq
  \CC(v)=\bigcup_{v'\in\child(v)}\CC(v')$, there are two distinct children
  $v',v''\in\child (v)$ such that there is $x\in Y\cap \CC(v')\setminus
  \CC(v'')\ne\emptyset$ and $y\in Y\cap \CC(v'')\setminus
  \CC(v')\ne\emptyset$. Since $\{x,y\}\subseteq Y$, we have
  $\lca(\{x,y\})\preceq_N v$ by Lemma~\ref{lem:lca-unique} and
  Obs.~\ref{obs:lca-PO}.  Suppose for contradiction that
  $\lca(\{x,y\})\prec_N v$.  Contraposition of Cor.~\ref{cor:lca-unique}
  and $\{x,y\}\subseteq \CC(v)$ implies that there is a child
  $v'''\in\child(v)$ with $\{x,y\} \subseteq \CC(v''')$. By the choice of
  $v'$ and $v''$ we have $v'''\notin \{v',v''\}$.

  We continue by showing that $v$, $v'$, and $v'''$ are located in a common
  block $B$ of $N$.  Consider first the case that $v'$ and $v'''$ are
  $\preceq_N$-comparable.  Then Lemma~\ref{lem:inclusion}, $y\in \CC(v''')$,
  and $y\notin \CC(v')$ imply $v'\prec_N v'''$.  Hence, the three vertices
  $v,v'$ and $v'''$ lie on an undirected circle formed by the arcs
  $(v,v')$ and $(v,v''')$ and a directed path from $v'''$ to $v'$. By Obs.\
  \ref{obs:identical-block}, the vertices $v$, $v'$, and $v'''$ are part of
  a common block $B$.  Assume now that $v'$ and $v'''$ are
  $\preceq_N$-incomparable, then $x\in \CC(v')\cap \CC(v''')$ and
  Lemma~\ref{lem:uvmin} implies that $v'$ and $v'''$ are contained in
  common non-trivial block $B$. If $v$ is not contained in $B$, then $v$
  and arcs $(v,v')$ and $(v,v''')$ can be added to $B$ without losing
  biconnectivity; contradicting that $B$ is a maximal biconnected
  subgraph. Hence, $v$ is also contained in $B$. Similarly, one shows that
  $v$, $v''$, and $v'''$ are located in a common block $B'$ of $N$.  Since
  $B$ and $B'$ share the arc $(v,v''')$, Obs.\
  \ref{obs:biConn-arc-disjoint} implies $B=B'$. In summary, $v$, $v'$,
  $v''$ and $v'''$ are all located in a common block $B$ of $N$.

  Now suppose again that $v'$ and $v'''$ are $\preceq_N$-comparable. As
  argued above, we have $v'\prec_N v'''$ and thus there is a directed path
  $P$ from $v'''$ to $v'$. Since $N$ is acyclic and $(v,v''')\in E(N)$, all
  vertices $w$ in $P$ satisfy $w\prec_N v$. Together with $(v,v')\in E(N)$,
  this implies that $v'$ has at least indegree 2 and thus, $v'$ is a
  hybrid-vertex of $B$.  Since the hybrid vertex in each block of a level-1
  network is unique, we have $v'=\min B$.  But then we have
  $x\in \CC(v')=\CC(\min B)\subseteq \CC(v'')$ by Cor.~\ref{cor:min-max-B}; a
  contradiction.  Hence, $v'$ and $v'''$ must be $\preceq_N$-incomparable
  and we can apply Lemma~\ref{lem:uvmin} to conclude that
  $\CC(v')\cap \CC(v''')=\CC(\min B)$ and thus $x\in \CC(\min B)$.
  Cor.~\ref{cor:min-max-B} therefore implies
  $x\in \CC(\min B)\subseteq \CC(v'')$; a contradiction. In summary, the case
  $\lca(\{x,y\})\prec_N v$ is not possible and hence we must have
  $\lca(\{x,y\})= v$.
\end{proof}

As an immediate consequence of Lemma~\ref{lem:lca-xy} (or alternatively
Lemmas~\ref{lem:inclusion} and~\ref{lem:uvmin}) and Prop.~\ref{prop:cwh},
we have:
\begin{corollary}
  \label{cor:l1->closed-weak-hierarchy}
  The clustering system of a level-1 network is a closed weak hierarchy.
\end{corollary}
This result also generalizes \cite[Prop.~1]{Gambette:12}, who showed
that $\mathscr{C}_N$ is a weak hierarchy for binary level-1 networks.

In many applications, vertex- or arc-labeled networks are considered as a
scaffold to explain genomic sequence data
\cite{HM:06,Huber:19,HS18,Hellmuth:15a,HW:16b,HS:21,Hellmuth2019-gs,BHS:21}.
In this context, it is of considerable interest to understand the structure
of  least-resolved networks that still explain the same data and
are obtained from the original network by shortcut removal and contraction
of arcs (cf.\ Def.~\ref{def:least-resolved}). Hence, it is important to
keep track of $\lca$'s after arcs have been contracted. To this end, we
provide the following
\begin{proposition}
  Let $N$ be a level-1 network with leaf set $X$ and $(v',v)$ be an arc
  such that $v$ is neither a hybrid vertex nor a leaf of $N$. Moreover, let
  $N'$ be the network obtained from $N$ by application of
  $\contract(v',v)$.  Then, for all $x,y\in X$, we have $\lca_{N'}(x,y) =
  \lca_N(x,y)$ whenever $\lca_N(x,y)\neq v'$ and, otherwise,
  $\lca_{N'}(x,y) =v$.
\end{proposition}
\begin{proof}
  Since $v$ is not a hybrid vertex, $e = (v',v)$ is not a shortcut. Thus,
  $\contract(v',v)$ is well-defined. Let $x, y \in X$. If $x=y$, then
  $\lca_N(x,y) = x = \lca_{N'}(x,y)$. Hence, assume that $x\neq y$.  Note,
  by Lemma~\ref{lem:contract-level-k}, $N'$ remains a level-1 network.  By
  Cor.~\ref{cor:lca-unique}, therefore, $\lca_N(x,y)$ and $\lca_{N'}(x,y)$
  are well defined and, in particular, correspond to unique vertices in $N$
  and $N'$, respectively.

  Assume first that $\lca_N(x,y)=v'$. Hence, $x,y\prec_N v'$ and there must
  be children $c,c'$ of $v'$ such that $x\preceq_N c$ and $y\preceq_N
  c'$. By construction, each of $c$ and $c'$ either equals $v$ or becomes a
  child of $v$ in $N'$. This and Lemma~\ref{lem:contraction}(1) implies
  that $x\preceq_{N'} c \preceq_{N'} v$ and $y\preceq_{N'} c' \preceq_{N'}
  v$.  Cor.~\ref{cor:lca-below} and $x,y\preceq_{N'} v$ imply
  $z\preceq_{N'} v$.  Assume, for contradiction, that $z\prec_{N'} v$ Since
  $x,y\preceq_{N'} z$, for $x$ and $z$ (resp.\ $y$ and $z$) one of the
  Cases~(i) or~(ii) as specified in Lemma~\ref{lem:contraction}(2) must
  hold.

  Assume that Case~(i) holds for both $x$ and $z$ as well as $y$ and $z$,
  i.e., we have $x,y\preceq_N z$. Note that Lemma~\ref{lem:contraction}(2)
  must hold for $z$ and $v$ as well.  Hence, we have~(i) $z\preceq_N v$
  or~(ii) $z\preceq w'$ for some child $w'\neq v$ of $v'$ in $N$. For both
  cases, we have $x,y\preceq_N z \prec_N v'$; a contradiction to $v' =
  \lca_N(x,y)$.

  Assume now that Case~(ii) is satisfied for $x$ and $z$.  In this case,
  $x\preceq_{N'} z$ implies, in particular, that $v\preceq_N z$.  This and
  Lemma~\ref{lem:contraction}(1) implies that $v\preceq_{N'} z$. This
  together with $z\preceq_{N'} v$ implies $z=v$; a contradiction. By
  similar arguments, Case~(ii) cannot hold for $y$ and $z$. Hence, neither
  of the Cases~(i) or~(ii) as specified in Lemma \ref{lem:contraction}(2)
  hold for $x$ and $z$ (resp.\ $y$ and $z$); a contradiction. Therefore,
  $\lca_N(x,y)=v'$ implies that $\lca_{N'}(x,y) = v$.

  Assume now that $\lca_N(x,y)\neq v'$. Let $z\coloneqq \lca_N(x,y)$ and
  $z'\coloneqq \lca_{N'}(x,y)$.  Since $z\neq v'$ and $z'\in V(N')$, we can
  conclude that $z,z'\in V(N)\cap V(N')$.  Lemma~\ref{lem:contraction}(1)
  together with $x,y\prec_N z$ implies $x,y\prec_{N'} z$.  Assume, for
  contradiction, that $z\neq z'$.  We distinguish Cases~(a) $z \prec_{N}
  z'$, (b) $z'\prec_{N} z$, and (c) $z$ and $z'$ are
  $\preceq_{N}$-incomparable.

  In Case~(a), $x,y \prec_{N} z \prec_{N} z'$ and
  Lemma~\ref{lem:contraction}(1) imply $x,y \prec_{N'} z \prec_{N'} z'$; a
  contradiction to $z'= \lca_{N'}(x,y)$.

  In Case~(b), suppose first, for contradiction, that $x\not\preceq_{N}
  z'$.  Together with $x\prec_{N'} z'$, this implies that Case~(ii) in
  Lemma~\ref{lem:contraction}(2) must hold, i.e., $v\preceq_{N} z'$ and
  $x\preceq_{N} w'$ for some $w'\in \child_{N}(v')\setminus \{v\}$.  In
  particular, we have $x\preceq_{N} w'\prec_{N} v'$.  Since $v'$ is the
  only parent of $v$ in $N$, the case $v\prec_{N} z'$ is not possible as it
  would imply $v'\preceq_{N} z'$ and thus $x\preceq_{N} w'\prec_{N}
  v'\preceq_{N} z'$; a contradiction. Hence, we have $v=z'$.  Since
  $z'\prec_{N} z$, $v'\ne z$, and $v'$ is the only parent of $v=z'$ in $N$,
  we must have $v'\prec_{N} z$.  Now consider $y\prec_{N'} z'(\prec_{N}
  v')$ which, by Lemma~\ref{lem:contraction}(2), implies (i) $y\preceq_{N}
  z'$ or (ii) $y\preceq_{N} w'$ for some $w''\in\child_{N}(v')\setminus
  \{v\}$. In any of the two cases, it holds $y\prec_{N} v'$. Hence, we have
  $x,y\prec_{N} v'\prec_{N} z$; a contradiction to $z=\lca_{N}(x,y)$.
  Therefore, it must hold $x\preceq_{N} z'$.  By analogous arguments, it
  holds $y\preceq_{N} z'$.  Hence, we have $x,y \preceq_{N} z' \prec_{N}
  z$; a contradiction to $z=\lca_N(x,y)$.

  In Case~(c), $z$ and $z'$ are $\preceq_{N}$-incomparable.  If
  $x,y\preceq_{N} z'$, then Cor.~\ref{cor:lca-below} implies
  $z=\lca_{N}(x,y)\preceq_{N} z'$; a contradiction.  Hence, suppose
  w.l.o.g.\ that $x\not\preceq_{N} z'$.  Re-using the arguments from
  Case~(b), this implies $v=z'$ and $x,y\prec_{N} v'$. The latter together
  with Cor.~\ref{cor:lca-below} and $z\ne v'$ implies
  $z=\lca_{N}(x,y)\prec_{N} v'$.  Hence, there is some child $c\in
  \child_{N}(v')$ with $z\preceq_{N} c$.  Since $z$ and $z'=v$ are
  $\preceq_{N}$-incomparable, it holds $c\in \child_{N}(v')\setminus\{v\}$.
  By construction, therefore, $c$ becomes a child of $v$ in $N'$, and thus,
  $c\prec_{N'} v$.  Together with Lemma~\ref{lem:contraction}(1), the
  latter arguments imply $x,y\prec_{N'} z\preceq_{N'} c\prec_{N'} v=z'$; a
  contradiction to $z'=\lca_{N'}(x,y)$.

  In summary, neither of Cases~(a), (b), and~(c) is possible.  Therefore,
  $z=z'$ must hold.
\end{proof}

\subsection{Property (L)}
\label{ssec:Level1-L}

\begin{lemma}
  \label{lem:overlaprule}
  Let $N$ be a level-1 network with clustering system $\mathscr{C}$,
  suppose $C_1,C_2\in\mathscr{C}$ overlap, i.e.,
  $C_1\cap C_2\notin\{C_1,C_2,\emptyset\}$. Then,
  $C_1\cap C_3\in\{ C_1, C_3, \emptyset, C_1\cap C_2\}$ for all
  $C_3\in\mathscr{C}$.
\end{lemma}
\begin{proof}
  Let $u_1, u_2\in V(N)$ be vertices such that $C_1=\CC(u_2)$ and
  $C_2=\CC(u_2)$.  Since $C_1$ and $C_2$ overlap, Lemma~\ref{lem:inclusion}
  implies that $u_1$ and $u_2$ are $\preceq_N$-incomparable, and thus, by
  Lemma~\ref{lem:uvmin}, $u_1$ and $u_2$ are located in the same block $B$,
  and $C_1\cap C_2= \CC(\min B)$. In particular, we have $u_1\ne \max B$
  since otherwise $u_2\preceq_{N} u_1$ (cf.\ Cor.~\ref{cor:min-max-B}).
  Now consider a vertex $w$ with $C_3=\CC(w)$.  If $w$ and $u_1$ are
  $\preceq_N$-comparable, then $C_3\subseteq C_1$ or $C_1\subseteq C_3$ by
  Lemma~\ref{lem:inclusion}, and thus $C_1\cap C_3\in\{C_1,C_3\}$.  Now
  consider the case where $w$ and $u$ are $\preceq_N$-incomparable.  If
  $C_1\cap C_3=\emptyset$, there is nothing to show.  Otherwise, by
  Lemma~\ref{lem:uvmin}, $u_1$ and $w$ are located in a common block $B'$
  and $C_1\cap C_3= \CC(\min B')$.  We have $u_1\ne \max B'$ since
  otherwise $w\preceq_{N} u_1$.  Hence, we have $u_1\notin\{\max B, \max
  B'\}$ and thus $B=B'$ by Lemma~\ref{lem:block-identity}.  Therefore
  $C_1\cap C_3= \CC(\min B)= C_1\cap C_2$.
\end{proof}

Inspection of the proof of Lemma~\ref{lem:overlaprule} shows that there
are overlapping clusters only if $N$ contains a non-trivial block and
thus a hybrid vertex. In particular, therefore, if $N$ is a rooted tree,
then $\mathscr{C}_N$ is a hierarchy.

Lemma~\ref{lem:overlaprule} can be rephrased in a more concise form
with help of the following
\begin{definition}[Property (L)]
  \label{def:L}
  A clustering system $\mathscr{C}$ satisfies property (L) if
  $C_1\cap C_2=C_1\cap C_3$ for all $C_1,C_2,C_3\in\mathscr{C}$ where $C_1$
  overlaps both $C_2$ and $C_3$.
\end{definition}
For later reference we record an equivalent way of expressing property (L):
\begin{corollary}
  \label{cor:L-more}
  A clustering system $\mathscr{C}$ satisfies property (L) if and only
  if $C_1\cap C_2\in \{\emptyset, C_1, C_2, C_1\cap C\}$
  for all $C\in \mathscr{C}$ that overlap with $C_1$.
\end{corollary}
\begin{proof}
  Let  $C_1,C_2$ be chosen arbitrarily from $\mathscr{C}$.
  If $C_1\subseteq C_2$, $C_2\subseteq C_1$ or $C_1\cap C_2 = \emptyset$,
  then $C_1\cap C_2\in \{\emptyset, C_1, C_2\}$ and there is nothing to show.
  If neither of the latter cases is
  satisfied, then $C_1$ and $C_2$ overlap.
  Let $\mathscr{C}'\subseteq \mathscr{C}$ be the subset of cluster $C\in
  \mathscr{C}$	that overlap with $C_1$.
  By construction, $C_1$  overlaps with all elements in $\mathscr{C}'$
  and $C_2\in \mathscr{C}'$. Hence, Property (L) holds if and only if
  $C_1\cap C_2 = C_1\cap C$ for all $C\in  \mathscr{C}'$.
\end{proof}

\begin{corollary}
  \label{cor:L}
  The clustering system $\mathscr{C}_N$ of every level-1 network $N$
  satisfies property (L).
\end{corollary}
\begin{proof}
  Let $C_1,C_2,C_3\in\mathscr{C}$ such that $C_1$ overlaps both $C_2$ and
  $C_3$. Hence, $C_1\cap C_3\notin\{C_1,C_3,\emptyset\}$.  This together
  with Lemma \ref{lem:overlaprule} implies that
  $C_1\cap C_3\in\{C_1\cap C_2\}$.
\end{proof}
\begin{corollary}\label{cor:L->wH}
  A clustering system $\mathscr{C}$ satisfying Property (L) is a
  weak hierarchy.
\end{corollary}
\begin{proof}
  If one of the three sets $C_1$, $C_2$ and $C_3$ is contained in another
  one, or if one of three pairwise intersections is empty, then the
  assertion follows immediately. If $C_1$ overlaps both $C_2$ and $C_3$
  then (L) implies $C_1\cap C_2=C_1\cap C_3=C_1\cap C_2\cap C_3$.
\end{proof}

In order to show that closed clustering systems with property (L) define
level-1 networks, we first demonstrate that one can identify the
non-trivial blocks directly in a clustering system provided it satisfies
(L). We start by introducing subsets of clusters with a given overlap: For
a given clustering system $\mathscr{C}$ and a set $C\in\mathscr{C}$, define
\begin{equation}
  \label{eq:B0(C)}
  \mathcal{B}^0(C) \coloneqq \{C'\in\mathscr{C}\setminus\{C\}
  \mid\ \text{there is a } C''\in\mathscr{C}\setminus\{C\} \text{
    s.t. } C'\cap C''=C\}\,.
\end{equation}
Note that the clusters $C'$ and $C''$ appearing Eq.(\ref{eq:B0(C)}) are
different from $C'\cap C''$ and thus, must overlap. Furthermore, we observe
that $C''\in\mathcal{B}^0(C)$ and $\mathcal{B}^0(C)=\emptyset$ if and only
if $C$ is not the intersection of two overlapping clusters. In particular,
we have
\begin{lemma}
  Let $\mathscr{C}$ be a clustering system. Then, $\mathscr{C}$ is closed
  and $\mathcal{B}^0(C)=\emptyset$ for all $C\in \mathscr{C}$ if and only
  if there is a phylogenetic tree $T$ with $\mathscr{C} = \mathscr{C}_T$.
\end{lemma}
\begin{proof}
  First, let $T$ be phylogenetic tree with $\mathscr{C} = \mathscr{C}_T$.
  By Cor.~\ref{cor:l1->closed-weak-hierarchy}, $\mathscr{C}$ must be
  closed. Moreover, since $T$ is a tree, we have
  $C'\cap C'' \in \{C',C'', \emptyset\}$ and thus,
  $\mathcal{B}^0(C)=\emptyset$ since $C',C''$ must satisfy $C',C''\neq C$.

  Assume now that $\mathscr{C}$ is closed and $\mathcal{B}^0(C)=\emptyset$
  for all $C\in \mathscr{C}$.  Since $\mathscr{C}$ is closed, Lemma
  \ref{lem:simple-closed} implies that $C'\cap C''\in \mathscr{C}$ for all
  $C',C''\in \mathscr{C}$ whenever $C'\cap C''\ne\emptyset$.  Therefore, we
  have $C'\cap C'' \in \{C',C'', \emptyset\}$ for all
  $C',C''\in \mathscr{C}$ since otherwise $C'$ and $C''$ overlap and we
  obtain $C',C'\in \mathcal{B}^0(D)\ne\emptyset$ for
  $D=C'\cap C''\in \mathscr{C}$; a contradiction.  Thus, $\mathscr{C}$ is a
  hierarchy and, by \cite[Thm.~3.5.2]{sem-ste-03a}, there is a 1-to-1
  correspondence between hierarchies and phylogenetic trees $T$ such that
  $\mathscr{C} = \mathscr{C}_T$.
\end{proof}

\begin{lemma}
  \label{lem:B0}
  Let $\mathscr{C}$ be a clustering system satisfying (L), and
  $C\in\mathscr{C}$ with $\mathcal{B}^0(C)\ne\emptyset$.  Then every
  cluster $D\in \mathscr{C} \setminus \mathcal{B}^0(C)$ satisfies ones the
  the following alternatives: (i) $D\subseteq C$, (ii) $D\cap C=\emptyset$,
  or (iii) $C'\subsetneq D$ for all $C'\in\mathcal{B}^0(C)$.
\end{lemma}
\begin{proof}
  Consider a cluster $D\in\mathscr{C}$ with $D\notin\mathcal{B}^0(C)$.  By
  contraposition, assume that none of the alternatives (i), (ii) and (iii)
  are satisfied.  Hence, suppose $D\not\subseteq C$, i.e.,
  $D\setminus C\ne\emptyset$, and $D\cap C\ne \emptyset$, and that there is
  some set $C'\in\mathcal{B}^0(C)$ such that $C'\setminus D\ne
  \emptyset$. The case $C'=D$ cannot occur since
  $D\notin\mathcal{B}^0(C)$. By definition, there is
  $C''\in\mathcal{B}^0(C)$ such that $C'$ and $C''$ overlap with
  $C'\cap C''=C$.  In particular, we have $C\subseteq C'$ and
  $C\subseteq C''$, and thus $D\cap C'\ne \emptyset$ and
  $D\cap C''\ne \emptyset$.  From $D\cap C'\ne \emptyset$ and
  $C'\setminus D\ne \emptyset$ we infer that either $D\subseteq C'$ or $C'$
  and $D$ overlap.  If $C'$ and $D$ overlap, then (L) implies
  $C'\cap D=C'\cap C''=C$, and thus $D\in\mathcal{B}^0(C)$; a
  contradiction. Thus we have $D\subseteq C'$ and hence
  $C''\setminus D\ne \emptyset$.  Moreover, $D\subseteq C'$,
  $C'\cap C''=C$, and $D\setminus C\ne\emptyset$ imply that
  $D\setminus C''\ne\emptyset$.  Together with $D\cap C''\ne \emptyset$, we
  obtain that $C''$ overlaps with both $C'$ and $D$. Now (L) implies
  $C''\cap D=C''\cap C'=C$, and thus $D\in\mathcal{B}^0(C)$.
\end{proof}

\begin{corollary}
  \label{cor:B0-C-no-overlap}
  Let $\mathscr{C}$ be a clustering system satisfying (L), and
  $C\in\mathscr{C}$ with $\mathcal{B}^0(C)\ne\emptyset$.  Then $C$ does not
  overlap with any cluster in $\mathscr{C}$.
\end{corollary}
\begin{proof}
  If $D\in \mathscr{C} \setminus \mathcal{B}^0(C)$, then Lemma~\ref{lem:B0}
  implies that (i) $D\subseteq C$, (ii) $D\cap C=\emptyset$, or (iii)
  $C\subsetneq C'\subsetneq D$ for all $C'\in\mathcal{B}^0(C)\ne
  \emptyset$. Hence, $D$ does not overlap with $C$ in any of the three
  cases.  If, on the other hand, $D\in \mathcal{B}^0(C)$, then $C\subsetneq
  D$, and thus, $D$ and $C$ also do not overlap.
\end{proof}

Alternative (iii) in Lemma~\ref{lem:B0} can be expressed equivalently with
the help of the set
\begin{equation}
  \label{eq:U(C)}
  U(C) \coloneqq \bigcup_{C'\in\mathcal{B}^0(C)} C' \,.
\end{equation}
\begin{lemma}
  \label{lem:B(C)subD}
  Let $\mathscr{C}$ be a clustering system satisfying (L) and let
  $C\in\mathscr{C}$ such that $\mathcal{B}^0(C)\ne\emptyset$.  Then,
  $D\in\mathscr{C}$ satisfies $C'\subsetneq D$ for all
  $C'\in\mathcal{B}^0(C)$ if and only if $U(C)\subseteq D$.
\end{lemma}
\begin{proof}
  The ``only if statement'' follows directly from the definition of $U(C)$.
  To see the ``if'' direction, note first that $U(C)\subseteq D$ implies
  $D\notin\mathcal{B}^0(C)$ since otherwise there would be a
  $C''\in \mathcal{B}^0(C)$ overlapping $D$, which is impossible because
  $C''\subseteq U(C)\subseteq D$. Furthermore, we have $C'\subsetneq D$ for
  all $C'\in \mathcal{B}^0(C)$ since $U(C)$, and thus also $D$, contains
  at least one set overlapping $C'$.
\end{proof}

As a consequence we can use the condition $U(C)\subseteq D$ instead of
alternative (iii) in Lemma~\ref{lem:B0}. If the clustering system
$\mathscr{C}$ is closed, $U(C)\subseteq D$ is equivalent to requiring
$\cl(U(C))\subseteq \cl(D)=D$. Therefore we define
\begin{equation}
  \label{eq:B(C)}
  \Top(C)\coloneqq
  \cl\left(U(C)\right)
\end{equation}
Alternative (iii) in Lemma~\ref{lem:B0} can now be expressed as
$\Top(C)\subseteq D$.

\begin{figure}
  \begin{center}
    \includegraphics[width=0.55\textwidth]{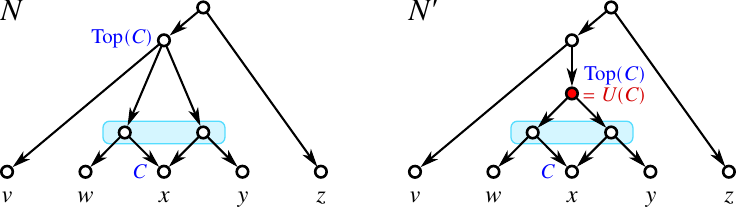}
  \end{center}
  \caption{$U(C)$ is not necessarily a cluster. In both networks
    $\mathcal{B}(C)$ for $C=\CC\{x\}=\{x\}$ is the only non-trivial block.
    The vertices with overlapping clusters, i.e., the set
    $\mathcal{B}^0(C)$ is highlighted in cyan. In both networks
    $U(C)=\{w,x,y\}$. L.h.s.: $U(C)\subsetneq \Top(C)$. The additional
    red vertex r.h.s., ensures that $U(C)=\Top(C)$ is a cluster.}
  \label{fig:UC}
\end{figure}

In general, $U(C)\notin\mathscr{C}$, see Fig.~\ref{fig:UC}. Nevertheless,
$\Top(C)\in\mathscr{C}$ for all $C$ of a closed clustering system
$\mathscr{C}$ that satisfy $\mathcal{B}^0(C)$.
\begin{lemma}\label{lem:BCinClusterSystem}
  Let $\mathscr{C}$ be a closed clustering system. Then,
  $\Top(C)\in\mathscr{C}$ and $C\subseteq \Top(C)$ for all
  $C\in\mathscr{C}$ with $\mathcal{B}^0(C)\neq \emptyset$.
\end{lemma}
\begin{proof}
  Let $\mathscr{C}$ be a closed clustering system on $X$.  Since
  $\mathscr{C}$ is closed we have, by definition, $\cl(A)=A$ if and only if
  $A\in \mathscr{C}$ for all non-empty sets $A\subseteq 2^X$. Let
  $C\in\mathscr{C}$ with $\mathcal{B}^0(C)\neq \emptyset$.  Hence,
  $U(C)\neq \emptyset$. Since $\cl$ is enlarging, we have
  $U(C)\subseteq \cl(U(C))=\Top(C)$ and, therefore, $\Top(C)\neq \emptyset$. In
  particular, $\Top(C)\subseteq 2^X$. Since $\cl$ is idempotent, we have
  $\cl(\Top(C)) = \cl(\cl(U(C))) = \cl(U(C)) = \Top(C)$.  Taking the latter
  arguments together, we obtain $\Top(C)\in \mathscr{C}$.  Moreover, if
  $\mathcal{B}^0(C)\neq \emptyset$, then it is straightforward to verify
  that $C\subseteq U(C)\subseteq \Top(C)$.
\end{proof}

Although $U(C)\notin\mathscr{C}$, it has an interesting property:
\begin{lemma}
  \label{lem:U-nix-overlap}
  Let $\mathscr{C}$ be a clustering system satisfying (L), and
  $C\in\mathscr{C}$ with $\mathcal{B}^0(C)\ne\emptyset$. Then $U(C)$ does
  not overlap any cluster $D\in\mathscr{C}$.
\end{lemma}
\begin{proof}
  Let $D\in\mathscr{C}$.  If $D\in\mathcal{B}^0(C)$, then
  $D\subseteq U\coloneqq U(C)$ by construction.  If
  $D\notin\mathcal{B}^0(C)$, we consider the three alternatives of
  Lemma~\ref{lem:B0}.  In case (i) we have $D\subseteq C\subseteq U$ and in
  case (iii) we have $U\subseteq \Top(C)\subseteq D$. In case (ii) we have
  $D\cap C=\emptyset$. If $D\cap C'=\emptyset$ for all
  $C'\in\mathcal{B}^0(C)$, then $D\cap U=\emptyset$. Otherwise
  $C'\cap D\ne\emptyset$ for some $C'\in\mathcal{B}^0(C)$. Since
  $C\subseteq C'$ and $D\cap C=\emptyset$, we have $C'\not\subseteq
  D$. However, $C'$ and $D$ cannot overlap since in this case (L) implies
  $C'\cap D=C$ and thus, $D\in{B}^0(C)$, a contradiction. Therefore,
  $D\subseteq C'\setminus C$, which implies $D\subsetneq U$.  Thus $D$ does
  not overlap $U$.
\end{proof}

So far, $\Top(C)$ is defined in terms of non-empty sets $\mathcal{B}^0(C)$.
We extend this notion to all clusters of a closed clustering system as
follows.  Let $\mathscr{C}$ be a closed clustering system.  We set
$\Top(X)\coloneqq X$. For the remaining clusters $C\in\mathscr{C}$, i.e.,
those that satisfy $\mathcal{B}^0(C)=\emptyset$ and $C\ne X$, we define
$\Top(C)$ as the unique inclusion-minimal cluster $C'\ne C$ that contains $C$.
To see that $\Top(C)$ is well defined in this case, recall first that
$X\in\mathscr{C}$ and hence there is a cluster properly containing
$C\ne X$. For uniqueness, suppose there are two distinct inclusion-minimal
clusters $C',C''\ne C$ that contain $C$. Clearly, these two supersets
overlap with $C\subseteq C^*\coloneqq C'\cap C''$. If $C= C^*$, then
$\mathcal{B}^0(C)\ne\emptyset$; a contradiction.  If $C\subsetneq C^*$, we
have $C^*\in\mathscr{C}$ since $\mathscr{C}$ is closed and thus
$C\subsetneq C^*\subsetneq C,C'$ contradicts inclusion-minimality of $C$
and $C'$. Now set
\begin{equation}
  \label{eq:calB(C)}
  \mathcal{B}(C)\coloneqq \mathcal{B}^0(C)\cup\{C,\Top(C)\} \text{ for all }
  C\in\mathscr{C}.
\end{equation}

\begin{corollary}\label{cor:BC-in-C}
  If $\mathscr{C}$ is a closed clustering system, then
  $\mathcal{B}(C)\subseteq \mathscr{C}$ for all $C\in\mathscr{C}$.
\end{corollary}
\begin{proof}
  If $\mathcal{B}^0(C) = \emptyset$, we have by construction,
  $\mathcal{B}(C)\coloneqq \emptyset \cup\{C,C'\} = \{C,C'\}$, where $C'$ is
  the unique inclusion-minimal element in $\mathscr{C}$ that contains $C$.
  Hence, $\mathcal{B}(C)\subseteq \mathscr{C}$. The latter covers in
  particular also the case $C=X$.  Otherwise, if
  $\mathcal{B}^0(C) \neq \emptyset$, then Lemma \ref{lem:BCinClusterSystem}
  implies that $\Top(C)\in \mathscr{C}$.  Moreover, by definition,
  $\mathcal{B}^0(C)\subseteq \mathscr{C}$.  Taken the latter together with
  $C\in\mathscr{C}$ implies $\mathcal{B}(C)\subseteq \mathscr{C}$.
\end{proof}

Lemma~\ref{lem:B0} then implies the following characterization of
$\mathcal{B}(C)$:
\begin{corollary}
  \label{cor:sandwich}
  Let $\mathscr{C}$ be a closed clustering system satisfying (L). Then, for
  all $C,D\in\mathscr{C}$ it holds that $D\in\mathcal{B}(C)$ if and only if
  $C\subseteq D\subseteq \Top(C)$.
\end{corollary}
\begin{proof}
  Let $C,D\in\mathscr{C}$ such that
  $D\in\mathcal{B}(C) = \mathcal{B}^0(C)\cup\{C,\Top(C)\}$.  If
  $D \in \mathcal{B}^0(C)$ or $D=C$, then $C\subseteq D$ and, by
  construction, $D \subseteq \Top(C)$.  If $D = \Top(C)$, we can apply
  Lemma \ref{lem:BCinClusterSystem} to conclude that
  $C\subseteq \Top(C)=D$.  Now, let $C\subseteq D\subseteq \Top(C)$ and
  assume, for contradiction, that $D\notin\mathcal{B}(C)$.  Thus,
  $C\subsetneq D\subsetneq \Top(C)$ and
  $D\notin\mathcal{B}^0(C)$. From $D\subsetneq \Top(C)$, we obtain
  $\mathcal{B}^0(C)\neq \emptyset$. Hence, we can apply Lemma~\ref{lem:B0} to
  conclude that the cluster $D$ satisfies one of the alternatives: (i)
  $D\subseteq C$, (ii) $D\cap C=\emptyset$, or (iii) $C'\subsetneq D$ for all
  $C'\in\mathcal{B}^0(C)$.  Based on the latter arguments only case (iii) can
  occur, and hence Lemma~\ref{lem:B(C)subD} and Eq.~\eqref{eq:B(C)} yield
  $\Top(C)\subseteq D$; a contradiction.
\end{proof}

As an immediate consequence of Cor.~\ref{cor:sandwich} we have
\begin{corollary}
  \label{cor:inducedsubgraph}
  Let $\mathscr{C}$ be a closed clustering system satisfying (L). Then, for
  all clusters $C\in\mathscr{C}$, the Hasse diagram
  $\Hasse[\mathcal{B}(C)]$ is an induced subgraph of the Hasse diagram
  $\Hasse[\mathscr{C}]$.
\end{corollary}

Furthermore, the sets $\mathcal{B}^0(C)$ are pairwise disjoint:
\begin{lemma}
  \label{lem:disjointB0}
  Let $\mathscr{C}$ be a closed clustering system satisfying (L) and let
  $C,C'\in\mathscr{C}$. Then $C'\in\mathcal{B}^0(C)$ implies
  $\mathcal{B}^0(C')=\emptyset$. Furthermore, if
  $\mathcal{B}^0(C)\cap \mathcal{B}^0(C')\ne\emptyset$, then $C=C'$.
  Consequently, $\mathcal{B}^0(C)\cap \mathcal{B}^0(C') =\emptyset$ for all
  distinct $C,C'\in\mathscr{C}$.
\end{lemma}
\begin{proof}
  Suppose that $C'\in\mathcal{B}^0(C)$. Then there is a set
  $C''\in\mathscr{C}$ such that $C'$ and $C''$ overlap and $C=C'\cap
  C''$. Assume for contradiction that $\mathcal{B}^0(C')\ne\emptyset$,
  i.e., there are two overlapping clusters $D,D'\in\mathscr{C}$ such that
  $C'=D\cap D'$. Since $C''$ overlaps $C'$ it cannot be contained in both
  $D$ and $D'$ since otherwise $C''\subseteq D\cap D'=C'$ and hence, $C'$
  and $C''$ would not overlap. Thus, at least one of $C''\setminus D$ and
  $C''\setminus D'$ is non-empty, say $C''\setminus
  D\ne\emptyset$. Moreover, $C'\subsetneq D$ and $C'$ and $C''$ overlapping
  each other imply that $D\setminus C''\ne\emptyset$ and
  $D\cap C''\ne\emptyset$. Hence, $C''$ and $D$ overlap.  By (L),
  $C = C'\cap C'' = C''\cap D = D'\cap D = C'$ and thus
  $C'=C\subseteq C''$; a contradiction to the assumption that $C'$ and
  $C''$ overlap. Hence, $\mathcal{B}^0(C')=\emptyset$ as claimed.
  Now assume that $\mathcal{B}^0(C)\cap \mathcal{B}^0(C')\ne\emptyset$,
  i.e., there is a cluster $C''\in \mathcal{B}^0(C)\cap \mathcal{B}^0(C')$
  and clusters $D\in \mathcal{B}^0(C)$ and $D'\in\mathcal{B}^0(C')$, both
  of which overlap with $C''$, such that $C=C''\cap D$ and $C'=C''\cap
  D'$. By (L), this implies $C=C'$.
\end{proof}

\begin{lemma}
  \label{lem:multiple-inneighbors}
  Let $\mathscr{C}$ be a closed clustering system satisfying (L).  Then
  $C\in\mathcal{C}$ has indegree greater than one in
  $\Hasse[\mathscr{C}]$ if and only if
  $\mathcal{B}^0(C)\ne\emptyset$.  In this case, all in-neighbors of $C$ in
  $\Hasse[\mathscr{C}]$ are contained in $\mathcal{B}^0(C)$.
\end{lemma}
\begin{proof}
  Suppose $C\in\mathcal{C}$ has indegree greater than one in
  $\Hasse[\mathscr{C}]$. Thus, let $D,D'\in \mathscr{C}$ be two
  distinct in-neighbors of $C$. Hence, $C\subseteq D\cap D'$ and thus, by
  closedness of $\mathscr{C}$ and definition of
  $\Hasse[\mathscr{C}]$, $C= D\cap D'$.  In particular, $D$ and $D'$
  overlap.  Therefore, $D,D'\in\mathcal{B}^0(C)\ne\emptyset$. In
  particular, since the in-neighbors $D$ and $D'$ were chosen arbitrarily,
  all in-neighbors of $C$ in $\Hasse[\mathscr{C}]$ are contained in
  $\mathcal{B}^0(C)$.  Now suppose $C\in\mathcal{C}$ has indegree zero or
  one in $\Hasse[\mathscr{C}]$. Clearly, $C$ has indegree zero if
  and only if $C=X$, in which case $\mathcal{B}^0(C)=\emptyset$.  Suppose
  $C$ has exactly one in-neighbor $C'$ and, for contradiction, that
  $\mathcal{B}^0(C)\ne\emptyset$.  Then there are two overlapping sets
  $D,D'\in \mathcal{B}^0(C)$ such that $C=D\cap D'$, and thus directed
  paths both from $D$ and $D'$ to $C$. Both of these paths must pass
  through $C'$ and thus, $C\subsetneq C'\subseteq D\cap D'$; a
  contradiction.  Hence, the \emph{if}-direction must also hold.
\end{proof}
Lemmas~\ref{lem:disjointB0} and~\ref{lem:multiple-inneighbors} imply
\begin{corollary}
  \label{cor:B0-indegree1}
  Let $\mathscr{C}$ be a closed clustering system satisfying (L) and with
  Hasse diagram $\Hasse$.  For every $C\in\mathscr{C}$, the elements
  $C'\in\mathcal{B}^0(C)$ have a unique in-neighbor in $\Hasse$. In
  particular, this unique in-neighbor of $C'$ is always contained in
  $\mathcal{B}(C)$.
\end{corollary}
\begin{proof}
  The statement is trivially true for all $C\in\mathscr{C}$ with
  $\mathcal{B}^0(C)=\emptyset$.  Thus consider a set $C\in\mathscr{C}$ with
  $\mathcal{B}^0(C)\ne\emptyset$.  Note that $C\ne X$ must hold.  Consider
  a cluster $C'\in \mathcal{B}^0(C)$. It overlaps with some
  $C''\in \mathcal{B}^0(C)$ and thus $C',C''\subsetneq \Top(C)$. Therefore,
  there is a directed path from $\Top(C)$ to $C'$ and thus, $C'$ has an
  in-neighbor $C^*$ that satisfies
  $C\subsetneq C'\subsetneq C^*\subseteq \Top(C)$. By Cor.~\ref{cor:sandwich},
  we have $C^*\in \mathcal{B}(C)$.  Lemma~\ref{lem:disjointB0} and
  $C'\in\mathcal{B}^0(C)$ imply $\mathcal{B}^0(C')=\emptyset$.  Hence, $C'$
  has indegree $1$ by Lemma~\ref{lem:multiple-inneighbors}, i.e.,
  $C^*\in \mathcal{B}(C)$ is the unique in-neighbor of $C$.
\end{proof}

\begin{lemma}
  \label{lem:biconnected}
  Let $\mathscr{C}$ be a closed clustering system satisfying (L). Let
  $C\in\mathscr{C}$ with $\mathcal{B}^0(C)\ne\emptyset$.  Then the induced
  subgraph $\Hasse[\mathcal{B}(C)]$ of $\Hasse[\mathscr{C}]$ is
  biconnected. In particular, $\Hasse[\mathcal{B}(C)]$ is a DAG with unique
  source $\Top(C)$ and unique sink $C$.
\end{lemma}
\begin{proof}
  By Cor.~\ref{cor:inducedsubgraph}, $\Hasse[\mathcal{B}(C)]$ is an induced
  subgraph of $\Hasse[\mathscr{C}]$.  By Lemma~\ref{cor:B0-indegree1}, all
  clusters in $\mathcal{B}^0(C)$ have a unique in-neighbor in
  $\mathcal{B}(C)$. By Cor.~\ref{cor:sandwich}, $C'\subseteq \Top(C)$ holds
  for all $C'\in\mathcal{B}^0(C)$.  Therefore, $\Top(C)$ has indegree $0$
  in $\Hasse[\mathcal{B}(C)]$ and, moreover, there exists a directed path
  from $\Top(C)$ to every cluster $C'\in\mathcal{B}^0(C)$. In particular,
  by Cor.~\ref{cor:sandwich}, of the clusters in such paths are again
  contained in $\mathcal{B}(C)$.  Taken together, these arguments imply
  that the Hasse diagram $\Hasse[\mathcal{B}(C)\setminus\{C\}]$ is a tree
  with root $\Top(C)$.  Note that this tree is not necessarily
  phylogenetic, i.e, there may exist clusters with outdegree $1$. However,
  the outdegree of the root $\Top(C)$ in
  $\Hasse[\mathcal{B}(C)\setminus\{C\}]$ is at least two.  To see this, let
  $C'$ be a cluster in $\mathcal{B}^0(C)\ne \emptyset$.  As argued above,
  there is a directed path from $\Top(C)$ to $C'$ and this path only
  contains clusters in $\mathcal{B}(C)$.  Therefore and since
  $C'\ne \Top(C)$, $\Top(C)$ has a child $D'$ in $\Hasse$ with
  $C\subsetneq C' \subseteq D' \subsetneq \Top(C)$. By
  Cor.~\ref{cor:sandwich}, we have $D'\in \mathcal{B}(C)$, and thus,
  $D'\in \mathcal{B}^0(C)$.  Hence, there is $C''\in \mathcal{B}^0(C)$ such
  that $D'$ and $C''$ overlap.  By similar argument as before, there is a
  child $D''\in \mathcal{B}^0(C)$ of $\Top(C)$ such that
  $C'' \subseteq D'' \subsetneq \Top(C)$. Now $C'' \subseteq D''$ and the
  fact that $D'$ and $C''$ overlap imply that $D'\ne D''$. Hence, $\Top(C)$
  has at least two children in $\Hasse[\mathcal{B}(C)\setminus\{C\}]$.
  Using Cor.\ref{cor:sandwich}, we see that each leaf of the tree induced
  by $\mathcal{B}(C)\setminus\{C\}$ is an in-neighbor of $C$.  It is now
  easy to verify the graph obtained from (i) a rooted tree whose root has
  at least two children and (ii) connecting its leaves to an additional
  vertex is biconnected.  Hence, $\Hasse[\mathcal{B}(C)]$ is
  biconnected. In particular, $\Hasse[\mathcal{B}(C)]$ features at least
  two internally vertex disjoint directed path connecting $\Top(C)$ and
  $C$, and any two vertices lie along a common ``undirected'' cycle (which
  necessarily passes through $C$).
\end{proof}

\begin{lemma}
  \label{lem:DD}
  Let $\mathscr{C}$ be a closed clustering system satisfying (L).  Let
  $D\in\mathcal{B}^0(C)$ for some $C\in\mathscr{C}$ and let
  $D'\notin\mathcal{B}(C)$ be adjacent to $D$ in the Hasse diagram $\Hasse$
  of $\mathscr{C}$. Then
  \begin{description}
  \item[(i)] $D$ is the unique in-neighbor of $D'$ in $\Hasse$ and thus
    $D'\subsetneq D$,
  \item[(ii)] $D'\cap C =\emptyset$, and
  \item[(iii)] if $D'$ overlaps with
    some $D''\in\mathscr{C}$, then there is $C'\in\mathscr{C}$ such that
    $D'\in\mathcal{B}^0(C')$ and $\Top(C')=D$.
  \end{description}
\end{lemma}
\begin{proof}
  We start with showing Property (i). By Cor.~\ref{cor:B0-indegree1}, the
  unique in-neighbor of $D$ is contained in $\mathcal{B}(C)$. Thus, $D'$
  must be an out-neighbor of $D$, i.e., $D'\subsetneq D$. If
  $\mathcal{B}^0(D')\ne\emptyset$, then
  Lemma~\ref{lem:multiple-inneighbors} implies $D\in\mathcal{B}^0(D')$.
  Lemma~\ref{lem:disjointB0} and $D\in \mathcal{B}^0(C)\cap
  \mathcal{B}^0(D')$ imply $D'=C\in \mathcal{B}(C)$; a contradiction.
  Hence, $\mathcal{B}^0(D')=\emptyset$ and in particular, by
  Lemma~\ref{lem:multiple-inneighbors}, $D'$ has indegree $1$, and thus
  $D$ is its unique in-neighbor.

  We continue with showing Property (ii). Since
  $\mathcal{B}^0(C)\neq \emptyset$ and $D'\notin\mathcal{B}^0(C)$,
  Lemma~\ref{lem:B0} implies that (a) $D'\subseteq C$, (b)
  $D'\cap C=\emptyset$, or (c) $C'\subsetneq D'$ for all
  $C'\in\mathcal{B}^0(C)$.  In Case~(a), $D'\notin\mathcal{B}(C)$ implies
  $D'\subsetneq C$.  Since moreover $D\in\mathcal{B}^0(C)$ and thus
  $C\subsetneq D$, we have $D'\subsetneq C\subsetneq D$, contradicting that
  $D'$ and $D$ are adjacent in $\Hasse$.  In Case~(c), we obtain
  $D\subsetneq D'$ since $D\in \mathcal{B}^0(C)$; contradicting
  $D'\subsetneq D$.  Hence, only Case~(b) $D'\cap C=\emptyset$ can hold.

  Finally, we show Property (iii). Suppose that $D'$ overlaps $D''$ and set
  $C'=D'\cap D''$ and thus, $D'\in\mathcal{B}^0(C')$. Since
  $D'\notin\mathcal{B}^0(C)\subseteq \mathcal{B}(C)$ it must hold that
  $C'\ne C$ and thus $\mathcal{B}^0(C')\cap\mathcal{B}^0(C)=\emptyset$ by
  Lemma~\ref{lem:disjointB0}. Since $D$ is the unique in-neighbor of $D'$
  in $\Hasse$, Cor.\ \ref{cor:B0-indegree1} implies $D\in\mathcal{B}(C')$
  and thus $D'\subsetneq D\subseteq \Top(C')$. On the other hand,
  $D\in\mathcal{B}^0(C)$ implies $D\notin\mathcal{B}^0(C')$ and hence
  $D\not\subsetneq \Top(C')$; Therefore $D=\Top(C')$.
\end{proof}

\begin{lemma}
  \label{lem:B-C-nontrivial-blocks}
  Let $\mathscr{C}$ be a closed clustering system satisfying (L). Then each
  subgraph $\Hasse[\mathcal{B}(C)]$ with $\mathcal{B}^0(C)\ne\emptyset$ is
  a non-trivial block of the Hasse diagram $\Hasse$ of $\mathscr{C}$.
\end{lemma}
\begin{proof}
  By Lemma~\ref{lem:biconnected}, $\Hasse[\mathcal{B}(C)]$ is
  biconnected. Therefore and since $\mathcal{B}^0(C)\ne\emptyset$, the set
  $\mathcal{B}(C)$ contains at least four clusters, i.e., $C$, $\Top(C)$,
  and at least two overlapping clusters in $\mathcal{B}^0(C)$.  Thus, it
  only remains to show that $\Hasse[\mathcal{B}(C)]$ is a maximal
  biconnected subgraph of $\Hasse$.  Since moreover, by
  Cor.~\ref{cor:inducedsubgraph}, $\Hasse[\mathcal{B}(C)]$ is an induced
  subgraph of $\Hasse$, $\Hasse[\mathcal{B}(C)]$ is a maximal if and only
  if there is no undirected cycle in $\Hasse$ that contains an arc of
  $\Hasse[\mathcal{B}(C)]$ and a vertex not contained in $\mathcal{B}(C)$
  (cf.\ Obs.\ \ref{obs:identical-block}).  Assume, for contradiction, that
  such a cycle $K$ exists. Since $K$ contains at least one arc of
  $\Hasse[\mathcal{B}(C)]$, we can find a maximal subpath $P$ of $K$ on at
  least two vertices and where all vertices of $P$ are contained in
  $\mathcal{B}(C)$.  In particular, the two distinct endpoints of $P$ are
  both incident with one cluster in $\mathcal{B}(C)$ and one cluster that
  is not in $\mathcal{B}(C)$. Clearly, at least one of the endpoints of $P$
  must be distinct from $\Top(C)$.  Hence, we can pick an endpoint $D\in
  \mathcal{B}(C) \setminus\{\Top(C)\}$ of $P$ that is adjacent to
  $C'\in\mathcal{B}(C)$ and $D'\notin\mathcal{B}(C)$, where both $C'$ and
  $D'$ are vertices in $K$.  Therefore, it suffices to consider the two
  mutually exclusive cases (a) $D=C$ and (b) $D\in\mathcal{B}^0(C)$:

  (a) $D=C$. Hence, $C'\neq C$ and thus, by Cor.~\ref{cor:sandwich},
  $C'\in\mathcal{B}(C)\setminus\{C\}$ implies $D=C\subsetneq C'$ and thus,
  $C'\in\overline{\mathcal{D}}(C)$ (cf.\ Equ.\ \eqref{eq:D}).  Suppose, for
  contradiction, that $C$ overlaps with some cluster
  $D''\in\mathscr{C}$. Then, since $\mathscr{C}$ is closed, we have $C\in
  \mathcal{B}^0(E)$ for $E=C\cap C''\in\mathscr{C}$. However, this together
  with Lemma~\ref{lem:disjointB0} implies $\mathcal{B}^0(C)\ne\emptyset$; a
  contradiction. Hence, $C$ does not overlap any cluster.  Furthermore, by
  Lemma~\ref{lem:multiple-inneighbors}, all in-neighbors of $C$ are
  contained in $\mathcal{B}^0(C)\subsetneq \mathcal{B}(C)$.  Therefore,
  $D'$ must be an out-neighbor of $C$ and thus $D'\subsetneq C$, which
  implies $D'\in\mathcal{D}(C)$.  Hence, we can apply
  Lemma~\ref{lem:cutvertex} to conclude that there is no cycle $K$
  containing $D'\in\mathcal{D}(C)$ and $C'\in\overline{\mathcal{D}}(C)$; a
  contradiction.

  (b) $D\in\mathcal{B}^0(C)$. By Lemma~\ref{lem:DD}, $D$ is the unique
  in-neighbor of $D'$. However, since $D'$ is located on the cycle $K$, it
  must be adjacent to another vertex $D''\neq D$ in $K$.  Since $D$ is the
  unique in-neighbor of $D'$ it follows that $D''$ must be an out-neighbor
  of $D'$ and thus, $D''\subsetneq D'$. By construction,
  $D''\in\mathcal{D}(D')$ and $D\in\overline{\mathcal{D}}(D')$.  If $D'$
  does not overlap any cluster in $\mathscr{C}$, then we can apply
  Lemma~\ref{lem:cutvertex} to conclude that there is no cycle $K$ in
  $\Hasse$ containing $D''\in\mathcal{D}(D')$ and
  $D\in\overline{\mathcal{D}}(D')$; a contradiction. Hence, $D'$ must
  overlap with some cluster in $\mathscr{C}$.  Then Lemma~\ref{lem:DD}(iii)
  implies that there is $E\in\mathscr{C}$ such that $D'\in\mathcal{B}^0(E)$
  and $D=\Top(E)$.  In particular, since $D\ne \Top(C)$, we have $D\ne E$.
  Moreover, by Lemma~\ref{lem:DD}(ii), we have $D_1\cap C=\emptyset$ for
  all children $D_1$ of $\Top(E)=D\in \mathcal{B}^0(C)$ with
  $D_1\notin\mathcal{B}^0(C)$.  In particular, $C\subsetneq D$, and thus we
  have $U\coloneqq U(E)=\bigcup_{D_1\in\mathcal{B}^0(E)}D_1 \subsetneq
  \Top(E)=D$. On the other hand, we have $D_1\subsetneq U$ for each of the
  children of $\Top(E)$. Since $C\subsetneq D$, $D$ has at least one child
  $F$ such that $C\subseteq F$.
  We continue with showing that $C\cap D_1=\emptyset$ for all
  $D_1\in\mathcal{B}^0(E)$. Hence, let $D_1\in\mathcal{B}^0(E)$. By
  Cor.~\ref{cor:B0-C-no-overlap}, $C$ does not overlap with any cluster in
  $\mathscr{C}$. In particular, this yield $C\ne D_1\in \mathcal{B}^0(E)$
  and $C$ and $D_1$ do not overlap.  The case $C\subsetneq D_1$ is not
  possible since otherwise $C\subsetneq D_1 \subsetneq \Top(E)=D\subsetneq
  \Top(C)$ and Cor.~\ref{cor:sandwich} would imply that
  $D_1\in\mathcal{B}^0(C)$. Together with Lemma~\ref{lem:disjointB0}, this
  would imply $C=E$; a contradiction. Now suppose $D_1\subsetneq C$. Thus,
  we have $E\subsetneq D_1\subsetneq C\subsetneq D = \Top(E)$. By
  Cor.~\ref{cor:sandwich}, this implies $C\in \mathcal{B}^0(E)$. However,
  this is not possible because $C$ does not overlap with any other cluster.
  Hence, $C\cap D_1=\emptyset$ must hold for all $D_1\in\mathcal{B}^0(E)$.
  Therefore, we obtain $U\cap C=\emptyset$. Together with $U\subseteq D$
  and $C\subsetneq D$, this implies $U\subsetneq D=\Top(E)$.  To summarize,
  since $\mathscr{C}$ is closed, it holds by definition that $\cl(U)=U \iff
  U\in \mathscr{C}$.  The latter arguments taken together with
  $\Top(E)=\cl(U)$ imply $U\notin\mathscr{C}$.  consider
  $\mathscr{C}^*\coloneqq \mathscr{C}\cup\{U\}$. Clearly, the Hasse diagram
  $\Hasse^*$ of $\mathscr{C}^*$ is obtained from $\Hasse$ by inserting a
  extra vertex $U$ as child of $D=\Top(E)$ and re-attaching the children
  $D_1$ of $\Top(E)$ with $D_1\in\mathcal{B}^0(E)$ in $\Hasse$ as children
  of $U$ in $\Hasse^*$, while the children $D_2$ of $\Top(C)$ with
  $D_2\notin\mathcal{B}^0(E)$ remain attached to $D$.  In $\mathscr{C}^*$
  we therefore have $\Top(E)=U$. Since $U$ does not overlap any set in
  $\mathscr{C}$ by Lemma~\ref{lem:U-nix-overlap}, $\mathscr{C}^*$ is again
  a closed clustering system and satisfies (L).  Moreover, since
  $U\subsetneq D$ we have $U\neq X$ and since $\mathcal{B}^0(E)\neq
  \emptyset$ we have $\vert U\vert>1$. Hence, we can apply
  Lemma~\ref{lem:cutvertex}
  to conclude that $U$ is a cut vertex in $\Hasse^*$ and that there is no
  cycle in $\Hasse^*$ containing both a vertex in $\mathcal{D}(U)$ and in
  $\overline{\mathcal{D}}(U)$.  Since $C'\in\mathcal{B}(C)$, we have
  $C\subseteq C'$, which together with $U\cap C=\emptyset$ implies that
  $C'\in \overline{\mathcal{D}}(U)$. Furthermore, $D'\in\mathcal{B}^0(E)$
  implies that $D\subsetneq U$ and thus, $D' \in \mathcal{D}(U)$. Taking
  the latter arguments together, there is no cycle in $\Hasse^*$ that
  contains both $C'$ and $D'$.  Since $\Hasse$ is recovered from $\Hasse^*$
  by ``contracting'' the arc $UD$, there is no cycle in $\Hasse$ that
  contains both $C'$ and $D'$; a contradiction.
\end{proof}

\begin{lemma}
  \label{lem:B-C-trivial-blocks}
  Let $\mathscr{C}$ be a closed clustering system on $X$ satisfying (L). If
  $C\in \mathscr{C}\setminus \{X\}$, $\mathcal{B}^0(C)=\emptyset$, and
  $C\notin\mathcal{B}^0(C')$ for all $C'\in\mathscr{C}$, then the arc
  $(\Top(C),C)$ is a block in $\Hasse[\mathscr{C}]$.
\end{lemma}
\begin{proof}
  We show that the arc $(\Top(C),C)$ is not contained in any cycle in
  $\Hasse$.  Since $\mathscr{C}$ is closed and $C\notin\mathcal{B}^0(C')$
  for all $C'\in\mathscr{C}$, we know that $C$ does not overlap any
  cluster. By Lemma \ref{lem:cutvertex}, there is no cycle that intersects
  both $\mathcal{D}(C)$ and $\overline{\mathcal{D}}(C)$.  Since
  $C\subsetneq \Top(C)$, we have $\Top(C)\in \overline{\mathcal{D}}(C)$.
  Furthermore, Lemma~\ref{lem:multiple-inneighbors} and
  $\mathcal{B}^0(C)=\emptyset$ imply that $\Top(C)$ is the only in-neighbor
  of $C$ in $\Hasse[\mathscr{C}]$. Therefore, any cycle that contains
  $(\Top(C),C)$ must contain some child $C'$ of $C$. Clearly, $C'\in
  \mathcal{D}(C)$ and thus such a cycle cannot exist as it would intersect
  both $\mathcal{D}(C)$ and $\overline{\mathcal{D}}(C)$.  Hence,
  $(\Top(C),C)$ is a cut arc, and thus a block.
\end{proof}

We summarize Lemmas~\ref{lem:B-C-nontrivial-blocks}
and~\ref{lem:B-C-trivial-blocks} in
\begin{proposition}
  \label{prop:blocks-in-Hasse}
  Let $\mathscr{C}$ be a closed clustering system on $X$ satisfying (L) and
  with Hasse diagram $\Hasse$. Then, $B$ is a block of $\Hasse$ if and only
  if $\vert X\vert=1$ or $\vert X\vert>1$ and $B= \Hasse[\mathcal{B}(C)]$ for
  some $C\in
  \mathscr{C}$ that satisfies either (i) $\mathcal{B}^0(C)\ne\emptyset$ or
  (ii) $C\neq X$ does not overlap any cluster and
  $\mathcal{B}^0(C)=\emptyset$.  If $\vert X\vert=1$ or in Case (ii) $B$ is a
  trivial block and, otherwise, in Case (i) a non-trivial one.
\end{proposition}
\begin{proof}
  If $\vert X\vert=1$, then $B= \Hasse[\mathcal{B}(C)] = \Hasse $ consists a
  single
  vertex only and is, therefore, a block of $\Hasse$.  Assume that $\vert
  X\vert>1$.
  By Lemma~\ref{lem:B-C-nontrivial-blocks} and
  Lemma~\ref{lem:B-C-trivial-blocks}, each subgraph
  $\Hasse[\mathcal{B}(C)]$ with $\mathcal{B}^0(C)\ne\emptyset$ is a
  non-trivial block and $\Hasse[\mathcal{B}(C)]$ for which $C\in
  \mathscr{C}\setminus \{X\}$ does not overlap any cluster and
  $\mathcal{B}^0(C)=\emptyset$ is a trivial block of the Hasse diagram
  $\Hasse$.

  For the converse, suppose first that $B$ is a trivial block of $\Hasse$,
  i.e., it only consists of the single vertex $C$ or the single arc
  $(C',C)$.  In the first case, we have $\vert X\vert=1$. Otherwise,
  $\Hasse$ consists of $(C',C)$ and hence $\vert X\vert>1$. Moreover, we
  have $C\subsetneq C'\subseteq X$ and thus $C\in \mathscr{C}\setminus
  \{X\}$. If $\mathcal{B}^0(C)\ne\emptyset$, then, by
  Lemma~\ref{lem:multiple-inneighbors}, $C'\in\mathcal{B}^0(C)\subsetneq
  \mathcal{B}(C)$.  Moreover, $\Hasse[\mathcal{B}(C)]$ is a non-trivial
  block of $\Hasse$ by Lemma~\ref{lem:B-C-nontrivial-blocks}.  In
  particular, the arc $(C',C)$ is contained in this block, contradicting
  that $(C',C)$ forms a trivial block. Hence, we have $\mathcal{B}^0(C) =
  \emptyset$. Assume, for contradiction, that $C$ overlaps with some
  cluster $C''\in\mathscr{C}$. Then, by closedness of $\mathscr{C}$,
  $C\in\mathcal{B}^0(D)$ for some $D\in\mathscr{C}$.  Then, by
  Cor.~\ref{cor:B0-indegree1}, $C'$ is the unique in-neighbor of
  $C\in\mathcal{B}^0(D)$ in $\Hasse$ and $C'\in\mathcal{B}(D)$.  Hence, $C$
  and $C'$ are contained in $\Hasse[\mathcal{B}(D)]$, which is non-trivial
  as a consequence of $C\in\mathcal{B}^0(D)$ and
  Lemma~\ref{lem:B-C-nontrivial-blocks}.  This again contradicts that
  $(C',C)$ forms a trivial block. In summary, we have $C\in
  \mathscr{C}\setminus \{X\}$, $\mathcal{B}^0(C)=\emptyset$, and $C$ does
  not overlap any cluster.  Suppose now that $B$ is a non-trivial block of
  $\Hasse$.  Hence, $\vert X\vert>1$ and $B$ contains an undirected cycle
  $K$ on at least $3$ clusters.  Since $\Hasse$ is a DAG, $K$ contains at
  least one cluster $C$ with two in-neighbors $C'$ and $C''$ in $K$ (and
  thus in $\Hasse$).  By Lemma~\ref{lem:multiple-inneighbors}, we have
  $C',C''\in\mathcal{B}^0(C)$.  Therefore,
  Lemma~\ref{lem:B-C-nontrivial-blocks} implies that
  $\Hasse[\mathcal{B}(C)]$ is a non-trivial block of $\Hasse$.  In
  particular, $\Hasse[\mathcal{B}(C)]$ contains the arcs $C'C$ and $C''C$,
  which are also arcs in $B$.  By Obs.~\ref{obs:biConn-arc-disjoint}, we
  therefore obtain $B=\Hasse[\mathcal{B}(C)]$.
\end{proof}

\subsection{Characterization of Clustering Systems of Level-1 Networks}
\label{ssec:level1-char}

We start with showing that a regular network is level-1 provided that its
clustering is closed and satisfied (L).
\begin{proposition}
  \label{prop:Hasse-is-level-1}
  Let $\mathscr{C}$ be a closed clustering system on $X$ satisfying (L).
  Then the Hasse diagram $\Hasse$ of $\mathscr{C}$ is a phylogenetic
  level-1 network with leaf set $X_{\Hasse}\coloneqq \{ \{x\} \mid x \in X
  \}$.
\end{proposition}
\begin{proof}
  By Lemma~\ref{lem:Hasse-is-Network}, $\Hasse$ is a phylogenetic network
  with leaf set $X_{\Hasse}\coloneqq \{ \{x\} \mid x \in X \}$.  To show
  that $\Hasse$ is level-1, we have to demonstrate that each block $B$ of
  $\Hasse$ contains at most one hybrid vertex that is distinct from the
  unique maximum $\max B$.  This holds trivially if $B$ is a trivial block
  consisting of a single arc or, if $\vert X\vert=1$, an isolated vertex.
  Now suppose that $B$ is a non-trivial block, and thus, by
  Prop.~\ref{prop:blocks-in-Hasse}, it contains exactly the clusters in
  $\mathcal{B}(C)$ for some $C\in\mathscr{C}$ with
  $\mathcal{B}^0(C)\ne\emptyset$.  By Lemma~\ref{lem:multiple-inneighbors},
  $C$ is a hybrid vertex.  From Cor.~\ref{cor:sandwich} and the
  construction of the Hasse diagram, we conclude that $\Top(C)=\max B$.  By
  Cor.~\ref{cor:B0-indegree1}, none of the clusters in $\mathcal{B}^0(C)$
  is a hybrid vertex.  Hence, $C$ is the only hybrid vertex in
  $\mathcal{B}(C)=\mathcal{B}^0(C)\cup\{C,\Top(C)\}$ that is distinct from
  $\Top(C)=\max B$.
\end{proof}

\begin{corollary}
  \label{cor:closed-L-imply-level-1}
  For every closed clustering system $\mathscr{C}$ on $X$ that satisfies
  (L), there is a level-1 phylogenetic network $N$ such that
  $\mathscr{C}_N=\mathscr{C}$.
  In particular, the unique regular network with clustering system
  $\mathscr{C}$ is level-1 and phylogenetic in this case.
\end{corollary}
\begin{proof}
  By Prop.~\ref{prop:Hasse-is-level-1}, the Hasse diagram $\Hasse[\mathscr{C}]$
  is a phylogenetic level-1 network. Since $\Hasse[\mathscr{C}]$ is graph
  isomorphic to the regular network $N$ for $\mathscr{C}$, $N$ is also a
  level-1 phylogenetic network.
\end{proof}

We summarize Cor.~\ref{cor:l1->closed-weak-hierarchy}, Cor.~\ref{cor:L},
and Cor.~\ref{cor:closed-L-imply-level-1} in the following characterization
of clustering systems that can be derived from level-1 phylogenetic
networks.
\begin{theorem}
  \label{thm:L1}
  Let $\mathscr{C}$ be a clustering system. Then there is a level-1 network
  $N$ such that $\mathscr{C}_N=\mathscr{C}$ if and only if $\mathscr{C}$ is
  closed and satisfies (L).
\end{theorem}

We emphasize, however, that there is no 1-to-1 correspondence between
level-1 networks and clustering systems. Recall that a network $N$ is
regular if $\varphi\colon V\to V(\Hasse[\mathscr{C}_N])\colon v\mapsto
\CC(v)$ is a graph isomorphism.  In contrast to the unique regular network
$\Hasse[\mathscr{C}]$, a level-1 network might have shortcuts and thus,
could even be not semi-regular and, therefore, not regular
(cf.\ Prop.~\ref{prop:semi-reg} and
Thm.~\ref{thm:semiregular}). Nevertheless, a level-1 network $N$ can easily
be edited into a level-1 network $N'$ that is isomorphic to
$\Hasse[\mathscr{C}_N]$ using two simple operations as specified in
Prop.~\ref{prop:edit-N-to-Hasse}.

\begin{proposition}
  \label{prop:edit-level-1-to-Hasse}
  For every level-1 network $N$, the regular network $N'$ with clustering
  system $\mathscr{C}_{N'} = \mathscr{C}_N$ is level-$1$ and can be
  obtained from $N$ by repeatedly removing shortcuts and contracting arcs
  $(u,w)$ with $\outdeg(u)=1$.  In particular, $N'$ is the unique
  least-resolved network w.r.t.\ $\mathscr{C}_N$ that can be obtained from
  $N$ in this way.
\end{proposition}
\begin{proof}
  Let $N$ be a level-1 network.  By Thm.~\ref{thm:L1}, $\mathscr{C}_N$ is
  closed and satisfies (L). By Cor.~\ref{cor:closed-L-imply-level-1},
  therefore, the regular network with clustering system $\mathscr{C}_N$ is
  level-$1$.  Now let $N'$ be the network obtained from $N$ by repeatedly
  (1) removing a shortcut and (2) applying $\contract(u,w)$ for an arc
  $(u,w)$ with $\outdeg(u)=1$ until neither operation is possible.  By
  construction, $N'$ is phylogenetic, shortcut-free, and contains no vertex
  with outdegree $1$.  It is easy to verify that the removal of shortcuts
  cannot increase the level of the network. This together with
  Lemma~\ref{lem:contract-level-k} implies that $N'$ is still level-$1$. By
  Lemma~\ref{lem:orderiff}, $N'$ satisfies (PCC), and thus, it is
  semi-regular.  Thm.~\ref{thm:semiregular} now implies that $N'$ is
  regular. Moreover, by Lemma~\ref{lem:shortcut-deletion} and
  Lemma~\ref{lem:outdeg-1-contraction}, we have
  $\mathscr{C}_{N}=\mathscr{C}_{N'}$.  By Prop.~\ref{prop:regular-unique},
  $N'$ is the unique regular network with
  $\mathscr{C}_{N}=\mathscr{C}_{N'}$. The latter, in particular, implies
  that the order of the operations ``shortcut removal'' and
  ``contractions'' to obtain $N'$ from $N$ does not matter.  By
  Cor.~\ref{cor:regular-IFF-least-resolved}, $N'$ is least-resolved.
  Moreover, a network that still contains a shortcut or an arc $(u,w)$ with
  $\outdeg(u)=1$ cannot be least resolved by
  Lemma~\ref{lem:shortcut-deletion} and
  Lemma~\ref{lem:outdeg-1-contraction}, respectively.  Taken together, the
  latter arguments imply that $N'$ is the unique least-resolved network
  w.r.t.\ $\mathscr{C}_N$ that can be obtained from $N$ by repeatedly
  removing shortcuts and contracting arcs $(u,w)$ with $\outdeg(u)=1$.
\end{proof}

As a direct consequence of Thm.~\ref{thm:L1} and
Prop.~\ref{prop:edit-level-1-to-Hasse} together with the fact that regular
networks are phylogenetic, we obtain
\begin{corollary}
  \label{cor:phyloL1}
  Let $\mathscr{C}$ be a clustering system. Then there is a phylogenetic
  level-1 network $N$ such that $\mathscr{C}_N=\mathscr{C}$ if and only if
  $\mathscr{C}$ is closed and satisfies (L).
\end{corollary}

\begin{corollary}
  \label{cor:unique-N-closed-L}
  Let $\mathscr{C}$ be a closed clustering system that satisfies (L).  Then
  there is a unique shortcut-free phylogenetic level-1 network $N$ with
  $\mathscr{C}_N=\mathscr{C}$ that moreover contains no vertex $v$ with
  $\outdeg_{N}(v)=1$. This network $N$ is regular and least-resolved.
\end{corollary}
\begin{proof}
  By Cor.~\ref{cor:closed-L-imply-level-1}, the regular network $N$ with
  $\mathscr{C}_N=\mathscr{C}$ is level-$1$. By Thm.~\ref{thm:semiregular},
  $N$ is shortcut-free and contains no vertex $v$ with $\outdeg_{N}(v)=1$.
  Thus $N$ is phylogenetic.  Now let $N$ be a shortcut-free phylogenetic
  level-1 network with $\mathscr{C}_{N}=\mathscr{C}$ that moreover contains
  no vertex $v$ with $\outdeg_{N}(v)=1$. By Lemma~\ref{lem:orderiff} and
  Thm.~\ref{thm:semiregular}, $N$ is a regular network, which is unique by
  Prop.~\ref{prop:regular-unique}.  By
  Cor.~\ref{cor:regular-IFF-least-resolved}, $N$ is least-resolved.
\end{proof}

Most publications on phylogenetic networks assume that leaves always have
indegree $1$, see e.g.\ \cite{huson_rupp_scornavacca_2010}.
\begin{corollary}
  \label{cor:unique-N-closed-L-leaves=indeg1}
  Let $\mathscr{C}$ be a closed clustering system that satisfies (L).  Then
  there is a unique shortcut-free phylogenetic level-1 network $N$ with
  $\mathscr{C}_N=\mathscr{C}$ such that every leaf has indegree $1$
  and all vertices $v$ with $\outdeg_{N}(v)=1$ are adjacent to leaves.
\end{corollary}
\begin{proof}
  By Cor.~\ref{cor:unique-N-closed-L}, there is a unique shortcut-free
  phylogenetic level-1 network $N'$ with $\mathscr{C}_{N'}=\mathscr{C}$ and
  for which no vertex has outdegree $1$.  In $N'$, all vertices with
  outdegree $0$ are leaves.  Hence, we can simply apply $\expand(x)$ for all
  leaves $x$ with $\indeg_{N'}(x)>1$. We can repeatedly (i.e., in each
  expansion step) apply Lemma~\ref{lem:expand} to conclude that the
  resulting digraph $N$ is a phylogenetic network,
  Cor.~\ref{cor:expand-shortcuts} to conclude that $N$ is shortcut-free,
  Lemma~\ref{lem:expand-level-k} to conclude that $N$ is level-1 and
  Lemma~\ref{lem:expand} to conclude that $N$ satisfies
  $\mathscr{C}_{N}=\mathscr{C}_{N'}=\mathscr{C}$.  In particular, every
  leaf in $N$ has indegree $1$ by construction and all vertices with
  outdegree $1$ must be adjacent to leaves.

  It remains to show that $N$ is unique w.r.t.\ these properties.  Let
  $\tilde N$ be a phylogenetic shortcut-free level-1 network with
  $\mathscr{C}_{\tilde N}=\mathscr{C}$ and such that every leaf has
  indegree $1$ and all vertices $v$ with $\outdeg_{\tilde N}(v)=1$ are
  adjacent to leaves.  Hence, after application of $\contract(v,x)$ to all
  vertices $v$ with outdegree $1$, we obtain a phylogenetic level-$1$
  network $\tilde N'$ that has no vertex with outdegree $1$ at all.
  Prop.~\ref{prop:semi-reg} implies that $\tilde N'$ is regular and
  Lemma~\ref{lem:outdeg-1-contraction} implies that $\mathscr{C}_{\tilde
    N'}=\mathscr{C}_{\tilde N}=\mathscr{C}$. By
  Cor.~\ref{cor:unique-N-closed-L}, $\tilde N'\simeq N'$. To obtain $N'$
  from $\tilde N$, we applied precisely the ``reversed'' operation of the
  operation to obtain $N$ from $N'$, which together with $\tilde N'\simeq
  N'$ implies that $\tilde N\simeq N$. Hence, $N$ is the unique network
  with the desired properties.
\end{proof}

As an immediate consequence we obtain a characterization of the level-1
networks that are completely determined by the last common ancestor
function, and equivalently by their clusters.
\begin{proposition}
  Let $N$ be a level-1 network without shortcuts. Then the following
  statements are equivalent:
  \begin{itemize}
  \item[(i)] $\outdeg(v)\ne 1$ for all $v\in V$.
  \item[(ii)] For every $v\in V$ there is a pair of leaves $x,y\in X$ such
    that $v=\lca(\{x,y\})$.
  \end{itemize}
\end{proposition}
\begin{proof}
  If $\outdeg(v)\ne 1$ for all $v\in V$, then Prop.~\ref{prop:semi-reg}
  implies that $N$ is regular, i.e., $\varphi\colon V\to
  V(\Hasse[\mathscr{C}_N])\colon v\mapsto \CC(v)$ is a graph isomorphism
  and thus a bijection.  Therefore, $\CC(u)=\CC(u')$ implies $u=u'$ for all
  $u,u'\in V$.  Together with Eq.~\eqref{eq:lca-C-equals-C}, i.e., the
  identity $\CC(v)=\CC(\lca(\CC(v)))$, we obtain, for all $v\in V$, that
  $v=\lca(\CC(v))$ and thus, by Lemma~\ref{lem:lca-xy}, there is a pair of
  leaves $x,y\in X$ such that $v=\lca(\CC(v))=\lca(\{x,y\})$.  Conversely,
  suppose there is a vertex $v\in V$ with a unique child $w$.  Moreover,
  assume for contradiction that there leaves $x,y\in X$ such that
  $v=\lca(\{x,y\})$. Using Obs.~\ref{obs:outdeg-1-cluster}, we have
  $\{x,y\}\subseteq \CC(v)=\CC(w)$. Together with $w\prec_N v$, this
  contradicts $v=\lca(\{x,y\})$.
\end{proof}

Finally, we show that every closed clustering system satisfying (L) is
represented by a unique ``minimal'' separated level-1 network.  More
precisely, we have
\begin{proposition}
  \label{prop:unique-separated-level-1}
  Let $\mathscr{C}$ be a closed clustering system satisfying (L). Then
  there is a unique separated phylogenetic shortcut-free level-1 network
  $N$ with $\mathscr{C}=\mathscr{C}_N$. The network $N$ is obtained from
  the unique regular network $\Hasse[\mathscr{C}]$ by applying $\expand(v)$
  to all hybrid vertices.
\end{proposition}
\begin{proof}
  By Cor.~\ref{cor:unique-N-closed-L}, the unique regular network
  $\Hasse[\mathscr{C}]$ is a level-1 network.  By
  Thm.~\ref{thm:unique-cluster-network}, there is a unique
  semi-regular separated phylogenetic network $N$ with
  $\mathscr{C}=\mathscr{C}_N$, and this network is obtained from
  $\Hasse[\mathscr{C}]$ by applying $\expand(v)$ to all hybrid vertices.
  The latter and Lemma~\ref{lem:expand-level-k} imply that $N$ is also
  level-1.  Since $N$ is semi-regular, it is shortcut-free. Hence, a
  network with the desired properties exist.  To see that $N$ is unique,
  let $\tilde{N}$ be a separated phylogenetic shortcut-free level-1 network
  $\tilde{N}$ with $\mathscr{C}=\mathscr{C}_{\tilde{N}}$.  By
  Lemma~\ref{lem:orderiff}, the shortcut-free network $\tilde{N}$ satisfies
  (PCC), and thus, it is semi-regular.  In summary, $\tilde{N}$ is a
  semi-regular separated phylogenetic network with clustering system
  $\mathscr{C}$ which is unique by
  Thm.~\ref{thm:unique-cluster-network}.
\end{proof}

\subsection{Compatibility of Clustering Systems and Intersection Closure}
\label{sec:compatibility}

A frequent task in phylogenetics is the construction of networks based on
partial information of putative networks, e.g.\ subtrees
\cite{aho1981inferring,jansson2006algorithms,vanI2009constructing,
  jansson2006inferring,vanI2011constructing}, subnetworks
\cite{HIMSW:17,van2017binets,semple2021trinets}, metrics or full
information about clusters \cite{Gambette:12}.  A property or properties of
networks can be thought of as a subset $\mathbb{P}$ of the set of all
rooted DAGs such that $N$ has property $\mathbb{P}$ whenever
$N\in\mathbb{P}$. In this case we simply call $N$ a $\mathbb{P}$-network. A
clustering system $\mathscr{C}\subseteq 2^X$ is \emph{compatible
w.r.t.\ $\mathbb{P}$-networks} if there is $\mathbb{P}$-network $N$ on $X$
such that $\mathscr{C}\subseteq \mathscr{C}_N$.

\begin{problem}\label{prob:cluster2network}
  Is a given clustering system $\mathscr{C}\subseteq 2^X$
  compatible w.r.t.\ to (separated, phylogenetic) level-$k$ networks?
\end{problem}

We show that this question can easily be answered for level-1 networks by
computing the so-called ``intersection-closure'' \cite{Bandelt:89}. To be
more precise, to every clustering system $\mathscr{C}$ one can associate
the set $\IC$ consisting of all non-empty intersections of an
arbitrary subset of clusters in $\mathscr{C}$.
Note that $A\in \mathscr{C}$ implies $A\cap A = A\in \IC$ and so
$\mathscr{C}\subseteq \IC$. Recall that a clustering system satisfying (L)
is in particular a weak hierarchy (cf.\ Cor.~\ref{cor:L->wH}). In this
case, only pairwise intersections need to be considered since the
intersection of arbitrary subset of clusters coincides with a pairwise
intersection. As an immediate consequence, we have
\begin{fact}
  \label{obs:IC-L->closed}
  Let $\mathscr{C}$ be a clustering system satisfying (L).
  Then $\IC=\mathscr{C} \cup \{C\cap C' \mid C,C'\in\mathscr{C} \text{
  overlap}\}$.
\end{fact}

Lemma~1 of \cite{Bandelt:89} asserts that $\mathscr{C}$ is a weak hierarchy
if and only if $\IC$ is a weak hierarchy. We use this fact to prove an
analogous result for property (L).
\begin{lemma}\label{lem:C-L<->IC-L}
  A clustering system $\mathscr{C}$ satisfies (L) if and only if $\IC$
  satisfies (L).
\end{lemma}
\begin{proof}
  Since (L) is a hereditary property, it suffices to show that, if
  $\mathscr{C}$ satisfies (L), then $\IC$ also satisfies (L).  We show
  first that the intersection of two overlapping clusters in $\mathscr{C}$
  cannot overlap any other cluster of $\mathscr{C}$.  To this end, consider
  $C_1,C_2,C_3\in\mathscr{C}$, and suppose that $C_1$ and $C_2$ overlap and
  $(C_1\cap C_2)\cap C_3\ne\emptyset$. Then, one easily verifies that
  either $C_3\subseteq C_1\cap C_2$, $C_1\cup C_2\subseteq C_3$, or $C_3$
  overlaps at least one of $C_1$ and $C_2$. In the latter case, (L) implies
  $C_1\cap C_3=C_1\cap C_2$ or $C_2\cap C_3=C_1\cap C_2$, and thus
  $(C_1\cap C_2)\cap C_3= C_1\cap C_2$. That is, the intersection of two
  overlapping clusters in $\mathscr{C}$ cannot overlap any other cluster of
  $\mathscr{C}$. It remains to show that the intersection $C_1\cap C_2$ of
  an overlapping pair of clusters $C_1,C_2\in\mathscr{C}$ also cannot
  overlap with the intersection $C_3\cap C_4$ of another overlapping pair
  $C_3,C_4\in\mathscr{C}$.  Assume, for contradiction, that $C_1\cap C_2$
  and $C_3\cap C_4$ overlap.  Hence we have $(C_1\cap C_2)\cap C_3\ne
  \emptyset$ and $(C_1\cap C_2)\cap C_4\ne \emptyset$ and also
  $C_3\setminus (C_1\cap C_2)\ne \emptyset$ and $C_4\setminus (C_1\cap
  C_2)\ne \emptyset$.  Moreover, $C_1\cap C_2 \subseteq C_3$ and $C_1\cap
  C_2 \subseteq C_4$ are not possible at the same time since otherwise
  $C_1\cap C_2 \subseteq C_3 \cap C_4$. Hence, $C_1\cap C_2$ overlaps at
  least one of $C_3$ and $C_4$; a contradiction. In summary, all
  overlapping pairs $C',C''\in \IC$ are formed by clusters
  $C',C''\in\mathscr{C}$, and thus $\IC$ also satisfies (L).
\end{proof}

\begin{algorithm}[t]
  \caption{\textnormal{\texttt{Check-L1-Compatibility}}}\label{algo1}
  \begin{algorithmic}[1]
    \Require Clustering system $\mathscr{C}$
    \Ensure Verifies if $\mathscr{C}$ is compatible w.r.t.\ to a
    phylogenetic level-1 network, and in the affirmative case, returns such
    a network
    \If{$\vert\mathscr{C}\vert>\binom{\vert X\vert}{2}+\vert X\vert$}
    \Return ``\texttt{no solution}''
    \EndIf
    \If{$\mathscr{C}$ satisfies Property (L)}
    \State $\IC \leftarrow \mathscr{C} \cup \{C\cap C' \mid C,C'\in\mathscr{C}
    \text{ overlap}\}$
    \State \Return $\Hasse[\IC]$
    \Else\ \Return 	\texttt{``no solution''}
    \EndIf
  \end{algorithmic}
\end{algorithm}

\begin{theorem}\label{thm:AlgL1Comp}
  Let $\mathscr{C}\subseteq 2^X$ be a clustering system.  Then,
  \textnormal{\texttt{Check-L1-Compatibility}} correctly verifies if there
  is a (separated, phylogenetic) level-1 network on $X$ such that
  $\mathscr{C}\subseteq \mathscr{C}_N$ and can be implemented to run in
  $O(\vert\mathscr{C}\vert^2\vert X\vert )\subseteq O(\vert X\vert^5)$ time.
  Moreover, such a
  network $N$ can be constructed in $O(\vert X\vert^5)$ time.
\end{theorem}
\begin{proof}
  The proof (in particular, the part concerning the time complexity) is
  rather lengthy and technical and is, therefore, placed to
  Section~\ref{sec:appx-algo} in the Appendix. We emphasize, that the
  proof, however, contains interesting insights for those readers who want
  to implement algorithm.
\end{proof}

\begin{theorem}\label{thm:comp-level1}
  For every clustering system $\mathscr{C}$ the following statements are
  equivalent:
  \begin{enumerate}[noitemsep,nolistsep]
  \item $\mathscr{C}$ is compatible w.r.t.\ to a (separated, phylogenetic)
    level-1 network;
  \item There is a (separated, phylogenetic) level-1 network with
    $\mathscr{C}_N = \IC$;
  \item $\mathscr{C}$ satisfies Property (L).
  \item $\Hasse[\IC]$ is a level-1 network.
  \end{enumerate}
\end{theorem}
\begin{proof}
  If Statement (1) is satisfied, then the network computed with
  \texttt{Check-L1-Compatibility} is a network with $\mathscr{C}_N = \IC$
  and thus, Statement (2) holds.  If there is a (separated, phylogenetic)
  level-1 network with $\mathscr{C}_N = \IC$, then Thm.~\ref{thm:L1}
  implies that $\IC$ satisfies (L).  Since (L) is a hereditary property,
  $\mathscr{C}$ must satisfy (L) as well. Hence, (2) implies (3).  Assume
  that $\mathscr{C}$ satisfies Property (L).  By Lemma
  \ref{lem:C-L<->IC-L}, $\IC$ satisfies (L) and, by definition, $\IC$ is
  closed.  By Thm.~\ref{thm:L1} and
  Prop.\ \ref{prop:unique-separated-level-1}, there is a (separated,
  phylogenetic) level-1 network such that $\IC=\mathscr{C}_N$. Since
  $\mathscr{C}\subseteq \IC= \mathscr{C}_N$, Item (1) is satisfied.  Hence,
  Statements (1), (2) and (3) are equivalent.  Assume that Statement (2)
  holds.  By Thm.~\ref{thm:L1}, $\IC$ is closed and satisfies (L).
  Cor.\ \ref{cor:unique-N-closed-L} implies that $\Hasse[\IC]$ is a level-1
  network and thus, Statement (4) holds.  Conversely, assume, that
  Statement (4) is satisfied.  Again, by Thm.~\ref{thm:L1}, $\IC$ is closed
  and satisfies (L).  Prop.\ \ref{prop:unique-separated-level-1} implies
  now Statement (2).  In summary, the four statements are equivalent.
\end{proof}

\section{Special Subclasses of Level-1 Networks}
\label{sec:specialL1}

\subsection{Galled Trees}
\label{sec:specialL1-galledtrees}

In level-1 networks, the structure a block $B$ is highly constrained if the
unique terminal vertex $v$ of $B$ has only two parents $v_1$ and $v_2$. The
absence of additional hybrid vertices implies, in particular, that the two
paths from $\max B$ to $v_1$ and $v_2$, are uniquely defined.
\begin{fact}
  \label{obs:cycle}
  Let $N$ be a level-1 network. Then every non-trivial block is an
  (undirected) cycle if and only if every hybrid vertex $v$ in $N$
  satisfies $\indeg(v)=2$.
\end{fact}
We note that a similar result does not hold for level-$k$ networks with
$k>1$. As an example, Fig.~\ref{fig:level2-multiset}(A) shows two networks
whose hybrid vertices have all indegree $2$ but whose blocks are not
(undirected) cycles.

\begin{lemma}
  \label{lem:outerplanar}
  If $N$ is a level-1 network and every hybrid vertex $v$ in $N$ satisfies
  $\indeg(v)=2$, then $N$ is outerplanar.
\end{lemma}
\begin{proof}
  By Obs.~\ref{obs:cycle}, every non-trivial block on $N$ is a cycle.
  Therefore, the underlying undirected graph of $N$ does not contain a
  subdivision of the graph $K_4$, i.e., the complete graph on $4$ vertices,
  nor of the complete bipartite graph graph $K_{2,3}$.  By Thm.~1 in
  \cite{Chartrand:67}, $N$ is outerplanar.
\end{proof}

In \cite{Gusfield:03}, \emph{galled trees} were introduced as phylogenetic
networks in which all cycles are vertex disjoint. Here we consider a more
general version, where cycles are allowed to share a cutvertex and the
network is not required to be phylogenetic. More constrained types of
networks will be discussed in the subsequent sections.
\begin{definition}
  \label{def:galledtree}
  A \emph{galled tree} is a network in which every non-trivial block is an
  (undirected) cycle.
\end{definition}
\begin{lemma}
  \label{lem:galled-level-1}
  Every galled tree is level-1.
\end{lemma}
\begin{proof}
  Let $N$ be a galled tree. Every trivial block $B$ contains at most one
  hybrid vertex distinct from $\max B$.  Thus consider a non-trivial block
  $B$ and assume, for contradiction, that $B$ properly contains two hybrid
  vertices $h$ and $h'$.  Since $B$ is a cycle, the two vertices in $B$
  that are adjacent with $h$ must exactly be the two in-neighbors of
  $h$. The same holds for $h$.  It is now easy to see that the two path in
  $B$ that connect $h$ and $h'$ must each contain a vertex whose two
  neighbors in the cycle $B$ are out-neighbors.  Hence, these two distinct
  vertices are $\preceq_{N}$-maximal in $B$, which contradicts the
  uniqueness of $\max B$.
\end{proof}
Lemma~\ref{lem:galled-level-1} and Obs.~\ref{obs:cycle} imply
\begin{corollary}
  \label{cor:galled-level-1}
  A network is a galled tree if and only if it is level-1 and satisfies
  $\indeg(v)=2$ for all hybrid vertices $v$.
\end{corollary}

\begin{fact}\label{obs:galled-outneighbors}
In a galled tree, every non-trivial block consists of two internally vertex
disjoint paths connecting $\max B$ and $\min B$. Moreover, every vertex
contained in $B$ that is distinct from $\max B$ and $\min B$ has precisely
one out-neighbor in $B$.
\end{fact}
As we shall see below, this implies that its clustering system satisfies
the following property:
\begin{definition}
  \label{def:N3O}
  \begin{description}
    \item[(N3O)] $\mathscr{C}$ contains no three distinct pairwise
    overlapping clusters.
  \end{description}
\end{definition}

\begin{lemma}
  \label{lem:galled-implies-N3O}
  If $\mathscr{C}$ is the clustering system of a galled tree, then
  $\mathscr{C}$ satisfies (N3O).
\end{lemma}
\begin{proof}
  Suppose there is a galled tree $N$ with $\mathscr{C}_{N}=\mathscr{C}$.
  In particular, $N$ is level-1 by Lemma~\ref{lem:galled-level-1}.  Now
  suppose, for contradiction, that (N3O) is not satisfied. Thus there are
  three distinct vertices $u_1,u_2,u_3\in V(N)$ such that
  $C_1\coloneqq\CC_N(u_1)$, $C_2\coloneqq\CC_N(u_2)$, and
  $C_3\coloneqq\CC_N(u_3)$ overlap pairwise.  By
  Lemma~\ref{lem:inclusion}, it must hold that $u_1$, $u_2$, and $u_3$ are
  pairwise $\preceq_{N}$-incomparable.  By Lemma~\ref{lem:uvmin},
  $C_1\cap C_2\ne \emptyset$ implies that $u_1$ and $u_2$ are located in a
  common block $B$. Clearly, $\preceq_{N}$-incomparability of $u_1$ and
  $u_2$ implies $u_1\ne \max B$.  By similar arguments, $u_1$ and $u_3$ are
  located in a common block $B'$ and $u_1\ne \max B'$.  Since $u_1\notin
  \{\max B, \max B'\}$, we can apply Lemma~\ref{lem:block-identity} to
  conclude that $B=B'$.  Hence, for every $i\in\{1,2,3\}$, there is a
  directed path $P_i$ in $B$ from $u_i$ to $\min B$.  Now consider, for
  distinct $i,j\in\{1,2,3\}$, the $\preceq_{N}$-maximal vertex $v$ in $P_i$
  that is also a vertex in $P_j$ (which exists since $\min B$ is a vertex
  of both paths). We have $v\notin\{u_i, u_j\}$ because $u_i$ and $u_j$ are
  $\preceq_N$-incomparable.  Therefore, the unique parents $v_i$ and $v_j$
  of $v$ in $P_i$ and $P_j$, resp., must be distinct. Therefore, $v$ is a
  hybrid vertex in $B$ and clearly distinct from $\max B$. Hence, it must
  hold that $v=\min B$.  Since $i$ and $j$ were chosen arbitrarily, the
  paths $P_1$, $P_2$, and $P_3$ only have vertex $\min B$ in common.  This
  together with the fact that $\min B\notin \{u_1,u_2,u_3\}$ implies that
  $\min B$ has at least indegree 3.  By Obs.~\ref{obs:cycle}, therefore,
  $N$ has a non-trivial block that is not an undirected cycle. Hence, $N$
  is not a galled tree; a contradiction.
\end{proof}
The converse of Lemma~\ref{lem:galled-implies-N3O} is not true since, in
addition to (N3O), closedness and (L) are required:
\begin{theorem}
  \label{thm:galled-no3overlap}
  $\mathscr{C}$ is the clustering system of a galled tree if and only if
  $\mathscr{C}$ is closed and satisfies (L) and (N3O).  Moreover, in this
  case, $\Hasse[\mathscr{C}]$ is a phylogenetic galled tree.
\end{theorem}
\begin{proof}
  Suppose first that $N$ is the clustering system of a shortcut-free galled
  tree. By Lemma~\ref{lem:galled-level-1}, $N$ is level-1, and thus, it is
  closed and $\mathscr{C}$ satisfies (L) by Thm.~\ref{thm:L1}.  By
  Lemma~\ref{lem:galled-implies-N3O}, $\mathcal{C}$ also satisfies (N3O).
  Now suppose, that $\mathscr{C}$ is closed, satisfies $(L)$, and does not
  contain three pairwise overlapping clusters.  By
  Cor.~\ref{cor:unique-N-closed-L}, the unique regular network $N\coloneqq
  \Hasse[\mathscr{C}]$ with clustering system $\mathscr{C}$ is a
  shortcut-free phylogenetic level-1 network.  Hence, it satisfies (PCC) by
  Lemma~\ref{lem:orderiff}.  Suppose, for contradiction that $N$ is not a
  galled tree, i.e., it contains a non-trivial block, that is not an
  undirected cycle. By Obs.~\ref{obs:cycle}, there is a hybrid vertex $w\in
  V(N)$ with (at least) three distinct in-neighbors $u_1$, $u_2$, and
  $u_3$. By Obs.~\ref{obs:shortcut} and since $N$ is shortcut-free, $u_1$,
  $u_2$, and $u_3$ must be pairwise $\preceq_{N}$-incomparable.  By
  (PCC), it therefore holds $\CC_N(u_i)\not\subseteq \CC_{N}(u_j)$ for all
  distinct $i,j\in\{1,2,3\}$.  Moreover, it holds
  $\emptyset\ne\CC_{N}(w)\subseteq \CC_{N}(u_i)$ for $i\in\{1,2,3\}$.
  Taken together, the latter two arguments imply that $\CC_{N}(u_1)$,
  $\CC_N(u_2)$, and $\CC_{N}(u_3)$ overlap pairwise; a contradiction.
  Hence, $N$ must be a galled tree.
\end{proof}

\begin{definition} \cite{Diday:86,Bertrand:13} A clustering system
  $(X,\mathscr{C})$ is \emph{pre-pyramidal} if there exists a total order
  $\lessdot$ on $X$ such that, for every $C\in\mathscr{C}$ and all
  $x,y\in C$, it holds that $x\lessdot u\lessdot y$ implies $u\in C$. That
  is, all clusters $C\in\mathscr{C}$ are intervals w.r.t.\ $\lessdot $.
\end{definition}

A necessary condition \cite{Nebesky:83,Changat:21w} for $\mathscr{C}$ to be
pre-pyramidal is
\begin{description}
\item[(WP)] If $C_1,C_2,C_3\in\mathscr{C}$ have pairwise non-empty
  intersections, then one of the three sets is contained in the union of
  the other two.
\end{description}
Taken together, (L) and (WP) imply (N3O). More precisely, we have
\begin{lemma}
  \label{lem:py-no3}
  Let $\mathscr{C}$ be a pre-pyramidal clustering system satisfying (L).
  Then $\mathscr{C}$ satisfies (N3O), i.e., there are no three pairwise
  overlapping sets.
\end{lemma}
\begin{proof}
  Assume, for contradiction, that $C_1$, $C_2$, and $C_3$ overlap
  pairwise. Then (L) implies that $C_1\cap C_2 = C_2\cap C_3 = C_1\cap
  C_3 = C_1\cap C_2\cap C_3\eqqcolon C\ne\emptyset$. Since $\mathscr{C}$ is
  pre-pyramidal and the three pairwise intersections are non-empty, (WP)
  implies that one of the three sets is contained in the union of the other
  two. W.l.o.g., suppose $C_1\subseteq C_2\cup C_3$. Equivalently, $C_1 =
  C_1\cap (C_2\cup C_3)=(C_1\cap C_2)\cup (C_1\cap C_3)=C\subsetneq C_1$, a
  contradiction.
\end{proof}

Pre-pyramidal set systems are also known as ``interval hypergraphs''.  A
characterization in terms of an infinite series of forbidden
sub-hypergraphs has been developed in
\cite{Tucker:72,Trotter:76,Duchet:84}. It can be used to obtain a simple
necessary and condition in the presence of (L).
\begin{proposition}
  \label{prop:N3O-iff-py}
  Let $\mathscr{C}$ be a clustering system satisfying (L). Then
  $\mathscr{C}$ is pre-pyramidal if and only if it satisfies (N3O).
\end{proposition}
\begin{proof}
  Starting from \cite[Thm.~7.2]{Duchet:84} one observed that condition (L)
  excludes all induced forbidden subhypergraphs with a single
  exception. The remaining configuration, called $M_1$ in
  \cite[Fig.~11]{Duchet:84}, comprises three pairwise overlapping sets
  that share at least one common point. Thus, if (N3O) holds, no
  $M_1$-subhypergraph is present in $\mathscr{C}$. Since (L) and (N3O)
  together exclude all forbidden subhypergraphs and thus $\mathscr{C}$ is
  pre-pyramidal. Lemma~\ref{lem:py-no3} now completes the proof.
\end{proof}

\begin{theorem}
  \label{thm:py}
  Let $N$ be a phylogenetic shortcut-free level-1 network with clustering
  system $\mathscr{C}$. Then $\mathscr{C}$ is pre-pyramidal if and only if
  $\indeg(v)\le 2$ for all $v\in V$, i.e., if and only if $N$ is a galled
  tree.
\end{theorem}
\begin{proof}
  ($\implies$) Suppose that $\mathscr{C}$ is pre-pyramidal with
  corresponding total order $\lessdot$ of $X$ and, moreover, assume, for
  contradiction, that $w$ is hybrid vertex with $\indeg_N(w)\ge 3$.  Hence,
  let $u_1,u_2,u_3\in\parent_N(w)$ be pairwise distinct. Since $N$ is
  shortcut-free, Obs.~\ref{obs:shortcut} implies that $u_1$, $u_2$, and
  $u_3$ are pairwise $\preceq_{N}$-incomparable.  Together with
  Lemma~\ref{lem:orderiff}, this implies $\CC(u_i)\not\subseteq \CC(u_j)$
  for distinct $i,j \in\{1,2,3\}$.  Moreover, $u_1,u_2,u_3\in\parent_N(w)$
  and Lemma~\ref{lem:inclusion} yield $\emptyset\ne \CC(w)\subseteq
  \CC(u_i)$ for all $i \in \{1,2,3\}$, i.e., $\mathscr{C}$ contains three
  pairwise overlapping clusters. On the other hand, $\mathscr{C}$
  satisfies (L) by Cor.~\ref{cor:L}, thus Lemma~\ref{lem:py-no3} implies
  that $\mathscr{C}$ cannot contain three pairwise overlapping clusters;
  a contradiction.

  ($\impliedby$) Suppose $\indeg(v)\le 2$ for all $v\in V$. By
  Obs.~\ref{obs:cycle}, this holds if and only if $N$ is a galled tree.  By
  Lemma~\ref{lem:galled-level-1}, Cor.~\ref{cor:L}, and
  Lemma~\ref{lem:galled-implies-N3O}, $\mathscr{C}$ satisfies (L) and
  (N3O).  Hence, $\mathscr{C}$ is pre-pyramidal by
  Prop.~\ref{prop:N3O-iff-py}.
\end{proof}

\begin{definition}\cite{Bertrand:08}\label{def:pairedH}
A clustering system $\mathscr{C}$ is a \emph{paired hierarchy} if a cluster
$C\in\mathscr{C}$ overlaps with at most one other cluster in $\mathscr{C}$.
\end{definition}

\begin{fact}
  \label{obs:paired-hierarchy}
  Every hierarchy is a paired hierarchy and every paired hierarchy
  satisfies (L) and (N30).
\end{fact}

\begin{proposition}
  \label{prop:paired-hierarchy}
  Let $\mathscr{C}$ be a closed clustering system. Then, $\mathscr{C}$ is a
  paired hierarchy if and only if there is a shortcut-free phylogenetic
  galled tree $N$ with $\mathscr{C}_N=\mathscr{C}$ where all non-trivial
  blocks consist of four vertices.
\end{proposition}
\begin{proof}
  Suppose first that $\mathscr{C}$ is a paired hierarchy. Since
  $\mathscr{C}$ satisfies (L) and (N3O) by Obs.~\ref{obs:paired-hierarchy}
  and is closed, we can apply Thm.~\ref{thm:galled-no3overlap} to conclude
  that $N \coloneqq \Hasse[\mathscr{C}]$ is a shortcut-free phylogenetic
  galled tree with $\mathscr{C}_{N} = \mathscr{C}$. Since $N$ is, in
  particular, a phylogenetic level-1 network
  (cf.\ Lemma~\ref{lem:galled-level-1}), Lemma~\ref{lem:orderiff} implies
  that $N$ satisfies (PCC).  By definition, every non-trivial block
  contains at least $3$ vertices. If a block $B$ woulds contain exactly
  three vertices, then one easily sees that $N$ contains the shortcut
  $(\max B, \min B)$; a contradiction. Hence, every non-trivial block in
  $N$ contains at least $4$ vertices.  Assume, for contradiction, that $N$
  contains a non-trivial block $B$ with at least $k\geq 5$ vertices. Note,
  $B$ refers to an (undirected) cycle in $N$. Hence, there are two internal
  vertex disjoint paths in $B$ connecting $\max B$ and $\min B$. Since $N$
  is shortcut-free and $k\geq 5$, we can conclude that one path contains a
  vertex $v$ and the other path contains vertices $u_1,u_2$ that are all
  distinct from $\max B$ and $\min B$.  Since $B$ is an undirected cycle,
  one easily verifies that $v$ and $u_1$ as well as $v$ and $u_2$ are
  $\preceq_N$-incomparable.  By Lemma~\ref{lem:inclusion}, we have
  $\emptyset \neq \CC(\min B) \subseteq \CC(v),\CC(u_1),\CC(u_2)$. This
  together with (PCC) and the fact that $v$ and $u_1$ as well as $v$ and
  $u_2$ are $\preceq_N$-incomparable implies that $\CC(v)$ must overlap
  with both $\CC(u_1)$ and $\CC(u_2)$; a contradiction.

  Assume now that there is a shortcut-free phylogenetic galled tree $N$
  with $\mathscr{C}_N=\mathscr{C}$ where all non-trivial blocks consists of
  four vertices.  Suppose, for contradiction, $\CC(v)\in \mathscr{C}$
  overlaps with two distinct clusters $\CC(u_2)$ and $\CC(u_2)$ in
  $\mathscr{C}$. Note, $v$, $u_1$, and $u_2$ must be pairwise distinct. By
  Lemma~\ref{lem:overlap-B0}, we have $v,u_1\in B_1^0$ and $v,u_2\in B_2^0$
  for non-trivial blocks $B_1$ and $B_2$ in $N$.  In particular, we have
  $v, u_1\notin \{\min B_1,\max B_1\}$ and $v, u_2\notin \{\min B_2,\max
  B_2\}$. We can therefore apply Lemma~\ref{lem:block-identity} to conclude
  that $B_1=B_2\eqqcolon B$. In particular, $B$ contains at least five
  pairwise distinct vertices $\min B$, $\max B$, $v$, $u_1$, and $u_2$; a
  contradiction.
\end{proof}

It is worth noting that for paired hierarchies, and in particular also for
hierarchies, $\mathscr{C}$ there are not only galled trees but also
shortcut-free and phylogenetic level-$k$ networks $N$ that are not
level-$(k-1)$ with
$\mathscr{C}_N=\mathscr{C}$. Fig.\ \ref{fig:tree-based}(A) serves as an
example.

\subsection{Conventional and Separated Level-1 Networks}
\label{sec:specialL1-sep}

The literature on phylogenetic networks often stipulates that the leaves
$v\in X$ have indegree $1$, see
e.g.\ \cite{huson_rupp_scornavacca_2010}. Furthermore, level-1 networks are
often defined such that every non-trivial block has exactly one hybrid
vertex.
\begin{definition}\label{def:conventional}
  A network $N$ is \emph{conventional} if (i) all leaves have indegree at
  most $1$ and (ii) every hybrid vertex is contained in a unique
  non-trivial block.
\end{definition}
We remark that, if $\vert X\vert>1$ all leaves have indegree $1$ in a
conventional network.  In Fig.~\ref{fig:network-8}, network $N$ is
conventional, while $N'$ is not.

\begin{proposition}
  \label{prop:almost-done2}
  Let $\mathscr{C}$ be a closed clustering system on $X$ satisfying (L).
  Then $\Hasse[\mathscr{C}]$ is conventional if and only if
  $\mathcal{B}^0(\{x\})=\emptyset$ for all $x\in X$ and
  $\mathcal{B}^0(\Top(C))=\emptyset$ for all $C\in\mathscr{C}$ with
  $\mathcal{B}^0(C)\ne\emptyset$.
\end{proposition}
\begin{proof}
  By Lemma~\ref{lem:Hasse-is-Network}, $\Hasse\coloneqq\Hasse[\mathscr{C}]$
  is a phylogenetic network with leaf set $X_{\Hasse}\coloneqq \{ \{x\}
  \mid x \in X \}$.  Moreover, Lemma~\ref{lem:multiple-inneighbors} implies
  that $\indeg(\{x\})_{\Hasse}\ge 2$ if and only if
  $\mathcal{B}^0(\{x\})\ne\emptyset$.  By Obs.~\ref{obs:identical-block}
  and Lemma~\ref{lem:block-identity}, two distinct non-trivial blocks $B$
  and $B'$ share at most one vertex $v\in\{\max B,\max B'\}$.  Thus every
  hybrid vertex is contained in a unique non-trivial block if and only if
  $\max B\ne\min B'$ for any pair of non-trivial blocks. Since the
  non-trivial blocks in $\Hasse[\mathscr{C}]$ are exactly the blocks
  $\mathcal{B}(C)\ne\emptyset$, this is equivalent to requiring that, for
  every $C\in\mathscr{C}$ with $\mathcal{B}(C)\ne\emptyset$ we have
  $\Top(C)$ is not the minimum of another non-trivial block, i.e.,
  $\mathcal{B}(\Top(C))=\emptyset$.
\end{proof}

\begin{proposition}\label{prop:sep-conv}
  If $N$ is separated, then $N$ is conventional.
\end{proposition}
\begin{proof}
  Suppose $N$ separated. Then all leaves in $N$ must have indegree at most
  $1$ since they have outdegree $0$ and, by definition, hybrid vertices
  have outdegree $1$.  By Lemma~\ref{lem:block-identity}, a vertex $v$ that
  is contained in two non-trivial blocks $B$ and $B'$ must be the unique
  maximal vertex of one of them. In this case,
  Cor.~\ref{cor:maxB-outdegree} implies $\outdeg_{N}(v)\ge 2$. However,
  since $v$ is hybrid vertex in a separated network, we have
  $\outdeg(v)=1$. Therefore, such a hybrid vertex $v$ that is contained in
  two non-trivial blocks cannot exist. Hence, $N$ is conventional.
\end{proof}

\subsection{Binary Level-1 Networks}
\label{sec:specialL1-binary}

Recall that a network is \emph{binary} if it is phylogenetic, separated,
and all vertices have in- and outdegree at most 2.  Equivalently, in a
binary network, every tree vertex is either a leaf or has exactly two
children, and every hybrid vertex has exactly two parents and one child.
As an immediate consequence of the definition,
Prop.~\ref{prop:phy-level1->TC}, Prop.~\ref{prop:sep-conv}, and
Cor.~\ref{cor:galled-level-1} we have:
\begin{fact}\label{obs:BinN-GalledTree}
  Binary level-1 networks are always phylogenetic, separated, conventional,
  tree-child, galled trees.
\end{fact}

\begin{lemma}
  Let $N$ be a binary level-1 network. Then, $\Hasse[\mathscr{C}_N]\simeq N$
  if and only if $N$ is a tree.
\end{lemma}
\begin{proof}
  Let $N$ be a binary level-1 network. By definition, $N$ is phylogenetic.
  If $N$ is a tree, then $\Hasse[\mathscr{C}_N]\sim N$
  (cf.\ Cor.~\ref{cor:hasse-tree}).  Assume now that $N$ is not a tree and
  thus, $N$ contains hybrid vertices all with outdegree $1$ in $N$.
  However, Prop.~\ref{prop:semi-reg} implies that $\Hasse[\mathscr{C}]$ does
  not contain any vertex with outdegree $1$.  Consequently,
  $\Hasse[\mathscr{C}_N]\not\sim N$.
\end{proof}
Hence, $\Hasse[\mathscr{C}_N]$ can never be binary in case $\mathscr{C}_N$
contains overlapping clusters.

\begin{definition}[2-Inc]
  \label{def:2-Inc}
  A clustering system $\mathscr{C}$ has Property (2-Inc) if, for all
  clusters $C\in\mathscr{C}$, there are at most two inclusion-maximal
  clusters $A,B\in\mathscr{C}$ with $A,B\subsetneq C$ and at most two
  inclusion-minimal clusters $A,B\in\mathscr{C}$ with $C\subsetneq A,B$.
\end{definition}
\begin{lemma}
  \label{lem:binary+PCC-->2inc}
  Let $N$ be a binary network that satisfies (PCC). Then $\mathscr{C}_N$
  satisfies (2-Inc).
\end{lemma}
\begin{proof}
  Let $N$ be a binary network on $X$ with clustering system $\mathscr{C}$.
  Assume, for contradiction, that $\mathscr{C}_N$ does not satisfy Property
  (2-Inc) for some cluster $C\in \mathscr{C}_N$.
  Assume first that there are (at least) three inclusion-minimal clusters
  $A_1,A_2,A_3\in \mathscr{C}$ that satisfy $C\subsetneq A_1,A_2,A_3$.
  Hence, $C\neq X$.  Since $C\in \mathscr{C}_N$, there is a
  $\preceq_N$-maximal vertex $v\in V(N)$ with $\CC(v) = C$. Note, $v$ has
  at least one but at most two parents in $N$ since $N$ is binary and
  $v\neq \rho_N$.  Let $v_i$ be a vertex in $N$ with $\CC(v_i) = A_i$,
  $i\in \{1,2,3\}$.  By Obs.~\ref{obs:subsetneq-implies-below}, we have
  $v\prec_N v_1,v_2, v_3$. Therefore, and because $v$ has at most two
  parents, at least two of $v_1,v_2,v_3$ must be ancestors of the same
  parent $w$ of $v$ in $N$.  W.l.o.g.\ assume that $w\preceq_N v_1,v_2$.
  Since $v$ is $\preceq_N$-maximal w.r.t.\ $\CC(v) = C$, it must hold that
  $\CC(v)\subsetneq \CC(w)$.  Lemma~\ref{lem:inclusion} implies that
  $\CC(w)\subseteq \CC(v_1),\CC(v_2)$.  Note, however, that
  $\CC(w) = \CC(v_1)$ is not possible, since then $A_1\neq A_2$ and
  $\CC(v)\subsetneq \CC(w)$ imply that
  $\CC(v)\subsetneq \CC(w) = \CC(v_1) \subsetneq \CC(v_2)$; a contradiction
  to the inclusion-minimality of $\CC(v_2)=A_2$.  By similar arguments,
  $\CC(w) = \CC(v_2)$ is not possible.  Hence,
  $\CC(v)\subsetneq \CC(w) \subsetneq \CC(v_1), \CC(v_2)$ must hold; again
  a contradiction to the inclusion-minimality of $\CC(v_1)=A_1$ and
  $\CC(v_2)=A_2$.

  Assume now that there are (at least) three inclusion-maximal clusters
  $A_1,A_2,A_3\in \mathscr{C}$ that satisfy $A_1,A_2,A_3\subsetneq C$.
  Hence, $C$ cannot be a singleton.  Since $C\in \mathscr{C}_N$, there is a
  $\preceq_N$-minimal vertex $v\in V(N)$ with $\CC(v) = C$.  Since $C$ is
  not a singleton and $N$ is binary, we can conclude that $v$ has at least
  one but at most two children in $N$.  Let $v_i$ be a vertex in $N$ with
  $\CC(v_i) = A_i$, $i\in \{1,2,3\}$.  Since $v$ has at most two children,
  at least two of $v_1,v_2,v_3$ must be descendants of the same child $w$
  of $v$ in $N$.  Since $v$ is $\preceq_N$-minimal w.r.t.\ $\CC(v) = C$ it
  must hold that $\CC(w)\subsetneq \CC(v)$.  Now, we can apply similar
  arguments as in the first case to obtain a contradiction.
\end{proof}

\begin{theorem}
  \label{thm:C-binary}
  Let $\mathscr{C}$ be a clustering system on $X$. Then, there is a binary
  level-1 network $N$ with $\mathscr{C}_N = \mathscr{C}$ if and only if
  $\mathscr{C}$ is closed and satisfies Properties~(L) and~(2-Inc).
  In this case, the (unique) cluster network with clustering system
  $\mathscr{C}$ is a binary level-1 network.
\end{theorem}
\begin{proof}
  Assume first that $\mathscr{C}$ is closed and satisfies Properties (L)
  and (2-Inc). Taken Thm.~\ref{thm:L1} and
  Prop.~\ref{prop:edit-level-1-to-Hasse} together, the regular network
  $\Hasse[\mathscr{C}]$ is a level-1 network.  Since $\mathscr{C}$
  satisfies (2-Inc), every vertex in $\Hasse[\mathscr{C}]$ must have in-
  and outdegree at most 2. By Thm.~\ref{thm:unique-cluster-network}, we
  can uniquely construct a cluster network $N$ by applying $\expand(v)$ to
  all hybrid vertices of $\Hasse[\mathscr{C}]$. Hence, $N$ is
  binary. Moreover, since $\Hasse[\mathscr{C}]$ is a level-1 network,
  Lemma~\ref{lem:expand-level-k} implies that $N$ is a level-1
  network. Hence, a binary level-1 network $N$ with
  $\mathscr{C}_N = \mathscr{C}$ exists. The latter, in particular, shows
  that the cluster network $N$ is a binary level-1 network.  Conversely,
  suppose that $N$ is a binary level-1 network on $X$ with
  $\mathscr{C}_N = \mathscr{C}$. By Thm.~\ref{thm:L1}, $\mathscr{C}$ is
  closed and satisfies Property~(L).  By Lemma~\ref{lem:orderiff}, $N$
  satisfies (PCC) and thus, by Lemma~\ref{lem:binary+PCC-->2inc}, is also
  satisfies (2-Inc).
\end{proof}

Since a phylogenetic level-1 network satisfies (PCC) by
Lemma~\ref{lem:orderiff}, it is semi-regular if and only if
it is shortcut-free. Theorem~\ref{thm:cluster-network-charac} therefore
yields the following characterization of level-1 cluster networks:
\begin{corollary}
  Let $N$ be a phylogenetic level-1 network. Then $N$ is a cluster network
  if and only if it is shortcut-free and separated.
\end{corollary}
Obs.~\ref{obs:BinN-GalledTree} then implies
\begin{corollary}
  Let $N$ be a binary level-1 network. Then $N$ is a cluster network if and
  only if it is shortcut-free.
\end{corollary}

We finally consider the problem as whether a clustering system
$\mathscr{C}\subseteq 2^X$ compatible w.r.t.\ to a binary level-1 network.
\begin{theorem}
  \label{thm:compatible-binary-l1}
  A given clustering system  $\mathscr{C}\subseteq 2^X$ is compatible w.r.t.\
  to a binary level-1 network if and only if $\mathscr{C}$ satisfies (L) and
  all hybrid vertices $w$ in $\Hasse[\IC]$ have $\indeg_{\Hasse[\IC]}(w)=2$.
\end{theorem}
\begin{proof}
  Assume first that $\mathscr{C}$ is compatible w.r.t.\ to a binary level-1
  network and let $N$ be such a network with
  $\mathscr{C}\subseteq \mathscr{C}_N$.  By Thm.~\ref{thm:C-binary},
  $\mathscr{C}$ satisfies (L).  Moreover, by Thm.~\ref{thm:L1},
  $\mathscr{C}_N$ is closed and thus
  $\IC\subseteq \mathcal{I}(\mathscr{C}_N) = \mathscr{C}_N$.  To see that
  all hybrid vertices $w$ in $\Hasse[\IC]$ have
  $\indeg_{\Hasse[\IC]}(w)=2$, suppose for contradiction that there is a
  vertex $w$ in $\Hasse[\IC]$ with indegree larger than $2$, i.e., there is
  $C\in\IC$ with at least three distinct inclusion-minimal supersets
  $C_1,C_2,C_3\in\IC\subseteq \mathscr{C}_N$. Since
  $C\subsetneq C_1,C_2,C_3$ and these cluster are inclusion-minimal (and
  thus not contained in one another), they overlap pairwise.  By
  Obs.~\ref{obs:BinN-GalledTree}, $N$ is a galled tree. Hence, by
  Thm.~\ref{thm:galled-no3overlap}, $\mathscr{C}_N$ contains no three
  pairwise overlapping clusters; a contradiction.

  Assume now that $\mathscr{C}$ satisfies (L) and that all hybrid vertices
  $w$ in $\Hasse\coloneqq\Hasse[\IC]$ have $\indeg_{\Hasse}(w)=2$.  In the
  following we use caterpillars $\mathrm{CAT}_n$, i.e, binary trees on $n$
  leaves such that each inner vertex has exactly two children and the
  subgraph induced by the inner vertices is a directed path with the root
  $\rho_{\mathrm{CAT}_n}$ at one end of this path.  By
  Thm.~\ref{thm:comp-level1}, $\Hasse$ is a level-1 network.  This together
  with Obs.~\ref{obs:cycle} implies that every non-trivial block in
  $\Hasse$ is a cycle and thus, $\Hasse$ must be a galled tree.  In
  particular, for every non-trivial block $B$, $\max B$ has exactly two
  children in $B$.  Let $v$ be vertex in $\Hasse$ with
  $\outdeg_{\Hasse}(v)>2$. We now ``resolve'' $v$ as follows: If
  $\indeg_{\Hasse}(v)=2$, then expand $v$. Otherwise, $v$ is a tree vertex.
  In this case, let $\mathcal{B}$ be the set of all non-trivial blocks $B$
  in $\Hasse$ with $v= \max B$ and $\mathcal{C}$ be the children of $v$
  that are not contained in some $B\in \mathcal{B}$.  We now replace $v$ by
  a caterpillar $\mathrm{CAT}_n$ with $n =
  \vert\mathcal{C}\vert+\vert\mathcal{B}\vert$ and thus, we can find a
  1-to-1 correspondence between the $n$ leaves of the caterpillar and the
  elements in $\mathcal{C}\cupdot \mathcal{B}$. The elements in
  $\mathcal{C}$ are now identified with their corresponding leaves.
  Observe that $\vert\mathcal{B}\vert>1$ is possible, i.e., $v=\max B$ for
  more than one non-trivial block $B$. We therefore re-attach, for each
  block $B\in \mathcal{B}$, the two children of $v=\max B$ that are
  contained in $B$ as children of the leaf of the caterpillar that
  corresponds to $B$. Note that this construction is well-defined since, by
  Lemma~\ref{lem:block-identity}, no two such children can be children of
  two distinct blocks $B,B'\in\mathcal{B}$.  Since we do not change the
  structure of non-trivial blocks, the network $N'$ obtained in this way
  remains a level-1 network whose hybrid vertices still have indegree $2$.
  Moreover, it is an easy task to verify that $\mathscr{C}\subseteq
  \mathscr{C}_{\Hasse}\subseteq \mathscr{C}_{N'}$. Repeated application of
  the latter steps to all vertices eventually results in a binary level-1
  network $N$ with $\mathscr{C}\subseteq \mathscr{C}_{\Hasse} \subseteq
  \mathscr{C}_N$.
\end{proof}

Using \texttt{Check-L1-Compatibility} and the results in
Thm.~\ref{thm:AlgL1Comp}, we obtain
\begin{corollary}
  Determining if a clustering system $\mathscr{C}\subseteq 2^X$ is
  compatible w.r.t.\ to a binary level-1 network and, in the affirmative
  case, the construction of such a network can be done in $O(\vert X\vert^5)$
  time.
\end{corollary}

\begin{corollary}
  \label{cor:l1-deg2->binary-l1}
  For all phylogenetic level-1 networks $N$ whose hybrid vertices $w$ have
  $\indeg_{N}(w)=2$, there is a binary level-1 network $N'$ with
  $\mathscr{C}_N\subseteq \mathscr{C}_{N'}$.
\end{corollary}
\begin{proof}
  If all hybrid vertices $w$ in $N$ have $\indeg_{N}(w)=2$, then $N$ is a
  galled tree by Obs.~\ref{obs:cycle}.
  By Thm.~\ref{thm:galled-no3overlap},
  $\mathscr{C}_N$ is closed (i.e., $\mathscr{C}_N=\IC$) and satisfies (L) and,
  moreover, $\Hasse[\mathscr{C}]= \Hasse[\IC]$ is a galled tree.
  Applying Obs.~\ref{obs:cycle} again to $\Hasse[\IC]$ yields that all hybrid
  vertices $w$ in $\Hasse[\IC]$ have $\indeg_{\Hasse[\IC]}(w)=2$. Now apply
  Thm.~\ref{thm:compatible-binary-l1}.
\end{proof}

The converse of Cor.~\ref{cor:l1-deg2->binary-l1} is not true, i.e., a
phylogenetic level-1 networks $N$ for which there is a binary level-1
network $N'$ with $\mathscr{C}_N\subseteq \mathscr{C}_{N'}$ may contain a
hybrid vertex $w$ with $\indeg_{N}(w)>2$.  To see this, consider a binary
level-1 network $N'$ with a non-trivial block $B$ such that $(\max B,\min
B)\notin E(N')$ and the non-binary network $N$ that is obtained from $N'$
by adding the shortcut $(\max B, \min B)$.  Clearly, $N$ is still
phylogenetic and level-$1$ and by Lemma~\ref{lem:shortcut-deletion}
satisfies $\CC_{N}=\CC_{N'}$ but the vertex $\min B$ has indegree $3$.

\subsection{Level-1 Networks Encoded by their Cluster Multisets}
\label{sect:BL1-MS}

In this part we focus on a particular subclass of networks:
\begin{definition}\label{def:quasi-bin}
  A network $N$ is a \emph{quasi-binary} if
  $\indeg_N(w)=2$ and $\outdeg_{N}(w)=1$ for every hybrid vertex $w\in V(N)$
  and,  additionally, $\outdeg_N(\max B) = 2$
  for every non-trivial block $B$ in $N$.
\end{definition}
We note that, in particular, all binary networks are quasi-binary.
Moreover, quasi-binary networks are separated. Therefore and by
Cor.~\ref{cor:galled-level-1}, we obtain
\begin{fact}\label{obs:quasiBinL1->Galled}
  Quasi-binary level-1 networks are separated
  galled trees.
\end{fact}

Theorem~\ref{thm:semireg-multiset} shows that all semi-regular networks are
encoded by their multisets of clusters. None of the conditions (PCC) or
shortcut-free can be omitted, as shown be the examples in
Fig.~\ref{fig:multiset-not-semireg}.  Nevertheless, replacing some of these
conditions by a different one might be possible. In the following, we
replace shortcut-freeness by requiring that $N$ is a phylogenetic
quasi-binary level-1 network. In this case, we obtain Property (PCC) as an
immediate consequence of Lemma~\ref{lem:orderiff}.  Again, we observe that
neither of the properties phylogenetic or quasi-binary or level-1 can be
dropped: Fig.~\ref{fig:level2-multiset}(A) shows two phylogenetic
quasi-binary networks that are not level-1;
Fig.~\ref{fig:level2-multiset}(B) shows two quasi-binary level-1 networks
that are not phylogenetic; Fig.~\ref{fig:level2-multiset}(C) shows two
phylogenetic level-1 networks that are not quasi-binary where, in all three
examples, the respective networks have the same multisets of clusters but
are not isomorphic.  As a by-product, we obtain the following result:
\begin{fact}
  Let $\mathbb{P}$ denote the class of \emph{all} networks $N$ for which
  either precisely one or at least one of the following conditions hold:
  \begin{enumerate}
  \item  $N$ is level-$k$, $k\geq 2$ but not level-1 and contains at least
    three
    leaves.
  \item $N$ does not satisfy (PCC);
  \item $N$ is not quasi-binary;
  \item $N$ is not shortcut-free;
  \item $N$ is not phylogenetic.
  \end{enumerate}
  Then, no network $N\in \mathbb{P}$ is encoded (w.r.t.\ $\mathbb{P}$) by
  its multiset $\mathscr{M}_N$ of clusters and thus, by its set
  $\mathscr{C}_N$ of clusters.
\end{fact}

\begin{figure}[t]
  \centering
  \includegraphics[width=0.95\textwidth]{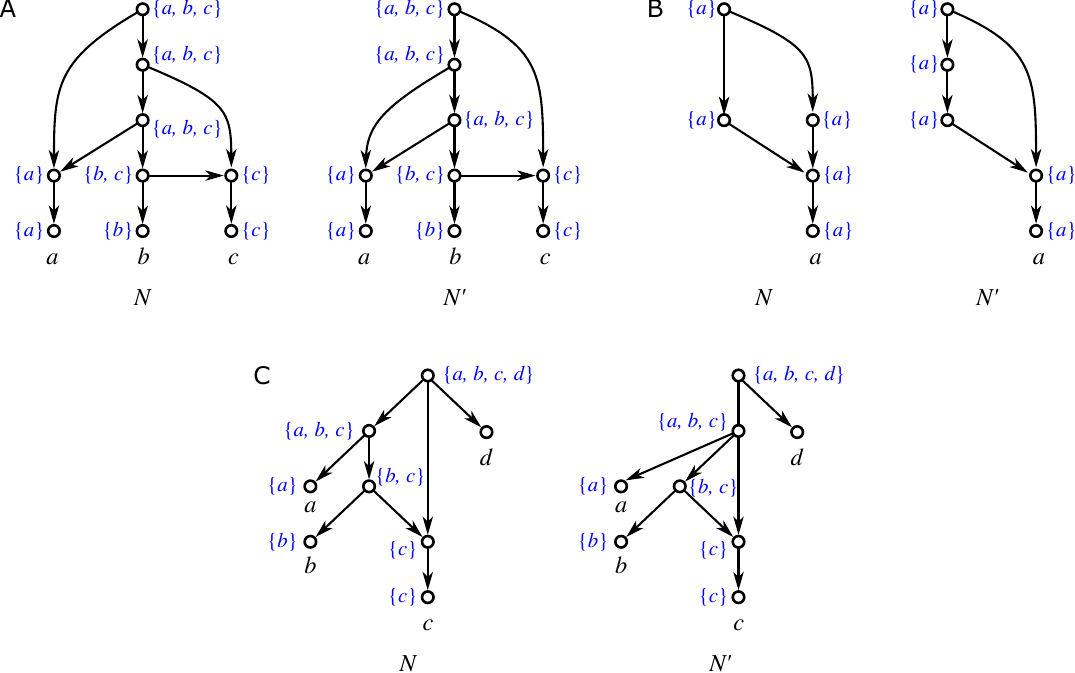}
  \caption{Three pairs of non-isomorphic networks $N$ and $N'$ for which
    $\mathscr{M}_{N} = \mathscr{M}_{N'}$.  (A) $N$ and $N'$ are binary
    level-2 tree-child networks. (B) $N$ and $N'$ are separated galled
    trees but not phylogenetic. (C) $N$ and $N'$ are separated phylogenetic
    level-1 networks.}
    \label{fig:level2-multiset}
\end{figure}

\begin{theorem}\label{thm:level-1-alt}
  Let $N$ be a  phylogenetic quasi-binary level-1 network. Then
  $N$ is the unique  phylogenetic quasi-binary level-1 network
  tree whose cluster multiset is $\mathscr{M}_N$.
\end{theorem}
\begin{proof}
  Suppose that both $N$ and $\tilde{N}$ are phylogenetic quasi-binary
  level-1 networks and $\mathscr{M}_N= \mathscr{M}_{\tilde{N}}$.
  Lemma~\ref{lem:orderiff} implies that both $N$ and $\tilde{N}$ satisfy
  (PCC).  We show that $\varphi\coloneqq \varphi_{PCC}\colon V(N) \to
  V(\tilde{N})$ is a graph isomorphism.  By Lemma~\ref{lem:varphi},
  $\varphi$ is a bijection between $V(N)$ and $V(\tilde{N})$ that is the
  identity on the common leaf set $X$.  In the following, we write
  $\tilde{v}\coloneqq \varphi(v)$ for all $v\in V(N)$, and make free use of
  the facts that, by Lemma~\ref{lem:varphi},
  $\CC_{N}(v)=\CC_{\tilde{N}}(\tilde{v})$ and $v$ is a leaf if and only if
  $\tilde{v}$ is a leaf, and moreover $u\prec_{N} v$ if and only if
  $\tilde{u}\prec_{\tilde{N}} \tilde{v}$ for all $u,v\in V(N)$.  In the
  following, we will make frequent use of the fact that both $N$ and
  $\tilde{N}$ are galled trees (cf.\ Obs.~\ref{obs:quasiBinL1->Galled}).

  It remains to show that, for all $u,v\in V(N)$, it holds $(v,u)\in E(N)$
  if and only if $(\tilde{v},\tilde{u})\in E(\tilde{N})$.  To this end,
  suppose $(v,u)\in E(N)$ and, for contradiction, that
  $(\tilde{v},\tilde{u})\notin E(\tilde{N})$.  Since $N$ is acyclic and
  finite, we can assume w.l.o.g.\ that $(v,u)\in E(N)$ is a
  $\preceq_{N}$-minimal arc that is ``missing'' in $\tilde{N}$, i.e., there
  is no arc $(v',u')\in E(N)$ with $u'\prec_{N} u$ and
  $(\tilde{v}',\tilde{u}')\notin E(\tilde{N})$.  We have $u\prec_{N} v$ and
  thus also $\tilde{u}\prec_{\tilde{N}} \tilde{v}$.  The latter together
  with $(\tilde{v},\tilde{u})\notin E(\tilde{N})$ implies that there is
  $\tilde{w}\in V(\tilde{N})$ such that $\tilde{u}\prec_{\tilde{N}}
  \tilde{w} \prec_{\tilde{N}} \tilde{v}$.  This in turn implies $u\prec_{N}
  w \prec_{N} v$, and thus, $(v,u)$ is a shortcut in $N$. In particular,
  $u$ is a hybrid vertex and, since by definition the non-trivial blocks in
  the galled tree $N$ correspond to the undirected cycles, $v=\max B$ for
  some non-trivial block $B$ of $N$ whose unique hybrid vertex is $u=\min
  B$.

  We continue with showing that $\tilde{u}$ is also a hybrid vertex.  Since
  $N$ is quasi-binary, the hybrid vertex $u$ has a unique
  child $c$. By the
  $\preceq_{N}$-minimal choice of $(v,u)\in E(N)$, $\tilde{c}$ must be a
  child of $\tilde{u}$ in $\tilde{N}$. Suppose, for contradiction, that
  $\tilde{u}$ is a tree vertex. Hence, since $\tilde{N}$ is phylogenetic,
  it must have a second child
  $\tilde{c}'\in\child_{\tilde{N}}(\tilde{u}) \setminus \{\tilde{c}\}$. Now
  $\tilde{c}'\prec_{\tilde{N}} \tilde{u}$ implies that $c'\prec_{N} u$.
  Therefore and by the choice of $(v,u)$, there is an arc $(p,c')\in V(N)$
  (where $p\ne u$ since $c$ is the only child of $u$) such that
  $(\tilde{p},\tilde{c}')$ is also an arc in $\tilde{N}$. Thus
  $\tilde{c}'$ is a hybrid vertex with distinct parents $\tilde{u}$ and
  $\tilde{p}$.  Hence, neither of $c'$ and $\tilde{c}'$ is a leaf.
  Therefore, and because $N$ is phylogenetic, $c'$ either has a second
  parent $p'$ (which is also distinct from $u$), or at least two
  children. By the choice of $(v,u)$, the images of these vertices are
  adjacent with $\tilde{c}'$ in both cases. Hence, if $c'$ has a second
  parent $p'$, then $\tilde{c}'$ has three distinct parents $\tilde{u}$,
  $\tilde{p}$, and $\tilde{p}'$.
  If on the other hand $c'$ has at least two children, then the
  hybrid vertex $\tilde{c}'$ also has at least two children.
  Both cases, therefore, contradict
  that $\tilde{N}$ is quasi-binary.
  Hence, $\tilde{u}$ must be a hybrid vertex.

  Let $z$ be the parent of the hybrid vertex $u$ that is not $v$ and
  observe that $z\preceq_{N} w\prec_{N}\max B=v$.  Moreover, $u\prec_{N} z$
  implies $\tilde{u}\prec_{\tilde{N}} \tilde{z}$.  Suppose, for
  contradiction, that $(\tilde{z}, \tilde{u})\notin E(\tilde{N})$.  In this
  case, $\tilde{u}\prec_{\tilde{N}} \tilde{z}$ implies that there is
  $\tilde{z}'$ with $\tilde{u}\prec_{\tilde{N}} \tilde{z}'\prec_{\tilde{N}}
  \tilde{z}$ and thus $u\prec_{N} z'\prec_{N} z$.  Hence, $u$ must have a
  parent $z''$ with $z''\preceq_{N} z'\prec_{N} z(\prec_{N} v)$. Since $N$
  is acyclic, it holds that $z'\neq v$ and thus, $z''\neq z,v$. This
  contradicts the fact that $z$ and $v$ are the only two parents of $u$ in
  $N$. Therefore, $\tilde{z}$ must be one of the two parents of $\tilde{u}$
  in $\tilde{N}$. In particular, it holds
  $\tilde{u}\prec_{\tilde{N}}\tilde{z}\prec_{\tilde{N}}\tilde{v}$.

  Now let $\tilde{q}$ be the parent of the hybrid vertex $\tilde{u}$ that
  is not $\tilde{z}$.  Since we assumed that $(\tilde{v},\tilde{u})\notin
  E(\tilde{N})$, it holds that $\tilde{q}\ne \tilde{v}$ (and thus $q\ne
  z$).  We have $\tilde{u}\prec_{\tilde{N}} \tilde{q}$ and thus $u\prec_{N}
  q$.  Therefore and since $q\notin \{v,z\}=\parent_{N}(u)$, we must have
  $z\prec_{N} q$ or $v \prec_{N} q$.  Since $z\prec_{N}\max B = v$, we
  have, in both of the latter two cases, that $z\prec_{N} q$. Moreover,
  since the non-trivial block $B$ consists of two paths that have only $v$
  and $u$ in common and which do not contain additional hybrid vertices and
  since $v$ and $z$ are the unique two parents of $u$, it holds that
  $z\prec_{N} q$ and (a) $q\prec_{N} v$ or (b) $v\prec_{N} q$. In
  particular, we also have $(\tilde{u}\prec_{\tilde{N}}) \;
  \tilde{z}\prec_{\tilde{N}} \tilde{q}$, which implies that
  $(\tilde{q},\tilde{u})$ is a shortcut.

  \begin{figure}[t]
    \begin{center}
      \includegraphics[width=0.9\textwidth]{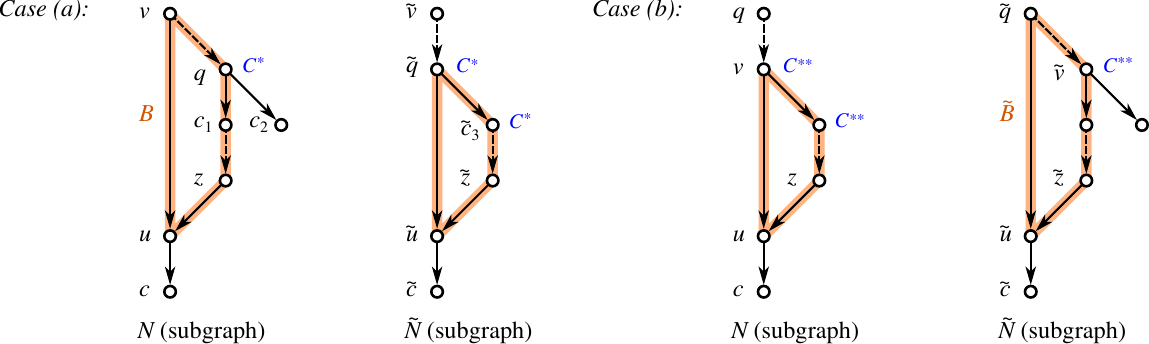}
    \end{center}
    \caption{Illustration of Cases~(a) and~(b) in the proof of
      Thm.~\ref{thm:level-1-alt}. Dashed arrows indicate directed paths
      (possibly consisting of a single vertex).}
    \label{fig:multiset-l1-binary-proof}
  \end{figure}

  \emph{Case~(a):} $q\prec_{N} v$. The situation that we will obtain in the
  following up to the final contradiction is illustrated in
  Fig.~\ref{fig:multiset-l1-binary-proof}.  Together with $z\prec_{N} q$
  and Lemma~\ref{lem:block-prec-sandwich}, $q\prec_{N} v$ implies that $q$
  is also contained in $B$. In particular, $\min B=u \prec_{N} q$ and thus
  $q$ is a tree vertex. Since $N$ is a galled tree, $B$ is exactly the
  undirected cycle that is formed by the shortcut $(v,u)$ and the directed
  path from $v=\max B$ to $u=\min B$ (which contains $q$ as an inner
  vertex).  Hence, $q$ has exactly one child $c_1$ in $B$ and, since $N$ is
  phylogenetic, at least one
  child that is not in $B$. Consider an arbitrary such child $c_2$ that is
  not in $B$.  Since all vertices of $B$ lie on a directed path and
  $z\prec_{N} q$, we must have $z\preceq_{N} c_1 \prec_{N} \max B$. In
  particular, $c_1$ is also a tree vertex since $N$ is level-$1$.  The
  vertices $c_1$ and $c_2$ must be $\preceq_{N}$-incomparable. To see this,
  suppose $c_1\prec_{N} c_2$. Then $(q,c_1)$ is a shortcut, and thus,
  because $N$ is a level-$1$ network whose non-trivial blocks correspond to
  undirected cycles, $N$ contains a non-trivial block $B'$ that is formed
  be the shortcut $(q,c_1)$ and a directed path from $q$ to $c_1$ that
  passes through $c_2$. In particular, $q=\max B'$ and both of $c_1$ and
  $c_2$ are contained in $B'$. Therefore, $B$ and $B'$ share the two
  vertices $q$ and $c_1$ and we have $B=B'$ by
  Obs.~\ref{obs:identical-block}; a contradiction to $c_2\notin V(B)$.
  Similarly, $c_2\prec_{N} c_1$ is not possible and thus $c_1$ and $c_2$
  are $\preceq_{N}$-incomparable.  Suppose, for contradiction, that there
  is a vertex $x\in \CC_{N}(c_1)\cap\CC_{N}(c_2)$.  Then,
  Lemma~\ref{lem:lower-path} implies that $c_1$ and $c_2$ are contained in
  some non-trivial block $B'$ of $N$. Since they are
  $\preceq_{N}$-incomparable, it holds $c_1\ne \max B'$, and thus, the
  unique parent $q$ of the tree vertex $c_1$ must also be contained in
  $B'$. Similar as before, we therefore obtain $B=B'$, a contradiction to
  $c_2\notin V(B)$.  Hence, $\CC_{N}(c_1)$ and $\CC_{N}(c_2)$ are
  disjoint. In particular, we have $\CC_{N}(c_1) \cupdot
  \CC_{N}(c_2)\subseteq \CC_{N}(q)$, and, because both of $\CC_{N}(c_1)$
  and $\CC_{N}(c_2)$ are non-empty, we obtain $\CC_{N}(c_i)\subsetneq
  \CC_{N}(q)$, $i=1,2$. Since $c_2$ was chosen arbitrarily, we have
  $\CC_{N}(c_i)\subsetneq\CC_{N}(q)$ for all children
  $c_i\in\child_{N}(q)$. Together with Lemma~\ref{lem:inclusion} and
  the fact that every vertex $w'$ with $w'\prec_{N} q$ satisfies
  $w'\preceq_{N} c_i$ for some $c_i\in\child_{N}(q)$, these
  inclusions imply that $\CC_{N}(w')\subsetneq \CC_{N}(q)$.  Hence, $q$
  is a $\preceq_{N}$-minimal vertex with cluster $C^*\coloneqq\CC_{N}(q) (
  = \CC_{\tilde{N}}(\tilde{q}))$. Since $(\tilde{q},\tilde{u})$ is a
  shortcut, $\tilde{q} = \max \tilde{B}$ of some non-trivial block
  $\tilde{B}$ in $\tilde{N}$. Since $\tilde{N}$ is quasi-binary,
  $\tilde{q}$ has precisely two children $\tilde{u}$ and $\tilde{c}_3$.
  Since $(\tilde{q},\tilde{u})$ is a shortcut and the non-trivial blocks of
  $N$ correspond to undirected cycles, we have $\tilde{u} \prec_{\tilde{N}}
  \tilde{c}_3$ and thus $\CC_{\tilde{N}}(\tilde{u}) \subseteq
  \CC_{\tilde{N}}(\tilde{c}_3)$ by Lemma~\ref{lem:inclusion}. Therefore, we
  obtain $\CC_{\tilde{N}}(\tilde{q}) = \CC_{\tilde{N}}(\tilde{u}) \cup
  \CC_{\tilde{N}}(\tilde{c}_3)= \CC_{\tilde{N}}(\tilde{c}_3)$. Together
  with $\tilde{c}_3 \prec_{N} \tilde{q}$, this implies that $\tilde{q}$ is
  not a $\preceq_{\tilde{N}}$-minimal vertex with cluster $C^*$ in
  $\tilde{N}$ (as opposed to $q$ in $N$); a contradiction to the
  construction of $\varphi$.  In summary, therefore, Case~(a) cannot occur.

  \emph{Case~(b):} $v\prec_{N} q$. This implies $\tilde{v}\prec_{\tilde{N}}
  \tilde{q}$. Since $(\tilde{q},\tilde{u})$ is a shortcut, it holds
  $\tilde{q}= \max \tilde{B}$ and $\tilde{u}= \min \tilde{B}$ for some
  non-trivial block $\tilde{B}$ in
  $\tilde{N}$. Lemma~\ref{lem:block-prec-sandwich}, together with
  $\tilde{v}\prec_{\tilde{N}} \tilde{q}$ and $\tilde{u}\prec_{\tilde{N}}
  \tilde{v}$, implies that $\tilde{v}$ is also contained in $\tilde{B}$.
  We can now apply similar arguments as in Case~(a), where the roles of $N$
  and $\tilde{N}$ are interchanged, to conclude that Case~(b) is also
  impossible. The situation up to the final contradiction is again
  illustrated in Fig.~\ref{fig:multiset-l1-binary-proof}.

  In summary, therefore, $(\tilde{v},\tilde{u})\in E(\tilde{N})$ must hold.
  By analogous arguments, $(\tilde{v},\tilde{u})\in E(\tilde{N})$ implies
  $(v,u)\in E(N)$. Hence, $\varphi$ is a graph isomorphism that is the
  identity on $X$ and thus $N\simeq \tilde{N}$. Therefore, $N$ is the
  unique phylogenetic quasi-binary level-$1$ network
  whose cluster multiset is $\mathscr{M}_N$.
\end{proof}

Our colleagues Simone Linz and Kristina Wicke drew our attention to an
alternative proof for Theorem \ref{thm:level-1-alt} that proceeds by
induction on the size of the leaf set of $N$.  It utilizes the concepts of
cherries and reticulated cherries that have been used extensively in the
literature, see e.g.~\cite{Bordewich:16,Murakami:19,semple2021trinets}.  We
opted for a non-inductive proof that remains closer to the construction
utilized throughout this contribution.

\begin{corollary}\label{cor:level-1-alt}
  Let $N$ be a binary level-1 network. Then $N$ is the unique binary
  level-1 network whose cluster multiset is $\mathscr{M}_N$.
\end{corollary}

Again, we can observe that neither of the properties binary or level-1 in
Cor.~\ref{cor:level-1-alt} can be dropped:
Fig.~\ref{fig:multiset-not-semireg}(A) shows two level-1 networks (where
one is not binary) and Fig.~\ref{fig:level2-multiset}(A) shows two binary
level-2 networks (that are not level-1) where, in both examples, the
respective networks have the same multisets of clusters but are not
isomorphic.

\section{Summary}\label{sec:sum}

\begin{table}[h!]
  \begin{center}
    \includegraphics[width=1.\textwidth]{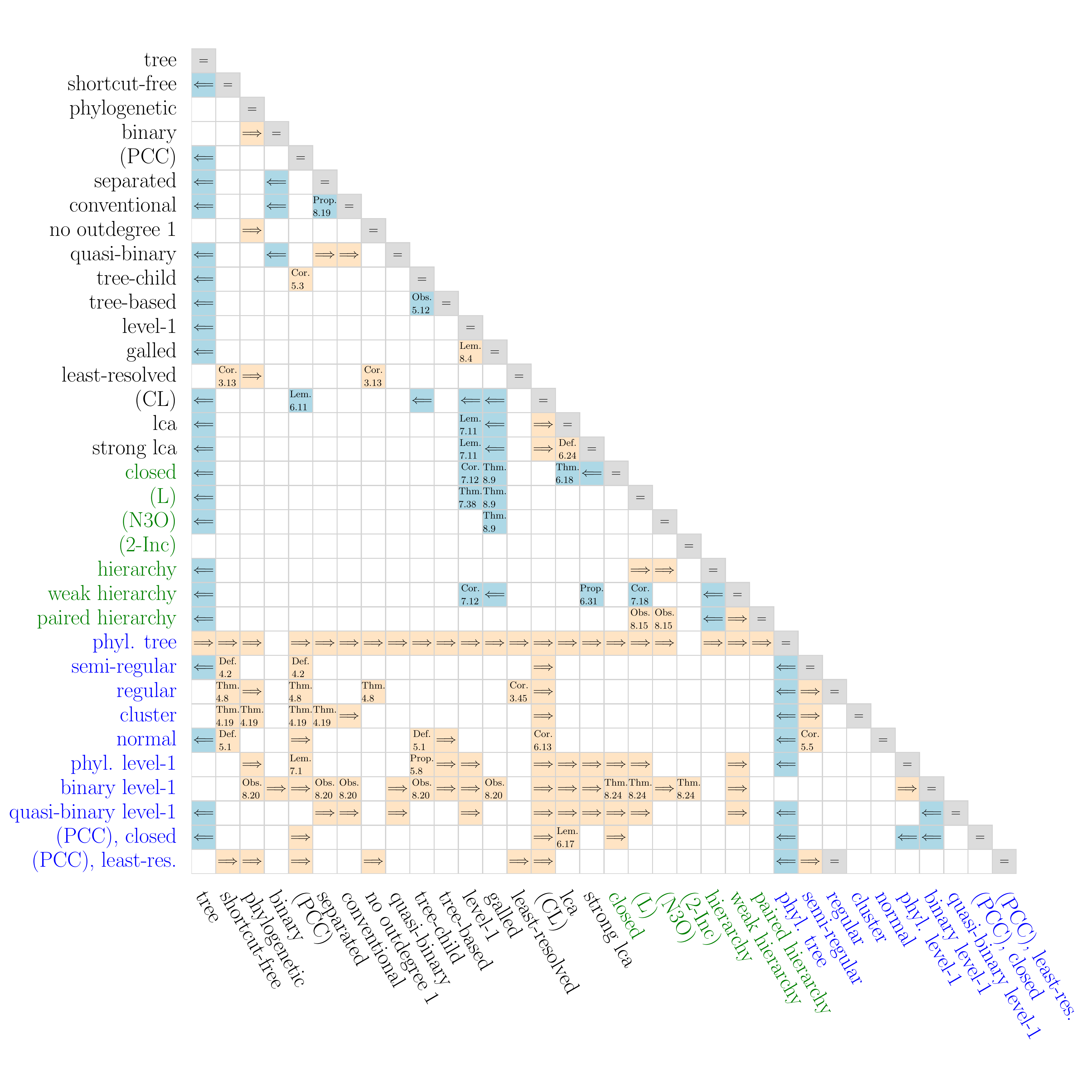}
  \end{center}
  \caption{Summary of main results: Mutual dependencies between the types
    of networks and clustering systems considered in this paper.
    The properties highlighted in black and green refer to
    properties of networks and clustering systems, respectively. Properties
    in blue text are combinations of ``basic'' properties of networks. An
    entry at position $(i,j)$ in the matrix is colored orange, turquoise,
    and white, if the property at pos.\ $i$ implies $j$, is implied by $j$,
    or does have a non-empty overlap with $j$, respectively.  Gray colored
    entries refer to equality. References within the matrix indicate the
    result where the respective dependencies are shown.  All other colors
    in the matrix are either trivial observations or were derived by
    computing the transitive closure over the proven implications.
    }
  \label{tab:summary-table}
\end{table}

In this contribution, we investigated the mutual dependencies between the
different concept of networks in the literature and their connection to
clustering systems. Most of our findings are summarized in
Table~\ref{tab:summary-table}. As one of the main results, level-1 networks
as well as some of their subclasses, such as galled trees or binary
phylogenetic level-1 networks, are characterized by the structure of their
clustering system $\mathscr{C}$.  Moreover, we showed that semi-regular
networks and phylogenetic quasi-binary level-1 (and thus, binary level-1
networks) are uniquely determined by their multiset of clusters.
Furthermore, regular and cluster networks (and their subclasses as
e.g.\ phylogenetic trees) are uniquely determined by their clustering
system.  We provided a plethora of examples that show, however, that most
classes of networks cannot be encoded in such way if there are not
sufficiently many extra restriction placed on such networks. In addition,
we showed that it is possible to determine in polynomial time whether a
clustering system is compatible with a level-1 network and to construct
such a network in the affirmative case.

It remains an open question as whether general level-$k$ networks can be
characterized by their clustering systems.  Moreover, under which
conditions is a clustering system compatible with other specified networks
and what is the computational complexity to determine them? While we have
shown that some types of networks can be encoded by the multisets of
clusters, a characterization of multisets that encode the underlying
networks as well as reconstruction algorithms are part of future research.

From the point of view of clustering systems, phylogenetic networks suggest
properties that may also be of relevance in practical data analysis beyond
applications in phylogenetics. A clustering systems between hierarchies and
weak hierarchies \cite{Bertrand:14}, (L)-clustering systems appear as an
attractive alternative e.g.\ to pyramidal clustering \cite{Bertrand:13} for
data that are not naturally linearly ordered.

\section*{Acknowledgments}

We are grateful to Simone Linz and Kristina Wicke for quite interesting
discussions about networks that, eventually, led to the main idea of
Theorem \ref{thm:level-1-alt}. This work was funded in part by the
  \emph{Deutsche Forschungsgemeinschaft} (DFG), proj.no.\ MI439/14-2.

\clearpage
\appendix

\section{Additional Results and Proofs}

\subsection{Expansion, Contraction, and Blocks}
\label{sec:appx-blocks}

We first provide the proof of
\begin{repeatedlemma}{\textnormal{\textbf{\ref{lem:contraction}}}}
  Let $N$ be a network and $(u,w) \in E(N)$ be an arc that is not a
  shortcut.  Then, $\contract(u,w)$ applied on $N$ results in a network
  $N'$ with leaf set $X$ or $X\setminus\{w\}$ and
  $V(N')=V(N)\setminus\{u\}$.  Moreover, for all $v,v'\in V(N')$,
  \begin{enumerate}
    \item $v\preceq_N v'$ implies $v \preceq_{N'} v'$, and
    \item $v\preceq_{N'} v'$ implies
    (i) $v \preceq_{N} v'$ or
    (ii) $w\preceq_N v'$ and $v\preceq_N w'$ for some
    $w'\in\child_N(u)\setminus \{w\}$ that is $\preceq_{N}$-incomparable
    with $w$.
  \end{enumerate}
  In particular, $v\prec_{N'} v'$ always implies $v\prec_{N} v'$ or $v$ and
  $v'$ are $\preceq_{N}$-incomparable.
\end{repeatedlemma}
\begin{proof}
  Assume that $(u,w)$ is not a shortcut in $N$ and recall that
  $\contract(u,w)$ consists in replacing arcs $(v,u)$ by $(v,w)$ for all
  $v\in\parent_N (u)$, replacing arcs $(u,v)$ by $(w,v)$ for all
  $v\in\child_N (u)\setminus \{w\}$, and deleting $(u,w)$ and $u$.  Observe
  that only the arcs incident with $u$ are removed and all inserted arcs
  are incident with $w$.  While $u$ does not exist in $N'$ anymore, both
  the in- and out-neighborhood of $w$ may change in such a way that $w$ may
  get additional in/out-neighbors of $u$.

  We show first that $N'$ has a single vertex with indegree $0$.  Suppose
  first $u\ne \rho_N$, and thus $\rho_N\in V(N')$. Since clearly
  $\rho_{N}\notin \child_N(u)$, we did not insert the arc $(w, \rho_N)$.
  Together with the fact that all inserted arcs are incident with $w$, this
  implies that $\rho_N$ still has indegree $0$ in $N'$. Since $u\ne
  \rho_N$, vertex $u$ has at least one in-neighbor, which becomes an
  in-neighbor of $w$ if it was not already an in-neighbor of $w$ in $N$.
  Now suppose, for contradiction, that there is a vertex $v\in
  V(N')\setminus\{\rho_N, w\}$ with $\indeg_{N'}(v)=0$.  Since $v\ne
  \rho_N$, there must be an arc $(u',v)$ in $N$ which is no longer
  contained in $N'$. By construction and since $v\ne u$, we must have
  $u'=u$. But then $(u,v)=(u',v)\in E(N)$ implies that $(w,v)$ is an arc in
  $N'$, contradicting that $\indeg_{N'}(v)=0$.  Now consider the case $u=
  \rho_N$.  Then $u$ is the unique in-neighbor of $w$ in $N$. To see this,
  assume, for contradiction, that there is $v\in\parent_N(w)\setminus\{u\}$
  and thus $w\prec_N v$.  Since $u=\rho_N$ is the unique root, there is a
  $uv$-path. Since $v\ne u$, this path passes through some child $w'$ of
  $u$. Since $w\prec_N v \preceq_{N} w'$, we must have $w\ne w'$ and thus
  $w\prec_N w'\prec_N u$, i.e., $(u,w)$ is a shortcut; a contradiction.
  Therefore, $u$ is the unique in-neighbor of $w$.  Moreover, since
  $\parent_N(u)=\emptyset$, $w$ does not get any new in-neighbors.  After
  deletion of $u$, $w$ has outdegree $0$ in $N'$.  Re-using the arguments
  from the case $u\ne \rho_N$, there is no vertex $v\in
  V(N')\setminus\{\rho_N, w\}$ with $\indeg_{N'}(v)=0$.  Hence, $N'$ is a
  directed graph that has a unique vertex with indegree $0$, i.e., a
  unique root $\rho_{N'}$, in both cases.

  We continue by showing that $N'$ is a DAG.  Assume, for contradiction,
  that $N'$ contains a directed cycle $K$ comprising the vertices $v_1,
  v_2, \dots, v_k$, $k\ge 2$, in this order, i.e., $(v_i,v_{i+1})$, $1\leq
  i\leq k-1$ and $(v_k,v_1)$ are arcs in $N'$.  If all arcs along $K$ are
  arcs of $N$, then $K$ is a directed cycle in $N$; a contradiction. Hence,
  at least one arc $e$ in $K$ cannot be contained in $N$. By construction,
  $e$ must be of the form $(v,w)$ with $v\in
  \parent_N(u)\setminus\parent_N(w)$ or of the form $(w,v)$ with $v\in
  \child_N(u)\setminus\child_N(w)$.  Since $w$ can appear in at most two
  arcs of this form, all other arcs in $K$ must also be arcs in $N$.
  Assume w.l.o.g.\ that $v_1=w$. Then the following cases have to be
  considered:
  \begin{description}[noitemsep,nolistsep]
  \item[{(a)}] $(w,v_2)=(v_1, v_2)\notin E(N)$ but all other arcs are
    contained in
    $N$,
  \item[{(b)}] $(v_k,w)=(v_k, v_1)\notin E(N)$ but all other arcs are
    contained in
    $N$, and
  \item[{(c)}] exactly the arcs $(w,v_2)=(v_1, v_2)$ and $(v_k,w)=(v_k,
    v_1)$ are not contained in $N$.
  \end{description}
  In case~(a), we must have $(u,v_2)\in V(N)$ by construction, i.e,
  $v_2\in\child_N(u)$. Since all arcs except $(v_1,v_2)$ exist in $N$ and
  $w=v_1\ne v_2$, there is a $v_2 w$-path in $N$, i.e., $w\prec_N
  v_2$. This together with $w, v_2\in\child_N(u)$ implies that $(u,w)$ is a
  shortcut; a contradiction.  In case~(b), we must have $(v_k,u)\in V(N)$
  by construction.  Then $N$ contains a directed cycle along
  $v_0\coloneqq u, v_1, v_2, \dots, v_k$ in this order, i.e.,
  $(v_i,v_{i+1})$, $0\leq i\leq k-1$ and $(v_k,v_0)=(v,u)$ are arcs in
  $N$; a contradiction.  In case~(c), we must have
  $(v_k,u),(u,v_2)\in V(N)$ by construction.  Replacing $v_1$ and its
  incident arcs in the cycle $K$ by vertex $u$ (which is not already in
  $V(K)\subseteq V(N')$) and arcs $(v_k,u)$ and $(u,v_2)$ thus yields a
  directed cycle $K'$ in $N$; a contradiction.  In summary, neither of the
  three cases is possible and thus $N'$ is acyclic. Since $N'$ has a unique
  root and acyclicity is preserved as well, $N'$ is a rooted network.

  Since $\outdeg_{N}(u)>1$, we do not delete any leaf of $N$, i.e.,
  $X\subseteq V(N')$. Moreover, the only vertices whose out-neighborhood
  changes are $w$ and the vertices in $\parent_N(u)$.  Since all vertices
  in $\parent_N(u)$ have $w$ as out-neighbor in $N'$, they do not become
  leaves. Hence, the leaf set of $N'$ is either $X$ or $X\setminus \{w\}$.

  We continue with showing that $v\preceq_N v'$ implies $v \preceq_{N'} v'$
  for all $v,v'\in V(N')$. Thus assume that $v\preceq_N v'$ i.e., there is
  a $v'v$-path $P\coloneqq (v'=v_1, \dots, v=v_k)$ in $N$.  If $P$ does not
  contain $u$, then $P$ is also a $v'v$-path in $N'$ since only arcs that
  are incident with $u$ are removed.  Now suppose $P$ contains $u$. Then
  clearly $u=v_i$ for some $1<i<k$.  Since $u$ appears in $P$ at most once
  and only arcs incident with $u$ were removed, all arcs in $P$ except
  $(v_{i-1},u)$ and $(u,v_{i+1})$ are also arcs in $N'$.  Observe that
  $(v_{i-1},w)\in E(N')$ holds by construction. If $v_{i+1}=w$, then
  $(v_{i-1}, v_{i+1})$ is an arc in $N'$ and thus
  $\tilde{P}=(v_1,\dots,v_{i-1}, v_{i+1}, \dots, v_k)$ is a $v'v$-path in
  $N'$. Otherwise, we have $v_{i+1}\in\child_N(u)\setminus \{w\}$ and thus
  $(w,v_{i+1})$ is an arc in $N'$. The vertex $w$ is not contained in
  $P$. To see this, observe first that $w\ne v_{i+1}$ and
  $w\in\child_N(u=v_i)$ implies $w\prec_N v_j$ for $1\le j\le i$.  If
  $w=v_j$ for some $i+1<j\le k$, then $w\prec_N v_{i+1}$ which implies that
  $(u,w)$ is a shortcut; a contradiction.  Taken together, these arguments
  imply that $\tilde{P}=(v_1,\dots,v_{i-1}, w, v_{i+1}, \dots, v_k)$ is a
  $v'v$-path in $N'$. Hence, we have $v\preceq_{N'} v'$ is all cases.

  Conversely, suppose $v \preceq_{N'} v'$, i.e., there is a $v'v$-path
  $P'\coloneqq (v'=v_1, \dots, v=v_k)$ in $N'$.  If all arcs in $P'$ are
  arcs in $N$, then $P'$ is a $v'v$-path in $N$ and thus $v \preceq_{N}
  v'$. Now suppose that $P'$ contains at least one arc that is not in
	$N$. By construction, any such arc must be incident with $w$ and thus
  $P'$ contains at most two arcs that are not in $N$. These are then
  consecutive in $P'$. We therefore have to consider three cases:
  \begin{description}[noitemsep,nolistsep]
  \item[{(a')}] $(w,v_{i+1})=(v_{i}, v_{i+1})\notin E(N)$ for some $1\le
    i<k$ but all other arcs in $P'$ are contained in $N$,
  \item[{(b')}] $(v_{i-1},w)=(v_{i-1}, v_{i})\notin E(N)$ for some $1< i\le
    k$ but all other arcs in $P'$ are contained in $N$, and
  \item[{(c')}] exactly the arcs $(v_{i-1},w)=(v_{i-1}, v_{i})$ and
    $(w,v_{i+1})=(v_{i}, v_{i+1})$ with $1< i< k$ in $P'$ are not
    contained in $N$.
  \end{description}
  In case~(a'), we have by construction that $(u,v_{i+1})\in E(N)$, i.e.,
  $v_{i+1}\in \child_N(u)$. Since all other arcs of $P'$ are arcs in $N$,
  the subpath of $P'$ from $v'$ to $w$ is a $v'w$-path in $N$ and thus
  $w\preceq_{N} v'$. Similarly, the subpath from $v_{i+1} (\ne w)$ to $v$
  is a $v_{i+1} v$-path in $N$ implying that $v\preceq_N v_{i+1}$ with
  $v_{i+1}\in\child_N(u)\setminus \{w\}$.  Clearly $w\prec_N v_{i+1}$ is
  not possible since otherwise $(u,w)$ would be a shortcut in $N$. Hence,
  we have either $v\preceq_N v_{i+1} \prec_N w \preceq_N v'$ or $w$ and
  $v_{i+1}$ are $\preceq_{N}$-incomparable children of $u$.
  In case~(b'), we have by construction that $(v_{i-1},u)\in E(N)$, i.e.,
  $v_{i-1}\in \parent_N(u)$.  Since all arcs of $P'$ except
  $(v_{i-1},v_{i})$ are also in $N$, we have $v\preceq_{N} v_{i}=w\prec_N u
  \prec_N v_{i-1} \preceq_N v'$.
  In case~(c'), we have by construction that $(v_{i-1},u), (u, v_{i+1})\in
  E(N)$. Since all arcs of $P'$ except $(v_{i-1}, v_{i})$ and $(v_{i},
  v_{i+1})$ are also arcs in $N$, we obtain $v\preceq_{N} v_{i+1}\prec_N u
  \prec_N v_{i-1} \preceq_N v'$.  In summary, in all cases it holds that at
  least one of (i) $v\preceq_N v'$ or (ii) $w\preceq_N v'$ and $v\preceq_N
  w'$ for some $w'\in\child_N(u)\setminus \{w\}$ that is
  $\preceq_{N}$-incomparable with $w$ is true.

  For the final statement, assume that $v\prec_{N'} v'$.  If $v'\preceq_{N}
  v$ would, then (1) implies $v'\preceq_{N'} v$; a
  contradiction. Consequently, we have $v\prec_{N} v'$ or $v$ and $v'$ are
  $\preceq_{N}$-incomparable.
\end{proof}

In the following, we show that the operation $\expand(w)$ does not
introduce shortcuts and that neither $\expand(w)$ nor $\contract(w',w)$
increases the level of a network.  Besides providing additional results, we
give here the proofs for statements that were omitted in the main text.

\begin{lemma}
  \label{lem:expand-shortcuts}
  Let $N_1$ be a network, $N_2$ be the network obtained from $N_1$ by
  applying $\expand(w_1)$ for some $w_1\in V(N_1)$, and $w_2$ be the unique
  vertex in $V(N_2)\setminus V(N_1)$.  Then, for every shortcut $(u,v) \in
  E(N_i)$, it holds either (i) $v\ne w_i$ and $(u,v)$ is a shortcut in
  $N_j$ or (ii) $v= w_i$ and $(u,w_j)$ is a shortcut in $N_j$ such that
  $i,j\in \{1,2\}$ are distinct.
\end{lemma}
\begin{proof}
  In the following, we put $N=N_1$, $N'=N_2$, $w=w_1$ and $w'=w_2$.  Recall
  that by Lemma~\ref{lem:expand}, $N$ and $N'$ are
  $(N',N)$-ancestor-preserving.

  Suppose first that $(u,v) \in E(N')$ is a shortcut in
  $N'$, i.e., there is a $v'\in\child_{N'}(u)\setminus\{v\}$ such that
  $v\prec_{N'} v'$.  The fact that $N'$ is acyclic, $v\prec_{N'}
  v'\prec_{N'} u$ and $(u,v) \in E(N')$ together imply that
  $\indeg_{N'}(v)\ge 2$.  By construction, it holds
  $\parent_{N'}(w)=\{w'\}$ and $\child_{N'}(w')=\{w\}$ which yield $v\ne w$
  and $u\ne w'$, respectively.

  (i) Suppose first that $v\ne w'$. Then $(u,v)$ is also an arc in $N$
  since moreover $u\ne w'$ and all newly inserted arcs are incident with
  $w'$.  Similarly, if in addition $v'\ne w'$, then $(u,v')\in E(N)$ and
  moreover $v\prec_{N'} v'$ implies $v\prec_{N} v'$.  Hence, $(u,v)$ is a
  shortcut in $N$.  If on the other hand $v'=w'$, then $w\in\child_N(u)$
  and $v\prec_{N'} v'=w'$ implies that there is a $w'v$-path in $N'$. Since
  $w$ is the unique child of $w'$ and $v\ne w'$, this path must pass
  through $w$ and thus $v\prec_{N'} w$ which together with $v,w\in V(N)$
  implies $v\prec_{N} w$.  Hence, $(u,v)$ is a shortcut in $N$.

  (ii) Suppose now that $v= w'$. Then, by construction, $(u, w)\in V(N)$.
  Moreover, we have $w\prec_{N'} w'\preceq_{N'} v'$. In particular, $u$,
  $w$, and $v'$ are all vertices in $N'$.  Hence, $w\preceq_{N'} v'$
  implies $w\preceq_{N} v'$ and, since all newly inserted arcs are
  incident with $w'$, $(u,v')\in V(N')$ is also an arc in $N$. In summary,
  therefore, $(u, w)\in V(N)$ and $w\preceq_{N} v'$ for
  $v'\in\child_N(u)\setminus\{w\}$, i.e., $(u,w)$ is a shortcut in $N$.

  Suppose now that $(u,v) \in E(N)$ is a shortcut in
  $N$, i.e., there is a $v'\in\child_{N}(u)\setminus\{v\}$ such that
  $v\prec_{N} v'$.

  (i') Suppose $v\ne w$. By construction, the arc $(u,v)$ still exists in
  $N'$. Since moreover $u,v,v'\in V(N)\subset V(N')$, $v\prec_{N} v'
  \prec_{N} u$ and $N$ and $N'$ are $(N',N)$-ancestor-preserving
  (cf.\ Lemma~\ref{lem:expand}), we have $v\prec_{N'} v' \prec_{N'}
  u$. Hence, $(u,v)$ must be a shortcut in $N'$.

  (ii') Finally suppose $v= w$. Observe that $u'\preceq_{N} v'(\prec_{N}
  u)$ for some $u'\in\parent_{N}(v)\setminus\{u\}$ and that $u, u',v'\in
  V(N)$.  In particular, it holds $u'\prec_{N'} u$ since $N$ and $N'$ are
  $(N',N)$-ancestor-preserving (cf.\ Lemma~\ref{lem:expand}).  By
  construction, we have $(u,w'),(u',w')\in V(N')$ and thus $w'\prec_{N'}
  u'$.  In summary, we have $w'\prec_{N'} u' \prec_{N'} u$, which implies
  that $(u,w')$ is a shortcut in $N'$.
\end{proof}

\begin{lemma}
  \label{lem:contract-common-block}
  Let $N$ be a network, $(w',w) \in E(N)$ be an arc that is not a
  shortcut, and $N'$ be the network obtained from $N$ by applying
  $\contract(w',w)$.  If two distinct vertices $u,v\in V(N')$ are in a
  common non-trivial block of $N'$, then either (i) $w\notin \{u,v\}$ and
  $u,v\in V(N)$ are in a common non-trivial block of $N$ or (ii) $w\in
  \{u,v\}$ and the unique element in $\{u,v\}\setminus \{w\}$ and $w'$ are
  in a common block in $N$.
\end{lemma}
\begin{proof}
  Suppose two distinct vertices $u,v\in V(N')=V(N)\setminus\{w'\}$ are in a
  common non-trivial block of $N'$.
  Thus, $u$ and $v$ are contained in an undirected cycle $K'$ in $N'$.
  If (the directed versions of) all arcs in $K'$ are arcs of $N$, then
  $K'$ is an undirected cycle in $N$, and thus, $u$ and $v$ are contained
  in a non-trivial block of $N$.  Now suppose that at least one arc $e$ in
  $K'$ is not contained in $N$.  By construction, $e$ must be of the form
  $(p,w)$ with $p\in \parent_N(w')\setminus\parent_N(w)$ or of the form
  $(w,c)$ with $c\in \child_N(w')\setminus\child_N(w)$.  Since such an arc
  is incident with $w$, the only arcs of $N'$ that possibly are not
  contained in $N$ are the two arcs incident with $w$ in $K'$. Hence,
  we have to consider the following cases:
  \begin{description}[noitemsep,nolistsep]
    \item[(1)] both of the incident arcs of $w$ in $K'$ are not contained in
    $N$ but all other arcs are contained in $N$,
    \begin{description}
      \item[(a)] $w$ is incident in $K'$ with two distinct parents $p_1,
      p_2\in\parent_N(w')\setminus\parent_N(w)$,
      \item[(b)] $w$ is incident in $K'$ with two distinct children $c_1,
      c_2\in\child_N(w')\setminus\child_N(w)$, or
      \item[(c)] $w$ is incident in $K'$ with a parent $p\in
      \parent_N(w')\setminus\parent_N(w)$ and a child $c\in
      \child_N(w')\setminus\child_N(w)$.
    \end{description}
    \item[(2)] exactly one of the arcs incident with $w$ in $K'$ is not
    contained in $N$, while all other arcs of $N'$ are contained in $N$.
    \begin{description}
      \item[(a)] $w$ is incident in $K'$ with a parent $p\in
      \parent_N(w')\setminus\parent_N(w)$, or
      \item[(b)] $w$ is incident in $K'$ with a child $c\in
      \child_N(w')\setminus\child_N(w)$.
    \end{description}
  \end{description}
  Observe that $w'\notin V(N')$ and thus $w'$ is not contained in $K'$.
  One therefore easily verifies that, in each of cases \textbf{(1a)},
  \textbf{(1b)}, and \textbf{(1c)}, replacing $w$ by $w'$ in $K'$ yields an
  undirected cycle $K$ in $N$.  If $w\notin \{u,v\}$, then $u$
  and $v$ are contained in $K$ and thus contained in a common non-trivial
  block of $N$.  If on the other hand $w\in \{u,v\}$, then the unique
  element in $\{u,v\}\setminus \{w\}$ and $w'$ are contained in $K$ and
  thus contained in a common non-trivial block of $N$.  In
  case~\textbf{(2a)}, $K'$ contains the arc $(p,w)$ by construction of
  $N'$.  Replacing this arc by vertex $w'$ and the arcs $(p,w'),(w',w)\in
  V(N)$ therefore yields an undirected cycle $K$ in $N$ that contains all
  vertices in $K'$.  In particular, therefore, $u,v\in V(K')$ are contained
  in a common non-trivial block in $N$.  The latter is also true in
  case~\textbf{(2b)} by similar arguments.
\end{proof}

From Lemma~\ref{lem:max-B-unique}, we obtain
\begin{corollary}
  \label{cor:block-identity}
  Let $N$ be a network and $u,v,w\in V(N)$ three distinct vertices.  If $u$
  and $v$ are contained in block $B_{uv}$, $u$ and $w$ are contained in
  block $B_{uw}$, and $v$ and $w$ are contained in block $B_{vw}$, then
  $B_{uv}=B_{uw}=B_{vw}$.
\end{corollary}
\begin{proof}
  If $u\notin \{\max B_{uv}, \max B_{uw}\}$, then
  Lemma~\ref{lem:block-identity} immediately implies that $B_{uv}= B_{uw}$.
  Thus suppose (a) $u= \max B_{uv}$ or (b) $u= \max B_{uw}$.  In Case~(a),
  we have $v\prec_N u$ and $v\ne\max B_{uv}$.  If $v\ne\max B_{vw}$, then
  $B_{uv}= B_{vw}$ by Lemma~\ref{lem:block-identity}.  Now assume $v=\max
  B_{vw}$ which implies $w\prec_N v$ and $w\ne\max B_{vw}$.  If $w\ne\max
  B_{uw}$, then $B_{uw}= B_{vw}$ by Lemma~\ref{lem:block-identity}.  The
  case $w=\max B_{uw}$ is not possible since otherwise $u\prec_N w$, and
  thus $u\prec_N w\prec_N v\prec_N u$; a contradiction.  One argues
  similarly in Case~(b).  Hence, in all possible case, $u$, $v$, and $w$
  are contained in a common non-trivial block $B$.  Since each of $B_{uv}$,
  $B_{uw}$, and $B_{vw}$ shares two vertices with $B$, it holds
  $B=B_{uv}=B_{uw}=B_{vw}$.
\end{proof}

\begin{lemma}
  \label{lem:contract-block-contained}
  Let $N$ be a network, $(w',w) \in E(N)$ be an arc that is not a shortcut,
  and $N'$ be the network obtained from $N$ by applying $\contract(w',w)$.
  Let $B'$ be a non-trivial block of $N'$. Then there is a non-trivial
  block $B$ of $N$ with $V(B')\setminus \{w\} \subseteq V(B)$.  Moreover,
  $w\in V(B')$ and $w\notin V(B)$ imply $w'\in V(B)$.
  \par\noindent
  If $v\in V(B')\setminus \{w\}$ is a hybrid vertex and properly contained
  in $B'$, then $v$ is a properly contained hybrid vertex in $B$.  If $w$
  is a properly contained hybrid vertex in $B'$, then at least one of $w$
  and $w'$ is a properly contained hybrid vertex in $B$.
\end{lemma}
\begin{proof}
  Suppose $B'$ is a non-trivial block of $N'$ and thus, it contains at
  least three vertices. In particular, we can find two distinct vertices in
  $u,v$ such that $w\notin \{u,v\}$.  By
  Lemma~\ref{lem:contract-common-block}, $u,v\in V(N)$ are in a common
  non-trivial block $B\coloneqq B_{uv}$ of $N$.  Now suppose there is
  $u'\in V(B')\setminus \{w,u,v\}$.  By
  Lemma~\ref{lem:contract-common-block}, $u,u'\in V(N)$ and $v,u'\in V(N)$,
  resp., are in common non-trivial blocks $B_{uu'}$ and $B_{vu'}$ of $N$.
  By Cor.~\ref{cor:block-identity}, $u$, $v$, and $u'$ are contained in
  $B=B_{uv}=B_{uu'}=B_{vu'}$ of $N$.  Since $u'\in V(B')\setminus
  \{w,u,v\}$ was chosen arbitrarily and blocks that share two vertices are
  equal by Obs.~\ref{obs:identical-block}, we have $V(B')\setminus
  \{w\}\subseteq V(B)$.

  Now suppose $w\in V(B')$ and $w\notin V(B)$.  By
  Lemma~\ref{lem:expd-common-block}(ii), $u,w'\in V(N)$ and $v,w'\in V(N)$,
  resp., are in common non-trivial blocks $B_{uw'}$ and $B_{vw'}$ of $N$.
  By Cor.~\ref{cor:block-identity}, $w'$ is contained in
  $B_{uw'}=B_{vw'}=B_{uv}=B$.

  Suppose $v\in V(B')\setminus \{w\}$ is a hybrid vertex and properly
  contained in $B'$. By Lemma~\ref{lem:properly-contained}, all of the at
  least two vertices in $\parent_{N'}(v)$ are contained in $B'$. Hence, let
  $p, p'\in\parent_{N'}(v)$ be two distinct parents, and assume
  w.l.o.g.\ that $p\ne w$.  Note that $v,p \in V(B')\setminus
  \{w\}\subseteq V(B)$. Moreover, since all newly inserted arcs to obtain
  $N'$ from $N$ are incident with $w$, the arc $(p,v)\in E(N')$ is also an
  arc in $N$.  Now consider $p'$. If $(p',v)\in E(N)$, then $v$ is a
  hybrid vertex in $N$.  If $(p',v)\notin E(N)$, then, by construction of
  $N'$, we must have $p'=w$ and $v\in \child_{N}(w')\setminus
  \child_{N}(w)$. Hence, we have $(w',v)\in E(N)$.  Together with $w'\ne p$
  (since $p\in V(N')=V(N)\setminus\{w'\}$) this implies that $v$ is a
  hybrid vertex in $N$ also in this case.  Therefore and since
  $p\in\parent_N(v)$ and $v,p\in V(B)$, Lemma~\ref{lem:properly-contained}
  implies that $v$ is properly contained in $B$.

  Suppose $w$ is a properly contained hybrid vertex in $B'$.  By
  Lemma~\ref{lem:properly-contained}, all of the at least two vertices in
  $\parent_{N'}(w)$ are contained in $B'$.  We distinguish cases~(a) $w\in
  V(B)$ and (b) $w\notin V(B)$.\\
  \textit{Case~(a):} $w\in V(B)$.  Suppose first that there is $p\in
  \parent_{N'}(w) \cap \parent_{N}(w)$.  We have $p\in V(B')\setminus\{w\}
  \subseteq V(B)$.  Moreover, $w$ has at least the two distinct parents
  $w'$ and $p$ in $N$, i.e., $w$ is a hybrid vertex in $N$.  By
  Lemma~\ref{lem:hybrid-properly-contained} and since $w$ and its parent
  $p$ are both contained in $B$, $w$ is properly contained in $B$.  Suppose
  now that $\parent_{N'}(w) \cap \parent_{N}(w)=\emptyset$.  By
  construction of $N'$, this implies that all of the at least two vertices
  in $\parent_{N'}(w)$ must be vertices in $\parent_{N}(w')$.  Hence, $w'$
  is a hybrid vertex in $N$ with at least two parents $p$ and $p'$ that are
  parents of $w$ in $N'$ and thus contained in $V(B')\setminus\{w\}
  \subseteq V(B)$.  By Lemma~\ref{lem:hybrid-properly-contained},
  $\{w',p,p'\}\in V(\tilde{B})$ for some block $\tilde{B}$ of $N$. Since
  $B$ and $\tilde{B}$ are blocks in $N$ that share the two vertices $p$ and
  $p'$, we conclude $B=\tilde{B}$. In particular, $w'$ is a properly
  contained hybrid vertex in $B$.\\
  \textit{Case~(b):} $w\notin V(B)$.  We have already seen that this
  implies that $w'\in V(B)$.  Since $w$ is properly contained in $B'$, all
  of its at least two parents in $N'$ are also contained in $B'$. Let $p\in
  \parent_{N'}(w)$.  Suppose, for contradiction, that $p\in
  \parent_{N}(w)$. Then $w$ has at least the two distinct parents $w'$ and
  $p$ in $N$, i.e., $w$ is a hybrid vertex in $N$. By
  Lemma~\ref{lem:hybrid-properly-contained}, $\{w,w',p\}\in V(\tilde{B})$
  for some block $\tilde{B}$ of $N$. Since $B$ and $\tilde{B}$ are blocks
  in $N$ that share the two vertices $w'$ and $p$, we conclude
  $B=\tilde{B}$ and thus $w\in V(B)$; a contradiction.  Hence, $p\notin
  \parent_{N}(w)$. Together with $p\in \parent_{N'}(w)$, this implies $p\in
  \parent_{N}(w')$.  Since the latter is true for all of the at least two
  vertices $p'\in \parent_{N'}(w)$, $w'$ must be a hybrid vertex in $N$.
  Therefore, and because $p\in\parent_N(w')$ and $w',p\in V(B)$,
  Lemma~\ref{lem:properly-contained} implies that $w'$ is properly
  contained in $B$.
\end{proof}

We are now in the position to prove Lemma~\ref{lem:contract-level-k}.  To
recall, Lemma~\ref{lem:contract-level-k} states that $\contract(w',w)$ of a
non-shortcut arc $(w',w)$ in a level-$k$ network $N$ preserves the property
of the resulting network $N'$ to be level-$k$.

\begin{proof}[\it Proof of Lemma~\ref{lem:contract-level-k}]
  Suppose, for contraposition, that $N'$ is not level-$k$.  Hence, there is
  a block $B'$ in $N'$ that properly contains at least $k+1$ hybrid
  vertices. Denote the set of these vertices by $A$.  By
  Lemma~\ref{lem:contract-block-contained}, it holds $V(B')\setminus \{w\}
  \subseteq V(B)$ for some non-trivial block $B$ of $N$, and in particular,
  all vertices in $A\setminus \{w\}$ are properly contained hybrid vertices
  in $B$.  If $w\notin A$, then $B$ properly contains at least $k+1$ hybrid
  vertices.  Otherwise, $w$ is a properly contained hybrid vertex in $B'$
  and $B$ properly contains at least $k$ hybrid vertices in $A\setminus
  \{w\}$, and, by Lemma~\ref{lem:contract-block-contained}, at least on of
  $w,w'\notin A\setminus \{w\}$ is an additional properly contained hybrid
  vertex in $B$.  Therefore, the block $B$ in $N$ properly contains at
  least $k+1$ hybrid vertices in both cases, and thus, $N$ is not
  level-$k$.
\end{proof}

\begin{lemma}
  \label{lem:expd-common-block}
  Let $N$ be a network, $N'$ be the network obtained from $N$ by applying
  $\expand(w)$ for some $w\in V(N)$, and $w'$ be the unique vertex in
  $V(N')\setminus V(N)$.  If two distinct vertices $u,v\in V(N')$ are in a
  common non-trivial block of $N'$, then either (i) $w'\notin \{u,v\}$ and
  $u,v\in V(N)$ are in a common non-trivial block of $N$ or (ii) $w'\in
  \{u,v\}$ and $w$ and the element in $\{u,v\}\setminus \{w'\}$ are in a
  common block in $N$.
\end{lemma}
\begin{proof}
  Suppose two distinct vertices $u,v\in V(N')$ are in a common non-trivial
  block $B'$ of $N'$.  Then $u$ and $v$ lie in an undirected cycle $K'$ in
  $N'$.

  (i) Assume first that $w'\notin \{u,v\}$, and thus $u,v\in V(N)$.  If
  $w'$ is not contained in $K'$, then all arcs in $K'$ are also arcs in
  $N$ since all newly inserted arcs are incident with $w'$. Hence, $u$ and
  $v$ lie on an undirected cycle in $N$ and are thus contained in a common
  non-trivial block $B$ in this case. Now suppose that $K'$ contains $w'$.
  Assume that $w'$ is incident with its unique child $w$ in $K'$. Then, by
  construction of $N'$, the second vertex that is incident with $w'$ in
  $K'$ must be a vertex $p\in\parent_{N'}(w')=\parent_N(w)$. In particular,
  we have $(p,w)\in E(N)$ and $(p,w),(w,p)\notin E(N')$. The latter implies
  that $p$ and $w$ are not incident in $K'$ and therefore $K'$ contains at
  least one additional vertex $x\notin\{p,w',w\}$. Together with the fact
  that $N$ contains all arcs in $K'$ except $(p,w')$ and $(w',w)$, the
  latter arguments imply that there is an undirected cycle $K$ in $N$
  formed by the vertices in $V(K')\setminus \{w'\}$ and
  %    (the undirected versions of)
  the arcs in $(E(K')\setminus\{(p,w'),(w',w)\})\cup \{(p,w)\}$.  Hence,
  $u,v\in V(K')\setminus \{w'\}$ lie in a common non-trivial block of $N$.
  Now suppose $w'$ is not incident with its unique child $w$ in $K'$, and
  thus the two incident vertices in $K'$ must be two distinct element
  $p_1,p_2\in \parent_{N'}(w')=\parent_N(w)$.  Thus, we have
  $(p_1,w),(p_2,w)\in E(N)$.  Since $(p_1,w')$ and $(p_2,w')$ are the only
  arcs incident with $w'$ in $K'$, all other arcs in $K'$ are also arcs
  in $N$.  Thus consider the (not necessarily induced) subgraph $K$ of $N$
  formed by the vertices $(V(K')\setminus \{w'\})\cup \{w\}$ and arcs
  $(E(K')\setminus\{(p_1,w'),(p_2,w')\})\cup \{(p_1,w),(p_2,w)\}$.  If
  $w\notin V(K')$, then one easily verifies that $K$ is an undirected cycle in
  $N$ that contains $u$ and $v$ and thus they are contained in a common block
  $B$ of $N$.  On the other hand, if $w\in V(K')$, then its two incident
  vertices in $K'$ are two distinct elements
  $c_1,c_2\in\child_{N'}(w)=\child_{N}(w)$ since, by assumption, $w$ is not
  incident in $K'$ with its unique parent $w'$.  In this case, one easily
  verifies that $K$ consists of two undirected cycles that share only
  the vertex $w$ and each of the two cycles contains
  exactly one of $p_1,p_2\in \parent_N(w)$ and one of
  $c_1,c_2\in\child_{N}(w)$.  If $u$ and $v$ are contained in the same of
  these two cycles, then they are contained in a common non-trivial block
  of $N$.  Otherwise, they are contained in non-trivial blocks $B$ and
  $B'$, resp., that each contain $w$ and one of its parents. The latter
  implies that $w\notin \{\max B, \max B'\}$. Hence, we have $B=B'$ by
  Lemma~\ref{lem:block-identity}.

  (ii) Suppose now that $w'\in \{u,v\}$, say $v=w'$.  We can essentially
  re-use the arguments from case~(i), with exception of the case that $w'$
  is not contained in $K'$ (which is impossible since $w'=v\in V(K')$),
  since we always have constructed an undirected cycle in $N$ that contains
  $u,w\in V(N)$. Hence, they are contained in a non-trivial block of $N$.
\end{proof}

\begin{lemma}
  \label{lem:expand-block-contained}
  Let $N$ be a network, $N'$ be the network obtained from $N$ by applying
  $\expand(w)$ for some $w\in V(N)$, and $w'$ be the unique vertex in
  $V(N')\setminus V(N)$.
  Let $B'$ be a non-trivial block of $N'$. Then there is a non-trivial block
  $B$ of $N$ with $V(B')\setminus \{w'\} \subseteq V(B)$.
  Moreover, $w'\in V(B')$ implies $w\in V(B)$.

  If $v\in V(B')\setminus \{w'\}$ is a hybrid vertex and properly contained
  in $B'$, then $v$ is a properly contained hybrid vertex in $B$.
  If $w'\in V(B')$ and $w'$ is a hybrid vertex, then $w$ is a properly
  contained hybrid vertex in $B$.
\end{lemma}
\begin{proof}
  Suppose $B'$ is a non-trivial block of $N'$ and thus, it contains at
  least three vertices. In particular, we can find two distinct vertices in
  $u,v$ such that $w'\notin \{u,v\}$.  By
  Lemma~\ref{lem:expd-common-block}, $u,v\in V(N)$ are in a common
  non-trivial block $B\coloneqq B_{uv}$ of $N$.  Now suppose there is
  $u'\in V(B')\setminus \{w',u,v\}$.  By Lemma~\ref{lem:expd-common-block},
  $u,u'\in V(N)$ and $v,u'\in V(N)$, resp., are in common non-trivial
  blocks $B_{uu'}$ and $B_{vu'}$ of $N$.  By Cor.~\ref{cor:block-identity},
  $u$, $v$, and $u'$ are contained in $B= B_{uv}$.  Since $u'\in
  V(B')\setminus \{w',u,v\}$ was chosen arbitrarily and blocks that share
  two vertices are equal by Obs.~\ref{obs:identical-block}, we have
  $V(B')\setminus \{w'\}\subseteq V(B)$.

  Now suppose $w'\in V(B')$. If $w\in V(B')$, then we have already seen
  that $w\in V(B)$. Hence, suppose $w\notin V(B')$.  By
  Lemma~\ref{lem:expd-common-block}, $u,w\in V(N)$ and $v,w\in V(N)$,
  resp., are in common non-trivial blocks $B_{uw}$ and $B_{vw}$ of $N$.  By
  Cor.~\ref{cor:block-identity}, $w$ must also be contained in $B$.

  Suppose $v\in V(B')\setminus \{w'\}$ is a hybrid vertex and properly
  contained in $B'$. Since $w$ has a unique parent in $N'$, we have $v\ne
  w$ and thus $\parent_{N'}(v)=\parent_N(v)$.  Moreover, since $v$ is
  properly contained in $B'$, we have $\parent_{N'}(v)\in V(B')$ and thus
  $\parent_{N'}(v)\in V(B)$. Taken together, the latter arguments imply
  that $v\ne \max B$, i.e., $v$ is properly contained in $B$.

  Suppose $w'\in V(B')$ and $w'$ is a hybrid vertex. Since $w'$ has a
  unique child and must lie on an undirected cycle in $N'$, $w'$ must be
  properly contained in $B'$. Therefore, all of its at least two parent are
  also contained in $B'$ and thus in $B$. Together with the fact that $w\in
  V(B)$ and $\parent_{N'}(w')=\parent_N(w)$, we obtain that $w$ is a hybrid
  vertex and properly contained in $B$.
\end{proof}

We now prove Lemma~\ref{lem:expand-level-k}.
To recall,  Lemma~\ref{lem:expand-level-k} states that two networks
$N$ and $N'$ are always level-$k$ whenever $N'$ is obtained from
$N$ by applying $\expand(w)$ for some $w\in V(N)$ 
and at least one of them is level-$k$.
\begin{proof}[\it Proof of Lemma~\ref{lem:expand-level-k}]
  Denote by $w'$ the unique vertex in $V(N')\setminus V(N)$.
  Suppose that $N'$ is not level-$k$.  Hence, there is a block $B'$ in $N'$
  that properly contains at least $k+1$ hybrid vertices. Denote the set of
  these vertices by $A$.  Observe that $w\notin A$ since it has a single
  parent $w$ in $N'$.  By Lemma~\ref{lem:expand-block-contained}, it holds
  $V(B')\setminus \{w'\} \subseteq V(B)$ for some non-trivial block $B$ of
  $N$, and in particular, all vertices $A\setminus \{w'\}$ are properly
  contained hybrid vertices in $B$.  If $w'\notin A$, then $B$ properly
  contains at least $k+1$ hybrid vertices.  Otherwise, $B$ properly
  contains at least $k$ hybrid vertices in $A\setminus \{w'\}$ and, by
  Lemma~\ref{lem:expand-block-contained}, additionally the hybrid vertex
  $w\notin A$.  Therefore, the block $B$ in $N$ properly contains at least
  $k+1$ hybrid vertices in both cases, and thus, $N$ is not level-$k$.

  Conversely, suppose that $N'$ is level-$k$. Observe that $N$ is recovered
  from $N'$ by applying $\contract(w',w)$.  By
  Lemma~\ref{lem:contract-level-k}, therefore, $N$ is also level-$k$.
\end{proof}

\subsection{Closed Clustering Systems}
\label{sec:appx-closed}

As promised, we provide here a short proof of Lemma~\ref{lem:simple-closed},
which states that a clustering system $\mathscr{C}$ is closed if and only if
$A,B \in \mathscr{C}$ and $A\cap B\ne\emptyset$ implies
$A\cap B\in \mathscr{C}$.
\begin{proof}[\it Proof of Lemma~\ref{lem:simple-closed}]
  Let $\mathscr{C}$ be closed and let $A,B\in \mathscr{C}$.  Since $\cl$ is
  enlarging, we have $A\cap B\subseteq \cl(A\cap B)$.  Moreover, since
  $A\cap B\subseteq A,B$, it holds by Equ.\ \eqref{eq:closure} that
  $\cl(A\cap B)\subseteq A\cap B$. Hence, $\cl(A\cap B)=A\cap B$ and the
  definition of ``closed'' implies $A\cap B \in \mathscr{C}$.  Assume now
  that $C\cap C'\in \mathscr{C}$ for all $C,C'\in \mathscr{C}$.
  Equ.~\eqref{eq:closure} implies $\cl(A)=A$ for all $A\in \mathscr{C}$. It
  remains to show that, for $A\in 2^X$, $\cl(A)=A$ implies
  $A\in\mathscr{C}$. By Equ.~\eqref{eq:closure}, $\cl(A)$ can be written as
  the intersection $\cl(A)=\bigcap_{i=1}^k C_i$ of a finite number $k$ of
  clusters $C_i\in \mathscr{C}$ with $A\subseteq C_i$. Since
  $A\subseteq X\in \mathscr{C}$ we have $k\ge 1$.  If $k\in\{1,2\}$,
  $\cl(A) = A \in\mathcal{C}$ follows immediately from the assumption that
  $C\cap C'\in \mathscr{C}$ for all $C,C'\in \mathscr{C}$.  Otherwise, we
  can construct a series of intersections $C'_i$, $1\leq i\leq k$, by
  setting $C'_1\coloneqq C_1$ and $C'_j\coloneqq C'_{j-1}\cap C_{j}$ for
  $2\leq j\leq k$.  By definition, $C'_1=C_1\in \mathscr{C}$. Moreover, if
  $C'_{j-1} \in \mathscr{C}$, then $C'_{j} \in \mathscr{C}$ holds by the
  assumption and the fact that $C_j\in\mathscr{C}$ for all $2\le j\le
  k$. By induction, therefore, we obtain $C'_k\in\mathcal{C}$.  We have
  $C'_k=\cl(A)$ by construction and thus $\cl(A) = A \in\mathcal{C}$.
\end{proof}

\subsection{Algorithmic Details}
\label{sec:appx-algo}

We show here the correctness and runtime results for
\texttt{Check-L1-Compatibility}.

\begin{proof}[\it Proof of Thm.~\ref{thm:AlgL1Comp}]
Let $\mathscr{C}\subseteq 2^X$ be a clustering system.  Suppose first
that \texttt{Check-L1-Compatibility} returns a network.  In particular,
$\mathscr{C}$ satisfies (L) in this case.  By Obs.~\ref{obs:IC-L->closed}
and Lemma~\ref{lem:C-L<->IC-L}, $\IC$ is correctly computed and satisfies
(L). By definition, $\IC$ is closed.  The latter two arguments together
with Prop.~\ref{prop:unique-separated-level-1} imply that there is a
(separated, phylogenetic) level-1 network on $X$ such that
$\mathscr{C}\subseteq \IC = \mathscr{C}_N $.

Conversely, suppose that \texttt{Check-L1-Compatibility} returns
\texttt{``no solution''}.  Assume for contradiction that there is a level-1
network $N$ with $\mathscr{C}\subseteq \mathscr{C}_N$. By
Thm.~\ref{thm:L1}, $\mathscr{C}_N$ is closed and satisfies (L).  By
Cor.~\ref{cor:L->wH}, Property (L) implies that $\mathscr{C}$ is weak
hierarchy.  As shown in \cite{Bandelt:89} (right below Lemma~1), this in
turn implies that $\vert\mathscr{C}\vert \le \vert\mathscr{C}_N\vert \le
\binom{\vert X\vert+1}{2} = \binom{\vert X\vert}{2}+\vert X\vert$.  Inputs
larger than this bound are therefore correctly rejected immediately.  If,
on the other hand, Property~(L) is not satisfied for $\mathscr{C}$, then by
the definition of Property~(L), its superset $\mathscr{C}_N$ also violates
(L); a contradiction. Therefore, such a network $N$ cannot exist and
\texttt{Check-L1-Compatibility} correctly exits with a negative answer.

We now proceed to show that the algorithm can be implemented to run in
$O(\vert\mathscr{C}\vert^2\vert X \vert)\subseteq O(\vert X\vert^5)$ time.
To this end, we first enumerate the elements in $X$ from $1$ to $\vert
X\vert$. We then initialize a list $\mathcal{L}$ containing a bitvector
$b_i$ of size $\vert X\vert$ for each $C_i\in\mathscr{C}$ that has a
$1$-entry at position $j$ if and only if the $j$th element of $X$ is
contained in $C_j$, and a $0$-entry otherwise; requiring a total effort of
$O(\vert\mathscr{C}\vert \vert X\vert)$ time. Moreover, we create an
initially arcless auxiliary graph $G$ whose vertices are (unique
identifiers of) the clusters in $\mathscr{C}$. In the end, two clusters
$C_i,C_j\in\mathscr{C}$ will be connected by an arc in $G$ precisely if
they overlap. Moreover, every arc $\{C_i,C_j\}$ will be associated with a
pointer to a bitvector that represents the intersection $C_i\cap C_j$.

To achieve this, we proceed as follows for every pair $C_i, C_j\in
\mathscr{C}$.  We compute the bitvector $b$ corresponding to the
intersection $C_i\cap C_j$ as $b\leftarrow b_i \wedge b_j$.  The
clusters $C_i$ and $C_j$ overlap if and only if the number of $1$-entries
in $b$ is greater than $0$ but less than $\vert C_i\vert$ and $\vert
C_j\vert$ (the latter cardinalities can be pre-computed for all clusters in
$\mathscr{C}$).  If the clusters do overlap, then we continue as follows.
If $C_i$ is already adjacent with some other cluster in
$\mathscr{C}\setminus\{C_i, C_j\}$, then we pick $C_k$ among them
arbitrarily. Let $b'$ be the bitvector associated with the arc
$\{C_i,C_k\}$ in $G$.  By (L), $b$ and $b'$ must be equal, which can be
checked in $O(\vert X\vert)$.  If they are not equal, we can exit
\texttt{``no solution''}.  If existent, we proceed analogously with some
neighbor $C_l\in\mathscr{C}\setminus\{C_i, C_j\}$ of $C_j$ in $G$ with
associated bitvector $b''$.  We add the arc $\{C_i,C_j\}$ to $G$. If at
least one of $C_i$ and $C_j$ previously had a neighbor $C_k$ or $C_j$,
respectively, then we associate the bitvector $b'$ or $b''$, respectively,
with the new arc $\{C_i,C_j\}$ and discard $b$. Otherwise, we associate $b$
with $\{C_i,C_j\}$ and add $b$ to $\mathcal{L}$.  In summary, for a pair
$C_i,C_j\in\mathscr{C}$, we only perform a constant number of operations,
all of which require at most $O(\vert X\vert)$ time.  Hence, we obtain a
total effort of $O(\vert\mathscr{C}\vert^2 \vert X\vert)$ time for
processing all pairs of clusters in $\mathscr{C}$.

One easily verifies that $\mathscr{C}$ satisfies (L) if the algorithms
did not exit at this point because all overlaps are represented by arcs
in $G$ and, for each $C\in\mathscr{C}$, the intersections with its
overlapping clusters are stepwisely added and compared to the overlaps
that where already computed.  Moreover, we have added at most $\lfloor
\vert\mathscr{C}\vert / 2 \rfloor$ bitvectors to $\mathcal{L}$.  To see
this, recall that we only added a new bitvector $b$ (associated with arc
$\{C_i,C_j\}$) if neither of $C_i$ and $C_j$ had any other neighbors at
that point. Since each arc is incident with two clusters, we can clearly
do this at most $\lfloor \vert\mathscr{C}\vert / 2 \rfloor$ times until
all but possible one clusters are adjacent with some other cluster.
Hence, $\mathcal{L}$ still contains only $O(\vert\mathscr{C}\vert)$
bitvectors.  In particular, $\mathcal{L}$ contains all clusters in
$\IC=\mathscr{C} \cup \{C\cap C' \mid C,C'\in\mathscr{C}\text{
	overlap}\}$ (represented by their bitvectors) at least once.

We now sort $\mathcal{L}$ lexicographically in $O( \vert X\vert
\vert\mathscr{C}\vert \log \vert\mathscr{C}\vert) = O( \vert X\vert\,
\vert\mathscr{C}\vert \log \vert X\vert)$, where the additional factor $\vert
X\vert$
originates from the fact that each comparison of bitvectors requires
$O(\vert X\vert)$ comparisons of their entries.  Removal of all duplicates in
$\mathcal{L}$ now takes $O(\vert\mathscr{C}\vert \vert X \vert)$ time, e.g.\ by
iterating through the list and comparing each vector to the last added
bitvector in a newly constructed list.  The bitvectors in this final list
$\mathcal{L}$ are thus in a 1-to-1 correspondence with the clusters in
$\IC$.

If we are only interested in the existence of a (separated, phylogenetic)
level-$1$ network $N$ such that $\mathscr{C}\subseteq \mathscr{C}_{N}$
and the clustering system $\mathscr{C}_{N}$, then we can stop here after
a total effort of $O(\vert\mathscr{C}\vert^2 \vert X\vert)\subseteq O(\vert
X\vert^5)$ time.

Otherwise we continue with the construction of the inclusion order, i.e.,
the Hasse diagram of $\IC$.  In the following, the cluster $C_i\in \IC$
corresponds to the $i$th bitvector in $\mathcal{L}$.  We initialize a
$\vert\IC\vert\times \vert\IC\vert$-matrix $M$ with all zero entries
(requiring $O(\vert\mathscr{C}\vert^2)$ time and space).  In the end, we
will have $M_{i,j}=1$ if and only if $C_i\subsetneq C_j$.  Since
$\mathcal{L}$ is still sorted lexicographically, observe that, if
$C_i\subsetneq C_j$, then $i<j$.  Hence, it suffices to check for all $1\le
i < j \le \vert\IC\vert$ whether all $1$-entries in $b_i$ are also
$1$-entries in $b_j$ and, if so, set $M_{i,j}=1$. Finally, the adjacency
matrix of $\Hasse[\IC]$ is obtained from $M$ by \textsc{Transitive
  Reduction}. As shown in \cite{Aho:72} for DAGs this task has the same
complexity as \textsc{Transitive Closure}, which, in our setting, is
bounded $O(\vert X\vert^{2\omega})$. Here, $\omega\le 2.3729$ is the
``matrix multiplication constant''. Thus $\Hasse[\IC]$ can also be
constructed in quintic time. Thus we obtain a regular (and thus
phylogenetic) level-1 network $N\sim\Hasse[\IC]$ in $O(\vert X\vert^5)$
time.  This network can be modified to be separated by arc expansion at all
hybrid vertices $v$ with $\outdeg(v)$. Since $\Hasse[\IC]$ has $O(\vert
X\vert^2)$ vertices, for each vertex, one checks in constant time whether
$\indeg(v)>1$ and $\outdeg(v)>1$, and $\expand(v)$ requires constant effort
for moving the list of out-neighbors from $v$ to the newly inserted vertex,
the total effort for the modification of the Hasse diagram is bounded by
$O(\vert X\vert^2)$.
\end{proof}

\bibliographystyle{plain}      % mathematics and physical sciences
\bibliography{Cluster-N}

\end{document}